\RequirePackage[log]{snapshot}
\documentclass{abnt}

\pdfoutput=1
\usepackage[T1]{fontenc}
\usepackage{amsmath,amsthm,amssymb,amsfonts,setspace}
\usepackage{bm}
\usepackage{bbold}
\usepackage{float}
\usepackage{graphicx}
\usepackage{mathtools}
\usepackage[nohints]{minitoc}
\usepackage[colorlinks,plainpages=false]{hyperref}
\usepackage{hypcap}
\usepackage[font=small]{caption}
\usepackage[caption=false]{subfig}
\usepackage{url}
\usepackage[usenames,dvipsnames]{color}
\newcommand{\refname}{References}
\usepackage[toc,page]{appendix}
\usepackage{pdfpages}
\usepackage{fancyvrb}
\usepackage{listings}

\makeatletter
\renewcommand{\@chap@pppage}{%
  \clear@ppage
  \thispagestyle{empty}%
  \if@twocolumn\onecolumn\@tempswatrue\else\@tempswafalse\fi
  \null\vfil
  \markboth{}{}%
  {\centering
   \interlinepenalty \@M
   \normalfont
   \Huge \bfseries \appendixpagename\par}%
  \if@dotoc@pp
    \addappheadtotoc
  \fi
  \vfil\newpage
  \if@twoside
    \if@openright
      \null
      \thispagestyle{plain}%
      \newpage
    \fi
  \fi
  \if@tempswa
    \twocolumn
  \fi
}
\makeatother

\mathtoolsset{centercolon}

\DeclareMathOperator{\SO}{SO}
\DeclareMathOperator{\SU}{SU}
\DeclareMathOperator{\poly}{poly}
\DeclareMathOperator{\diag}{diag}

\newcommand{\bra}[1]{\left\langle{#1}\right\vert}
\newcommand{\ket}[1]{\left\vert{#1}\right\rangle}
\newcommand{\eq}[1]{Eq.~\hyperref[eq:#1]{(\ref*{eq:#1})}}
\newcommand{\eqbeg}[1]{Equation~\hyperref[eq:#1]{(\ref*{eq:#1})}}
\newcommand{\eqs}[2]{Eqs.\ \hyperref[eq:#1]{(\ref*{eq:#1})} and \hyperref[eq:#2]{(\ref*{eq:#2})}}
\renewcommand{\sec}[1]{\hyperref[sec:#1]{Section~\ref*{sec:#1}}}
\newcommand{\appdx}[1]{\hyperref[appdx:#1]{Appendix~\ref*{appdx:#1}}}
\newcommand{\chap}[1]{\hyperref[chapter:#1]{Chapter~\ref*{chapter:#1}}}
\newcommand{\fig}[1]{\hyperref[fig:#1]{Figure~\ref*{fig:#1}}}
\newcommand{\sfig}[2]{\hyperref[fig:#1]{Figure~\ref*{fig:#1}(#2)}}
\newcommand{\thm}[1]{\hyperref[thm:#1]{Theorem~\ref*{thm:#1}}}
\newcommand{\conj}[1]{\hyperref[conj:#1]{Conjecture~\ref*{conj:#1}}}
\newcommand{\lem}[1]{\hyperref[lem:#1]{Lemma~\ref*{lem:#1}}}
\newcommand{\defn}[1]{\hyperref[def:#1]{Definition~\ref*{def:#1}}}
\newcommand{\ex}[1]{\hyperref[ex:#1]{Example~\ref*{ex:#1}}}
\newcommand{\post}[1]{\hyperref[post:#1]{Postulate~\ref*{post:#1}}}
\newcommand{\tabl}[1]{\hyperref[tab:#1]{Table~\ref*{tab:#1}}}

\newcommand{\swap}{\textsc{swap}}
\newcommand{\cnot}{\textsc{cnot}}
\newcommand{\fswap}{\textrm{f-}\swap}
\newcommand{\iswap}{\textrm{i-}\swap}
\newcommand{\fs}{\textrm{f}\textsc{s}}
\newcommand{\is}{\textrm{i}\textsc{s}}
\newcommand{\cz}{\textsc{cz}}
\newcommand{\A}{\mathcal{A}}
\newcommand{\PP}{P.~P.}
\newcommand{\B}{\mathcal{B}}

\newcommand{\mathsym}[1]{{}}
\newcommand{\unicode}[1]{{}}

\newtheorem{lemma}{Lemma}
\newtheorem{theorem}{Theorem}
\newtheorem{postulate}{Postulate}
\newtheorem{conjecture}{Conjecture}

\theoremstyle{definition}
\newtheorem{example}{Example}
\newtheorem{definition}{Definition}
\newtheorem*{definition*}{Definition}

\numberwithin{theorem}{chapter}
\numberwithin{table}{chapter}
\numberwithin{lemma}{chapter}
\numberwithin{definition}{chapter}
\numberwithin{example}{chapter}
\numberwithin{figure}{chapter}
\numberwithin{postulate}{chapter}

\begin{document}

\pagenumbering{Alph}

 \autor{DANIEL JOST BROD}
 \titulo{The Computational Power of Non-interacting Particles}
 \orientador{Prof.\ Dr.\ Ernesto Fagundes Galv\~ao}
 \comentario{Tese apresentada ao Curso de P\'os-Gradua\c{c}\~ao em F\'{i}sica da Universidade Federal Fluminense, como requisito parcial para obten\c{c}\~ao do T\'{i}tulo de Doutor em F\'{i}sica.}
 \instituicao{Universidade Federal Fluminense\par
 Instituto de F\'{i}sica}
\local{Niter\'oi}
 \data{Março - 2014}
 \capa
 \folhaderosto

\newpage
\chapter*{}
\vfill\vfill

\begin{flushright}

"Music washes away from the soul \\ the dust of everyday life."  \\
\textit{Berthold Auerbach}
\end{flushright}

\vfill \ 

\newpage

\chapter*{Agradecimentos}

{\singlespacing
A tarefa de avaliar o impacto, na vida de uma pessoa, de seus amigos, familiares e professores, é intimidadora e infalivelmente injusta. Gostaria de agradecer a todos que, de uma forma ou de outra, tocaram minha vida e a conduziram até este ponto. Já me desculpo por quaisquer omissões, que são devidas apenas à minha memória falível, e não significam muito.

Meus agradecimentos mais que sinceros a meu orientador, Ernesto Galvão, pelo apoio e orientação constantes ao longo dos últimos cinco anos. A sua determinação e confiança me ajudaram a manter o foco tanto nas horas mais frustrantes como nas mais empolgantes. No dia em que eu conseguir demonstrar, na relação com um estudante, uma fração da paciência e clareza que ele demonstra, me considerarei apto para seguir este caminho. Espero também ter ganhado um amigo e colaborador para a vida toda.

Também sou muito grato ao professor Andrew Childs, que me recebeu durante os seis meses que passei no IQC (Institute for Quantum Computing). Seu amplo conhecimento da área e seus insights penetrantes nunca deixaram de me surpreender, e ele dedicou mais tempo e atenção a me orientar do que eu jamais teria o direito de exigir.

Em meu Doutorado, tive a feliz oportunidade de colaborar com os excelentes grupos de óptica quântica de Roma e Milão e gostaria de agradecer em especial a Fabio Sciarrino, Paolo Mataloni, Chiara Vitelli, Nicolò Spagnolo, Andrea Crespi, Roberta Ramponi e Roberto Osellame. Sendo um físico teórico, é difícil para mim imaginar o esforço que precisa ser empregado na preparação e execução de tantos projetos interessantes, de forma que o empenho e a dedicação de todos eles são para mim uma poderosa fonte de inspiração.

Ao longo do meu Doutorado, tive a sorte de conhecer amigos fantásticos e inspiradores. Gostaria de agradecer em especial a Raphael Dias, que foi o primeiro a me apresentar à área de Computação Quântica, e a Tiago Debarba. Eles foram companheiros de aventuras e de inúmeras horas de conversas agradáveis, e ambos deixaram marcas profundas em minha visão de mundo, da minha área de pesquisa e da carreira acadêmica. Também gostaria de agradecer a todos os amigos do IQC e da UFF, incluindo (mas não só) David Gosset, Rajat Mittal, Zak Webb e Ingrid Hammes. Na UFF, também tive diversos excelentes professores, aos quais devo muito da minha formação acadêmica. Gostaria de agradecer em particular a Nivaldo Lemos, Antônio Zelaquett, Marco Moriconi, Luis Oxman e Daniel Jonathan. Um muito obrigado especial a Matt Fries, que se desdobrou para fazer de Waterloo o local mais hospitaleiro possível.

Toda a minha gratidão a Letícia, Vera, Gabriel, Hermano e Rayssa, amigos de longa data que ajudaram a tornar minha vida e esta jornada mais fáceis, divertidas e significativas. Um lugar de destaque para o André, meu amigo desde sempre, que ajudou a moldar uma grande parte do ser humano que eu sou hoje, e com quem atualmente falo muito, muito menos do que deveria, gostaria, ou que ele merece.

Também gostaria de agradecer a minha família pelo amor e apoio que permitiram que eu me tornasse quem eu sou hoje. Meu pai e meu avô apontaram o caminho, e é com um sentimento profundo de honra e humildade que me esforço por seguir seus passos. Minha mãe e minha avó estiveram presentes em cada passo do caminho, oferecendo constantemente amor, apoio e encorajamento, que me deram forças para superar todos os obstáculos.

Finalmente, gostaria de agradecer à Raissa. Esse é o mais injusto de todos os agradecimentos, pois não há palavras que bastem para expressá-lo. Ela é o amor da minha vida, minha alma gêmea, minha crítica mais gentil e a companheira mais amorosa com que eu poderia esperar dividir minha vida. Absolutamente nada disso teria sido possível sem ela.

\textbf{Apoio financeiro:} Agradeço o apoio financeiro total do CNQp (Conselho Nacional de Desenvolvimento Científico e Tecnológico). Os experimentos aqui relatados também foram parcialmente financiados pelo ERC-Starting Grant 3D-QUEST.

}

\chapter*{Acknowledgments}

{\singlespacing 
The task of gauging the impact, on one's life, of friends, families, and teachers, is daunting and almost never fair. I would like to thank everyone who, in one way or another, has touched my life and led it to this point. I apologize in advance for any omissions, they are due only to my fallible memory, please do not read too much into it.

My warmest thanks to my advisor, Ernesto Galvão, for the unwavering support and guidance throughout these last five years. His determination and confidence helped me keep the momentum at both the most exciting and frustrating times. The day that I can display a fraction of the patience and clarity, towards a student, that he does, is the day I will consider myself ready for this path. I hope to have also gained a lifelong friend and collaborator.

I am also very grateful to Andrew Childs, who hosted me during the six months I spent at IQC. Andrew's broad knowledge of the field and sharp insights never ceased to amaze me, and he devoted more time and attention helping me than I ever had the right to ask for.

I had the opportunity to work in collaboration with the outstanding quantum optics groups of Rome and Milan. I would like to specially thank Fabio Sciarrino, Paolo Mataloni, Chiara Vitelli, Nicolò Spagnolo, Andrea Crespi, Roberta Ramponi, and Roberto Osellame. As a theoretician I cannot begin to imagine the effort necessary to bring to fruition so many interesting projects, and I take their drive and dedication as a very powerful inspiration.

During my Ph.D., I was also fortunate enough to meet very amazing and inspiring friends. I would like to specially thank Raphael Dias, who first introduced me to this field, and Tiago Debarba. With both I shared adventures and countless hours of enjoyable conversation, and both have made profound impressions on my views on life, my field, and my academic career. I would also like to thank everyone else at IQC and UFF, including (but not limited to) David Gosset, Rajat Mittal, Zak Webb, and Ingrid Hammes. I also had a host of excellent teachers at UFF, to whom I owe much of my academic formation, and I would particularly like to thank Nivaldo Lemos, Antônio Zelaquett, Marco Moriconi, Luis Oxman, and Daniel Jonathan. Finally, I would especially like to thank Matt Fries, who went to great lengths to make Waterloo seem the most hospitable place possible.

I would also like to thank Letícia, Vera, Gabriel, Hermano, and Rayssa, longtime friends who helped make my life and this journey easier, fun, and more meaningful. A very special place for André, who has been my friend since forever, who has helped shape a major portion of the human being I am today, and who I currently talk to much, much, much less than I should. Or want. Or he deserves.

I also would like to thank my family for the love and support that allowed me to be who I really am. My father and grandfather showed me the way, and it honors and humbles me deeply that I am allowed to try and follow their footsteps. My mother and grandmother were present in every step of the path, providing constant love, support, and encouragement, that gave me strength to overcome all obstacles.

Finally, and most importantly, I would like to thank Raissa. This is the most unfair of all acknowledgments, as no words will ever be sufficient to express it. She is the love of my life, my soul mate, my kindest critic, and the most supporting and loving companion I could have ever hoped to share my life with. Absolutely none of this would have been possible without her.

\textbf{Financial support:} I acknowledge full financial support by CNPq (Conselho Nacional de Desenvolvimento Científico e Tecnológico). The experiments reported here were also partially supported by ERC-Starting Grant 3D-QUEST.

}

\begin{abstract}

We can study the computational power of restricted models of computation in order to shed light on the nature of quantum computational speedup. From a theoretical perspective, it can help determine what resources are necessary and/or sufficient for universal quantum computation. This issue is also relevant in experimental settings where the available operations or resources may be restricted. In this thesis, I study two different restricted models of quantum computation that stem from the behavior of free indistinguishable quantum-mechanical particles.

The dynamics of noninteracting fermions correspond to a restricted set of two-qubit gates known as matchgates. Matchgates are known to be classically simulable when acting on nearest-neighbor qubits on a path, but are universal for quantum computation when the gates can also act on more distant qubits or, equivalently, when SWAP gates are available. Here, I generalize these known results in two ways. First, I show that SWAP is only one in a large family of gates that can uplift matchgates to full quantum universality. More specifically, I show that the set of all matchgates plus any nonmatchgate parity-preserving two-qubit gate is universal, and I give an interpretation of this fact in terms of local invariants of two-qubit gates. Second, I investigate the power of two-qubit matchgates between qubits in an arbitrary connectivity graph, showing that they are universal on any connected graph other than a path or a cycle, and that they are classically simulable on a cycle. This same dichotomy holds for the XY interaction, a proper subset of matchgates that arises naturally in several implementations of quantum computing.

Noninteracting bosons (e.g.\ linear optics) give rise to a recently proposed restricted model known as BosonSampling. The BosonSampling task consists of (i) preparing an initial Fock state of $n$ identical photons, (ii) interfering these photons in an $m$-mode linear interferometer, and (iii) measuring the resulting output distribution in the Fock basis. It can be shown that sampling approximately from the resulting distribution should be classically hard, under reasonable complexity assumptions. Here I show, under similar assumptions, that exact BosonSampling remains hard even if the linear-optical circuit has constant depth. I also report several experiments performed in collaboration with Quantum Optics groups in Rome and Milan, where three-photon interference was observed in integrated interferometers of various sizes, providing some of the first implementations of BosonSampling in this regime. The experiments also focus on the bosonic bunching behavior in multimode interferometers, and on the validation of BosonSampling devices. This thesis also contains detailed descriptions of the numerical analyses done on the experimental data, and which were omitted from the corresponding publications.

\end{abstract}

\begin{resumo}

Podemos estudar o poder computational de modelos restritos de computação para ajudar a esclarecer a natureza do \textit{speedup} computacional. Do ponto de vista teórico, pode ajudar a determinar que recursos são necessários e/ou suficientes para computação quântica universal. Essa questão também é de interesse no caso de implementações experimentais em que haja restrições nas operações ou recursos disponíveis. Esta tese dedica-se ao estudo de dois modelos restritos de computação quântica, provenientes da descrição da evolução de partículas idênticas não interagentes em Mecânica Quântica.

A dinâmica de férmions não interagentes corresponde a um conjunto restrito de portas de dois \textit{qubits} conhecidas como \textit{matchgates}. Circuitos de \textit{matchgates} são simuláveis classicamente se os \textit{qubits} estão organizados em um grafo linear e as portas só atuam entre primeiros vizinhos, e universais para computação quântica se as portas podem atuar entre \textit{qubits} distantes ou, de forma equivalente, se a porta SWAP está disponível. Nesta tese, eu generalizo esses resultados de duas formas. Primeiro, mostro que a SWAP pertence a uma família contínua de portas capazes de tornar \textit{matchgates} universais. Mais especificamente, mostro que qualquer porta de dois \textit{qubits} que preserve a paridade (e não seja um \textit{matchgate}) pode ser adicionada ao conjunto completo de \textit{matchgates} para se obter computação universal e, além disso, dou uma interpretação desse fato em termos de invariantes locais de portas de dois \textit{qubits}. Em seguida, estudo o poder computacional de \textit{matchgates} entre \textit{qubits} em grafos de conectividade arbitrários. Mostro que \textit{matchgates} podem realizar computação universal em qualquer grafo que não seja um ciclo ou um caminho, e que eles são simuláveis classicamente se o grafo é um ciclo. Essa dicotomia persiste se restringimos o conjunto somente à interação XY, um subconjunto de \textit{matchgates} diretamente relacionado a diversas implementações de computação quântica.

Bósons não interagentes (e.g.\ ótica linear) dão origem a um modelo, proposto recentemente, conhecido como amostragem bosônica (\textit{BosonSampling}). A tarefa de amostragem bosônica consiste em: (i) preparar um estado de Fock de $n$ fótons, (ii) evoluí-lo de acordo com um interferômetro linear de $m$ modos e (iii) medir as saídas do interferômetro na base de Fock. Pode-se mostrar que, partindo de algumas conjecturas razoáveis relativas a classes de complexidade, não é possível produzir, de forma eficiente em um computador clássico, uma amostra da distribuição resultante desse sistema, nem de forma aproximada. Nesta tese mostro que, sob conjecturas semelhantes, a versão exata da amostragem bosônica é difícil mesmo se o circuito ótico tem profundidade constante. Também descrevo alguns experimentos, realizados em colaboração com grupos experimentais de Roma e Milão, em que foi observada a interferência de três fótons em \textit{chips} fotônicos de vários tamanhos. Esses experimentos estão entre as primeiras implementações de amostragem bosônica nesse regime. Os experimentos também evidenciam o efeito de agrupamento (\textit{bunching}) bosônico em interferômetros multimodo e a aplicação de protocolos de validação desses dispositivos. Esta tese contém descrições detalhadas de análises numéricas realizadas sobre os dados experimentais, que foram omitidas das respectivas publicações.

\end{resumo}

\newpage
\listoffigures
\newpage
\tableofcontents

\newpage
\pagenumbering{arabic}
\chapter{Introduction and thesis outline} \label{chapter:outline}

Quantum computing is a new computing paradigm, in which the basic components that make up the computer work in a regime (e.g.\ they are sufficiently small) where they behave according to the laws of Quantum Mechanics. One could expect the counterintuitive properties of the quantum world to hinder the engineering of our electronic devices---this is true, in a sense, as it imposes limitations on the miniaturization of electronic components, and conventional computer processors will soon reach a fundamental limit on the number of transistors that can be packed on a silicon chip. However, quantum computing is based on viewing the unusual properties of quantum systems not as a hindrance, but rather as a resource to be exploited for the performance of computational tasks.

The idea of using quantum systems to process information was put forth in the 1980s, most notably by the work of Feynman \cite{Feynman1982}. It was known that classically simulating a quantum system was a very hard problem, especially because (but not only because) the Hilbert space used to describe a system of $n$ particles grows exponentially with $n$, and the runtime of any naive algorithm quickly blows up for systems composed of more than a few particles. It was then suggested that, rather than using a classical computer to simulate a quantum system, it might be simpler to use a suitably-controllable quantum system to simulate another. This became known as a quantum simulator, and was the first candidate of a task in which a quantum processor can, in principle, outperform a classical one. As a matter of fact, to this day quantum simulation remains one of the most prominent expected applications for a quantum computer \cite{Lloyd1996}.

After Feynman's original proposal, quantum computing gathered further momentum with the works of Deutsch \cite{Deutsch1985}, Simon \cite{Simon1997}, Shor \cite{Shor1995, Shor1997}, Grover \cite{Grover1996}, Lloyd \cite{Lloyd1996}, and many others. Most notably, two contributions of Shor in the 1990s helped gather attention to this then-developing field: Shor's algorithm for factoring \cite{Shor1997} and the concept of quantum error-correcting codes \cite{Shor1995}. The factoring algorithm consists of a routine to obtain the prime factorization of an $n$-digit number exponentially faster than any known classical algorithm (and, in fact, faster than any classical algorithm is conjectured to be). This algorithm was the first example of a practical and ``non-academic'' task with a quantum computational speedup, and it remains the quantum algorithm most well-known by the general public, given that it has drastic consequences to modern day cryptography. However, even then there was still a lot of skepticism regarding quantum computing, mostly due to the belief that imperfections in real-world systems would make it impossible to achieve the high level of precision necessary for quantum computing, and that errors would acumulate fast enough to disrupt any useful computation. The answer to these fears came with the works of Shor \cite{Shor1995} and Steane \cite{Steane1996} on quantum error correction, which showed that it is, in fact, possible to measure and correct errors that happen during the quantum computation. The several advances made since then, including improved error correcting techniques \cite{Gottesman2009}, fault-tolerance and the threshold theorem \cite{Aharonov2008b}, and the development of several alternative models of quantum computation, such as adiabatic \cite{vanDam2001,Aharonov2007}, topological \cite{Nayak2008}, and measurement-based \cite{Raussendorf2001, Raussendorf2012} quantum computation, have been making the skeptics' job progressively harder, as one needs to formulate error models that are tailored to rendering all of these results useless whilst not contradicting the numerous experimental observations of quantum mechanics \cite{Kalai2009,Kalai2011}.

However, as much as the motivation and theoretical feasibility of quantum computing have been thoroughly established in the last two decades, experimental and technological limitations still hold practical quantum computing in the distant future---to illustrate this point, note that the largest integer factored with Shor's algorithm to date is $21$ \cite{Martinlopez2012}. In fact, it is still not clear which physical platform will be the most feasible for implementation of quantum computers, nor if this should be done in the circuit model or using an alternative model. These issues naturally lead to the study of \emph{restricted} models of computation, which are the main focus of the results reported in this thesis.

For our purposes, a restricted model of quantum computation consists of some set of operations that do not, a priori, contain the standard ``formula'' for a quantum computer: (i) preparation of a polynomial number of qubits in the computational basis, (ii) a sequence of a polynomial number of arbitrary two-qubit gates, and (iii) final arbitrary single-qubit measurements, on the computational basis, of a suitably large subset of the output. We may define a restricted model by imposing further limitations on one, or all, of these ingredients. By studying restricted models of computation we can better understand the physical origin of the quantum computational speedup, what the minimal and most feasible resources for implementation of a fully scalable quantum computer are, about the computational properties inherent to particular physical systems, etc. For example, one can consider the computational power of a quantum computer that operates with limited amounts of entanglement \cite{Jozsa2003,Nest2013}, or with a restricted set of operations (such as e.g.\ Clifford gates \cite{Gottesman1999a,Jozsa2014} or matchgates \cite{Valiant2002,Terhal2002}), or what happens if it can only perform a fixed number of rounds of operations \cite{Terhal2004}. These restrictions may arise directly from known experimental implementations, such as e.g.\ the hardness of obtaining suitably large and controllable optical nonlinearities that motivated the study of quantum computing with linear optics \cite{Knill2001b}; they may arise from the necessity of implementing a given computational task, such as e.g.\ the hardness of implementing fault-tolerant computation with non-Clifford gates; finally, they may arise from the modeling of more exotic scenarios, such as quantum computing with closed timelike curves \cite{Deutsch1991, Svetlichny2011, daSilva2011,Lloyd2011}. There are too many scenarios and motivations to list at length here, but we will review several examples in \chap{introduction}.

In this thesis, we will be mainly concerned with two restricted models of computation that arise from the description of the evolution of noninteracting particles. The first is the model of \emph{matchgate} quantum computing, and the second is a recent proposal known as \emph{BosonSampling} \cite{Aaronson2013a}. 

Matchgates are a restricted set of two-qubit gates that can be related to free fermions via the Jordan-Wigner transformation, which maps spin operators to fermionic operators. They provide a transition from classical to quantum computational power based on a seemingly ``mild'' change in underlying restrictions: if the qubits are arranged on a path, the output of a circuit composed only of nearest-neighbor matchgates can be efficiently simulated classically \cite{Valiant2002, Terhal2002}, whereas any quantum computation can be efficiently simulated by a circuit of matchgates acting on first and second neighbors \cite{Jozsa2008b,Kempe2002}---alternatively, this change in connectivity can be simulated by suitable use of the $\swap$ gate. In this thesis, I generalize these results in two main ways: (i) I use the theory of local invariants of two-qubit gates to investigate what property the $\swap$ gate has that allows it to enact this transition between classical and quantum computational power on the set of matchgates, and describe a continuous family of two-qubit gates that can replace it; and (ii) I show that changing the underlying qubit connectivity graph can also bridge this gap in computational power, and that matchgates are in fact universal if they can act according to any connected graph other than a path or a cycle, and they are classically simulable on a cycle. I also show that this same dichotomy holds for the XY interaction, a proper subset of matchgates that arises naturally in several implementations of quantum computing \cite{Imamoglu1999,Quiroga1999,Zheng2000,Mozyrsky2001}. It should be noted that, while matchgates are related to the evolution of free fermions, our results treat them as two-qubit gates in their own right, and points (i) and (ii) in fact depart from the correspondence with the free-fermion picture.

BosonSampling is a restricted model directly related to free bosons, and well-suited to recent advances in linear-optical quantum information processing. The BosonSampling task consists of (i) preparing an initial Fock state of $n$ identical photons, (ii) interfering these photons in an $m$-mode linear interferometer, and (iii) measuring the resulting output distribution in the Fock basis. As shown in \cite{Aaronson2013a}, sampling from the distribution obtained at the output of a device that performs this task should be classically hard, under reasonable complexity assumptions. This result is similar in spirit to others concerning quantum circuits built out of a constant number of two-qubit gate layers \cite{Terhal2004}, circuits built only out of commuting quantum gates \cite{Bremner2011}, among others. In this thesis, I show that this result also holds if BosonSampling is performed only with a constant number of beam splitter layers. 

The authors of \cite{Aaronson2013a} go even further, and their main technical contribution is to show that the BosonSampling task should be hard even if we allow the classical simulator to sample only from a distribution that \emph{approximates} the ideal one. This approximate BosonSampling result brings the model closer to real-world implementations, where experimental imperfections mean that any quantum device will also only approximate the ideal output distribution. In fact, since its inception in 2011, several quantum optics groups have reported small-scale experiments related to BosonSampling \cite{Broome2013, Crespi2013b,Spring2013,Tillmann2013,Spagnolo2013b,Carolan2013,Spagnolo2013c}. In this thesis, I report several of these experiments that were performed by the Quantum Optics groups in Rome and Milan \cite{Crespi2013b,Spagnolo2013b, Spagnolo2013c}, in a collaboration that I was part of. The experiments consisted of observation of three-photon interference in integrated interferometers of various sizes. The experiments also focused on other aspects of the photonic behavior in this regime, such as bosonic bunching, and on the certification of the observed output distributions. I describe the setup and results of these experiments, but focus on the theoretical motivation and numerical analyses that were my main contribution, providing a great deal of detail that was omitted from the corresponding publications.

I will elaborate on the motivation, technical definitions, and historical background for matchgates and BosonSampling in the introductions of their respective chapters. For now I just point out the fascinating contrast between these two models: while free fermions are classically simulable and thus presumably cannot display any quantum computational power, bosons outperform classical computers in a particular task that, furthermore, seems to consist only of ``behaving naturally''. This contrast suggests a deep connection between the physical nature of a system and its computational power, arguably rivalled only by the possibility of universal quantum computing itself.

\section{Thesis outline}

By necessity, this thesis contains a great deal of material reviewing previous results. Chapters \ref{chapter:introduction}, \ref{chapter:fermreview} and \ref{chapter:bosonreview} contain mostly revision material pertaining, respectively, to general quantum computing, matchgates, and BosonSampling. Chapters \ref{chapter:fermnew} and \ref{chapter:bosonnew} consist of new results on matchgates and BosonSampling, respectively. \chap{fermnew} is based on my co-authored papers \cite{Brod2011,Brod2012}, published in \textit{Physical Review A}, as well as \cite{Brod2013} recently accepted for publication in \textit{Quantum Information and Computation}. \chap{bosonnew} is based on my co-authored paper \cite{Crespi2013b}, published in \textit{Nature Photonics}, as well as \cite{Spagnolo2013b}, published in \textit{Physical Review Letters} and \cite{Spagnolo2013c}, currently submitted for publication. In more detail, this thesis is organized as follows:

In \chap{introduction}, I review some basic definitions, which are standard in quantum computing research. In \sec{introuniv} I review the notions of quantum universality and encoded universality. In \sec{introsimul} I describe several notions of exact and approximate classical simulation of quantum circuits. In \sec{introduction_b} I describe the formalism of second quantization, the standard physical toolset for describing identical quantum particles, which is also particularly suitable for the subsequent description of the corresponding computational models. In \sec{introduction_d} I give a brief overview of complexity theory, including formal descriptions of the complexity classes that will play the most prominent roles in the remainder of the thesis. Finally, in \sec{introduction_c}, I give a general overview of restricted models of computation, including several that seem to have intermediate computational power between classical and quantum computation. The main purpose of \chap{introduction} is to give some basic definitions that we will need, and which may be taught in a Physics course but not in a Computer Science course and vice-versa. Most quantum computing researchers will be familiar with most of this chapter, but these concepts are included for completeness.

In \chap{fermreview}, I review some previously-known results regarding matchgates. Although these results are not new, their proof details aid in the exposition of our new results in \chap{fermnew}. I begin with the formal definition of this model in \sec{ferm_overview}, as well as a historical overview. In \sec{fermreview_a} I reproduce the proof that matchgates are classically simulable when acting on nearest-neighboring qubits aligned on a path, which uses the Jordan-Wigner transformation to map matchgates into fermionic operators. In \sec{fermreview_b} I show how a small change in the underlying restrictions---namely, allowing matchgates to act between second-neighboring qubits---enables universal quantum computation. In \sec{fermreview_c} I briefly review, for completeness, results connecting the power of nearest-neighbor matchgates to that of other restricted models, such as log-space quantum computing.

In \chap{bosonreview}, I review some previously-known results regarding quantum computing with linear optics. In \sec{boson_overview} I give a brief historical overview, and discuss the advantages and disadvantages of processing quantum information with photons. In \sec{bosonreview_a} I describe the seminal KLM scheme \cite{Knill2001b} that enables universal and scalable quantum computing with linear optics if operations can be adapted on the outcome of intermediate measurements. Although our new results in \chap{bosonnew} will not concern the KLM scheme directly, they use some of its constructions. In \sec{bosonreview_b} I give a simple proof, based on the KLM scheme, that exact BosonSampling is hard up to some reasonable complexity assumptions, which was first shown in \cite{Aaronson2013a}. I discuss how this proof relates to other restricted models, such as quantum computers that have constant depth, or are built out of commuting gates. I discuss the underlying restrictions, as well as general pros and cons of the BosonSampling model. \sec{bosonreview_c} is devoted to the approximate version of BosonSampling. Since the original paper that proposed it is very long and highly technical \cite{Aaronson2013a}, there is no hope of giving a complete review, so we focus on some of the specific technical aspects that are more relevant to our new results in \chap{bosonnew}.

In \chap{fermnew}, I present our new results concerning matchgates. In \sec{fermnew_a} I study how to supplement the computational power of nearest-neighbor matchgates with other two-qubit gates. The $\swap$ gate is one previously-known such example, and I generalize this results to show that any parity-preserving two-qubit gate---that is, a gate that does not connect two-qubit states of different parity---can replace the $\swap$ and provide universal computational power to matchgates. In \sec{fermnew_b} I study the bridging of this gap from classical to quantum computational power in a completely different way. Rather than including new operations, we investigate the power of matchgates acting according to different connectivity graphs. We show how the previous results (i.e.\ classical simulability on nearest-neighbors and quantum universality on nearest and next-nearest neighbors) can be recast in this formalism, and we generalize them to show that: (i) matchgates are universal acting on any graph that is not a path or a cycle; and (ii) matchgates are classically simulable on a cycle. This establishes a dichotomy, showing that the jump in computational power is abrupt, and that there is no connectivity graph for which an intermediate computational power (such as that of BosonSampling) may arise. I also show that this same dichotomy holds for the XY interaction, which forms a subset of matchgates. Although this latter result implies the former, the construction using general matchgates is more explicit, and more efficient in general. Finally, in \sec{fermnew_c} I make some final remarks, discussing possible open questions, as well as the relation between our results and other works.

In \chap{bosonnew}, I present our new theoretical and experimental results concerning BosonSampling. In \sec{bosonnew_a} I show that the exact version of the BosonSampling result, reviewed in \sec{bosonreview_b}, holds even if the linear-optical circuit only has a constant number of layers of beam splitters. In \sec{bosonnew_b} I describe several aspects of the experiments, such as experimental setup, numerical analysis, etc, that were used in the experiments reported in \sec{bosonnew_c}. I give special focus to the numerical analyses, where I show simulations of the expected behavior of different matrix ensembles for the BosonSampling task, as well as a numerical procedure for refining the process tomography for linear interferometers. \sec{bosonnew_c} is devoted to the results of the three recently performed experiments, namely: (i) one of the first implementations of BosonSampling in the 3-photon, 5-mode regime (\sec{BosonSamplingExp}), (ii) observation of bosonic bunching effects on three photons interfering in interferometers of up to 16 modes (\sec{BosonicBunchingExp}), and (iii) validation of BosonSampling devices of 3 photons in interferometers of 5, 7, and 9 modes (\sec{validation}). Finally, \sec{bosonnew_d} consists of some concluding remarks, where I discuss the open questions and major challenges still remaining for scalable implementations of the BosonSampling model, inspired both by theoretical and experimental motivations. 

Finally, \chap{conclusions} is devoted to concluding remarks. I make a brief summary of the results obtained here, and how they fit into the larger picture of current research in quantum computing. I also describe a few more questions that were left open, as well as other directions in which our results can be expanded.

The results of \chap{bosonnew} are supported by \appdx{app1} and \appdx{app2}, which contain the Mathematica$^\copyright$ code used for the numerical simulations, and \appdx{app3} which contains tables describing the specifications of each linear interferometer.

\section{Notation and conventions}

In this section, I summarize the basic notation that will be used in this thesis. My intention is that the notation and graphical representations described here be maintained consistently throughout the thesis. However, it will often prove necessary to locally change the adopted notation, especially in \chap{bosonnew} where we report the experimental results. I opted for using the original figures, published in the respective journals, and so I considered it more instructive to adapt the local notation to reflect that of the figures. To avoid any confusion, these small changes in notation and graphical representations will be accompanied by corresponding remarks, besides being clear from context.

For reference, the recurring single-qubit gates that we will need, which are standard to quantum computing literature, are the Hadamard ($H$) gate, the $\pi/4$ ($P$), the $\pi/8$ ($T$) gates:
\begin{align*}
H = \frac{1}{\sqrt{2}} \begin{pmatrix} 1 & 1  \\ 1 & -1 \end{pmatrix}, \quad
P = \begin{pmatrix} 1 & 0  \\ 0 & i \end{pmatrix}, \quad
T = \begin{pmatrix} 1 & 0  \\ 0 & e^{i \pi/4} \end{pmatrix},
\end{align*}
as well as the Pauli gates:
\begin{align*}
X = \begin{pmatrix} 0 & 1  \\ 1 & 0 \end{pmatrix}, \quad
Y = \begin{pmatrix} 0 & i  \\ -i & 0 \end{pmatrix}, \quad
Z = \begin{pmatrix} 1 & 0  \\ 0 & -1 \end{pmatrix}.
\end{align*}
The recurring two-qubit gates we will need are the well-known controlled-NOT ($\cnot$), controlled-phase ($\cz$), and $\swap$ gates:
\begin{align*} 
\cnot = \begin{pmatrix} 1 & 0 & 0 & 0 \\ 0 & 1 & 0 & 0 \\ 0 & 0 & 0 & 1 \\ 0 & 0 & 1 & 0 \end{pmatrix}, \quad 
\cz = \begin{pmatrix} 1 & 0 & 0 & 0 \\ 0 & 1 & 0 & 0 \\ 0 & 0 & 1 & 0 \\ 0 & 0 & 0 & -1 \end{pmatrix}, \quad 
\swap = \begin{pmatrix} 1 & 0 & 0 & 0 \\ 0 & 0 & 1 & 0 \\ 0 & 1 & 0 & 0 \\ 0 & 0 & 0 & 1 \end{pmatrix},
\end{align*}
as well as the the less common fermionic $\swap$ ($\fswap$) and $\iswap$:
\begin{align*}
\fswap = \begin{pmatrix} 1 & 0 & 0 & 0 \\ 0 & 0 & 1 & 0 \\ 0 & 1 & 0 & 0 \\ 0 & 0 & 0 & -1 \end{pmatrix}, \quad 
\iswap = \begin{pmatrix} 1 & 0 & 0 & 0 \\ 0 & 0 & i & 0 \\ 0 & i & 0 & 0 \\ 0 & 0 & 0 & 1 \end{pmatrix}.
\end{align*}
We will also mention in passing the three-qubit Toffoli gate, given by:
\begin{align*}
\textrm{Toffoli} = \begin{pmatrix} 1 & 0 & 0 & 0 & 0 & 0 \\ 0 & 1 & 0 & 0 & 0 & 0 \\ 0 & 0 & 1 & 0 & 0 & 0 \\ 0 & 0 & 0 & 1 & 0 & 0 \\ 0 & 0 & 0 & 0 & 0 & 1 \\ 0 & 0 & 0 & 0 & 1 & 0 \end{pmatrix}.
\end{align*}
We will also refer to arbitrary parity-preserving two-qubit gates by the shorthand $G(A,B)$, which stands for
\begin{align*} 
G(A,B) := \begin{pmatrix} A_{11} & 0 & 0 & A_{12} \\ 0 & B_{11} & B_{12} & 0 \\ 0 & B_{21} & B_{22} & 0 \\ A_{21} & 0 & 0 & A_{22} \end{pmatrix},
\end{align*}
where $A$ and $B$ are unitary $2 \times 2$  matrices. In \sfig{Notation}{a-c} we display some standard circuit graphical notation.

\begin{figure}
\capstart
\centering
\subfloat[]{\centering \includegraphics[width=0.5\textwidth]{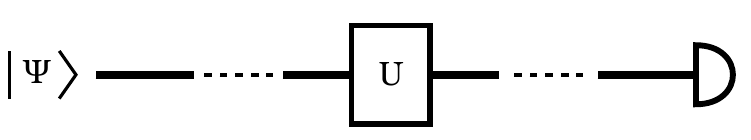}} \\
\subfloat[]{\centering \includegraphics[width=0.2\textwidth]{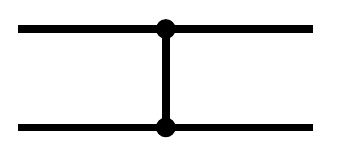}} \quad
\subfloat[]{\centering \includegraphics[width=0.2\textwidth]{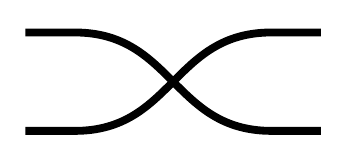}} \\
\subfloat[]{\centering \includegraphics[width=0.5\textwidth]{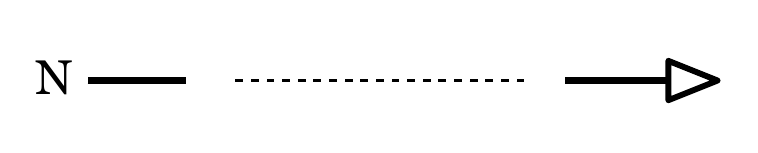}} \\
\subfloat[]{\centering \includegraphics[width=0.2\textwidth]{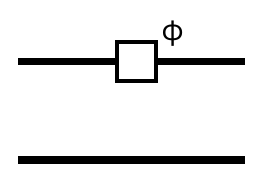}} \quad
\subfloat[]{\centering \includegraphics[width=0.23\textwidth]{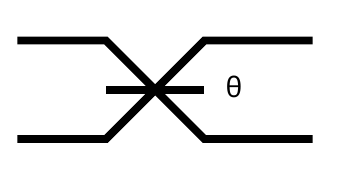}}
\caption[Graphical circuit representation]{(a) A qubit prepared in the state $\ket{\psi}$, then subject to the single-qubit gate $U$ at some point in the circuit, and measured at the end of the circuit in the computational basis. (b) The two-qubit $\cz$ gate. (c) The two-qubit $\swap$ gate. (d) Similar to (a), but for an optical circuit. An optical mode is prepared in a Fock state of $N$ photons, and subject to a final photon-number-resolving detection. (e) A phase shifter by angle $\phi$. (f) A beam splitter with angle $\theta$. Often, it $\theta$ is omitted, it represents a 50:50 beam splitter (i.e.\ $\theta = \pi/4$). Notice that there are marked differences between the two representations, mainly in the preparation, measurement and application of two-qubit/two-mode operations. In this thesis, it will always be clear by context which system we are talking about.}
\label{fig:Notation}
\end{figure} 

In most of the thesis we will consider continuous families of single- or two-qubit gates, so it will be convenient to denote the family of $n$-qubit unitary gates generated by the $n$-qubit Hamiltonian $K$ as $R_K (\theta) = \exp{(i \theta K)}$.

In several sections in this thesis, we will refer to linear-optical circuits rather than standard qubit circuits. The main elements of these circuits are phase shifters and beam splitters. If there are only two modes, a phase shift by an angle $\phi$ (on the first mode) and a beam splitter by an angle of $\theta$ correspond respectively to the $2 \times 2$ matrices:
\begin{align*}
U_{PS}(\phi) = \begin{pmatrix} e^{i \phi} & 0  \\ 0 & 1 \end{pmatrix}, \quad
U_{BS}(\theta) = \begin{pmatrix} \cos (\theta) & i \sin (\theta)  \\ i \sin (\theta) & \cos (\theta) \end{pmatrix},
\end{align*}
If these optical elements are embedded in larger interferometers, the unitary matrices that describe them are the straightforward generalization that acts as above on the involved modes and the identity on the rest. A beam splitter can be alternatively parameterized by the transmissivity $t:=\cos(\theta)$ or the transmission probability $T:=t^2$, which are parameters more natural to the quantum optics literature. The graphical notation for a linear optical circuit is very similar to that for a standard quantum circuit, but in the latter the lines correspond to qubits whereas in the former they correspond to optical modes. In \sfig{Notation}{d-f} we display some standard graphical representation for optical circuits, and the figure legend makes explicit the differences between the two representations. We will deviate from this graphical representation from optical circuits in \sec{bosonnew_c}, where we will give preference to a representation that more closely resembles the integrated interferometer architectures, but that will be clear from context.

For a single-qubit state, the basis $\{ \ket{0}, \ket{1} \}$ is the computational basis, and the basis $\{ \ket{+}:= H \ket{0}, \ket{-}:= H \ket{1} \}$ is the $X$ basis (because it is the basis of eigenvectors of the Pauli $X$ gate). The computational basis of an $n$-qubit state is the corresponding set of all tensor product of elements of the computational bases of each qubit. 

Finally, we denote $\{0,1\}^{*}$ as the set of all $n$-bit strings for every $n \in \mathbb{N}$. We will also say that a quantity has poly$(n)$ growth if there is some constant $k$ such that the referred quantity grows as O$(n^k)$.

\newpage
\chapter{Preliminary definitions} \label{chapter:introduction}

Quantum Computation is an extensively multidisciplinary research area and, as such, brings together researchers with vastly different backgrounds and formations. As a consequence, this thesis contains Physics material that is not typically found in a Computer Science course and vice versa. With this in mind, I decided to include an introductory chapter to ``lay the groundwork'', so to speak, for the remainder of the thesis. This chapter is not intendend as a complete review on any particular subject, but rather focuses on specific concepts from several areas that will play major roles throughout the subsequent chapters, whilst also fixing important terminology. 

This chapter is organized as follows. \sec{introuniv} discusses the concept of universal quantum circuits, including a generalization of the concept known as encoded universality. \sec{introsimul} reviews several notions of classical simulations of quantum circuits, namely strong simulation, and exact and approximate weak simulation. \sec{introduction_b} is devoted to identical particles in Quantum Mechanics and the formalism of second quantization. In \sec{introduction_d} I discuss complexity theory, and define the most prominent complexity classes that will be important in later chapters. Finally, in \sec{introduction_c} I discuss restricted models of quantum computation, reviewing some important known results about a few such models.

\section{Quantum universality} \label{sec:introuniv}

Let us begin this introduction with a brief review on quantum circuits and quantum universality (most of the information contained here can be found in any standard textbook, such as \cite{LivroNielsen}).

A classical circuit can be seen as a mapping from bit strings to bit strings. In \sec{introduction_d}, when we review complexity theory, we will address the more interesting matter of using circuits to solve problems, in which case a circuit is interpreted as mapping an input bit string encoding a \emph{question} to an output bit string encoding an \emph{answer}, and the central question is whether this can be done in a feasible amount of time and/or space. In a similar spirit, a quantum circuit can be viewed simply as a mapping between quantum \emph{states}. Fortunately, the formalism for describing transformations between different states was developed decades ago, in the context of quantum mechanics, to describe the dynamical evolution of physical systems. Here, we use a simplified version of this formalism, which is standard in the theory of quantum computing, where we only consider composite quantum systems consisting of a collection of two-level systems (i.e.\ qubits).

Consider the $n$-bit string $x=x_1 x_2 x_3 \ldots x_n$. The corresponding computational state is the tensor product $\ket{x} = \ket{x_1} \ket{x_2} \ldots \ket{x_n}$, and the collection of these states for all $x$ forms a basis of the Hilbert space, known as the computational basis. A quantum circuit can then described by some $2^n \times 2^n$ unitary matrix $U$, which takes as input some state $\ket{\psi_{in}}$ and maps it to an output state $\ket{\psi_{out}}$ as
\begin{equation*}
\ket{\psi_{out}}=U \ket{\psi_{in}}.
\end{equation*}
It is a well-known fact (see e.g.\ \cite{Hoffman1972, Reck1994, LivroNielsen}) that any such unitary matrix $U$ can be decomposed in terms of matrices acting on only two qubits at a time---each of these building blocks is denominated a quantum gate. We could, more generally, consider quantum gates acting on any small constant number of qubits at a time (e.g.\ the three-qubit Toffoli gate), but these will not appear in important roles throughout this thesis. The set of all two-qubit gates is a particular example of a \emph{universal} set. More generally, a set of quantum gates is universal for quantum computation if it densely generates the group of all $2^n \times 2^n$ unitary matrices.

Throughout most of this thesis, we will consider continuous families of two-qubit gates (or two-mode optical elements, in the case of linear optics). This is not the most realistic scenario: no real-word implementation of quantum computers will have perfectly controllable parameters, and it becomes necessary to search for discrete sets of gates, defined up to some experimental error tolerance, that are also universal. For example, it is well-known that products of $T$ and $H$ gates form a set that is dense in the set of all single-qubit gates, $\SU(2)$. It is also known that the two-qubit \textsc{cnot} gate together with arbitrary single-qubit gates generates the set of two-qubit gates, $\SU(4)$. By previous considerations, the set of all two-qubit gates densely generates the set of $n$-qubit gates, $\SU( 2^n )$, for any $n$. Concatenating these claims, we conclude that the set $\{ H, T, \textsc{cnot} \}$ is universal for quantum computation in the sense defined above.

The need for a universal discrete set of gates concerns mostly the theory of fault-tolerant quantum computing, that is, the ability to measure and correct errors that occur during the process sufficiently fast so that they do not completely disrupt the computation. As mentioned previously, in most of this thesis we will consider families of gates parameterized by continuous parameters, and so will not be concerned with fault-tolerance \textit{per se}. This choice is based on two main reasons. First, our results adopt a more conceptual point of view (e.g.\ what resources make a particular restricted set of operations universal) rather than a practical one (e.g.\ what is the most efficient way to perform a particular computation). In this spirit, proving non-fault-tolerant universality is a good first step towards proving its fault-tolerant version. Second, the well-known Solovay-Kitaev theorem (\cite{Kitaev1997}, see also \cite{LivroNielsen}) guarantees that, if we have a quantum circuit built out of a specific universal set of gates, we can rewrite it in terms of any other universal set with a modest overhead in the number of operations. Importantly, this is true even if one or both of the sets are discrete, in which case the Solovay-Kitaev theorem guarantees that any gate from one of the sets can be approximated within accuracy $\epsilon$ by sequences of length O$\left[\textrm{log}^c(1/\epsilon)\right]$ of gates of the other set, for some constant $c$. As such, all universal sets are equivalent, in the sense that any problem that can be solved efficiently (in a sense to be defined precisely later) by one can also be solved efficiently by any other.

It is important to point out that, while a universal set of gates generates a set dense in $\SU(2^n)$, by definition, this does not say anything about efficiency. In fact, a simple counting argument (see e.g.\ \cite{LivroNielsen}) shows that an \emph{exponential} number of two-qubit gates is needed to approximate an arbitrary $\SU(2^n)$ matrix, and this is considered highly unfeasible. In light of this, a unitary transformation is only considered feasible if it can be implemented by a circuit of polynomially-many quantum gates\footnote{In this sense, the Solovay-Kitaev theorem is fundamental since the polylogarithmic overhead it provides results in a mapping between universal sets that preserves the polynomial \textit{versus} exponential gap.} (we will return to this distinction between exponential and polynomial growth in \sec{introduction_d}).

Another standard definition we will need is that of a \emph{uniform} family of quantum circuits. For our purposes, a uniform family of (quantum or classical) circuits is a set of circuits $\{ C_n | n \in \mathbb{N} \}$, where $n$ is the number of qubits the circuit takes as input, and whose description can be classically computed from $n$ in poly($n$) time. In practice, this means that the circuits do not depend on the inputs, only on their size, and that the gate sequences can be computed efficiently (including all matrix elements of each gate to some desired precision, if the gate set is continuous). This is a slightly technical condition, but if uniformity is not imposed, some computational properties of the devices could be ill-defined. For example, we could somehow hide away the answer to some problem on the decimal expansion of the matrix elements of the gates, and it would even be possible to compute functions which are known to be uncomputable.

Another concept that will be central to most of our results is that of \emph{encoded universality}. Often, a certain restricted gate set clearly cannot be universal in the sense defined before for some fundamental reason. For example, in \chap{fermreview} we will encounter a class of two-qubit gates that is parity-preserving---that is, it not does connect two-qubit states of different parity. Any circuit composed only of these gates cannot take a computational state of even parity (say, $\ket{0 0 \ldots 0}$) into another of odd parity (say $\ket{1 0 \ldots 0}$), and thus it trivially does not generate SU($2^n$) when acting on $n$ qubits. However, we will see that in some cases these gates can still perform universal quantum computation in a generalized sense. By encoding each qubit of information into more than one physical qubit, we can circumvent this limitation. For parity-preserving gates, for example, it will be convenient to interpret the two-qubit state $\ket{00}$ as encoding bit 0, and $\ket{11}$ as encoding bit 1, and simply disregard states $\ket{01}$ and $\ket{10}$. As we will see, it will be possible to construct a gate set that: (i) preserves this encoding, and (ii) is universal on the encoded space (i.e., generates any single- and two-qubit gate on the encoded qubits), and thus performs any desired quantum computation at the cost of doubling the number of qubits. Throughout this thesis we will occasionally denote an encoded state or gate by a subscript $L$ (for logical) if there is chance for ambiguity. We will also denominate the encoded states as encoded, logical or computational qubits, whereas the actual qubits that compose them will be denoted as physical qubits.

Finally, we point out that this model of quantum computation based on circuits, which is called simply the circuit model, while being the most common setting in which to describe quantum computers, is not the only one. An alternative model we will occasionally need is the measurement-based model of quantum computation (MBQC). MBQC is a very successful and extensively studied model, and a complete discussion is unfortunately beyond the scope of this thesis, so we will restrict ourselves to a description of what the computation consists of, and further discussion can be found e.g.\ in \cite{Raussendorf2001, Raussendorf2012}. In essence, an MBQC protocol can be implemented in two steps: first, prepare a highly-entangled state, and then sequentially measure each qubit in a specific single-qubit basis. One important requirement is that the measurements must be \emph{adaptive}, in the sense that the basis on which each measurement is performed, as well as the meaning of the outcome, depends on the results of previous measurements. By measuring the state in a well-defined (and efficiently-computable) order, one ``forces'' the information to propagate from qubit to qubit, and the final output qubits encode the result of the computation. A more explicit description will be given in \chap{bosonnew}.

\section{Efficient classical simulability} \label{sec:introsimul}

One way to show that a physical system or restricted model of quantum computation has limited computational power is to show that it is efficiently classically simulable\footnote{By efficient, we mean in polynomial time, as it is always possible to do a classical simulation of a quantum system if unbounded time is allowed, just by keeping track of the vectors in Hilbert space. In this thesis, ``classically simulable'' will always refer to efficient simulation.}. Naturally, this conclusion relies on the conjecture that quantum computers are more powerful than classical computers---often, in the study of computational complexity, results must be conditioned on plausible assumptions, and unconditional proofs are remarkably rare. As we will see in \sec{introduction_c}, some restricted classes of quantum computation may display features considered essential for a quantum speedup, such as unbounded entanglement \cite{Jozsa2003}, but nonetheless be classically simulable. In this case, they cannot offer a quantum speedup, no matter how entangled or otherwise complex the states of the system may look a priori. Since classical simulability is such a central concept throughout this thesis, let us define it formally.

There are several possible definitions of classical simulability, with varying degrees of strength. I will now give a description of a few that will be useful later on, with some remarks about where each type of simulation arises naturally and how they relate to each other. Since these are common concepts in quantum computing literature, the references provided just refer to those where the definition was taken from, stated in the most convenient form for our purposes. The most restrictive notion of simulation is the aptly named:

\begin{definition} \label{def:strongsimul}
\textbf{Strong simulation} \cite{Jozsa2008b,Nest2011b}: Let $C_n$ be any uniform family of quantum circuits, defined together with a specified class of (i) input states and (ii) output measurements. $C_n$ is classically \emph{strongly simulable} if any probability or conditional probability of the output measurements can be computed by classical means to $m$ digits of accuracy in poly$(n,m)$ time.
\end{definition}

In most cases, input states are taken to be either product states or just computational basis states, and output measurements are taken as single-qubit measurements in the computational basis\footnote{Technically, the class $C_n$ is strongly simulable \emph{relative to the choices of (i) and (ii)}, as there are numerous examples of families of circuits where the gap between classical and quantum computational power is bridged just by a change in (i) or (ii). Throughout this thesis, however, this choice will always be made explicit.}. If the output consists on several disjoint measurements, it may be necessary to calculate both the probabilities and conditional probabilities---that is, the probability of observing a measurement outcome conditioned on results of previous measurements---in order to also capture their correlations. This condition will guarantee that strong simulation always implies weak simulation, as we will see shortly.

To understand why this is known as strong simulation, note that computing the probabilities to $m$ digits of accuracy in poly$(n,m)$ time corresponds to \emph{exponential} precision. Quantum computers are probabilistic in nature, and outcome probabilities are accessed only by performing repeated measurements and compiling the results, which cannot achieve exponential precision in polynomial time. To illustrate this, consider a quantum experiment repeated 1000 times. The table of experimental results provides each outcome probability with precision no larger than 1 in 1000, i.e., just 3 digits of accuracy. If the number of digits of accuracy provided by the quantum computation itself only grows logarithmically with the elapsed time, clearly \defn{strongsimul} makes for an unfair comparison. To drive the point home further: even the quantum computer cannot strongly simulate itself. Nonetheless, this is an unfair fight that classical computers often win---as we will see in \sec{introduction_c} and \chap{fermreview}, there are several restricted classes of quantum circuits where strong classical simulation can be achieved.

This discussion naturally suggests a less stringent notion of classical simulation:

\begin{definition}
\textbf{(Exact) Weak simulation} \cite{Bremner2011,Nest2011b}: Let $C_n$ be any uniform family of quantum circuits under the same conditions as \defn{strongsimul}. $C_n$ is classically \emph{weakly simulable} if the probability distribution over the measurement outcomes can be sampled by purely classical means in poly$(n)$ time.
\end{definition}

The main difference between this notion of simulation and the previous one is that, rather than asking the classical computer to compute the outcome probabilities, we ask only that it sample from the same output distribution. By previous considerations, this simulation only provides the probabilities with a logarithmic number of digits (i.e.\ polynomial precision). As the names suggest, strong simulation implies weak simulation \cite{Terhal2004}. This is not entirely obvious---in general, the output distribution we wish to simulate may be defined on an exponentially large space, and it would not be possible to compute and store the probabilities of all outcomes. Nonetheless, a weak simulation can be obtained from a strong simulation as follows: (i) choose one measurement and calculate its associated probabilities; (ii) fix the value of the corresponding output bits by a classical ``coin toss'' according to the computed probabilities; (iii) choose another measurement and repeat steps (i) and (ii), calculating probabilities conditioned on previously fixed outcomes, until all outcomes have been fixed. It is clear that the final fixed bit string has been sampled from the same distribution as that generated by the quantum computer. 

As strong simulation implies weak simulation, one may wonder whether the latter is relevant. In fact it is, as there are certain quantum computations that are very hard to simulate strongly, but can be weakly simulated (see \cite{Nest2011b} and references therein).

Weak simulation is clearly a fairer comparison between classical and quantum computers, as we now only require the classical computer to ``behave'' like the quantum computer, in a well-defined sense. For example, suppose we are given a black-box device and told it is a quantum computer that performs some particular calculation. If the black box was actually a classical computer, a weak simulation of the quantum computer would be sufficient to trick us.

However, we can define even weaker notions of simulation, based on the fact that a real-world quantum computer will be subject to noise, experimental imperfections, etc. Let us now define notions of \emph{approximate} weak simulation, where the classical simulator is only required to sample from a distribution close to the ideal quantum one---presumably, that is the best a real-world quantum computer itself can do. First, let us define two notions of approximate distributions: let $p(x)$ and $q(x)$ be two distributions defined over the same sample space. We say $p$ is close to $q$ with multiplicative error $c \geq 1$ if, for every $x$ in the sample space,
\begin{equation}
\frac{1}{c} q(x) \leq p(x) \leq c q(x).
\end{equation}
And $p$ and $q$ are $\epsilon$-close in total variation distance if
\begin{equation}
||p-q|| = \frac{1}{2} \sum_x | p(x)-q(x) | < \epsilon
\end{equation}

We now define two natural notions of weak simulation based on these definitions of error \cite{Bremner2011}:

\begin{definition}
Let $C_n$ be any uniform family of quantum circuits under the same conditions as \defn{strongsimul}. $C_n$ is classically \emph{weakly simulable with multiplicative error $c$} if there is a family of probability distributions parameterized by $n$, close to the ideal quantum output distribution with multiplicative error $c$, and that can be sampled by purely classical means in poly$(n)$ time.
\end{definition}

\begin{definition}
Let $C_n$ be any uniform family of quantum circuits under the same conditions as \defn{strongsimul}. $C_n$ is classically \emph{weakly simulable within $\epsilon$ total variation distance} if there is a family of probability distributions parameterized by $n$, $\epsilon$-close in total variation distance to the ideal quantum output distribution, and that can be sampled by purely classical means in poly$(n)$ time.
\end{definition}

These two definitions are clearly weaker than exact weak simulation. Furthermore, it is not hard to show that the first implies the second: if two probability distributions are close in multiplicative error $c$, then they are $\epsilon$-close in total variation distance with $\epsilon=c-1$. However, the converse is not true. Suppose we have a distribution with a few extremely unlikely outcomes, and we approximate it by setting the probability of these outcomes to 0 (and normalizing the rest). The resulting distribution will be $\epsilon$-close in total variation distance to the original one, for some suitable $\epsilon$ bounded by the discarded probabilities, but it will not be multiplicatively close for any value of $c$. 

For usual quantum computation, these notions of approximate simulation are often unnecessary, due to the power of quantum error correction. If the error rates are smaller than a certain value, the threshold theorem (see e.g.,\ \cite{Gottesman2009} for an introduction to this subject) guarantees, under some reasonable assumptions, that we can correct the errors faster than they happen. However, as we will see in \chap{bosonreview} and \chap{bosonnew}, there are restricted models of computation that do not (yet) have a notion of error correction, and the notion of approximate simulation is more meaningful.

Finally we point out that, from a computational complexity perspective, all these definitions of simulation might already be too strong---conceivably, a classical computer might be able to solve the same \emph{decision problems} as a particular class of quantum computers, without ever needing to explicitly simulate them. We will return to this point on \sec{introduction_d}, when we review computational complexity classes.

\section{Identical particles in QM} \label{sec:introduction_b}

In the classical world, we are used to describing the states and trajectories of objects as if they were perfectly distinguishable. It is natural to assume that one can simply pin labels to particles and follow their dynamics, and their ``identity'' will somehow be preserved. However, in the quantum world the situation is not so simple. Consider a system of two particles, identical in every aspect (charge, spin orientation, etc.), evolving according to some dynamics. To distinguish these particles, one would need to follow their trajectories individually---but quantum mechanics clearly forbids this, as the particles do not even \emph{have}	 well-defined trajectories. In fact, quantum mechanical particles are fundamentally \emph{indistinguishable}. 

To enforce this indistinguishability, quantum theory is supplemented with the following symmetrization postulate (see any standard textbook on quantum mechanics, such as \cite{LivroCohen}):

\begin{postulate} \label{post:symm}
When a system includes several identical particles, only certain elements of its state space can describe physical states. Physical states are, depending on the nature of the particles, either completely symmetric or completely antisymmetric with respect to permutation of these particles. Those particles for which the physical states are symmetric are called bosons, and those for which they are antisymmetric, fermions.
\end{postulate}

In other words, if $\Psi \left( x_1,\ldots,x_N \right)$ is the wave function describing the state of an $N$-particle system, it must additionally satisfy
\begin{equation} \label{eq:symmet}
\Psi \left( x_1,\ldots,x_i,\ldots,x_j,\ldots,x_N \right) = \pm \Psi \left( x_1,\ldots,x_j,\ldots,x_i,\ldots,x_N \right)
\end{equation}
for all $i$ and $j$. The meaning of \post{symm} is simple: since the particles are indistinguishable, their collective state can only change by a global phase upon permutations of the particles, which has no observable consequence. Furthermore, this global phase is only $\pm 1$\footnote{Actually, this restriction only holds in three dimensions, and naturally applies to any fundamental particle. In two-dimensional systems, however, we may obtain other phases under permutations, giving rise to quasi-particles known as anyons. Anyons are the basis for the very elegant model of topological quantum computation which, unfortunately, is beyond the scope of this thesis.}, with each sign defining a class of particles. A very well-tested empirical rule is the following: every known particle of integer spin, such as photons or He$^4$ nuclei, is a boson, whereas every known particle of half-integer spin, such as electrons and nucleons, is a fermion. This rule is also known as the spin-statistics theorem, as it can be proven to follow from some very general assumptions in quantum field theory---for a discussion on this, see e.g.\ \cite{LivroCohen}. \eqbeg{symmet}, despite its apparent simplicity, also predicts strikingly different behaviors for bosons and fermions.

For fermions, \eq{symmet} leads to the well-known \emph{Pauli exclusion principle}, which states that two fermions may not occupy the same state (this follows trivially from \eq{symmet} by setting, e.g.,\ $x_i = x_j$). Another well-known consequence is that an ensemble of fermions in thermodynamical equilibrium satisfies the \emph{Fermi-Dirac statistics} (hence, the name fermion). This, in turn, is central for understanding the behavior of various physical systems, from the electronic theory of metals to neutron stars.

In contrast to fermions, an ensemble of bosons in thermodynamical equilibrium satisfies the \emph{Bose-Einstein statistics} (again, hence the name boson). The most remarkable difference is that bosons are not limited to one particle per mode---in fact, bosons have a greater tendency to occupy modes ``in groups'' (a behavior known as bosonic bunching). Examples of this bunching behavior are the laser, and Bose-Einstein condensation and the related phenomena of superfluidity and superconductivity. This bunching behavior will also be important on \chap{bosonnew}, where we will report some new experimental and theoretical results concerning the bunching of photons at the output of linear-optical devices.

This classification of particles in bosons and fermions raises a natural question: does the Fermi-Dirac/Bose-Einstein statistics have any consequence for quantum computation? Since quantum computers rely on very precise control of microscopic systems, it is conceivable that the very fermionic/bosonic nature of the particles could help or hinder the experimental efforts. The answer to this question is in fact affirmative---there is a fundamental relation between the computational power of (noninteracting) particles and their statistics. Throughout this thesis we will investigate aspects of both noninteracting fermions, which correspond to a class of computations that cannot outperform classical computers (Chapters \ref{chapter:fermreview} and \ref{chapter:fermnew}), and noninteracting bosons, the behavior of which cannot be simulated classically (in a precise sense to be defined later, see Chapters \ref{chapter:bosonreview} and \ref{chapter:bosonnew}). Before that we must introduce the formalism of second quantization, which will be very convenient for the description of identical particles from a computational point of view.

\subsection{Second quantization} \label{sec:intro_secondquant}

The description of multi-particle states in terms of symmetric/anti-symmetric wave functions, as in \eq{symmet}, is known as first quantization. By contrast, second quantization is a formalism that describes system states by the number of particles that occupy each mode, and will be much more convenient for our purposes. An introduction to second quantization and the equivalence of the formalisms can be found, e.g., in \cite{livroBallentine}.

Consider a set of modes, either bosonic or fermionic, labeled by some index $i=1, \ldots, m$ (that can be discrete or continuous) encoding the set of relevant dynamical properties of the particles: direction of propagation, polarization, spin, frequency, atomic orbital, etc. The basis of the state space (or Fock space) consists of a vacuum state $\ket{~}$, containing no particles, together with all states of the form $\ket{n_1, n_2, n_3, \ldots}$, where each $n_i$ represents the occupation number for mode $i$ and $\sum_i n_i = N$, for all $N \in \mathbb{N}$. For bosonic modes, each $n_i$ can assume any nonnegative integer value, whereas for fermionic modes each $n_i$ can only be 0 or 1, in accordance with Pauli's exclusion principle. 

For fermions, we define the creation and annihilation operators $f_i$ and $f_i^{\dagger}$ by their action on the Fock states:
\begin{subequations} \label{eq:fcrea}
\begin{align} 
f_i^{\dagger} \ket{n_1, n_2, \ldots, 0_i, \ldots} &= \ket{n_1, n_2, \ldots, 1_i, \ldots}, \\ 
f_i \ket{n_1, n_2, \ldots, 1_i, \ldots} &= \ket{n_1, n_2, \ldots, 0_i, \ldots}, \\
f_i \ket{n_1, n_2, \ldots, 0_i, \ldots} &= 0 = f_i^{\dagger} \ket{n_1, n_2, \ldots, 1_i, \ldots}.
\end{align}
\end{subequations}
where $1_i$ and $0_i$ represent mode $i$ being occupied or empty, respectively. The fermionic operators satisfy the anti-commutation relations
\begin{subequations}\label{eq:fcommut}
\begin{align} 
\{ f_i , f_j \} &= 0 = \{ f_i^{\dagger} , f_j^{\dagger} \}, \\
\{ f_i , f_j^{\dagger} \} &= \delta_{ij},
\end{align} \end{subequations}
which are a consequence of the anti-symmetrization required by \post{symm}. From the anti-commutation relations we also obtain that $(f_i^{\dagger})^2=0$, which is nothing more than the Pauli exclusion principle.

The bosonic creation and annihilation operators $b_i$ and $b_i^{\dagger}$ are defined analogously
\begin{subequations} \label{eq:bcrea}
\begin{align} 
b_i^{\dagger} \ket{n_1, n_2, \ldots, n_i, \ldots} &= \sqrt{n_i+1} \ket{n_1, n_2, \ldots, n_i+1, \ldots}, \\ 
b_i \ket{n_1, n_2, \ldots, n_i, \ldots} &= \sqrt{n_i} \ket{n_1, n_2, \ldots, n_i-1, \ldots}. 
\end{align}
\end{subequations}
These bosonic operators satisfy the following commutation relations, which are also a consequence of the symmetrization required by \post{symm}:
\begin{subequations} \begin{align} \label{eq:bcommut}
[ b_i , b_j ] &= 0 =  [ b_i^{\dagger} , b_j^{\dagger} ], \\
 [ b_i , b_j^{\dagger} ] &= \delta_{ij}.
\end{align} \end{subequations}

\subsubsection{Free-particle dynamics} \label{sec:dynamics}

Most of the results presented throughout this thesis relate to the computational power of noninteracting particles. In view of this, we will only describe in detail the formalism of free-particle dynamics, rather than considering the most general case of fermionic and bosonic interactions. In terms of the second quantization formalism this assumption is overly restrictive, of course, but it will provide a more simplified and instructive discussion, while still sufficient for our purposes.

We will also restrict ourselves to discrete-time transformations, where the evolution of the system is not described by the action of some Hamiltonian during a given time, but directly by the action of some unitary operator. This can be done without loss of generality, since every unitary matrix can be written as the complex exponential of some Hermitian matrix---the point is only that we will not consider continuous-time evolution explicitly. This choice will be more convenient both when we study fermions, which will correspond to unitary gates in a quantum circuit, and when we study bosons, where the evolution will be due to discrete linear-optical devices.

Consider now a collection of identical particles evolving according to some (possibly infinite-dimensional) unitary matrix $U_F$ acting on the Fock space. In the Heisenberg representation, an operator $A_{in}$ will evolve into an operator $A_{out}$ according to
\begin{equation*}
A_{out} = U_F A_{in} U_F^{\dagger}.
\end{equation*}
Since $U_F$ is a matrix acting on the Fock space, to define it completely we would need, in principle, to describe its action on every basis element of this space or, equivalently, on every possible monomial of creation and annihilation operators. However, if $U_F$ describes the evolution of noninteracting particles, we can make a major simplification in its description, reducing it to a linear transformation of the modes themselves. The crucial point is that, since the particles do not interact, the evolution must be completely determined \emph{on the single-particle sector of the Fock space}. More explicitly, consider a single particle initially in mode $i$ (whenever the equations are the same for bosons and fermions, we denote annihilation and creation operators generically by $a_i$ and $a_i^{\dagger}$). The final state of the system must be a linear combination of single-particle states or, equivalently,
\begin{equation} \label{eq:Bogoliubov}
U_F a_i^{\dagger} U_F^{\dagger} = \sum_{j=1}^m U_{ij} a_j^{\dagger},  \quad i = {1,2,\ldots, m}
\end{equation}
for some matrix $U$. We made the additional assumption that, besides being linear, the transformation is also particle-number-preserving---in the most general case, the right-hand side of \eq{Bogoliubov} could contain a linear combination of both creation and annihilation operators\footnote{In this case, the evolution would also describe particle creation or absorption by the media, but not interaction between the particles.}. One can also check that $U$ must be unitary to preserve the (fermionic or bosonic) commutation relations. The action of $U_F$ on an arbitrary Fock state can be easily obtained by mapping every particle operator that makes up the state via \eq{Bogoliubov}.

The transformation, as described by \eq{Bogoliubov}, also known as a Bogoliubov transformation \cite{LivroKok}, can arise both as a passive or an active transformation. As a passive transformation, \eq{Bogoliubov} represents simply a change of basis for the modes. It arises, for example, when we shift the description of photonic polarization from linear to circular, or change the orientation of the $Z$ axis of the electronic spin. In this sense, it does not represent a dynamical evolution at all. On the other hand, as an active transformation\footnote{We use the word active to describe the fact that it is a dynamical transformation, much like rotations in classical mechanics, which are denominated passive when they consist of a rotation of the frame of reference, but active when they describe an actual rotation of some physical object. However, in the quantum optics literature the devices that implement transformations of the type of \eq{Bogoliubov}, such as phase shifters and beam splitters, are often called passive optical elements to distinguish them from active elements, such a nonlinear media, that mediate interactions between the particles.} \eq{Bogoliubov} describes, for example, an optical interferometer, where photons can enter any of $m$ input modes and exit in a superposition of the output modes. The correspondence between these two types of transformation only holds in the free-particle setting.

\subsubsection{Elementary two-mode transformations} \label{sec:twomode}

Let us now work out some simple examples of linear transformations that will be important later on. For simplicity, suppose first that there are only two modes, which for now may be bosonic or fermionic---we will give preference to a linear-optical terminology, as this formalism will be needed more explicitly when we discuss linear optics in \chap{bosonreview} and \chap{bosonnew}, but the equations will describe transformations valid for both types of particles. First, consider the following single-mode unitary matrix $U_{PS}$:
\begin{equation*}
U_{PS}(\phi ):= e^{i \phi a_i^{\dagger} a_i}.
\end{equation*}
It induces the transformation
\begin{equation} \label{eq:phaseshifter}
a_i^{\dagger} \rightarrow e^{i \phi a_i^{\dagger} a_i} a_i^{\dagger} e^{-i \phi a_i^{\dagger} a_i} =  e^{i \phi} a_i^{\dagger}.
\end{equation}
That is, $U_{PS}$ describes a phase shifter. Physically, it arises whenever particles in one mode gain a phase relative to the particles in other modes, for example due to a difference in the optical lengths of two paths, or due to a difference in the local magnetic field acting on distant electrons. 

Another important linear transformation is given by 
\begin{equation*}
U_{BS}(\theta):= e^{i \theta \left( a_1^{\dagger} a_2 + a_1 a_2^{\dagger} \right ) }.
\end{equation*}
Its action can be written as
\begin{subequations} \label{eq:beamsplitter} 
 \begin{align} 
a_1^{\dagger} & \rightarrow \cos{\theta} a_1^{\dagger} + i \sin{\theta} a_2^{\dagger}, \\
a_2^{\dagger} & \rightarrow i \sin{\theta} a_1^{\dagger} + \cos{\theta} a_2^{\dagger}.
\end{align} \end{subequations}
The operator $U_{BS}$ describes a beam splitter, and arises as a mechanism allowing a particle to jump between modes. For photons, it may correspond to an actual beam splitter, or a waveplate if $a_1$ and $a_2$ correspond to polarization modes, while for fermions it may correspond to the hopping term between different sites of a lattice. In \chap{bosonnew}, when we use \eq{beamsplitter} to describe the beam splitter transformations, we will refer to $T:=\sin{\theta}$ as the transmission probability (or fraction), and $t:=\sqrt{T}$ as the transmissivity.

It is well-known that the transformations described by \eqs{phaseshifter}{beamsplitter} suffice to construct an arbitrary number-preserving two-mode linear transformation (see, e.g., \cite{Simon1990}). Furthermore, in an $m$-mode system, an arbitrary transformation such as the one described by $U$ in \eq{Bogoliubov} (which we from now on call a multimode interferometer, or an $m$-port) can be decomposed in terms of two-mode elements only \cite{Reck1994}. From an experimental perspective this is a very powerful result, as it allows the construction of arbitrary multimode transformations from network of simpler elements. These results are mentioned here only in passing, as they will be reviewed in more detail when we discuss the computational models associated with fermionic (\chap{fermreview}) and bosonic (\chap{bosonreview}) linear optics.

Note also that both $U_{BS}$ and $U_{PS}$ are generated by Hamiltonians quadratic in the particle operators. This is, in fact, a general feature: unitaries describing the evolution of noninteracting particles precisely correspond to Hamiltonians quadratic in the particle operators. This is true because, for any quadratic Hamiltonian $H$, it is always possible to enact a change of basis (that is, a \emph{passive} Bogoliubov transformation) that takes $H$ into a sum of terms each acting on a single mode---in a sense, this amounts to diagonalizing the $m$-port unitary $U$.

Let us consider now an example of the evolution of a multi-particle state. Suppose, for simplicity, that two particles initially occupy modes $1$ and $2$, and evolve according to some arbitrary $m \times m$ unitary $U$ such as that of \eq{Bogoliubov}. We write the initial state as
\begin{equation*}
\ket{in}=\ket{1 1 0 0 \ldots 0} = a_1^{\dagger} a_2^{\dagger} \ket{~}.
\end{equation*}
By \eq{Bogoliubov}, the final state after the action of $U$ can then be written as
\begin{equation}
\ket{out}= U_F \ket{in} = \sum_{i,j=1}^m U_{1i} U_{2j} a_i^{\dagger} a_j^{\dagger} \ket{~},
\end{equation}
Suppose now we want the amplitude associated with particles exiting in modes $1$ and $2$. At this point we obtain different results for fermions and bosons, so let us first assume that the particles are bosons. In this case, we have
\begin{align} \label{eq:bpermanent1}
\bra{1 1 0 0 \ldots 0} U_F \ket{1 1 0 0 \ldots 0} & = \bra{1 1 0 0 \ldots 0} \sum_{i,j=1}^m U_{1i} U_{2j} b_i^{\dagger} b_j^{\dagger} \ket{~} \notag \\
& = U_{11} U_{22} + U_{12} U_{21}
\end{align}
The two terms that contribute correspond to the two possible combinations of bosons exiting in the output modes, and the plus sign is due to the bosonic commutation relations. Suppose now that the particles are fermions. In this case, we have analogously
\begin{align} \label{eq:fdeterminant1}
\bra{1 1 0 0 \ldots 0} U_F \ket{1 1 0 0 \ldots 0} & = \bra{1 1 0 0 \ldots 0} \sum_{i,j=1}^m U_{1i} U_{2j} f_i^{\dagger} f_j^{\dagger} \ket{~} \notag \\
& = U_{11} U_{22} - U_{12} U_{21}
\end{align}
We obtain a similar result, but with a minus sign due to the fermionic anti-commutation relations. Note that \eq{fdeterminant1} equates the desired amplitude to the \emph{determinant} of a $2 \times 2$ sub-matrix of $U$, whereas \eq{bpermanent1} equates this amplitude to a similar matrix function known as the \emph{permanent}. The determinant and permanent of an $n \times n$ square matrix $A$ are defined, respectively, by
\begin{align}
\textrm{det}(A) & = \sum_{\sigma \in S_n} \textrm{sgn}(\sigma) \prod_{i=1}^{n} A_{i,\sigma_i}, \\
\textrm{perm}(A) & = \sum_{\sigma \in S_n} \prod_{i=1}^{n} A_{i,\sigma_i}.
\end{align}
The sums are taken over all permutations $\sigma$ of the set $\left \{ 1, 2, \ldots, n \right \}$, and sgn$(\sigma)$ is the signature of the permutation, equal to $+1$ if the permutation is even and $-1$ if it is odd. The difference between the definitions of determinant and permanent lies only on the minus signs introduced in the determinant for odd permutations---this resembles very closely the distinction between the commutation and anti-commutation relations in bosons and fermions. This seemingly trivial observation is in fact a particular case of more general rules, which we now state without proof (see e.g.\ \cite{Troyansky1996, Scheel2004, Aaronson2013a}).

Let $\ket{T}=\ket{t_1 t_2 t_3 \ldots t_m}$, with $\sum_i^m t_i = N$, be the initial state of an $N$-particle system. Suppose the system evolves according to some $m$-port unitary $U$, which induces a unitary $U_F$ on the Fock state. Let also $\ket{S}=\ket{s_1 s_2 s_3 \ldots t_m}$, with $\sum_i^m s_i = N$, and let $U_{S,T}$ be the sub-matrix of $U$ obtained by taking $t_i$ copies of the $i^\textrm{th}$ column of $U$ and $s_j$ copies of its $j^\textrm{th}$ row. We then have the following lemmas

\begin{lemma} \label{lem:bosonperm}
If the particles are bosons, the amplitude that they exit in the state $\ket{S}$ is given by
\begin{equation} \label{eq:bpermanent2}
\bra{S}U_F\ket{T}= \frac{\textrm{perm}(U_{S,T})}{\sqrt{s1! \ldots s_m! t_1! \ldots t_m!}}.
\end{equation}
\end{lemma}

\begin{lemma} \label{lem:fermdet}
If the particles are fermions,  the amplitude that they exit in the state $\ket{S}$ is given by
\begin{equation}  \label{eq:fdeterminant2}
\bra{S} U_F \ket{T}= \textrm{det}(U_{S,T}),
\end{equation}
\end{lemma}

Note that, in \lem{fermdet}, each $s_i$ and $t_j$ can only be 0 or 1, so the product of their factorials is simply 1. 

This distinction between bosons and fermions may seem unremarkable, but in fact it has very profound implications for the computational power of these particles, and is the cornerstone of the results presented in this thesis. The permanent and the determinant, although they seem similar, are vastly different with respect to the complexity of their calculation---more specifically, the determinant is easy to compute efficiently in a classical computer, whereas the best-known classical algorithm for the permanent takes exponentially long, and it is not expected to be efficiently computable even on quantum computers. We will return to this distinction in \sec{introduction_d}, when we talk about computational complexity classes. For now, it suffices to say that, even if a collection of particles is noninteracting, the mere fact that they are identical quantum particles and thus obey one statistics or the other may imply an intrinsic difficulty in simulating their behavior.

\subsubsection{The Hong-Ou-Mandel effect} \label{sec:HOMeffect}

We end this section with a simple illustration of \eqs{bpermanent2}{fdeterminant2} of great historical interest. Consider two photons impinging on different ports of a balanced beam splitter (also denominated a 50:50 beam splitter), which is a simple two-mode interferometer described by
\begin{equation*} 
U = \frac{1}{\sqrt{2}} \begin{pmatrix}
1 & i  \\
i & 1
\end{pmatrix}.
\end{equation*}
The amplitude that these photons exit in different output ports is
\begin{equation*}
\bra{1 1} U_F \ket{1 1} = \textrm{perm} \begin{pmatrix}
\frac{1}{\sqrt{2}} & \frac{i}{\sqrt{2}}  \\
\frac{i}{\sqrt{2}} & \frac{1}{\sqrt{2}}
\end{pmatrix} = 0
\end{equation*}
On the other hand, the amplitude that they exit both in the first mode is
\begin{equation*}
\bra{2 0} U_F \ket{1 1} = \frac{1}{\sqrt{2!}} \textrm{perm} \begin{pmatrix}
\frac{1}{\sqrt{2}} & \frac{i}{\sqrt{2}}  \\
\frac{1}{\sqrt{2}} & \frac{i}{\sqrt{2}}
\end{pmatrix} = \frac{i}{\sqrt{2}},
\end{equation*}
with the same value for $\bra{0 2} U_F \ket{1 1}$. Thus, two photons incident on different modes of a balanced beam splitter \emph{always exit together}, and furthermore in equal superposition of each output mode (see \sfig{HOMdip}{a}). This is the well-known HOM (Hong-Ou-Mandel) effect \cite{Hong1987}, a manifestation of the bunching behavior of bosons mentioned in the previous section, and that has several applications in quantum optics.

\begin{figure}[t]
\capstart
\centering
\subfloat[]{\centering \raisebox{0.6in}{\includegraphics[width=0.4\textwidth]{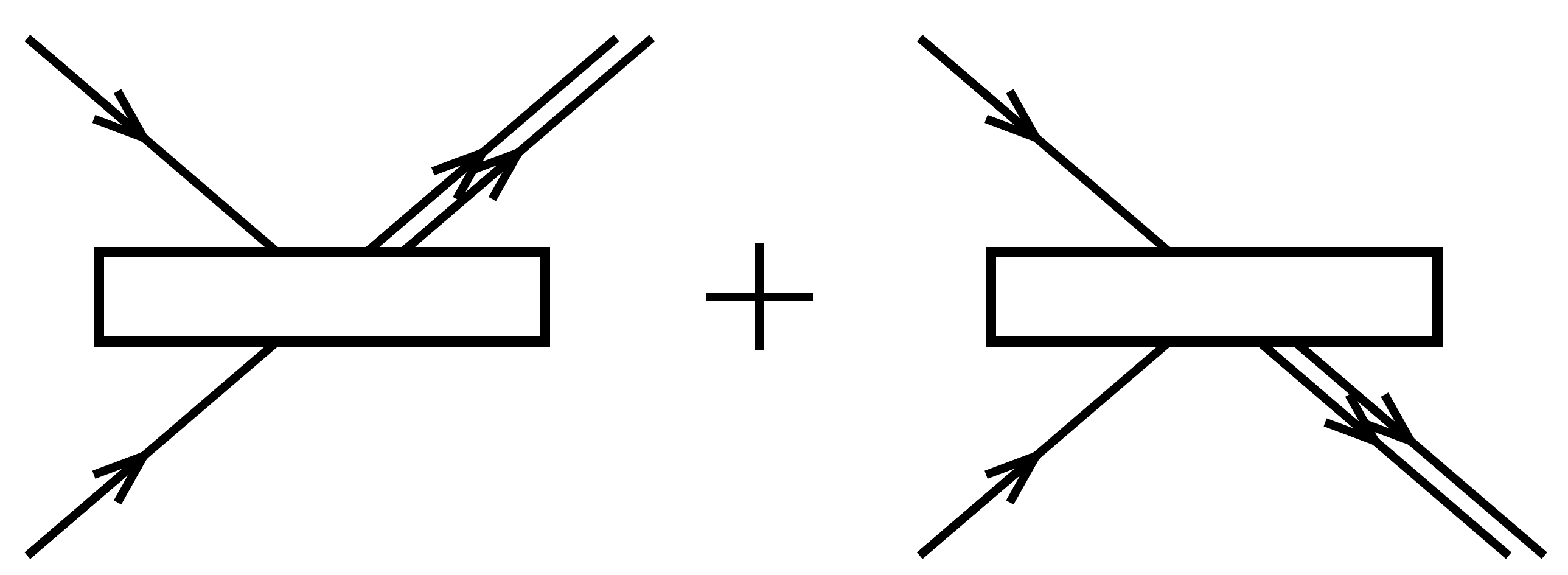}}} \qquad
\subfloat[]{\centering \includegraphics[width=0.5\textwidth]{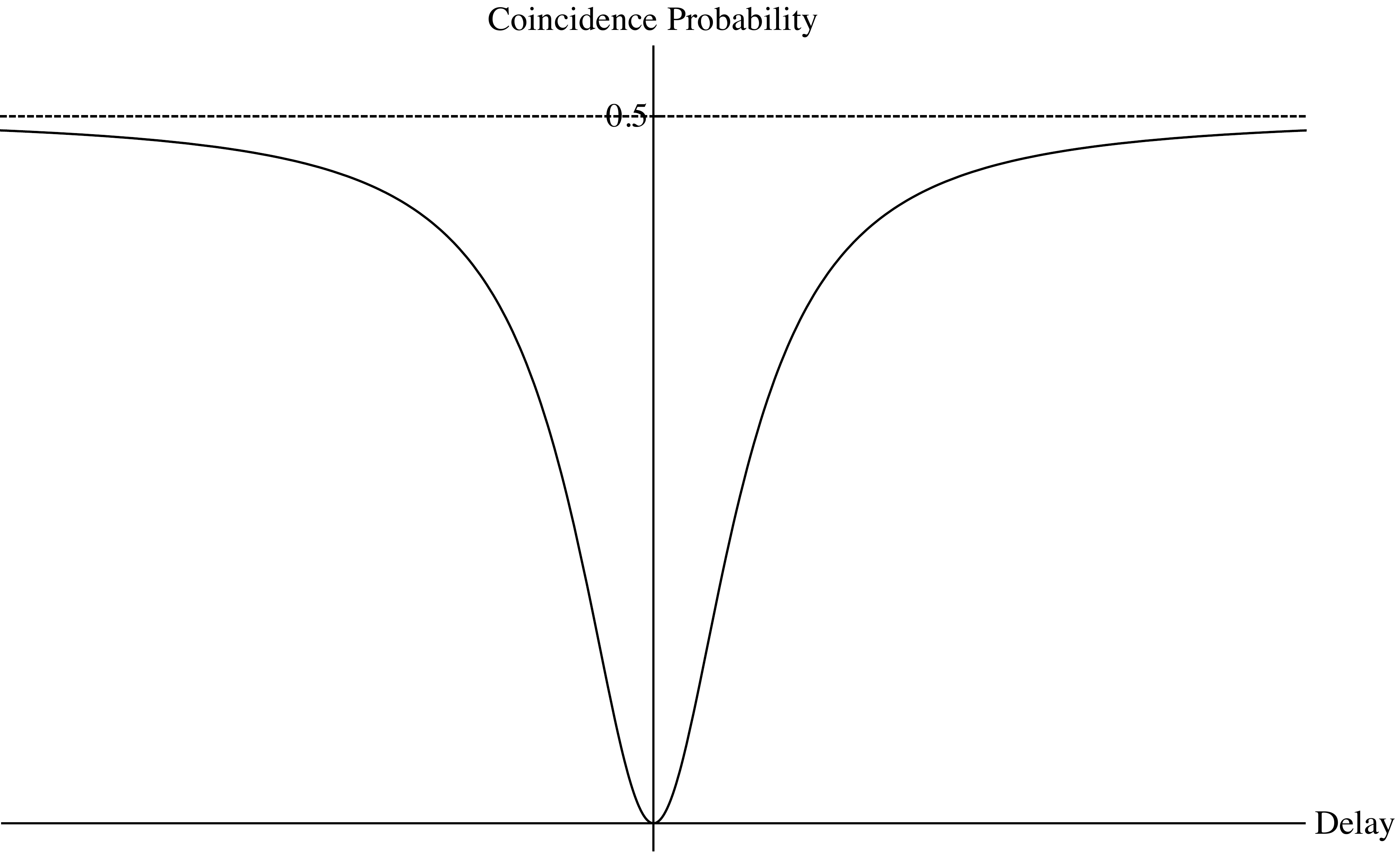}}
\caption[The HOM effect]{The HOM effect. (a) Two photons incident on different ports of a balanced beam splitter exit in equal superposition of states $\ket{20}$ and $\ket{02}$. (b) Typical profile of a HOM curve. By introducing a delay between the photons, it is possible to make them partially distinguishable. The HOM curve interpolates continuously between two extremal regimes: completely indistinguishable (quantum), where the coincidence probability is 0, and completely distinguishable (classical), where the coincidence probability is 1/2.}
\label{fig:HOMdip}
\end{figure} 

One application of the HOM effect that is particularly relevant to the experimental results reported in \chap{bosonnew} is the characterization of single-photon sources. More specifically, we know that perfectly indistinguishable photons always exit the balanced beam splitter together. However, several experimental imperfections (collectively known as mode mismatch) introduce some level of distinguishability---the photon wave packets might not be produced and/or arrive at the detectors at the exact same time, or the apparatus may induce some difference in polarization, frequency, etc. In this case, the fraction of events where photons do not exit together provides a measure of these imperfections. In \sfig{HOMdip}{b} we show the typical profile of a HOM curve, used in this type of characterization.

In \chap{bosonreview} and \chap{bosonnew} we will also consider generalizations of the HOM effect, where the bunching behavior is observed for more than two photons in larger interferometers.

In contrast to the HOM effect, two fermions incident on different modes of a beam splitter must always exit in different modes. This is, of course, nothing more than the Pauli exclusion principle, but can also be seen from \eq{fdeterminant2}. The determinant of any matrix with two identical rows or columns is always 0, while the determinant of the $U$ itself must be a phase, since $U$ is unitary---thus the only allowed output state for fermions is $\ket{11}$, which furthermore does not depend on $U$.

\section{Computational Complexity classes} \label{sec:introduction_d}

So far, we saw some computational tasks which are of interest to us, loosely classifying them as ``easy'' (i.e.\ efficient) or ``hard''.  More formally, we define that a certain task is easy for a computational model (i.e.\ it can be efficiently solved in that model) if there exists a procedure in that model, such as an algorithm, to solve it with only a polynomial amount of resources. Conversely, a task is hard for that computational model if the best possible procedure for solving it demands an exponential amount of resources\footnote{This choice of polynomial \textit{versus} exponential is not unique and, while natural and convenient, is not without criticism. In particular, an algorithm that takes, say, $n^{1000}$ time steps for an $n$-sized input would certainly not be considered efficient \emph{in practice}. However, the discovery of such algorithms is often followed by drastic improvements in the degree of the polynomial, which has given rise to the general belief that every natural problem solvable in polynomial-time is also solvable in a ``reasonably efficient'' manner. For further discussion, see e.g.\ \cite{LivroGoldreich}.}. Intuitively, these definitions capture the distinction between that which may be a matter of developing new theoretical and experimental techniques, and that which is inherently unfeasible. The term ``resource'' has been left intentionally ambiguous since it can be, in principle, anything relevant for a practical realization of that computational model, such as time, space, energy, etc. Most notably, we will be interested in efficiency in terms of time (or number of computational steps of an algorithm) and space (number of bits, qubits, or particles). The classification of the hardness of computational tasks, the relationships between them and different computational models is the subject of computational complexity theory.

We stated that problems can be easy or hard, but that of course is a major over-simplification---in fact, complexity theory is a very rich and diverse field, based partly on the fact that not only there are seemingly hard problems, but that there seem to be several different levels and types of ``hardness'', and tasks are separated in what are called \emph{complexity classes}. Complexity classes group problems by some shared characteristic, commonly restrictions in the definition of the problem or the underlying computational model (we will see several examples shortly). Unfortunately, it is notoriously difficult to prove that a particular problem is hard, as this must take into account all possible procedures for solving it that anyone could ever come up with. Instead, what is often obtained is a \emph{reduction} from one problem to another---that is, some problem $A$ reduces to another problem $B$ if an efficient algorithm for $B$ can be used to efficiently solve $A$. In this way we know that $B$ is at least as hard as $A$, even if we cannot show that either $A$ or $B$ are hard in the first place. This allows us, for example, to reduce the hardness of vast complexity classes to the hardness of a couple of natural problems, and then conjecture these to be hard based on some well-justified intuition. We already saw this implicitly in \sec{introsimul}: we stated that there are classes of quantum circuits believed not to be universal for quantum computation simply because they can be simulated classically, thus conditioning their limitations on the conjectured separation between classical and quantum computers.

We now give brief descriptions of some computational complexity classes. We will restrict our attention to those classes relevant later in this thesis, describing them with the corresponding level of detail. A complete survey of complexity theory, with the proper amount of technical details, can be found e.g.\ in \cite{LivroGoldreich,LivroPapadimitriou}. For a compendium of complexity classes and known relations between them, see also the Complexity Zoo \cite{CompZoo}.

$\bullet$ \textbf{P, BPP, BQP, and NP}

We begin by defining classes P and BPP. Informally, these classes capture precisely the capabilities of classical computers---that is, P (resp.\ BPP) corresponds to those set of problems efficiently solvable by a deterministic (resp.\ probabilistic) classical computer in polynomial time. For now we restrict our attention to \emph{decision} problems, which are problems that, for any input $x$, have a simple yes or no answer. More formally we will define a decision problem as a subset of all strings L $\subseteq \{0,1\}^{*}$ (also known as a language) such that, for every polynomial-sized bit-string $x$,  $x \in$ L if and only if $x$ is an input to the problem for which the outcome is yes. As an example, consider the problem of deciding whether a number is prime: $x$ could be the binary representations of integers, in which case L would be the the set of all $x$ encoding prime numbers. We can then define

\begin{definition}
P (Polynomial-time) is the set of all languages L $\subseteq \{0,1\}^{*}$ for which there exists a uniform family (cf.\ \sec{introuniv}) of polynomial-size deterministic classical circuits $\{C_n\}$ such that, for all inputs $x$, 
\begin{itemize}
\item[(i)] if $x \in$ L, the first output bit is 1;
\item[(ii)] if $x \notin$ L, the first output bit is 0.
\end{itemize}
\end{definition}

\begin{definition}
BPP (Bounded-Error Probabilistic Polynomial-Time) is the set of all languages L $\subseteq \{0,1\}^{*}$ for which there exists a uniform family of polynomial-size probabilistic classical circuits $\{C_n\}$ such that, for all inputs $x$, 
\begin{itemize}
\item[(i)] if $x \in$ L, the first output bit is 1 with probability at least 2/3;
\item[(ii)] if $x \notin$ L, the first output bit is 1 with probability at most 1/3.
\end{itemize}
\end{definition}

The values $1/3$ and $2/3$ are irrelevant---as long as the gap between the yes and no instances is not small, we can efficiently amplify the success probability by repeating the computation and taking a majority vote\footnote{This, of course, relies on the majority vote function itself being efficiently computable. This is the case for P and BPP, but in \chap{fermreview} we will come across a complexity class of sub-classical computational power that cannot implement the majority vote function deterministically, and thus this remark does not hold.}. It is an important open question whether P = BPP, although the general belief it that this is true.

Using the notion of quantum circuit defined in \sec{introuniv}, we can define the quantum analogue of BPP, called BQP, in a similar manner:

\begin{definition}
BQP (Bounded-Error Quantum Polynomial-Time) is the set of all languages L $\subseteq \{0,1\}^{*}$ for which there exists a uniform family of polynomial-size quantum circuits $\{Q_n\}$ such that, for all inputs $x$, after applying $Q_n$ to the state $\ket{0 \ldots 0} \otimes \ket{x}$, 
\begin{itemize}
\item[(i)] if $x \in$ L, the probability of measuring the first qubit in the state $\ket{1}$ is at least 2/3;
\item[(ii)] if $x \notin$ L, the probability of measuring the first qubit in the state $\ket{1}$ is at most 1/3.
\end{itemize}
\end{definition}

Note that this definition of BQP contains both a polynomial-size quantum circuit and a polynomial-size classical circuit, the latter implicit in the definition of the uniform family $\{Q_n\}$. In simpler terms, there must be a classical computer capable of efficiently finding the description of the quantum circuit. Also remark that quantum computers are inherently probabilistic, which is why the definition of BQP is more closely related to that of BPP rather than P. One fundamental difference is that BPP can be defined as a deterministic computer supplied with an additional random input string, whereas for BQP no such alternative definition is possible (i.e.\ its randomness is intrinsic). We can also formalize a remark made at the end of \sec{introsimul}: even if quantum computers cannot be simulated classically in the weakest sense, it is still conceivable that BQP = BPP, that is, all decision problems in BQP (such as factoring) would also be in BPP.

Finally, for completeness, we define the class NP. Intuitively, NP consists of those decision problems which may be hard to solve, but for which a solution can be efficiently checked. To illustrate, consider the problem of integer factoring\footnote{The standard formulation of the factoring problem is not technically a decision problem, but this issue can easily be sidestepped.}: given an arbitrary integer $x$, there is no known (classical) algorithm to efficiently factor it but, if we are given a set of numbers and told they are the factors of $x$, we can efficiently check whether this is true. We can formally define NP in the following manner:

\begin{definition}
NP (Nondeterministic Polynomial-Time) is the set of all languages L $\subseteq \{0,1\}^{*}$ for which there exists a function $V \in$ P (called a verifier) such that, for all inputs $x$, 
\begin{itemize}
\item[(i)] if $x \in$ L, there exists some polynomial-size witness $y$ such that $V(x,y)=1$;
\item[(ii)] if $x \notin$ L, for all witnesses $y$ we have $V(x,y)=0$.
\end{itemize}
\end{definition}

Also important are NP-complete problems: a problem is NP-complete if it is in NP and if it is NP-hard, by which we mean that any other NP problem can be reduced to it. These problems, in a sense, capture the essence of the class NP---not only they are the hardest among the NP problems, but an efficient solution for any one of them would collapse the whole class NP and make it equal to P. It is astounding that there are hundreds of NP-complete problems \cite{LivroGarey}, stemming from various areas of Computer Science, Mathematics, and Physics, and not a single efficient algorithm has been developed for any one of them. This is taken as a strong argument towards the conjecture that P $\neq$ NP, which remains one of the most famous open problems in Mathematics.

Between these four classes, several relations are known or conjectured. It is trivial that P $\subseteq$ NP, P $\subseteq$ BPP and BPP $\subseteq$ BQP. It is also believed that BQP $\subsetneq$ BPP, and arguably the most well-known evidence towards this conjecture is Shor's quantum algorithm for factoring. Factoring is believed to be an NP-intermediate problem---that is, a problem neither in P nor NP-complete---and thus, conditioned on this belief, Shor's algorithm places BQP strictly greater than BPP, but not necessarily as large as NP. As we will see, one of the main motivations for the study of the computational complexity of linear optics (with which half of this thesis is concerned) is that it is based on even weaker assumptions than factoring $\notin$ P. 

$\bullet$ \textbf{PP, $\#$P, and PH}

We now enter the domain of some less usual complexity classes\footnote{Less usual for physicists, anyway, but not anywhere close to the most exotic classes studied by complexity theory.}. In a sense, PP, $\#$P and PH are all different levels of ``incredibly hard'', much beyond what classical or quantum computers are expected to achieve. The reason we are interested in them is mostly proof techniques. Recall that it is very hard to prove that a certain problem cannot be efficiently solved in a classical (or quantum) computer, and often we must resort to reductions. One way to provide evidence that some computational task is hard is to show that, if that task could be performed efficiently, this would have consequences deemed unlikely for the structure of larger complexity classes such as PH. We now proceed to define these classes, although we will omit the formal definitions of $\#$P and PH, as they are too cumbersome and not very enlightening, so we give preference to a conceptual discussion of their relevance to this thesis.

We begin with the following definition:

\begin{definition}
PP (Probabilistic Polynomial-Time) is the set of all languages L $\subseteq \{0,1\}^{*}$ for which there exists a uniform family of polynomial-size probabilistic classical circuits $\{C_n\}$ such that, for all inputs $x$, 
\begin{itemize}
\item[(i)] if $x \in$ L, the first output bit is 1 with probability greater than 1/2;
\item[(ii)] if $x \notin$ L, the first output bit is 1 with probability smaller than 1/2.
\end{itemize}
\end{definition}

Remark the similarity with the definition of BPP. The difference now is that there is no bounded gap between the yes or no cases---for a given problem, if the probabilities for both cases are exponentially close to 1/2, we would need an exponential number of repetitions to amplify this gap and distinguish the correct answer. Thus, it is clear that the class PP does not correspond to any notion of ``feasible''. 

Closely related to PP is the class $\#$P that, rather than decision problems, concerns \emph{counting} problems. Informally, $\#$P is the class of problems associated with counting the number of solutions to an NP problem. As an example, consider the problem of finding a perfect matching\footnote{A perfect matching in a graph G is a set of edges E such that every vertex in G is the endpoint of exactly one edge in E.} in a graph---this problem is in NP, whereas the related problem of \emph{counting} the number of perfect matchings in a graph is in $\#$P. Curiously, even though finding a perfect matching is actually easy (in P), its counting version is $\#$P-complete, and thus among the hardest problems in $\#$P\footnote{In fact, although this problem is $\#$P-complete in general \cite{LivroPapadimitriou}, it is in P if the graph is planar \cite{Kasteleyn1961,Temperley1961}.}. 

Another very important $\#$P-complete problem is exactly calculating the permanent of a $\{0,1\}$-matrix \cite{Valiant1979}. Recall that the permanent function already appeared in the previous section, in the calculation of transition amplitudes in photons in a linear-optical network. Its $\#$P-hardness is the basis of the results of the computational complexity of linear optics---although there are several subtleties, most notably that \emph{linear optical devices cannot solve $\#$P-complete problems}, which we return to in \chap{bosonreview}.

Finally, we define the polynomial hierarchy PH, which is a generalization of NP. Recall that NP was the class of problems of the form: ``given an input $x$, does there exist $y$ such that $V(x,y)=1$?''. In this spirit, PH can be defined as the following hierarchy of classes: for each natural number $k$, the $k^{\textrm{th}}$ level of PH, denoted $\Sigma_k$, corresponds to problems of the form ``given an input $x$, does there exist $y_1$ such that for all $y_2$, there exists a $y_3$, such that \ldots $y_k$ such that $V(x,y_1,y_2,y_3,\ldots,y_k)=1$?''. As a particular case, $\Sigma_1$=NP. 

The polynomial hierarchy is strongly believed to be infinite, so much so that several hardness results for specific problems are of the form ``if such and such problem were in P, the polynomial hierarchy would collapse to a certain level'', including the ones we review in \chap{bosonreview}. Informally, it would be very surprising if all problems described by any number of existential and universal quantifiers could be reduced to problems with only a few quantifiers. However, the belief that PH is infinite is not as strong as the belief that P $\neq$ NP. More specifically, even if P$\neq$NP, it is still possible for the polynomial hierarchy to collapse to some upper level.

The main result we will need relating these classes is Toda's Theorem \cite{Toda1989}, which states that PH is contained in P with a PP oracle or, equivalently, P with a $\#$P oracle. An oracle, also called a ``black box'', is an abstract device capable of solving some computational problem in a single step, and is an extremely useful concept in complexity theory. When we say ``A with an oracle for B'', which is denoted $\textrm{A}^\textrm{B}$, we mean a device that can solve problems in A with the added capability of making queries to a black box device that instantly solves problems in B\footnote{Among several other applications, oracles formalize the idea of reduction. For instance, NP-hardness, which we defined previously, can be given the following alternative characterization: a problem L is NP-hard if NP$\subseteq \textrm{P}^\textrm{L}$.}. Since $\#$P consists essentially of counting problems, it may seem strange to formally compare its computational power with other classes based on decision problems. Thus, for our purposes, $\textrm{P}^\textrm{\#P}$ merely formalizes this comparison: it consists of decision problems that could be solved efficiently if we had a device for solving counting (i.e.\ $\#$P) problems. In this sense, Toda's theorem informally states that the computational power of PH is smaller than that of PP or $\#$P. In the light of this statement, the difference between fermions and bosons described in \sec{intro_secondquant} is remarkable: while fermions evolve according to the determinant, which is in P, bosons evolve according to the permanent, which is $\#$P-complete and thus expected to be a drastically harder problem.

$\bullet$ \textbf{Post-selected classes}

Finally, we define a model of computation based on an operation known as post-selection. In essence, post-selection consists on the ability to condition a probabilistic computation on the outcome of a subset of the output register, no matter how unlikely that outcome may be. To illustrate, consider there is some number $m$ we wish to factor, and we apply the following naive probabilistic algorithm: we sample two random numbers $a$ and $b \in [1,\sqrt{m}]$, and check whether $a.b=m$. If they are, we assign one output bit $p$ to $1$ to flag that we found the correct answer, otherwise we assign it to $0$ and try again. This algorithm is obviously very inefficient, as it must repeat an exponential number of times, on average, until the right answer happens to appear. But suppose now that we had the astonishing ability of, in one computational step, parse through all random possibilities of $a$ and $b$ and pick one for which $p$ is 1. Informally, that is the power of post-selection.

For now, we will define classes postBQP and postBPP, corresponding respectively to quantum and randomized classical computers with post-selection, but other post-selected classes will arise in \chap{bosonreview} and \chap{bosonnew}. For these definitions, it is convenient to define the output of the circuit as consisting of a single-line output register ($o$), which encodes the answer to the decision problem, and a poly$(n)$-sized post-selection register ($p$), on which the success of the circuit will be conditioned\footnote{Often, the post-selection register is defined with only one bit \cite{Aaronson2005}. For postBPP and postBQP this can be done without loss of generality, since at the end of the circuit we can simply append a few extra operations to encode whether $p$ is the desired bit string or not on a single bit. However, for some of the classes to be defined later this is not always possible.}. We then define

\begin{definition}
postBPP (BPP with post-selection, also known as BPP$_\textrm{PATH}$) is the set of all languages L $\subseteq \{0,1\}^{*}$ for which there exists a uniform family of polynomial-size probabilistic classical circuits $\{C_n\}$ such that, for all inputs $x$
\begin{itemize}
\item[(i)] The probability of $p=00\ldots0$ is nonzero;
\item[(ii)] if $x \in$ L then, conditioned on $p=00\ldots 0$, the probability of  $o=1$ is at least 2/3;
\item[(iii)] if $x \notin$ L then, conditioned on $p=00\ldots 0$, the probability of  $o=1$ is at most 1/3.
\end{itemize}
\end{definition}

\begin{definition} \label{def:postBQP}
postBQP (BQP with post-selection) is the set of all languages L $\subseteq \{0,1\}^{*}$ for which there exists a uniform family of polynomial-size quantum circuits $\{Q_n\}$ such that, for all inputs $x$, after applying $Q_n$ to the state $\ket{0 \ldots 0} \otimes \ket{x}$, 
\begin{itemize}
\item[(i)] The probability of measuring $p$ in the state $\ket{00\ldots 0}$ is nonzero;
\item[(ii)] if $x \in$ L then, conditioned on measuring $p$ on state $\ket{00\ldots 0}$, the probability measuring $o$ on state $\ket{1}$ is at least 2/3;
\item[(iii)] if $x \notin$ L then, conditioned on measuring $p$ on state $\ket{00\ldots 0}$, the probability measuring $o$ on state $\ket{1}$ is at most 1/3.
\end{itemize}
\end{definition}

Clearly, post-selection is a very powerful (and unrealistic) resource, since it allows us to single out exponentially-unlikely outcomes. It is easy to see, for example, that postBPP contains NP: similar to the algorithm we described for factoring, for any NP problem we can just post-select on all possible random strings as witnesses for the solution. However, this relies on the property of NP problems of having a witness to begin with: it is not so trivial whether other problems in, say, higher levels of PH, could be contained in postBPP. Nonetheless, it is known that postBPP is contained \emph{within} the third level of the polynomial hierarchy \cite{Han1997}. PostBQP, on the other hand, was shown to be equal to PP \cite{Aaronson2005}. Now recall that PP is considered a larger class than the whole polynomial hierarchy. This shows that, from the point of view of post-selection, classical computers have a very modest power compared to quantum computers and, if postBPP=postBQP, PH collapses to its third level. It is tempting to think that BPP=BQP also implies a collapse of the polynomial hierarchy, but in fact this is not true. It would be true if BPP were equal to BQP due to some simulation scheme as defined in \sec{introsimul}: if probabilistic classical computers could efficiently sample from the same distributions as quantum computers we could post-select on obtaining the same results, and thus postBPP would be equal to postBQP. But, as remarked previously, if BPP is equal to BQP simply because they happen to contain the same set of decision problems, this reasoning no longer applies.

As we will review in the next section, the hardness of simulating some restricted classes of quantum computers, such as constant-depth quantum circuits or linear optics, follows precisely from this conjectured separation between postBPP and postBQP. 

\section{Restricted models of computation} \label{sec:introduction_c}

The results reported in this thesis fit into the larger program of understanding the gap between classical and quantum computing. Rather than searching for quantum algorithms to efficiently solve some important problem, our goal is to better characterize the resources needed to build a quantum computer in the first place. In other words, our focus is not on characterizing BQP itself, which is of course a different but extremely important research program, but rather on understanding what resources bridge the gap between BPP and BQP in different contexts, and whether there are intermediate computational classes. To that end, we must investigate the computational capabilities of \emph{restricted models} of quantum computation.

Consider a model of computation Q, consisting of quantum computers subject to some restriction (motivated e.g.\ by some real-world physical implementation). For the sake of discussion, let us classify our knowledge about Q into one of three (oversimplified) scenarios:

\begin{itemize}
\item[(i)] We can show that Q is classically simulable (cf.\ \sec{introsimul}),
\item[(ii)] we can show that Q is universal for quantum computation (cf.\ \sec{introuniv}), or
\item[(iii)] we can show that the ability to classically simulate Q contradicts some plausible conjecture (cf.\ \sec{introduction_d}), thus providing evidence that Q has some supra-classical computational power. 
\end{itemize}

Suppose that we modify the underlying restrictions defining Q to obtain another model, Q', and ask what can we prove about the computational power of Q'. For example, Q may be in regime (i), and Q', obtained by adding some missing ingredient to Q, may be in regime (ii) or (iii), in which case that ingredient is somehow fundamental for the computational power of Q'. On the other hand, if Q is in regime (ii) and Q', obtained by imposing additional restriction on Q, is also in regime (ii), we conclude that this restriction is irrelevant for the computational power of Q. Let us give some examples to illustrate these concepts.

The jump from (i) to (ii) is the most sought-after, as it provides direct practical applications for Q', such as Shor's algorithm. For example, it is known that a circuit of single-qubit gates, acting on a separable input state and followed by single-qubit measurements is classically simulable, but by the addition of any entangling two-qubit gate we obtain quantum universality \cite{Bremner2002}, thus suggesting that entanglement is somehow a necessary resource for (pure-state) quantum computation\footnote{If the global state of the quantum computer is mixed, however, this does not hold \cite{Knill1998}.}. This intuition is strengthened by the result that, for any quantum circuit acting on pure states, if the entanglement at every intermediate step is bounded (in the sense that only a small number of qubits is entangled) the output is classically simulable \cite{Jozsa2003}. Another interesting example is that of circuits composed only of the Toffoli gate. Not only are these circuits classically simulable, they are in fact universal for classical computation \cite{Fredkin1982}, whereas circuits composed of Toffoli and any non-basis-preserving single-qubit gate \cite{Shi2003} are universal for quantum computation. In this case, clearly what the Toffoli gate is missing is the ability to create quantum superpositions.

However, the claim that a certain resource is essential for quantum computation must also be made with care. Consider the family of well-known Clifford gates\footnote{The Clifford group consists of the unitaries that preserve the Pauli group under conjugation.}. A computation consisting of a computational basis input state, a circuit of Clifford gates and a final computational basis measurement may display unbounded entanglement throughout, but it is nonetheless classically simulable by the Gottesman-Knill theorem \cite{Gottesman1999a}. Despite the presence of entanglement, there must be some other resource missing for Clifford circuits. In fact, Clifford gates become universal when supplemented by non-Clifford gates \cite{Shi2003} or special input states \cite{Shor1996,Gottesman1999a}. A recent collection of results on the classical simulability of Clifford circuits under the addition of different ingredients can be found in \cite{Jozsa2014}. Another well-known class of quantum computations that is classically simulable despite the large amounts of entanglement generated is that of nearest-neighbor matchgates (acting on qubits arranged on a path). Since matchgates are the topic of \chap{fermreview} and \chap{fermnew}, we defer the definition and discussion of this model to those chapters. For now it suffices to say that, curiously enough, matchgates become universal when supplemented by the $\swap$ gate, which does not seem a particularly ``quantum'' resource---in \chap{fermnew} we will identify the property that makes the $\swap$ gate special, and prove several other scenarios where matchgates become universal.

While (ii) is the regime of most general interest, there are also several results along the lines of (iii). More specifically, in this thesis we will be interested in four restricted models in this regime:

\begin{itemize}
\item[(a)] Depth-4 quantum circuits \cite{Terhal2004}: This model consists of a computational basis input, followed by three rounds of arbitrary two-qubit gates acting on arbitrary pairs of qubits, and ending in a measurement of a polynomial number of qubits in the computational basis. Curiously enough, if we restrict this model further to only two rounds of two-qubit gates, it becomes classically simulable.
\item[(b)] Circuits of commuting gates \cite{Bremner2011}: Also called IQP (for Instantaneous Quantum Polynomial-time\footnote{This is a joke from the authors, of course. The point is that all gates commute and thus can be done in any order, not at the same time.}), this model consists of a computational basis input, followed by a circuit of gates diagonal in the $\{ \ket{+},\ket{-}\}$ basis, with a final measurement of a polynomial number of qubits in the computational basis. 
\item[(c)] Non-adaptive measurement-based quantum computation \cite{Hoban2013}: As the name suggests, consists of a subset of measurement-based quantum computations where there is no adaptation in the measurements, and so they can all be performed at once.
\item[(d)] Non-adaptive linear optics, or BosonSampling \cite{Aaronson2013a}: The main topic of \chap{bosonreview} and \chap{bosonnew}. Consists of preparing an $n$-photon Fock state, evolving it through some $m$-mode interferometer, and doing a final round of number-resolving measurements. The qualifier ``non-adaptive'' distinguishes it from the KLM scheme \cite{Knill2001b}, that uses adaptive measurement to perform universal linear optical quantum computation.
\end{itemize}

For all of these models, the following is true: if an efficient weak classical simulation of its output is possible, the polynomial hierarchy collapses to its third level. Furthermore, these proofs all follow the same recipe. The starting point is some trick to show that the model is capable of arbitrary quantum computations if post-selection is allowed. As a consequence, the post-selected version of the model contains BQP, and trivially also postBQP, as nothing is gained by more than one round of post-selection. Finally, by the discussion at the end of the previous section, it is clear that if we could weakly simulate the model by classical means we could solve the same set of problems simply by post-selecting on the corresponding outcomes of the classical simulator. This would imply that the model is contained within postBPP, which in turn would imply postBPP = postBQP and the collapse of the polynomial hierarchy. For a proof that this remains true even if the classical simulation is approximate within multiplicative error, see \cite{Bremner2011}. Notably, the complexity classes defined by these devices comprise only sampling problems, rather than decision problems. Furthermore, an additional limitation  is that the sampling must be done over a polynomial number of output qubits---it can be shown, at least for models (a) and (b), that the distribution over any subset of the output comprising only a logarithmic number of qubits can in fact be simulated classically\cite{Terhal2004, Bremner2011}.

This hand-waving argument outlines how the hardness-of-simulation of these models reduces to the conjectured separation between postBPP and postBQP (and conjectured infinitude of PH). We omit further details for now since, in \chap{bosonnew}, we will construct a class of quantum circuits that is, at the same time, a subset of (a), (b) and (c), and use it to prove that the aforementioned hardness reduction also follows through for \emph{constant-depth linear optics}.

One drawback of these proofs is that they do not provide a practical application for these devices, such as solving some important problem (say, factoring), and the only nontrivial computational task they can perform a priori is \emph{simulating themselves}. In this respect, Shor's algorithm has two notable advantages: first, factoring is in NP, meaning that, when we have sufficient technology to build a full-blown quantum computer, we will be able to \emph{verify} whether it is actually running Shor's algorithm, whereas the tasks in reach of these restricted models are not known to be in NP, making their validation highly nontrivial. The second is that factoring is an inherently useful task, given its cryptographic applications, and this alone can drive the experimental efforts (and funding) to build an universal, fault-tolerant quantum computer\footnote{Unless someone eventually develops an efficient classical factoring algorithm, of course.}.

However, despite this, restricted models such as (a)-(d) have been drawing increasing interest, for two main reasons. The first is that results concerning their complexity are based on milder assumptions than ``factoring $\notin$ P''---in other words, even if a classical algorithm for factoring is found, these restricted models will still provide evidence that quantum devices can perform some nontrivial computational task. The second reason is that, precisely because they are restricted quantum computers, they might require less stringent experimental control of some physical system, providing an intermediate milestone for experimentalists that may be feasible in a much nearer future. 

In terms of simplifying experimental efforts, BosonSampling (d) is ahead in the race. The post-selection argument given above provides evidence that linear-optical devices cannot be weakly classically simulated with multiplicative error, but in fact most of the work done in \cite{Aaronson2013a} was to prove a stronger result: classical simulability remains unlikely even for approximation \emph{close in total variation distance} (recall that this is a weaker requirement), thus more faithfully describing real-world experimental devices. Crucially, what separates linear optics from the other restricted models (a)-(c) is that there is a second way to prove its hardness-of-simulation reduction, based on the $\#$P-completeness of approximating the permanent. This provides a more robust result, but also leads to a direct prescription of an experiment expected to be hard to simulate classically\footnote{At the cost of some extra assumptions, most notably a conjecture that the permanent of a random matrix is, typically, as hard to compute as the permanent of an arbitrary matrix. We will return to these technicalities in \chap{bosonreview}.}, especially since the model itself is inspired on a physical system (noninteracting bosons). Indeed, this result sparked a flurry of experimental interest that culminated on four quantum optics groups reporting, within two days of each other, demonstrations of small-scale implementations of BosonSampling devices on integrated photonic chips---one of these results is reported in \chap{bosonnew}, as it is part of an ongoing collaboration between us and quantum optics groups in Rome and Milan.

To summarize, we see that there is an inherent richness in the study of restricted models of quantum computation, with different sets of restrictions highlighting the role of different resources. In \chap{fermnew} we will report our new results in the context of matchgates, where we show that several resources have the capability of uplifting matchgates from classically simulable to quantum universal and, furthermore, this jump will be abrupt in every case. In \chap{bosonnew}, we report several theoretical and experimental results related to linear optics and BosonSampling.
\newpage
\chapter{Review: Matchgates} \label{chapter:fermreview}

In this chapter, I will review some previously known results regarding the computational power of matchgates. In \sec{ferm_overview} I give a brief historical overview on matchgates, their connection to noninteracting fermions, and various aspects of their computational power. In the subsequent sections, I revisit some of these results in more detail, with special focus on those definitions and proofs that will be important for our main results in \chap{fermnew}. 

\section{Definitions and historical background} \label{sec:ferm_overview}

Matchgates are a restricted class of quantum operations originally defined by Valiant \cite{Valiant2002} in graph-theoretical terms, and shown to be closely related to systems of noninteracting fermions  \cite{Terhal2002}. We define matchgates as follows:

\begin{definition} \label{def:matchgates}
Let $G(A,B)$ denote the unitary gate that acts as unitaries $A$ and $B$, respectively, on the even- and odd-parity subspaces of a 2-qubit Hilbert space:
\begin{equation} \label{eq:Matchgate}
G(A,B) = \begin{pmatrix}
A_{11} & 0 & 0 & A_{12} \\
0 & B_{11} & B_{12} & 0 \\
0 & B_{21} & B_{22} & 0 \\
A_{21} & 0 & 0 & A_{22}
\end{pmatrix}.
\end{equation}
The gate $G(A,B)$ is a \emph{matchgate} if $\det A = \det B$. 
\end{definition}

Throughout this chapter we will cover known results pertaining to the computational properties of matchgates in several different contexts. Often, in the literature, the definition of matchgate is more general (e.g., by including additional non-unitary gates such as in \cite{Valiant2002}), or more restrictive (e.g., by restricting the gates to only act on nearest-neighboring qubits, such as in \cite{Terhal2002}) than the one we give here, depending on what the author intends to investigate. Since most of the results presented in \chap{fermnew} are based on variations of the underlying assumptions, it will be convenient to take \defn{matchgates} as a working definition and add or remove further restrictions as we go. Hopefully, no confusion will arise with other definitions found elsewhere in the literature. Additionally, throughout this chapter it should be implicit that the qubits in a circuit are arranged on a path (i.e., a one-dimensional array with open boundary conditions), as this is the standard setting for the results which we will revisit here. Variations of this setting will be considered in \chap{fermnew}.

Matchgate circuits were originally introduced by Valiant in \cite{Valiant2002} as a class of restricted classically simulable quantum computations, derived from the problem of counting the number of perfect matchings (recall the definition of perfect matching from \sec{introduction_d}) in planar graphs. We mention Valiant's result in passing for historical interest, but our approach will follow that of subsequent work, and we will not concern ourselves with the technical aspects of the original definition. It suffices to say that the class of computations defined by Valiant corresponds to circuits of matchgates (as per our definition) acting on nearest-neighboring qubits, together with some non-unitary gates and a broader set of two-qubit gates on the first two qubits. From hereon we consider only the restricted class of circuits composed of unitary matchgates as defined by \eq{Matchgate}.

Valiant's result was soon after reinterpreted, by Terhal and DiVincenzo \cite{Terhal2002} and Knill \cite{Knill2001a}, in terms of the evolution of systems of noninteracting fermions. More specifically, the authors use the Jordan-Wigner transformation \cite{Jordan1928} to show that the Hamiltonians that generate the group of nearest-neighbor matchgates correspond precisely to Hamiltonians quadratic in fermionic creation and annihilation operators. Recall from \sec{dynamics} that such quadratic Hamiltonians describe the evolution of noninteracting fermions---as this correspondence between quadratic Hamiltonians and noninteracting particles is analogous for bosons (i.e.\ linear optics), matchgates are occasionally referred to by the expression ``fermionic linear optics''. However, in terms of computational power, the parallel between fermionic linear optics and its bosonic counterpart breaks down when nontrivial measurements are allowed. As shown in \cite{Knill2001a}, the action of particle detectors can be described as the limit of (non-unitary) linear optical operators for fermions, but not for bosons. Hence, noninteracting bosons become universal for quantum computation when augmented with adaptive measurements via the KLM scheme \cite{Knill2001b} (incidentally, that is topic of \chap{bosonreview}), whereas fermions remain simulable, as shown by \cite{Terhal2002}, becoming universal only with nondestructive two-qubit measurements \cite{Beenakker2004}.

We may also study matchgates as the subgroup of two-qubit unitaries $G(A,B)$ in its own right, with many results obtained from their algebraic properties with no explicit reference to their fermionic nature. For instance, it was shown by Jozsa and collaborators in \cite{Jozsa2010} that circuits of nearest-neighbor matchgates on $n$ qubits are equivalent to general quantum circuits on $O(\log n)$ qubits. It was also shown by Van den Nest that the class of Boolean functions computable by circuits of nearest-neighbor matchgates corresponds to that of so-called linear threshold gates \cite{Nest2011a}. These results connect the computational power of nearest-neighbor matchgates to that of other computational classes, irrespective of their underlying fermionic nature. This approach, which we also take here, focuses on the algebraic structure of the matchgate group, and allows us to investigate its properties beyond the fermionic formalism.

One can also consider deviations from the original setting by relaxing some of the (occasionally implicit) restrictions, and ask whether the computational power of the system changes. For instance, it was shown that matchgates become universal for quantum computation if allowed to act on both nearest and next-nearest neighbors---or, equivalently, by allowing a modest use of the $\swap$ gate (remark that the $\swap$ gate, while being of the form $G(I,X)$, is not a matchgate since det$I \neq$ det$X$). This was first explicitly stated in the formalism of matchgates by Jozsa and Miyake in \cite{Jozsa2008b}, although it was already implicitly known in the context of the XY interaction \cite{Kempe2002}. Jozsa and Miyake also extended the original simulability result by showing that a circuit of matchgates remains simulable if conjugated by a suitable circuit of Clifford gates. These kinds of results often push the boundaries of our understanding of matchgates in directions which do not arise naturally in the fermionic formalism. All contributions of our work to the understanding of matchgates are of this sort, and are described in \chap{fermnew}. 

Finally, note that we can also investigate the computational power of proper subsets of matchgates. Any such subset acting only on nearest neighbors must be classically simulable, as implied by the simulability of matchgates themselves. However, in those regimes where matchgates become universal (say, when allowed to act on more distant neighbors) we may ask whether proper subsets of matchgates retain this computational power. The most notable of these subsets is the one defined by the XY (or Heisenberg anisotropic) interaction. This interaction is an idealized model of the interactions present in several proposed physical implementations of quantum computing, such as quantum dots \cite{Imamoglu1999, Quiroga1999}, atoms in cavities \cite{Zheng2000}, and quantum Hall systems \cite{Mozyrsky2001}. The XY interaction was shown to be as powerful as general matchgates\footnote{In fact, this result preceded the one for general matchgates.} if allowed to act on both nearest and next-nearest neighbors \cite{Kempe2002}. In \chap{fermnew} we will show that several of our results for matchgates also follow through for the XY interaction.

For the remainder of this chapter I will review the proofs for some of these results in detail, with special emphasis on those techniques that will be useful for our purposes in \chap{fermnew}. In \sec{fermreview_a}, I review the proof of the simulability of nearest-neighbor matchgates, following the formalism of Jozsa and Miyake in \cite{Jozsa2008b}. In \sec{fermreview_b}, I review the proof of the universality of matchgates when complemented with the $\swap$ gate, also showing how this translates to matchgates acting on more distant qubits. In that section I also review what is known about the XY interaction in this context. Finally, in \sec{fermreview_c} I provide, for completeness, a brief outline of the relations between matchgates and other quantum/classical computational complexity classes. Throughout this chapter, by ``matchgates'' we will mean nearest-neighbor matchgates, unless stated otherwise, and we reiterate the assumption that the qubits in the circuit are arranged on a path (in other words, each qubit has at most two neighbors). 

\section{Simulability of nearest-neighbor matchgates} \label{sec:fermreview_a}

The starting point of this section is the Jordan-Wigner transformation \cite{Jordan1928}. It is a result of great historical interest, that originated as a mapping between fermionic operators and spin operators, and which often helps in obtaining solutions for solid-state models, such as the 1D Ising or XY spin-chains. For our purposes, the role of spins will be taken by qubits, and the spin operators will be replaced by logical unitary gates (more specifically, matchgates).

We begin by defining the Jordan-Wigner operators \cite{Jordan1928} acting on $n$ qubits:
\begin{subequations} \label{eq:JWc}
\begin{align} 
c_{2j-1} &:= \left( \prod_{i=1}^{j-1} Z_i \right) X_j \notag \\
c_{2j} &:= \left( \prod_{i=1}^{j-1} Z_i \right) Y_j
\end{align}
\end{subequations}
for $j \in \{1,\ldots,n\}$. From the commutation relations of the Pauli operators, one can obtain the (anti-)commutation relations for the $c_{i}$ operators:
\begin{equation} \label{eq:commutc}
\left \{ c_i, c_j \right \} = 2 \delta_{i j} I, \qquad \qquad i,j \in \{1,\ldots,2n\}
\end{equation}
To see how these operators relate to fermions, it's helpful to rewrite them in the following combinations:
\begin{align*}
f_{k} &:= (c_{2k-1}+i c_{2k})/2 \notag \\
f_{k}^{\dagger} &:= (c_{2k-1}-i c_{2k})/2. 
\end{align*}
Using \eq{commutc} one obtains the commutation relations for the $a_i$'s:
\begin{align*}
\{ f_i, f_j \} & = 0 = \{ f_i^{\dagger}, f_j^{\dagger} \}, \qquad \textrm{and} \notag \\
\{ f_i, f_j^{\dagger} \} & = \delta_{i j},
\end{align*}
with $i,j \in \{1,\ldots,n\}$. The above relations are precisely what one expects for fermionic creation and annihilation operators [cf.\ \eq{fcommut} and \sec{introduction_b}]. Another way to understand this correspondence is to notice that the Pauli $X$ gate is just a bit-flip (i.e.\ exchanges states $\ket{0}$ and $\ket{1}$)---if we interpret the qubit states $\ket{0}$ and $\ket{1}$ as occupation numbers for a fermionic mode\footnote{Since a fermionic mode cannot be occupied by more than one particle, identification between qubit states and fermionic modes is well-defined.}, the $X$ gate just creates a fermion if the mode is empty and annihilates a fermion if it is occupied, and we expect it to be a simple linear combination of $a$ and $a^{\dagger}$. This is the content of \eq{JWc}, with the caveat that we have to include some $Z$ gates to obtain the correct commutation relations between operators on different fermionic modes.

Now consider the following products of the $c_i$ operators:
\begin{align}
c_{2k-1} c_{2k} & = i Z_k \label{eq:JWquad1}
\end{align}
for $k \in \{1,\ldots,n\}$ and
\begin{subequations} \label{eq:JWquad2}
\begin{align}
c_{2k} c_{2k+1} & = i X_k X_{k+1}  \\
c_{2k-1} c_{2k+2} & = -i Y_k Y_{k+1} \\
c_{2k-1} c_{2k+1} & = -i Y_k X_{k+1} \\
c_{2k} c_{2k+2} & = i X_k Y_{k+1} 
\end{align}
\end{subequations}
for $k \in \{1,\ldots,n-1\}$. Although each $c_i$ operator can correspond to Pauli operators acting on any number of qubits, in these particular quadratic combinations the $Z$ gates present in \eq{JWc} cancel out, and we are left with two-qubit Hamiltonians between nearest neighbors. Also recall from \sec{introduction_b} that Hamiltonians quadratic in the fermionic operators describe the evolution of noninteracting fermions, as we can always apply a trivial transformation to take such a Hamiltonian into a sum of terms acting each on a single fermionic mode.

Now consider the unitary operators generated by these quadratic Hamiltonians---since we know these have to be nearest-neighbor two-qubit gates, let us first assume that there are only two qubits and drop the qubit labels. Then, by explicit exponentiation, we obtain the following unitary matrices:

\begin{subequations} \label{eq:JWtoMG}
\begin{align}
e^{i (a X \otimes X + b Y \otimes Y)} & = \left(\begin{smallmatrix}
\cos (a-b)  & 0 & 0 & i \sin{(a-b)} \\
0 & \cos (a+b)  & i \sin{(a+b)} & 0 \\
0 & i \sin{(a+b)} & \cos{(a+b)}  & 0 \\
i \sin{(a-b)} & 0 & 0 & \cos (a-b)
\end{smallmatrix}\right) = G[R_x(a-b),R_x(a+b)],  \qquad \\
e^{i (c X \otimes Y + d Y \otimes X)} & =  \left(\begin{smallmatrix}
\cos (c+d)  & 0 & 0 & \sin{(c+d)} \\
0 & \cos (c-d)  & -\sin{(c-d)} & 0 \\
0 & \sin{(c-d)} & \cos{(c-d)}  & 0 \\
-\sin{(c+d)} & 0 & 0 & \cos (c+d)
\end{smallmatrix}\right)  = G[R_y(c-d),R_y(-c-d)],  \\
e^{i (f Z \otimes I + g I \otimes Z)} & =  \left(\begin{smallmatrix}
e^{i(f+g)}  & 0 & 0 & 0 \\
0 & e^{i(f-g)} & 0 & 0 \\
0 & 0 & e^{-i(f-g)} & 0 \\
0 & 0 & 0 & e^{-i(f+g)}
\end{smallmatrix}\right) = G[R_z(e+f),R_z(e-f)].
\end{align}
\end{subequations}

All matrices in Eq.~\hyperref[eq:JWtoMG]{(\ref*{eq:JWtoMG})} are of the form $G(A,B)$ with det$A$=det$B$---that is, they are matchgates. Not only that but, since $G(A,B)G(C,D)=G(AC,BD)$, by a specific choice of the parameters and by composing gates from \eq{JWtoMG} we can obtain any single-qubit gates $A,B \in \SU(2)$ in $G(A,B)$. This hand-waving argument shows that the group generated by the six quadratic Hamiltonians of \eqs{JWquad1}{JWquad2} is exactly the group of matchgates, up to some irrelevant global phases.

We saw that nearest-neighbor matchgates correspond to Hamiltonians quadratic in the $c_i$ operators, and these in turn correspond to noninteracting fermions. We can now show that a poly-sized circuit of such gates is classically simulable. Suppose that the circuit being simulated has an initial product state input $\ket{\psi}=\ket{\psi_1} \ket{\psi_2}\ldots\ket{\psi_n}$, a sequence of nearest-neighbor matchgates, and a final single-qubit measurement in the computational basis. To simulate the final measurement of qubit $k$, it suffices to calculate the expectation value $\langle Z_k \rangle$ = $-i \langle c_{2k-1} c_{2k} \rangle = -i \bra{\psi} U^{\dagger} c_{2k-1} c_{2k} U \ket{\psi}$, where $U$ is the unitary corresponding to the action of the matchgate circuit. This suffices because, if $p_i$ is the probability of qubit $k$ being measured in state $\ket{i}$, we have that $\langle Z_k \rangle = p_0 - p_1 = 2p_0 -1$. This is a case of strong simulation, as defined in \sec{introsimul}, although with the restriction that the output of the circuit being simulated consists of one single-qubit measurement. For a more general set of output measurements, we would also need to compute conditional probabilities, as done e.g.\ in \cite{Terhal2002}. However, here we follow along the lines of \cite{Jozsa2008b}, considering only a single-qubit output, with the justification that this suffices, for example, to subsume any decision problem solvable by a matchgate circuit. So, to show that $\langle Z_k \rangle$ can be computed efficiently, we begin with the following Theorem (cf.\ \cite{Knill2001a,Terhal2002,Jozsa2008b}, as proved in \cite{Jozsa2008b}):

\begin{theorem} \label{thm:quadratic}
Let H be any Hamiltonian given by a sum of terms quadratic in the $c_k$ operators and let $U = e^{iH}$ be the corresponding unitary. Then, for all $i \in \{1,\ldots,2n\}$,
\begin{equation}
U^{\dagger} c_i U = \sum_{j=1}^{2n} R_{i,j} c_j,
\end{equation}
where $R \in\SO(2n)$, and we obtain all of $\SO(2n)$ this way.
\end{theorem}

\begin{proof}
First, let us write
\begin{equation}
H = i \sum_{i\neq k=1}^{2n} h_{ij} c_i c_j,
\end{equation}
where $h_{ij}$ is a real antisymmetric matrix, since $H$ is hermitian and the operators $c_i,c_j$ anti-commute if $i \neq j$. Also notice that terms with $i=j$ can be omitted as $c_i$ squares to the identity. We now go explicitly to the Heisenberg representation to write $c_i$ as $c_{i}(0)$, $c_{i}(t)=U(t) c_i(0) U(t)^{\dagger}$ and, consequently,
\begin{equation*}
\dot{c}_{i}(t)= i \left [H, c_{i}(t) \right ] = \sum_j 4 h_{ij} c_{j}(t).
\end{equation*}
The solution to this differential equation is simply
\begin{equation*}
c_{i}(t)= \sum_j R_{ij}(t) c_{j}(0),
\end{equation*}
where $R=e^{4ht}$. By setting $t=1$, we obtain a direct expression for $R$ in terms of the Hamiltonian which generates the unitary $U$. Also notice that, in principle, $h$ can be any antisymmetric matrix, and so $R$ can be any matrix in $\SO(2n)$.
\end{proof}

Since the Hamiltonians that generate nearest-neighbor matchgates are quadratic in the $c_i$ operators, the expectation value for $Z_k$ after the application of the circuit is
\begin{align} \label{eq:expectedZ}
\langle Z_k \rangle & = -i \bra{\psi} U^{\dagger} c_{2k-1} c_{2k} U \ket{\psi} \\ \notag
& = -i \sum_{a,b=1}^{n} R_{2k-1, a} R_{2k, b} \bra{\psi} c_a c_b \ket{\psi}.
\end{align}

If $t$ is the number of matchgates in the circuit, $R \in \SO(2n)$ can be calculated in $\poly(n,t)$ time as the product of the $\SO(2n)$ rotations corresponding to each matchgate. Also notice that the sum in \eq{expectedZ} has only $O(n^2)$ terms. Finally, note that $\ket{\psi}$ is a product state, and any monomial $c_a c_b$ is a tensor product of Pauli matrices, as is clear from \eq{JWc}. Thus, each term in the sum factors as a product of the form $\prod_{i=1}^{n} \bra{\psi_i} \sigma_i \ket{\psi_i}$, which involves $n$ efficiently computable terms. Since $\langle Z_k \rangle$  is a sum of a polynomial number of efficiently computable terms, it can be computed efficiently, which completes the proof of classical simulability of nearest-neighbor matchgates. We suggested, in \chap{introduction}, that the reason why noninteracting fermions are classically simulable is that they evolve according to determinants (cf.\ \lem{fermdet}), which marks a strong distinction with their bosonic counterparts. This did not appear explicitly in the above proof, since we opted for a presentation of the result that is more convenient for future discussions, but this simulability can in fact be traced to \lem{fermdet}, see e.g.\ \cite{Terhal2002}.

Before I finish this section, I would like to detail two ways in which this simulability result can be generalized, both discussed in \cite{Jozsa2008b}. First, the operators found in \eqs{JWquad1}{JWquad2} are the only quadratic operators corresponding to one- or two-qubit gates, but there are other quadratic operators which translate to gates acting on more than two qubits. For example, the operator $c_2 c_5$ translates to $X_1 Z_2 X_3$---in fact, it is easy to see that all such quadratic operators\footnote{In \cite{Jozsa2008b} the authors denote the unitaries generated by these Hamiltonians as Gaussian operators.} have the form $A_i Z_{i+1} \ldots Z_{j-1} B_{j}$ with $A,B \in \{ X, Y \}$ and $1 \leq i < j \leq n$. Thus these can correspond to gates acting nontrivially on many qubits at a time. However, there is a sense in which these gates do not add anything qualitatively new to circuits of matchgates, as they can be implemented by a circuit of nearest-neighbor matchgates. To see that, consider the $G(Z,X)$ matchgate, which acts as a $\swap$ gate followed by a $\cz$ gate. We dub this entangling gate the fermionic $\swap$ (or $\fswap$), and it will play a central role in our results in \sec{fermnew_b}. For now, it suffices to state that this gate acts on the matchgate Hamiltonians as e.g.,
\begin{equation} \label{eq:fswaponXX}
\fs_{23} (X_1 X_2) \fs_{23} = X_1 Z_2 X_3,
\end{equation}
where $\fs$ is a shorthand for the $\fswap$. This can be easily generalized for other matchgates and qubit pairs, and thus the operators $A_i Z_{i+1} \ldots Z_{j-1} B_{j}$ mentioned earlier can be implemented using regular nearest-neighbor matchgates and repeated use of the $\fswap$ gate.

The second generalization of the result concerns a normal form for matchgate circuits. Recall from \thm{quadratic} that each matchgate in the circuit effectively corresponds to a $\SO(2n)$ rotation $r$ in the space of the $c_i$ operators, and that the action of the whole circuit is obtained as some rotation $R$ which is the product of these individual rotations. We can now decompose $R$ into a product of $O(n^2)$ rotations which act nontrivially only on two dimensions (e.g.\ the space spanned by $c_i$ and $c_j$). An algorithm for this decomposition can be found in \cite{Hoffman1972}---it follows the same lines of the decomposition of general $n$-qubit unitaries into $2$-qubit gates \cite{LivroNielsen}, or of a general linear-optical transformation into two-mode interferometers \cite{Reck1994} which we will use in \chap{bosonnew}.

More specifically, let $r_{i,j}$ denote the matrix appearing in the decomposition of $R$ acting nontrivially only on coordinates $i$ and $j$. That is, we have

\begin{align*}
R = r_{2n-1,2n} \, r_{2n-2,2n} \, \ldots \, r_{2,2n} \, r_{1,2n} & \\ 
		 \times \, r_{2n-2,2n-2} \, \ldots \, r_{2,2n-1} \, r_{1,2n-1}&\\
		 \ldots &\\
							\times \, r_{2,3} \, r_{1,3} & \\
							\times \, r_{1,2} &.
\end{align*}
Let $r_{i,j} = e^{h_{ij}}$, where $h_{ij}$ is an antisymmetric matrix with nonzero values only on rows and columns with coordinates $i$ and $j$. Recall then from \thm{quadratic} that this rotation on the space of the $c_i$'s corresponds to a unitary $u_{i,j} = e^{i \theta_{ij} c_i c_j}$ on the space of the qubits, for some parameter $\theta_{ij}$. This naturally defines a unitary matrix $U$ corresponding to the rotation $R$ and with the structure:

\begin{align}\label {eq:normalform} 
U = u_{2n-1,2n} \, u_{2n-2,2n} \, \ldots \, u_{2,2n} \, u_{1,2n} & \notag \\ 
		 \times \, u_{2n-2,2n-2} \, \ldots \, u_{2,2n-1} \, u_{1,2n-1}& \notag\\
		 \ldots & \notag\\
							\times \, u_{2,3} \, u_{1,3} & \notag \\
							\times \, u_{1,2} &.
\end{align}
By construction, this unitary matrix acts on the $c_i$ operators in the same way as the unitary $V$ defined by the matchgate circuit. Hence the actions of $U$ and $V$ must be the same on any monomial $c_{i_1} \ldots c_{i_k}$. Since these monomials span the space of all $n \times n$ matrices, $U$ and $V$ must be the same matrix up to some irrelevant global phase. We thus say that $U$ is a normal form for $V$, in the sense that it is the minimal form which only needs $O(n^2)$ gates [or $O(n^3)$, if we want to use only nearest-neighbor matchgates and the trick of \eq{fswaponXX}], regardless on the number of matchgates in the original circuit. In \fig{normalform} we show what such a normal form would look like for a 4-qubit circuit\footnote{There is one omitted step in obtaining the circuit of \fig{normalform}: each rectangle actually involves two operations from \eq{normalform} that have been brought together by exploiting the commutation of some of the $u$'s, such as $u_{1,4}$ and $u_{2,3}$}. Remark that the labels in \eq{normalform} do not refer to the qubits, but to the $c_i$ operators, thus both labels $2j-1$ and $2j$ refer to gates acting on qubit $j$.

\begin{figure}
\capstart
\centering
\includegraphics[width=0.7\textwidth]{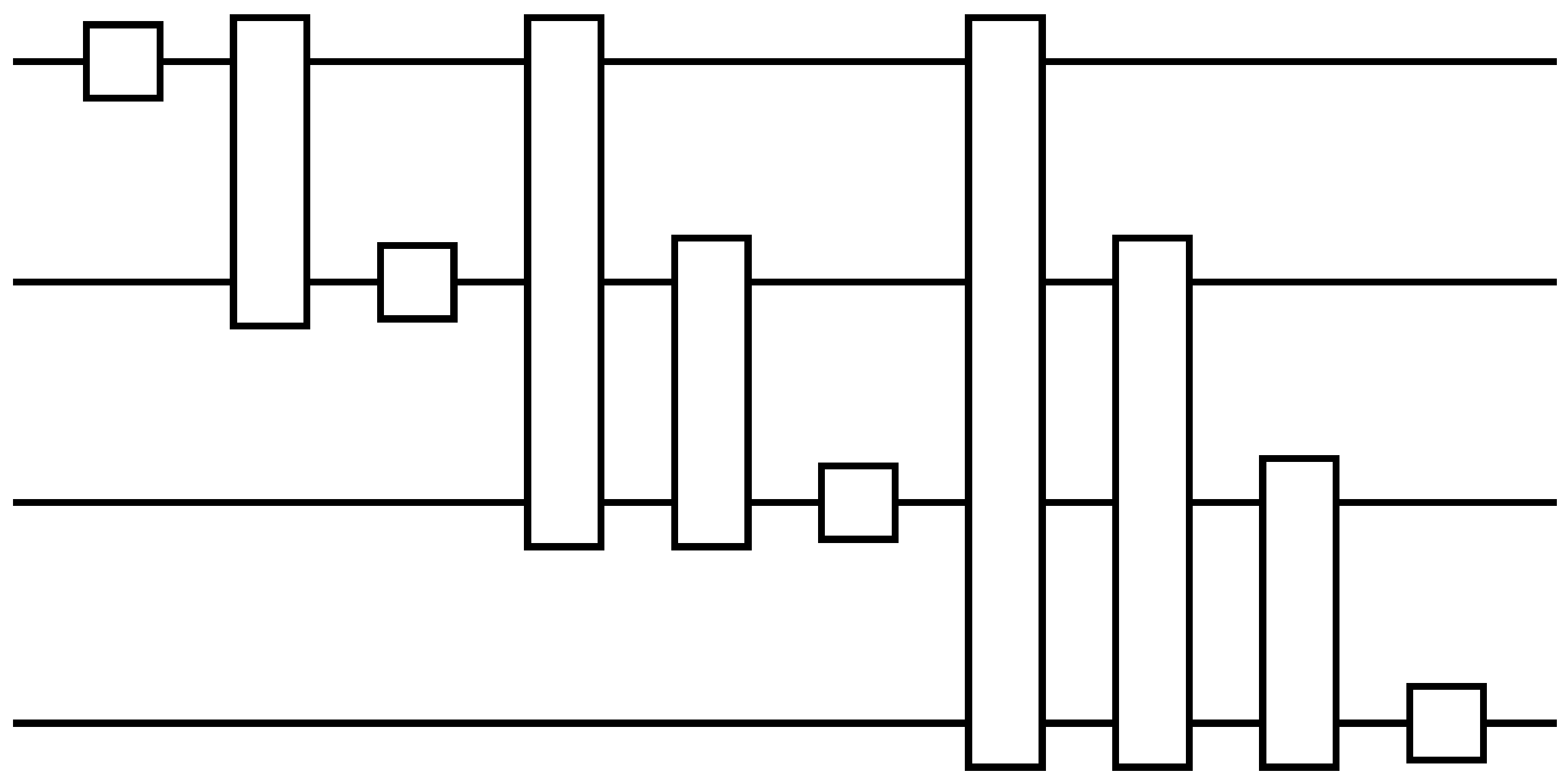}  
\caption[A 4-qubit matchgate circuit in normal form]{A 4-qubit matchgate circuit in normal form. We represent each gate by a simple rectangle, as their expressions and labels can be easily worked out from \eq{normalform}. Squares represent single-qubit $Z$ rotations, and rectangles spanning many lines represent unitaries acting nontrivially on all these qubits, but that can be implemented using a ladder of $\fswap$ gates and a nearest-neighbor matchgate [cf.\ \eq{fswaponXX}].}	
\label{fig:normalform}
\end{figure}

\section{Positive universality results for matchgates} \label{sec:fermreview_b}

In this section, I review the positive results for universal quantum computation with matchgates, in contrast to the previous section. Clearly some restriction must be relaxed, or some extra ingredient must be added, in order to obtain this universality. For instance, it was shown that nondestructive two-electron charge measurements---that is, in the qubit picture, a measurement that only distinguishes the numbers of 1's in the two-qubit state---enables universal computation with fermionic linear optics\footnote{Evidently, such a measurement cannot be implemented in the qubit picture with nearest-neighbor matchgates and computational basis measurements.}. Here, we focus on specific quantum gates and/or changes in the connectivity restrictions that can uplift matchgates to universal quantum computation. I begin by showing, along the lines of \cite{Jozsa2008b}, how matchgates become universal when complemented with the $\swap$ gate. 

First note that matchgates are parity-preserving operations. This means that they do not connect $n$-qubit states with different overall parity and so it is not possible, for instance, to even approximate an arbitrary two-qubit gate without using some ancilla qubits. One way to sidestep that issue is to consider only encoded universality, as discussed in \sec{introuniv}---to my knowledge, all universality results for matchgates adopt some similar definition. Consider then an encoding of a logical qubit into two physical qubits, given by
\begin{align}
\ket{0}_L = \ket{00},  \notag \\
\ket{1}_L = \ket{11}. \label{eq:evenencoding}
\end{align}

Clearly an encoded single-qubit gate $A_L$ can be implemented simply by applying the matchgate $G(A,A)$ to the pair of physical qubits that encode the logical qubit (see \sfig{JozsaDem}{a}).

To obtain a universal set we also require an entangling 2-qubit gate, such as the $\cz$ gate. Consider two adjacent logical qubits encoded in physical pairs of qubits labeled $\{1,2\}$ and $\{3,4\}$, respectively. Then a $\cz_L$ between the logical qubits can be implemented simply by a $\cz$ between the neighboring qubits 2 and 3. Note that this is not a matchgate, as this would contradict the simulability results of the previous section. Therefore the entangling gate must be implemented with the aid of some non-matchgate operation. One such example is the sequence (see \sfig{JozsaDem}{b})
\begin{equation} \label{eq:CZL}
\cz = \fswap \cdot \swap.
\end{equation}
Recall that $\swap = G(I,X)$ is not a matchgate, while the closely related gate $\fswap=G(Z,X)$ is a matchgate that swaps the two qubits and induces a minus sign when both are in the $\ket{1}$ state (we already used this gate with another purpose in the previous section). In \eq{CZL} we can interpret the $\swap$ as undoing an undesired interchange of the qubits induced by the $\fswap$ during the entangling operation.

\begin{figure}[t]
\capstart
\centering
\subfloat[]{\centering \raisebox{0.45in}{\includegraphics[width=0.3\textwidth]{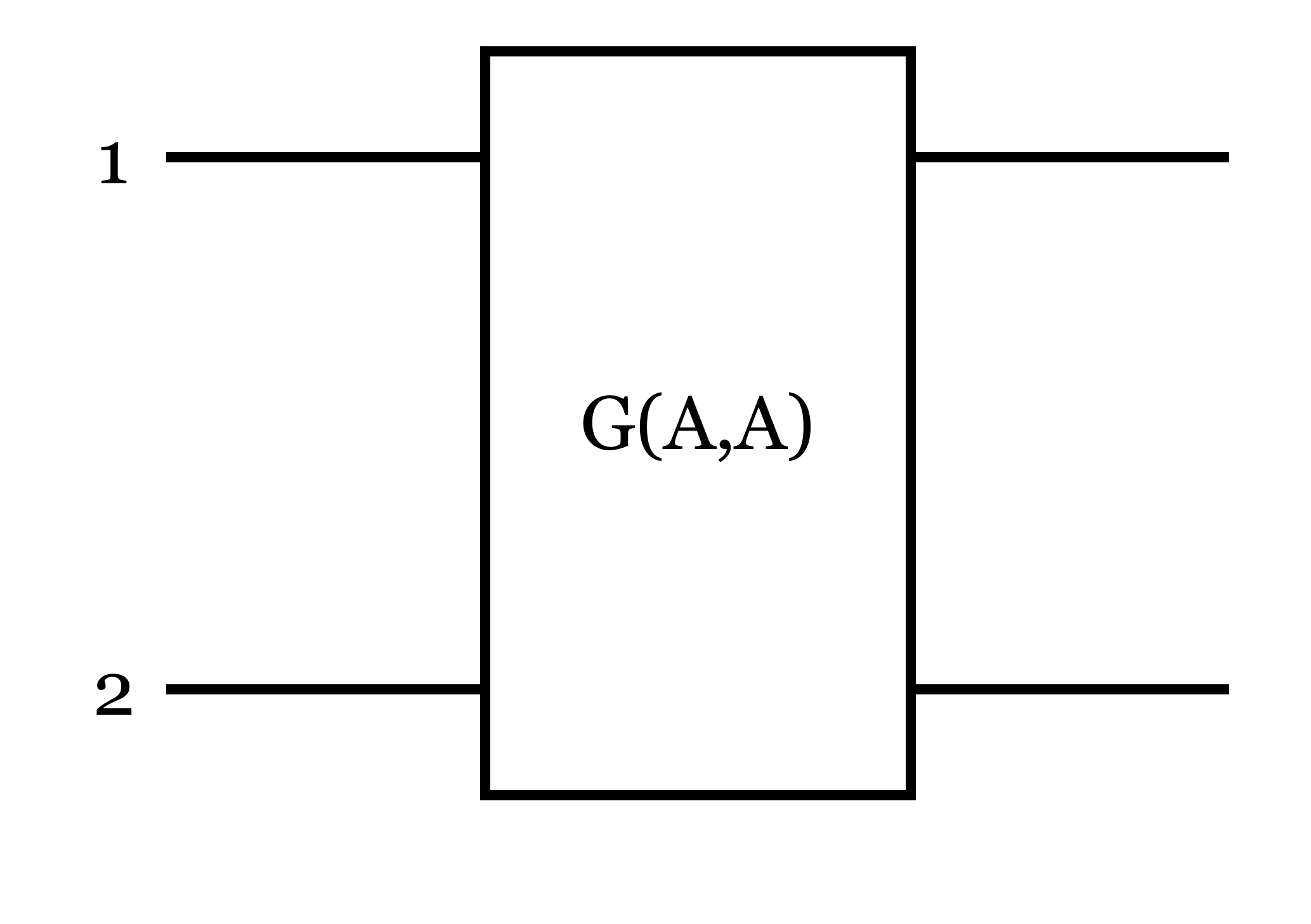}}} \qquad
\subfloat[]{\centering \includegraphics[width=0.5\textwidth]{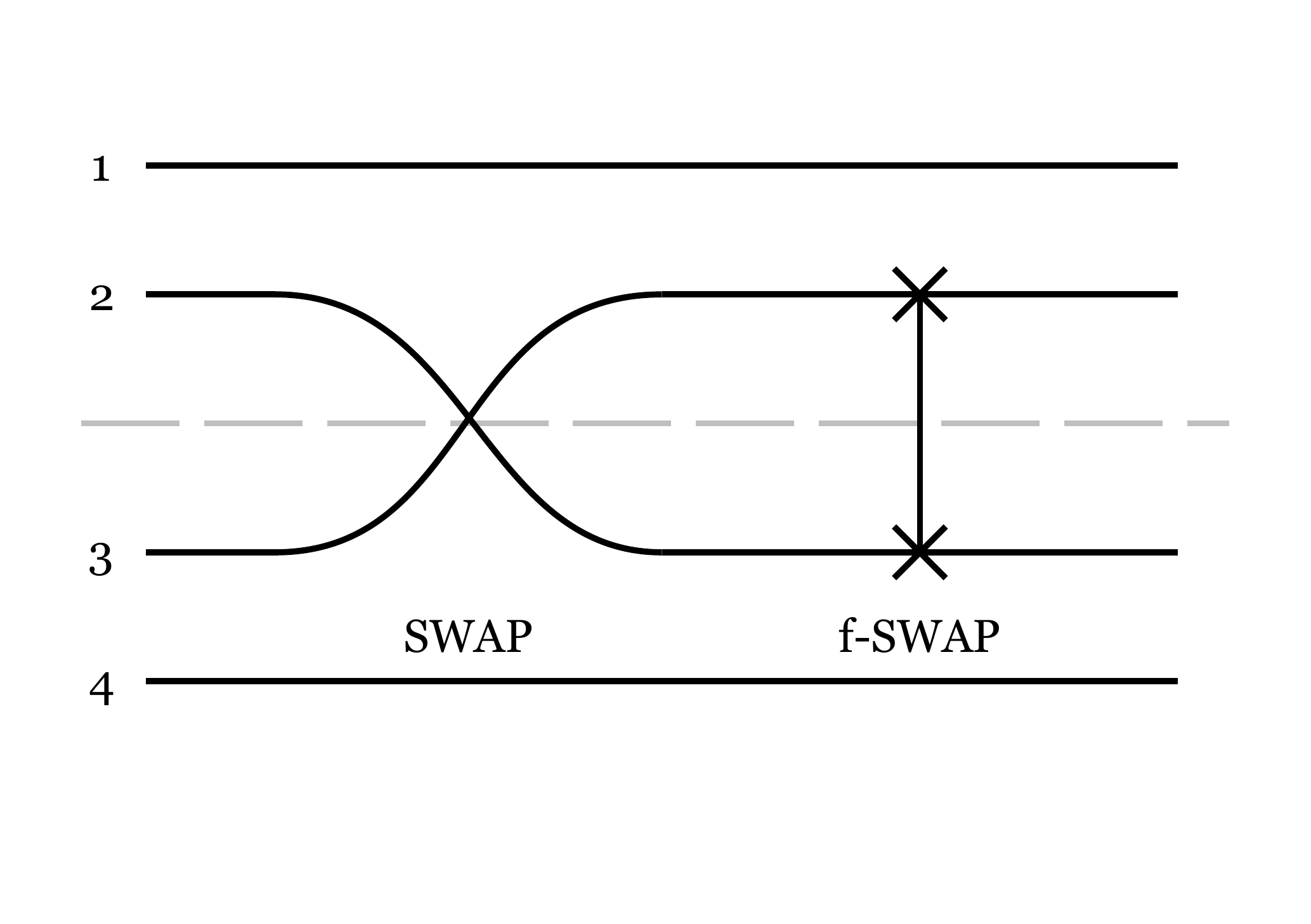}} \\
\caption[Universal set of gates built from matchgates and the $\swap$]{Universal set of gates built from matchgates and the $\swap$: (a) a logical single-qubit $A_L$, and (b) an entangling logical $\cz$ gate between encoded pairs $\{1,2\}$ and $\{3,4\}$. Note that the symbol we used for the $\fswap$ is occasionally used for the $\swap$ gate in the literature, but we opted to use a more intuitive graphical representation for the $\swap$.}
\label{fig:JozsaDem}
\end{figure} 

As we have obtained arbitrary (encoded) single-qubit gates and a two-qubit gate, we thus conclude that matchgates, when supplemented by the $\swap$, form a universal set. Furthermore, the $\swap$ is only applied on disjoint sets of physical qubits (i.e., $\{2i, 2i+1\}$ for $1 \leq i \leq n$, where $n$ is the total number of logical qubits), so no qubit is swapped more than one position away from its original place \cite{Jozsa2008b}. We can then commute all the $\swap$ gates to the end of the circuit, at the cost of allowing some of the matchgates to act on more distant neighbors. Since no qubit is moved farther than one position away, this means that the $\swap$ gate in this construction can be replaced by allowing matchgates to also act on second and third neighbors. In fact, matchgates between only nearest and next-nearest neighbors are already universal as shown in \cite{Jozsa2008b} using an alternative encoding of one logical qubit into $4$ physical qubits.

The universality result, as described, was first put forth in the context of matchgates by Jozsa and Miyake in \cite{Jozsa2008b}. However, it was already known since the work of \cite{Kempe2001b} in the context of the XY (or anisotropic Heisenberg) interaction. I now make a brief digression to introduce the XY interaction, as it provides a proper subset of matchgates that arises more naturally in real-world implementations of quantum computing, while still retaining most of the computational properties of the whole set, as we will see in \chap{fermnew}.

\subsection{The XY interaction} \label{sec:ferm_reviewXY}

Consider the following Hamiltonians

\begin{subequations} \label{eq:HeisenHamil}
\begin{align}
H_A & = X \otimes X + Y \otimes Y, \\
H_I & = X \otimes X + Y \otimes Y + Z \otimes Z.
\end{align}
\end{subequations}

These are idealized models of interactions present in several proposed solid state implementations of quantum computing. $H_A$ is known as the XY interaction, or Heisenberg anisotropic interaction, and arises in systems such as quantum dots \cite{Imamoglu1999, Quiroga1999}, atoms in cavities \cite{Zheng2000}, and quantum Hall systems \cite{Mozyrsky2001}. $H_I$ is known as the exchange or isotropic Heisenberg interaction \cite{DiVincenzo2000a}. Notice that $H_A$ is spanned by Hamiltonians which generate matchgates [cf.\ \eq{JWquad2}], while $H_I$ is not, and this reflects directly on their computational power: unitaries generated by $H_I$ are universal for quantum computation acting on nearest-neighboring qubits \cite{DiVincenzo2000a,Kempe2001a}, whereas the XY interaction is simulable on nearest neighbors (as we know it should be from \sec{fermreview_a}), but universal when acting also on next-nearest neighbors. This can be shown using an encoding of a logical qubit into $3$ physical qubits \cite{Kempe2001b,Kempe2002}, but we do not review this proof here as it is not qualitatively different from that of general matchgates in the previous section, and our results of \chap{fermnew} will subsume it. For now we just point out that, although the XY interaction generates a proper subset of matchgates (acting nontrivially only on the odd parity subspace of the two-qubit state), it bridges the gap between (sub-)classical and quantum computational power under curiously similar conditions as matchgates.

To bring the idealized models closer to real-world implementations, work has also been done on generalizations of the XY interaction, where other spurious anisotropic terms such as $X\otimes Y$ are present \cite{Vala2002}, arising for example from surface effects or spin-orbit coupling \cite{Dzyaloshinsky1958,Moriya1960}, or where there is the presence of a background field $\sum_i{\epsilon_i Z_i}$ \cite{Lidar2001}. In both situations it was found that such effects can be canceled out by clever sequences of the allowed operations. One such example is the technique known as encoded selective recoupling, used in \cite{Lidar2001}. Intuitively, this technique works by finding one controllable Hamiltonian $A$ such that, for some background Hamiltonian $B$, we have exp$(-i A \tau)$ exp$(iB)$ exp$(i A \tau)=$ exp$(-i B)$\footnote{Note that this is true whenever $A$ and $B$ anti-commute, $A$ squares to the identity and $\tau=\pi/2$, thus these are easy criteria to apply for Pauli matrices.}. Then, by alternating pulses of $A$, one can cancel out the effect of $B$ whenever necessary. A more detailed description of such results is beyond our purposes, so we refer to \cite{Lidar2001,Vala2002} and references therein.

\section{Other computational power results} \label{sec:fermreview_c}

In the last section of this chapter, I would like to give a brief description of other computational results concerning matchgates. 

The first of these results shows an equivalence between matchgate circuits on $n$ qubits and general quantum circuits on $O(\log n)$ qubits \cite{Jozsa2010}. More specifically, Jozsa and collaborators show the following

\begin{theorem} \label{thm:logspace}
The following equivalence between matchgate circuits and general quantum circuits holds:
\begin{itemize}
\item Given a circuit MG of nearest-neighbor matchgates with an $n$-qubit input $\ket{x_1 \ldots x_n}$, $N$ gates, and final measurement on qubit $k$, there exists an equivalent quantum circuit QC with an input of $\lceil \log n \rceil + 3$ qubits initialized in the 0 state, composed of $O(N \log n)$ gates, and with final measurement on the first qubit. Moreover, the encoding of the circuit QC can be computed from the encoding of the matchgate circuit MG by means of a classical space $O(\log n)$ computation.
\item Conversely, given any quantum circuit QC with an $m$-qubit input $\ket{y_1 \ldots y_m}$, $M$ gates, and final measurement on the first qubit, there exists an equivalent matchgate circuit MG with an input of $2^m-1$ qubits initialized in the $0$ state, composed of $O(M 2^{2m})$ gates and with final measurement on the first qubit. Moreover, the encoding of matchgate circuit MG can be computed from the encoding of the circuit QC by means of a classical space $O(m)$ computation.
\end{itemize}
\end{theorem}
 
A proof for this Theorem can be found in \cite{Jozsa2010}. In this Theorem, two circuits are equivalent if their output measurements have the same probability distribution. Also, the feature that the encoding of the circuits be computable in a classical computer with limited space guarantees that the computational power is due to the circuits themselves, not ``hidden'' in their classical description. 

In order to use \thm{logspace} to make claims about actual complexity classes, there are some technicalities involved in defining families of computational tasks that can be solved by such circuits. The complete formal treatment can be found in \cite{Jozsa2010}, we will only concern ourselves with the significance of this result. In essence, \thm{logspace} states that a polynomially-sized circuit of nearest-neighbor matchgates can be simulated by a logarithmically-sized quantum circuit---this is done by directly interpreting the rotation $R \in \SO(2n)$ of \thm{quadratic} as a unitary matrix describing some circuit and, since this unitary matrix has only dimensions $2n \times 2n$, it corresponds to a transformation on a space of $O(\log n)$ qubits. Conversely \thm{logspace} states that an arbitrary quantum circuit on $O(\log n)$ qubits can be simulated by a circuit of nearest-neighbor matchgates on $n$ qubits. Thus, these results together not only provide an equivalence between these two types of restricted quantum circuits, but suggest that real-world physical systems where the evolution is described by matchgates could be simulated by very small quantum computers. In fact, building on this result other works have been published describing compressed quantum simulations of properties of the 1D Ising \cite{Kraus2011} and 1D XY \cite{Boyajian2013} models.

Finally, I would like to mention in passing that the computational power of matchgates has also been related to a \emph{classical} complexity class. In \cite{Nest2011a}, Van den Nest showed that the class of Boolean functions computable by circuits of nearest-neighbor matchgates coincides to the so-called linear threshold gates\footnote{A Boolean function $f$ on $n$ bits is called a linear threshold gate if there exists an $n$-dimensional real vector $w$ and a real constant $\theta$, such that $f(x)$ equals 0 if and only if $w^T+\theta$ is strictly positive \cite{Nest2011a}.}. The author also shows that, if the computation is to succeed with probability greater than $3/4$, then the only computable functions are trivial in the sense that they only depend on one bit of the input. Note that, while a probabilistic computation that succeeds with probability greater that $1/2 + \epsilon$, for constant $\epsilon$, can usually be amplified to probability arbitrarily close to 1, this requires the use of the majority vote function (cf. \sec{introduction_d}). While the majority vote is in principle computable by a matchgate circuit, there is a caveat: to be used for this amplification, the majority vote itself must succeed with high probability, a condition unsatisfied in the case of matchgates \cite{Nest2011a}. For a complete discussion of these issues, of linear threshold gates and a comparison between the results of Van den Nest and of Jozsa \textit{et al.}\ of the beginning of the section, see \cite{Nest2011a} and references therein.
\newpage
\chapter{Review: Linear optics and BosonSampling} \label{chapter:bosonreview}

In this chapter, I will review several known results regarding the computational power of linear optics. In \sec{boson_overview} I give a brief historical overview the development of linear optical quantum computing, from its inception in the seminal KLM paper \cite{Knill2001b}, to the current state-of-the-art and the BosonSampling model \cite{Aaronson2013a}. In subsequent sections I describe some of these results in more detail, particularly those important for \chap{bosonnew}. Unfortunately, a complete account of the BosonSampling model is beyond the scope of this thesis, as there is an unbounded amount of technical details one could choose to include, both from the complexity-theoretical and the experimental aspects of the model. Our focus is precisely on this interface between theory and experiment, and in \chap{bosonnew} I report new results of both types. In view of this, I will restrict the level of detail of this chapter to only that necessary for our purposes, at the risk of omitting important technical discussions which can, nonetheless, be found in the provided references.

\section{Historical background} \label{sec:boson_overview}

Photons are considered excellent information carriers---they hardly interact with their environment, and thus loss of the information through decoherence is a relatively minor problem in optical devices, in contrast to other physical implementations of quantum information processing. As such, photons are among the best candidates for quantum \emph{communication}, given that they can carry information over dozens or hundreds of kilometers, either of free space \cite{Yin2012, Ma2012} or optical fibers \cite{Jouguet2013}. However, for more general quantum information processing tasks, their main advantage also becomes a serious hindrance: not only do they hardly interact with their environment, but they also hardly interact with \emph{each other}. Until 2001, scalable linear-optical quantum computing was regarded as unlikely due to precisely this reason, since the intuition was that strongly nonlinear media would be required for the implementation of two-qubit gates.

In 2001, Knill, Laflamme and Milburn (KLM) published their seminal paper \cite{Knill2001b}, proving that it is in fact possible to implement scalable quantum computing using only linear optics. What they realized was that measurement is also an inherently nonlinear process, and thus photon interaction can be replaced by measurement-induced nonlinearities \cite{Scheel2003}. Unfortunately, measurements are also \emph{probabilistic}, and consequently any attempt to implement a two-qubit gate via measurements has some nonzero chance of failure. For instance, in \cite{Knill2002}, Knill proposed a linear-optical $\cz$ gate between two optical modes that works with probability $2/27$ (of which we give an explicit construction in \sec{bosonreview_a}). 

Clearly, a naive construction based only on probabilistic gates is not scalable, as the success probability of the overall circuit decays exponentially with the number of gates. To remedy this, the KLM scheme also includes an error-correction step. By a clever combination of encoding, gate teleportation, and error correction, it is possible to increase the probability of success of the two-qubit gates to arbitrarily close to one, and to protect the information from being destroyed when the gate fails. The crucial feature of these gates is that they are heralded---that is, the success of the gate is signaled by the measurement outcome of some auxiliary mode. Thus, an essential ingredient for universal linear-optical quantum computing is \emph{measurement feedforward}, or adaptive measurements, which consists on performing operations conditioned on the outcomes of previous measurements (recall that this is also the defining feature of measurement-based quantum computation, described in \sec{introuniv}).

Photons also have additional advantages over other physical systems. For example, they can encode information in several different degrees of freedom: spatial modes, polarization, orbital angular momentum, frequency, time-bin, etc\footnote{Calling these degrees of freedom of a photon is an abuse of terminology, as technically what we have are different \emph{modes}, each labeled by a set of quantum numbers, and each populated by a certain number of photons (cf.\ \sec{intro_secondquant}). However, it is often a good approximation to consider them independent degrees of freedom, and ignore some of them completely if they do not couple dynamically to the others in a given experiment, which is the approach we take here.}. This provides great flexibility for encoding information, and for coupling photons with other physical systems into hybrid proposals. An interesting example is the encoding of more than one qubit into a same photon---for example, a photon can encode a qubit in its polarization, another in the orbital angular momentum, and so on. Of course, this only provides a well-defined model if we can manipulate each degree of freedom independently, and there are obvious complications with the fact that photonic measurements are usually destructive. However, it also gives rise to curious phenomena, such as entanglement between different modes of a single photon\footnote{There is some controversy on whether this should be denominated entanglement. In this thesis we take the approach that, if the degrees of freedom are sufficiently independent, and if the computational complexity results can be trivially adapted to treat all degrees of freedom as equivalent, then the specific denomination is not very important.}, or between several degrees of freedom of different photons (what is known as hyperentanglement \cite{Barreiro2005}). These reasons, among others, have driven the field of linear-optical quantum computing, both in improvements of the original KLM scheme and in the development of new experimental techniques. 

Unfortunately, the current state-of-the-art in linear optical experiments is still far from the required for a practical large-scale implementation of the KLM scheme, and one of the most experimentally-challenging steps is measurement feedforward. Given the magnitude of the speed of light in optical media, very fast electronics is required to measure an optical mode, decide the next computational step conditioned on the result and change the optical network accordingly, all before the remaining photons arrive at the next stage of the computation. One could then ask what the computational power of linear optics is \emph{without} adaptive measurements. Even if they seem necessary for arbitrary quantum computations, are there any nontrivial tasks that can be performed without them? These could provide intermediate milestones to drive the development of the field in the near future. A partial affirmative answer to this question can be found in the \emph{BosonSampling} model, which is the main subject of this chapter and \chap{bosonnew}.

The BosonSampling model was recently proposed by Aaronson and Arkhipov \cite{Aaronson2013a}. BosonSampling consists essentially of non-adaptive linear-optical quantum computing: one starts with an $n$-photon Fock state, evolves it according to some random $m$-port interferometer and measures the output distribution. The task then is to sample from a distribution that's close in total variation distance to the ideal output distribution predicted by Quantum Mechanics. Sampling \emph{exactly} from the ideal distribution is presumably harder, of course, but it may too hard even for the quantum device itself, due to experimental imperfections. The more realistic task of approximate sampling cannot be performed efficiently by classical computers, modulo some plausible conjectures, unless the polynomial hierarchy collapses to its third level (cf.\ \sec{introduction_c}). The authors report that this model did not arise from experimental motivations, but rather it was driven by the $\#$P-hardness of the permanent function and a way to harness this to propose a task in which quantum systems could outperform classical computers---that it can be adapted so naturally to a particular physical system came as a bonus.

Since its inception, BosonSampling has received quite a lot of attention both from theoretical and experimental sides. Several quantum optics groups have reported small-scale BosonSampling (and related) experiments \cite{Broome2013, Crespi2013b,Spring2013,Tillmann2013,Spagnolo2013b,Carolan2013,Spagnolo2013c}, some of which are reported in \chap{bosonnew} as I was part of the collaboration that co-authored these results. From the theoretical side, several results were reported attempting to relax even more the conditions on the model, by e.g.\ investigating to what extent faulty source/detectors or other experimental errors can drive the model into a computationally uninteresting regime \cite{Leverrier2013,Rohde2012,Motes2013}, by proposing alternative BosonSampling models with different types of sources \cite{Lund2013,Shchesnovich2013}, or even by proposing implementations of the model in completely different physical system, such as phonon scattering in trapped ions \cite{Shen2014}.  

The model has also suffered criticism, mainly due to the fact that \emph{certifying} whether a BosonSampling device is performing as desired might be as hard as simulating it. This criticism questions the interest of large-scale implementations, as it might be impossible to verify whether the device is operating in a computationally-interesting regime \cite{Gogolin2013}. This criticism has since been (partially) answered by the original authors \cite{Aaronson2013b}, although the possibility of a complete classically-efficient certification of BosonSampling devices remains an open question since the original paper \cite{Aaronson2013a}.

For the remainder of this chapter, I will describe the KLM scheme and the BosonSampling model in some detail. Although the new results reported in \chap{bosonnew} do not focus on the KLM scheme per se, in \sec{bosonreview_a} I will review some of its constructions, as they provide a simple and straightforward way to prove some of the subsequent results. In \sec{bosonreview_b} (resp.\ \sec{bosonreview_c}) I will describe the exact (resp.\ approximate) version of the BosonSampling model in more detail, also highlighting its pros and cons relative to other quantum computational tasks of interest. 

\section{The KLM scheme}  \label{sec:bosonreview_a}

There are several variants to the original KLM scheme \cite{Knill2001b}. Here we will describe the most convenient one for our purposes in the following sections, even if it is not the best one in terms of efficiency.	

We begin with what is known as dual-rail encoding\footnote{In contrast, single-rail encoding consists in directly encoding the qubits states 0 and 1 in a mode being empty or occupied with one photon, respectively. However, in this case single-qubit gates rely on creating superpositions of states with different photon numbers, and thus are not implementable only with linear optics.}. In this case, each qubit state corresponds to a photon being in a different mode, as 
\begin{align}
\ket{0}_L = \ket{10},  \notag \\
\ket{1}_L = \ket{01}, \label{eq:KLMencoding}
\end{align}
which we represent as in \sfig{KLMscheme}{a}. Remark the similarity to \eq{evenencoding}, the encoding used for matchgates. The main difference is that here the right-hand side of \eq{KLMencoding} does not represent qubit states, but rather \emph{mode} states. In particular, states such as $\ket{20}$ are physically allowed, as they represent two photons (or any other number, in the most general case) occupying a single mode. A collection of $n$ qubits then corresponds to $n$ photons in $2n$ modes, with the modes paired two-by-two. To exemplify, suppose we pair the modes sequentially, such that the state $\ket{011001}$, of 3 photons in 6 modes, corresponds the logical state $\ket{101}_L$. In this case, as long as the system is in a state where each pair of modes $\{i,i+1\}$ , for every odd $i$, contains a single photon, it is in a valid computational state. If the photons ``leak'' into the wrong modes, producing states such as $\ket{110010}$ or $\ket{020010}$, the system is no longer in the encoding space, signaling a failure in the computation. Finally, note that we did not specify what degrees of freedom these modes represent. The notation of \eq{KLMencoding} is more usual for an encoding in which-path degrees of freedom---i.e., where each mode corresponds to a different spatial direction the photons can travel in (a beam splitter is an example of a transformation that connects two such modes). We stick to this notation, as the experiments reported in \chap{bosonnew} are all included in this case. However, the protocol would be perfectly well-defined if the modes represented different states of polarization, orbital angular momentum, etc. In this case, the difference would be the actual physical device implementing the required transformations: for e.g.\ polarization, beam splitters would be replaced by wave plates.

\begin{figure}[t]
\capstart
\centering
\subfloat[]{\centering \includegraphics[width=0.3\textwidth]{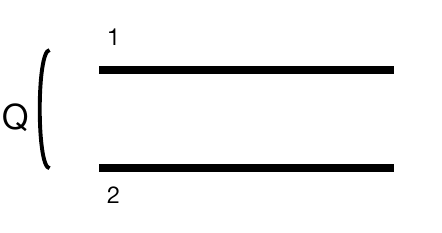}}
\subfloat[]{\centering \includegraphics[width=0.4\textwidth]{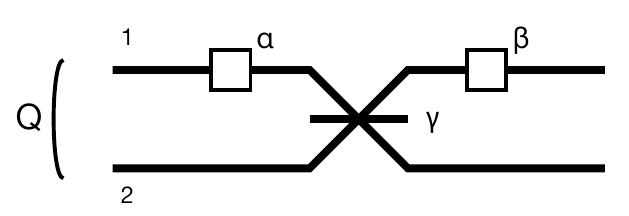}}\\
\subfloat[]{\centering \includegraphics[width=0.7\textwidth]{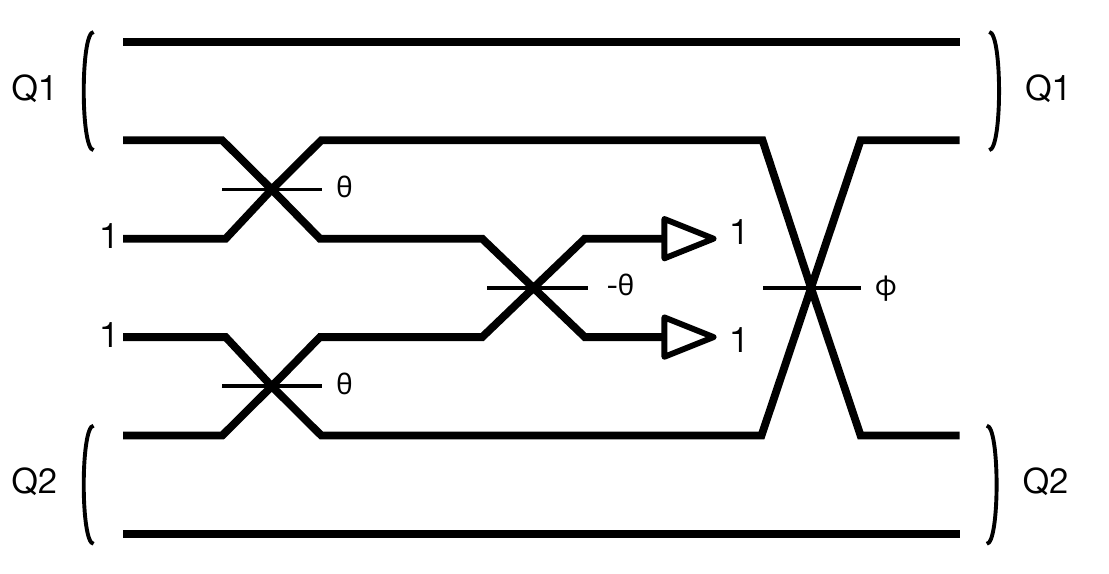}} 
\caption[The KLM scheme]{The KLM scheme. (a) Two modes encoding a qubit Q: mode 1 corresponds to qubit state $\ket{0}$, and mode 2 corresponds to state $\ket{1}$. (b) An arbitrary single-qubit gate, up to a global phase, on qubit Q. (c) A probabilistic two-qubit gate. Q1 and Q2 encode the qubit, ordered such that the outermost modes encode states $\ket{0}$ and the innermost modes encode states $\ket{1}$ of each qubit. The two central modes are ancilla modes, initialized each with a single photon. If $\phi \approx 17.63^{\circ}$ and $\theta \approx 54.74^{\circ}$,  and one photon is measured in the output of each ancilla mode, the overall action of the circuit is a $\cz$ gate.}
\label{fig:KLMscheme}
\end{figure} 

Using this encoding, single-qubit gates are easy to perform deterministically. An arbitrary unitary transformation connecting the states of \eq{KLMencoding} is nothing more than an arbitrary two-mode transformation, which can be implemented by beam splitters and phase shifters, as described in \sec{twomode}. We depict such an arbitrary transformation in \sfig{KLMscheme}{b}. One can easily see that the corresponding logical transformation is described by the matrix
\begin{equation*}
S = \begin{pmatrix}
\cos (\gamma) e^{i (\alpha+\beta)} & i \sin (\gamma) e^{i \beta}\\
i \sin (\gamma) e^{i \alpha} & \cos (\gamma) 
\end{pmatrix},
\end{equation*}
which is one possible parameterization of an arbitrary single-qubit gate (up to a global phase). Of course, as explained in \sec{twomode}, the action of the optical circuit of \sfig{KLMscheme}{b} on the complete Fock space is much more complicated, but as long as the system does not leave the encoded space, the single-qubit gates will never act in a multi-particle state, and thus the matrix $S$ correctly describes the transformation.

Finally, we need an entangling logical two-qubit gate. The simplest candidate is the \cz, since it only acts nontrivially on the $\ket{11}_L$ state, and furthermore only by adding a $-1$ phase. Thus, it clearly suffices to implement a \emph{physical} $\cz$ gate, that is, some two-mode operation that adds a $-1$ phase only if both modes are occupied by a single photon, and does nothing otherwise. Applying such an operation on the second and fourth modes of a four-mode state, for example, would only induce a phase on the $\ket{0101}=\ket{11}_L$ state. This seems to require an interaction between the photons: informally, the physical transformation would affect one photon conditioned on the other photon ``being there''. However, as we alluded to before, there is a way to do this \emph{probabilistically}. 

In \sfig{KLMscheme}{c} we depict an optical circuit due to Knill \cite{Knill2002}\footnote{This is not the original gate in the KLM paper, but it is more convenient for our purposes.}. The two qubits are represented by lines Q1 and Q2, and we consider that the innermost modes represent the $\ket{1}_L$ state of the respective qubits. We then add two extra ancilla modes, each carrying a photon, and measure them at the output of the circuit. If we measure one photon in each of the two ancilla modes, as depicted in the figure, the circuit is successful, and the overall action is a physical $\cz$ gate. If we measure any different configuration, the gate has failed. This gate succeeds with probability 2/27.

A scalable scheme for linear-optical quantum computation requires, of course, gates which are not probabilistic, or at least which have high enough probability for the well-known threshold theorem to apply \cite{Aharonov2008b}. The success probability of a quantum algorithm implemented directly with the gate of \sfig{KLMscheme}{c} decays exponentially with the number of gates. To deal with this, the KLM scheme also includes an error-correcting step. As described above, the gate of \sfig{KLMscheme}{c} works whenever we observe a particular measurement outcome of the ancilla modes. But this means that, when we observe a different outcome, we know that the gate has failed, and we know what state the system is left in---in other words, the gate's success is \emph{heralded}. One important advantage is that this allows the use of a gate teleportation trick, as first defined by Gottesman and Chuang in \cite{Gottesman1999b}, that consists of the following: whenever we want to perform some probabilistic $\cz$ gate, rather than performing it directly on the computational qubits (that presumably carry information we want to preserve), we can perform it ``off-line'', on some ancilla qubits, and, if it succeeds, use the standard trick to teleport the computational qubits onto those ancillas. This allows us to make repeated attempts of the probabilistic gate until it succeeds, and only then use it for the computation.

The heralded errors\footnote{Note that not all errors are of the type described above. Other types of errors include e.g.\ photon loss or imperfections in the optical elements. These errors are not heralded and so cannot be detected without destroying the coherence of the computation, consequently requiring more sophisticated error correction schemes usual to standard circuit-model quantum computing.} and the gate teleportation trick are the reasons why adaptive measurements seem necessary in this protocol. Whenever we detect an outcome corresponding to a failed two-qubit gate, the measurement result indicates what error-correcting steps should follow. Even if we restrict the protocol and only perform two-qubit gates off-line, the decision to use these ancillas must be conditioned on the gate's heralded success. The complete scalable construction, including the encoding and error-correction steps, can be found in the original KLM paper \cite{Knill2001b}, and we omit further discussion on this subject. We will only use the KLM construction as an intermediate step in other results related to BosonSampling, in which case the circuits of \fig{KLMscheme} will be sufficient.

One drawback of the KLM scheme is that, albeit scalable, the overhead induced in the number of optical elements is not practical. It can be estimated that, to achieve a $\cz$ gate with 95\% success probability in this protocol, it would be required of the order of $10^4$ linear-optical elements \cite{Kok2007, Hayes2004}. Since then, several improvements have been developed, most notably inspired on measurement-based quantum computing \cite{Yoran2003, Nielsen2004, Browne2005}, although there is a proposal more closely related to the circuit model \cite{Gilchrist2007}, which reduce the number of required operations by two orders of magnitude. Unfortunately, a review of these results is beyond the scope of this thesis. We would just like to point out that the circuits of \fig{KLMscheme} are not anywhere close to the state-of-the-art, but they are more convenient for our purposes. In particular, in \sec{bosonnew_a} we will be interested in the depth-complexity of the BosonSampling model---that is, how many layers of parallel elements are necessary for the linear-optical circuit to display some interesting computational power. For that purpose, the $\cz$ gate depicted in \sfig{KLMscheme}{c} is the most convenient as it displays the minimal depth for a two-qubit gate.

For an extensive review of the landscape of (experimental and theoretical) linear-optical quantum computing prior to the advent of the BosonSampling model, see e.g., the review article \cite{Kok2007}, and the textbook \cite{LivroKok}. 

\section{Exact BosonSampling} \label{sec:bosonreview_b}

In the previous section, I described one possible set of circuits that enables quantum computation within the KLM scheme. The circuits presented are manifestly not scalable, since the gates are probabilistic and I did not include details on error-correction. Nonetheless, in this section I will show how that construction can be adapted to give rise to the exact version of the BosonSampling model. In \sec{bosonreview_c} we will address some more technical aspects of the approximate version of BosonSampling, which is a proposal intended to provide more robust results, closer to what can be done with real-world imperfect physical systems.

We begin by recalling the definition of the complexity class postBQP, or BQP with post-selection, from \sec{introduction_d}. Recall that, in this definition, $p$ is a poly-sized post-selection register, and $o$ is a single-qubit output register.

\begin{definition*}
postBQP (BQP with post-selection) is the set of all languages L $\subseteq \{0,1\}^{*}$ for which there exists a uniform family of polynomial-size quantum circuits $\{Q_n\}$ such that, for all inputs $x$, after applying $Q_n$ to the state $\ket{0 \ldots 0} \otimes \ket{x}$, 
\begin{itemize}
\item[(i)] The probability of measuring $p$ in the state $\ket{00\ldots 0}$ is nonzero;
\item[(ii)] if $x \in$ L then, conditioned on measuring $p$ on state $\ket{00\ldots 0}$, the probability measuring $o$ on state $\ket{1}$ is at least 2/3;
\item[(iii)] if $x \notin$ L then, conditioned on measuring $p$ on state $\ket{00\ldots 0}$, the probability measuring $o$ on state $\ket{1}$ is at most 1/3.
\end{itemize}
\end{definition*}

We can also define the analogous class postLO, consisting of (non-adaptive) linear optics with post-selection. A subtlety is that now we must consider, as input, states of $n$ photons in $m$ modes, where $m=$poly$(n)$, and which include both the computational modes as encoded per \eq{KLMencoding} and the ancilla modes necessary for two-qubit gates as in \sfig{KLMscheme}{c}. Similarly, the output of the circuit (both the output and the post-selection registers) now consists of measurements of occupation numbers of photonic modes, and the post-selection is done on particular configurations of the photon states. Taking these two factors into account, however, the definition for postLO is straightforward. 

We can now define the \emph{exact BosonSampling} task\footnote{In this thesis, we incur in a slight abuse of terminology: the word ``BosonSampling'', for which we follow the spelling of the original authors, will occasionally be used to represent the task, as defined here, occasionally to represent the natural linear-optical device capable of performing this task, and occasionally to denote the corresponding model of computation. Each use of the word, however, will be made clear by the context.}: given a unitary matrix $U$, describing a linear-optical interferometer, and given a specific input state $\ket{S}=\ket{s_1 s_2 s_3 \ldots s_m}$ of  $n=\sum_i^m s_i$ photons in $m$ modes, sample from the output distribution obtained by measuring the occupation number of each output mode. In other words: do a weak simulation of the linear-optical device, as defined in \sec{introsimul}. By patching together the discussions of \sec{introduction_c} and \sec{bosonreview_a}, we arrive at the following result:

\begin{theorem} \label{thm:exactbs}
\cite{Aaronson2013a} If the exact BosonSampling task could be performed efficiently by a classical computer, then postBPP=PP and the polynomial hierarchy would collapse.
\end{theorem}

\begin{proof}[Proof sketch]
We begin showing that postLO=postBQP, i.e., that post-selected linear optics is equivalent to arbitrary post-selected quantum computing. 

That postLO $\subseteq$ postBQP is trivial: for an $n$-photon, $m$-mode linear optical circuit, we simply treat each mode as an $n$-level system, and use a standard mapping from $m$ qudits to O$(m \; \textrm{log} n)$ qubits \cite{Aaronson2013a,LivroNielsen}. 

For postBQP $\subseteq$ postLO, we use the KLM scheme from \sec{bosonreview_a}. For any quantum circuit, decomposed in terms of the standard universal set of single-qubit gates and $\cz$ gates, just translate it gate-by-gate to linear optics using the circuits of \fig{KLMscheme}. For every $\cz$ in the circuit, post-select the particular outcome that heralds its success, according to \sfig{KLMscheme}{c}. After this post-selection round, the remaining photonic modes will correctly encode the same output as the quantum circuit. This proves that BQP $\subseteq$ postLO but, since there is obviously no advantage in having two distinct rounds of post-selection, this also proves that postBQP $\subseteq$ postLO.

Now suppose that the BosonSampling task could be performed efficiently by a classical computer. Then, for any linear-optical circuit, there is some randomized classical circuit that samples from the exact same output distribution. This means that any quantum circuit that can be implemented in postLO, as described above, can also be implemented in postBPP, simply by running the randomized classical algorithm and post-selecting on the corresponding output. Thus, this would imply postBPP = postBQP.

Finally, recall from \sec{introduction_d} that postBQP = PP. Thus, BosonSampling $\subseteq$ BPP implies that postBPP=PP which, by the discussion of \sec{introduction_d}, implies a collapse of the polynomial hierarchy to its third level.
\end{proof}

One possible extension of this theorem is to allow for weak simulation with multiplicative error. A rigorous proof of this extended version can be found in \cite{Bremner2011} for the alternative restricted model IQP (cf.\ \sec{introduction_d}), but a simple adaptation provides the analogous result for BosonSampling. The intuition is the following: imagine we can classically sample from an approximate distribution where we know that both (i) the probability of obtaining the correct output for $p$ and (ii) the conditional probabilities for $o$ given the right outcome for $p$ are sufficiently close to the ideal ones. This suggests that the success probabilities of 1/3 and 2/3, for the yes and no instances (cf.\ the definition of postBQP), might not be altered too much. But, as discussed in \sec{introduction_c}, these values are not rigid: as long as the yes and no instances have some nonvanishing gap, for example of $2 \epsilon$ for some constant $\epsilon$, the complexity class remains unchanged. It is possible to show, then, that as long as the weak simulation is approximate within multiplicative error $c$ for some $1<c<1/\sqrt{2}$, \thm{exactbs} still holds (see e.g.\ \cite{Bremner2011}). This does not seem to hold, however, for weak simulation that is only close in total variation distance, which is the subject of the next section.

As previously mentioned, the proof of \thm{exactbs} is analogous to an equivalent result obtained for IQP. In fact, the reasoning behind it establishes a recipe of sorts: if a given restricted model, when imbued with the ability to post-select on some output register, can implement arbitrary quantum computation, then efficient weak simulation of that model would imply a collapse of the polynomial hierarchy. This recipe was used not only for BosonSampling and IQP, but also for constant-depth quantum circuits \cite{Terhal2004} and non-adaptive measurement-based quantum computation \cite{Hoban2013}, which are the four types of restricted models described in \sec{introduction_c}. In \chap{bosonnew}, we will again follow along these lines to give a unified proof of some of these results, and show that the same results hold for \emph{constant-depth linear optics}.

Exact BosonSampling has several advantages over other, more usual, quantum information processing tasks. The most prominent task a quantum computer can solve efficiently is factoring, via Shor's algorithm. However, the hardness-of-simulation of BosonSampling relies on milder assumptions than that factoring is hard. More specifically, even if factoring turns out to be in BPP, BosonSampling might remain an example of a task where quantum devices outperform their classical counterparts\footnote{Of course, the drawback is that we are restricted to sampling tasks, rather than decision problems.}. At the same time, BosonSampling is presumably easier to implement physically than Shor's algorithm, at least in optical systems, since it does not require measurement feedforward, which seems crucial to the KLM scheme. As we will see later on, computationally-interesting regimes for BosonSampling are indeed closer to the reality of nowadays experiments.

However, there are also some disadvantages. The most notable is that, while Shor's algorithm solves a problem with several real-world applications (e.g.\ cryptography), to date there is no known practical task solvable by a BosonSampling computer\footnote{One could ask whether there is any hope for non-adaptive linear optics to be universal for quantum computing. As is usual in complexity theory, this cannot be proven, but is considered highly unlikely. A plausibility argument for this intuition is that, while arbitrary $n$-qubit quantum circuits live in an operator space of dimension $4^n$, the most general $m$-mode interferometer lives in an operator space of dimension $m^2$, despite the exponentially-large Hilbert space it acts on.}. The only task achievable by such a device is simulating itself, that is, generating an output distribution that cannot be simulated by a classical computer, which of is purely academic interest---it provides evidence that quantum computers have some nontrivial computational power, and it may drive the experimental efforts in the near future, but it has no practical application. Another notable disadvantage is that there is no known efficient way to certify the operation of the device. That is, contrary to factoring\footnote{And, more generally, any problem in BQP $\cap$ NP. More general tasks in BQP suffer from the same certification problem.}, where the problem may be hard to solve but the answer is easy to verify, there is no straightforward way to determine what distribution a given sample produced by the BosonSampling device came from. We will return to this in the next section, when we review some criticism (and its rebuttal) to the approximate BosonSampling model.

Finally, we point out two disadvantages of BosonSampling, as defined so far, which can in fact be dealt with. The first is that \thm{exactbs} shows that an arbitrary linear-optical circuit should be hard to simulate classically, but does not provide a prescription for a direct (and feasible) experimental implementation of a hard instance of the problem. The second is that a multiplicative error in the simulation is still very stringent. To illustrate, suppose we constructed a linear-optical instance of Shor's algorithm via the KLM scheme, and we replace the adaptive measurements with post-selection, as described previously. The resulting circuit will only work with exponentially small probability, since it will be conditioned on observing the heralded success of every $\cz$ gate. Given that this probability of success is exponentially small, the experimental imperfections would probably drown it out completely, which is incompatible with a multiplicative-error approximation (cf.\ \sec{introsimul}). An approximation close in total variation distance would correspond to a more realistic noise model, since in this case the device could get several of the exponentially-small probabilities completely wrong and still be sufficiently close to the computationally-interesting regime. The total variation distance also has better composability properties, akin to gate fidelity in usual quantum computing: since it satisfies the triangle inequality, if each operation in a linear-optical circuit (e.g., each beam splitter) has some error $\epsilon$, the total variation distance between the ideal and the real distributions will only scale linearly in $\epsilon$. In the next section, we review the approximate BosonSampling model defined by \cite{Aaronson2013a} that deals with these two issues. 

\section{Approximate BosonSampling}\label{sec:bosonreview_c}

In this section, I give a general overview of the approximate BosonSampling model of \cite{Aaronson2013a}, with the goal of setting the motivation for the experiments reported in \chap{bosonnew}. As stated previously, the original paper is heavy in technical details. While most of these technicalities do not concern us here, some are important in order to better justify the choices made during the aforementioned experiments. Thus, for the sake of clarity, this section has been further subdivided. In \sec{bosonreview_c_motiv} I delineate the motivation and formal definition of the result, and in subsequent sections I give an overview of its technical aspects that had some impact in the choices reported in \chap{bosonnew}. 

\subsection{Motivation and Formal definition} \label{sec:bosonreview_c_motiv}

The previous discussions in this chapter should provide enough motivation for the computational complexity of linear optics in general. However, the results reviewed so far only concern the hardness of exact classical simulation of a BosonSampling device, or at most a simulation within multiplicative error. As suggested previously, allowing only for a multiplicative error is too stringent a requirement, likely beyond even the reach of realistic quantum devices, thus we reserve the title of \emph{approximate} BosonSampling to the case where we only require the simulation to be close in total variation distance to the ideal distribution. We now proceed to define the BosonSampling task, as we will use it here.

Let $\Phi_{m,n}$ be the set of all tuples $ S = (s_1, \ldots, s_m)$ such that $s_i \geq 0$ and $\sum_{i=1}^m s_i = n$. As described in \sec{intro_secondquant}, the set of all Fock states $\ket{S}$, for $S \in \Phi_{m,n}$, forms a basis for the Hilbert space of a $n$-boson, $m$-mode system, the dimension of which is $\binom {m+n-1} {n}$. For future discussion, let us also define $\Phi_{m,n}^{*}$ as the restriction of $\Phi_{m,n}$ over no-collision outcomes, that is, over states for which each $s_i$ is 0 or 1. Consider now the following procedure:

\begin{itemize}
\item[(i)] Initialize the system in a fixed Fock state $\ket{T}$, which can be taken, without loss of generality, to be $\ket{1, \ldots, 1, 0, \ldots, 0}$, i.e.\ the state where each boson enters in one of the first $n$ modes.
\item[(ii)] Evolve these bosons according to some Haar-random (i.e.\ sampled from the uniform distribution) $m$-port interferometer $U$ (cf.\ \sec{dynamics}).
\item[(iii)] Measure the system at the output of the interferometer, and register the observed outcome $\ket{S}$.
\end{itemize}

These measurement outcomes are distributed according to an output probability distribution, $\mathcal{D}_U$, over the sample space $\Phi_{m,n}$, defined by (cf.\ \lem{bosonperm})
\begin{equation} \label{eq:DUquantum}
\underset{\mathcal{D}_U}{\textrm{Pr}}[S] = \left | \frac{\textrm{perm}(U_{S,T})}{\sqrt{s1! \ldots s_m! t_1! \ldots t_m!}} \right |^2,
\end{equation}
where $U_{S,T}$ denotes the sub-matrix of $U$ obtained by taking $s_i$ copies of the $i^{th}$ row of $U$ and $t_j$ copies of its $j^{th}$ column. Note that $\mathcal{D}_U$ depends both on $U$ and $T$, although we omit the latter dependence as we will assume $\ket{T}=\ket{1, \ldots, 1, 0, \ldots, 0}$ fixed unless stated otherwise. The BosonSampling task, then, is simply defined as sampling from $\mathcal{D}_U$. In the regime that $m=\textrm{O}(n^k)$, for some constant $k$ (the best known lower bound for $k$ is 6, although the authors conjecture that it can be reduced to 2---see the discussion in \cite{Aaronson2013a}), the main result of \cite{Aaronson2013a} reads as follows:
\begin{conjecture} \label{conj:BSishard}\cite{Aaronson2013a}
Suppose there exists a classical algorithm that takes as input a description of $U$ and $T$, as well as an error bound $\epsilon$, and efficiently samples from a probability distribution $\mathcal{D}^{'}_U$ such that $|| \mathcal{D}_U - \mathcal{D}^{'}_U ||< \epsilon$. Then, the polynomial hierarchy collapses.
\end{conjecture}
The original result actually consists of two major parts. The first is a highly nontrivial theorem that connects the ability to sample from $\mathcal{D}^{'}_U$ and the ability to obtain a suitable approximation to the permanent of Gaussian-random matrices, which already contains important implications for the underlying complexity classes. The second is a set of natural conjectures about properties of the permanents of such random matrices. At the risk of unfairly oversimplifying the authors' achievements, we condensed their main theorem and conjectures into one claim that is most convenient for our purposes\footnote{Which, of course, does not have the status of Theorem as it includes unproven-yet-plausible conjectures as ingredients.}. In the next few sections, we will describe some of these intermediate steps in more details.

Besides its robustness in terms of the allowed type of approximation, this model has the feature of suggesting a very simple and straightforward experimental test. Any ``convincing'' demonstration of a scalable quantum computer performing, say, Shor's algorithm, is far beyond the current experimental capabilities---to this day, the record for factoring a number with a quantum computer is 21 \cite{Martinlopez2012}. On the other hand, as we will explain shortly, \conj{BSishard} suggests that a linear-optical experiment with, say, $\sim 20$ photons in $\sim 400$ modes, would already be in a regime that would take noticeably long for a classical computer to simulate. Although an experiment in this range is also beyond current technology, it is within much closer reach than general scalable quantum computing (especially given its simpler resource requirements), thus providing an intermediate milestone for the field. In fact, that is the content of part of the results reported in \chap{bosonnew}, as we were part of a collaboration that performed this exact experiment with $3$ photons in linear-optical networks of $5-9$ modes.

We now discuss some of the most relevant aspects, for this thesis, of the BosonSampling task and of \conj{BSishard}.

\subsection{Permanents and (inefficient) classical simulation} \label{sec:bosonreview_c_simulation}

Throughout \chap{introduction}, at several points, we suggested a connection between the fact that bosons evolve according to permanents and the computational complexity of BosonSampling. Although the post-selection-based proof given in \sec{bosonreview_b} at no point used this connection explicitly, it is crucial for the approximate case. 

Recall from the previous section that $\mathcal{D}_U$, the output distribution of the BosonSampling experiment, relies on calculating the permanents of $n \times n$ sub-matrices of $U$. Curiously, it is also possible to write the same transition probabilities for classical particles in terms of permanents. It is not hard to show that, if the photons are \emph{distinguishable}, in which case they behave much like classical billiard balls for the purpose of this model, the corresponding distribution is given by
\begin{equation} \label{eq:MUclassical}
\underset{\mathcal{M}_U}{\textrm{Pr}}[S] = \frac{\textrm{perm}(|U_{S,T}|^2)}{s1! \ldots s_m!},
\end{equation}
where $|A|^2$, in this case, represents the matrix obtained from $A$ by taking the element-wise squared absolute value. It is curious that \eqs{DUquantum}{MUclassical} can be written so similarly\footnote{Note that the $t_i!$ terms only appear in \eq{DUquantum}, not in \eq{MUclassical}, but they are all equal to 1 if, as we assume here, the input contains no mode with more than one photon.}, but it is the difference between them that is the crucial point of this model. 

As mentioned before (cf.\ \sec{introduction_d}), it is a well-known fact that exactly computing the permanent of a $\{0,1\}$-matrix \cite{Valiant1979} is a $\#$P-complete problem . However, an algorithm due to Jerrum, Sinclair, and Vigoda, shows that it is possible to \emph{approximate} the permanent of a matrix of nonnegative entries in polynomial time \cite{Jerrum2004}. The same does not hold for complex matrices: it can be shown  \cite{Aaronson2013a} that approximating the permanent of arbitrary complex matrices to within a constant factor is $\#$P-complete. The intuition behind this is the following: each permanent consists, a priori, of a sum of an exponential number of terms. If all terms are positive, it is reasonable to expect that their sum, even if hard to compute exactly, might be easy to approximate. However, if the terms are complex numbers, it is possible that most of them cancel out, and all that is left is an exponentially small residue, which is much harder to approximate. This also has a profound implication in terms of computational complexity of physical systems---as can be checked in \eqs{DUquantum}{MUclassical}, this distinction between complex and nonnegative matrices is precisely what separates the classical and quantum regimes of BosonSampling.

It is tempting to think, at this point, that linear-optical devices can be used to solve $\#$P-complete problems, but that is misleading. As mentioned, the $\#$P-hardness is related to approximating an exponentially-small residue. This will be better formalized in the next sections, in a natural conjecture that matrices of Gaussian random entries provide $\#$P-hard instances of the (approximate) permanent problem. For now we just point out that, typically,  all permanents involved in the proposed experimental setup will be exponentially small. As explained in \sec{introsimul}, since the quantum device is inherently probabilistic, the values of the permanents can only be approximated by repeating the experiment several times, and constructing the relative frequency tables. If the experiment is only repeated a polynomial number of times, this can only approximate each permanent to polynomial precision---but, since they are in fact exponentially small, this clearly is not enough, as $0$ is already also a good approximation up to polynomial precision. Thus, even though the BosonSampling device somehow draws its computational complexity from the permanents that Nature is ``computing'', it cannot actually be used to compute these values itself (this, of course, would be too good to be true).

It is also possible to estimate how hard it would be to classically simulate the proposed experiments, in the interest of providing a rough regime of computational interest. It seems very hard to completely certify the operation of a BosonSampling device (we will return to this matter shortly), and thus the interesting regime for experimental demonstrations of the model would be that where classical computers can still simulate the experiment, but take noticeably long to do so. It is estimated that this corresponds to a regime of $20-40$ photons, evolving in an interferometer of $400-1000$ modes \cite{Aaronson2013a}. A typical desktop computer, running a standard algorithm for the permanent on the Mathematica $\copyright$ software, computes a $15 \times 15$ permanent in 20ms, a $25 \times 25$ permanent in 10s, and a $35 \times 35$ permanent in an estimated 10hrs\footnote{An algorithm developed by the Mathematica community, motivated by a forum user interested in BosonSampling, can be found in \url{http://mathematica.stackexchange.com/questions/38177/can-compiled-matrix-permanent-evaluation-be-further-sped-up}. The estimate I reported combines the reported benchmark values of their algorithm and a fitting for larger values, where computer memory also becomes an issue.}. This might still have room for optimization, but an estimate of $n \approx 50$ will likely remain nontrivial for classical machines in the near future. In an experiment with $n > 100$, for example, it would be impossible to check the working of the device without circularly assuming that the physical laws hold in the first place, at least in the current state of both theory and technology.

However, a classical simulation consists of more than calculating a single $n \times n$ permanent. If we want to algorithmically generate samples from a given distribution, we need, a priori, to be able to calculate all probabilities in the sample space. Recall that, in an $n$-photon, $m$-mode system, for a fixed input, there are $\binom {m+n-1} {n}$ different outputs, thus requiring the calculation of an exponential number of permanents of $n \times n$ sub-matrices of the $m \times m$ unitary $U$. To illustrate this, notice that in the 20-photon and 400-mode regime, we would have $\binom {419} {20} = 7.2 \times 10^{33}$ possible outcomes. This suggests that the task of \emph{simulating} a BosonSampling device might be, itself, exponentially harder than the task of \emph{verifying} it, which is already expected to be hard, and this might bring the computationally-interesting experimental regime a few levels down\footnote{Interestingly, even though the sample spaces are of the same size, the ``classical'' regime does not suffer from this problem. In that case, since the photons are all distinguishable, an efficient algorithm could simply simulate each photon at a time.}. The theoretical standing of such a claim, however, is still very nebulous. On one hand, the result of \cite{Aaronson2013a} only guarantees that simulating BosonSampling is as hard as calculating a single $n \times n$ permanent. It is possible that a simulation algorithm could exploit the constraints between the different probabilities (e.g.\ the unitarity of $U$) to perform the simulation without having to calculate all the individual permanents. On the other hand, no such algorithm seems to be known \cite{Aaronsonprivcomm}. Still, even if an algorithm is developed that requires calculating only a small polynomial number of permanents, it might improve the boundaries of the interesting range of experimental parameters.

\subsection{Random Interferometers} \label{sec:bosonreview_c_interf}

An aspect of BosonSampling that we stealthily introduced without explanation in \sec{bosonreview_c_motiv} was the need for the interferometer unitary $U$ to be Haar-random. The reason for this is a conjectured worst-case/average-case equivalence for the BosonSampling task. One issue with the exact BosonSampling model of \sec{bosonreview_b} was that it did not provide a very natural instance of the problem, amenable to being implemented experimentally. However, if BosonSampling with a random interferometer is typically as hard as one defined by a complex concatenation of the circuits of \fig{KLMscheme}, this greatly simplifies experimental efforts\footnote{Note that there are problems where the random instances are typically easy. One example would be factoring: a random integer has a very high probability of being easy to factor---in fact, there are no NP problems with worst-case/average-case equivalence.}. This result is almost achieved by the authors in \cite{Aaronson2013a}. In fact, what they do is posit the following Permanent-of-Gaussians conjecture:

\begin{conjecture} \label{conj:POG} \cite{Aaronson2013a}
Given as input a random matrix $X$, whose elements are i.\ i.\ d.\ random Gaussian complex variables (with mean 0 and variance 1), together with error bounds $\epsilon$, $\delta >0$, the problem of estimating perm$(X)$ to within error $\pm \epsilon |$perm$(X)|$ with probability at least $1-\delta$ over $X$ in poly$(n, 1/\epsilon,1/\delta)$ time is $\#$P-hard.
\end{conjecture}

They then prove that, if \conj{POG} is true\footnote{Actually, they need a second conjecture, called the Permanent Anti-Concentration Conjecture, for the purpose of connecting two distinct notions of ``approximating the permanent'', but it is much too technical to discuss here.}, efficient classical simulation of a BosonSampling device with a Haar-random interferometer implies a collapse of the polynomial hierarchy (i.e.\ our statement of \conj{BSishard}). 

What connects the Gaussian and the Haar ensemble is the fact that if $m$ is sufficiently larger than $n$ [more specifically, $m = \textrm{O} \left(n^5 / \delta \log^2 n/\delta \right)$], then $n \times n$ sub-matrices of an $m \times m$ Haar-random unitary are close, in variation distance, to matrices of i.\ i.\ d.\ Gaussian variables. In simpler terms: if we only look at small sub-matrices of a uniformly-drawn $U$, they all look independently Gaussian, and only by looking at larger sub-matrices does the underlying unitarity constraint become apparent. This fact is then used to ``hide'' a given Gaussian matrix into a larger Haar-random matrix, which is given as input to the problem. This may seem like an odd requirement, but is in fact quite natural in complexity theory, where results need to cover even the worst imaginable scenario. Suppose that the BosonSampling device is completely adversarial (e.g., controlled by an untrusted third-party), and tries to trick us into believing that it is performing some classically-hard task. If we provide, as input to the problem, a Gaussian matrix $X$ of interest (i.e., the one described in \conj{POG}) \emph{hidden} in a larger Haar-random matrix $U$, we guarantee that the device has no knowledge of which permanent we care about in the first place and, in order to sample from a distribution that is close in total-variation distance to the ideal one, it must approximate most of the probabilities sufficiently well. 

This adversarial device also has an experimental interpretation, where it is better known as Nature. If an experiment samples from a distribution close in total variation distance to the ideal one, this means Nature must be ``computing'' most of the probabilities sufficiently well and so, with great likelihood, it is computing the permanent we desire sufficiently well too. Of course, there is still the delicate matter of proving, based only on experimental outcomes, whether the device really is operating so well---this question is complicated by the fact that this output typically consists of an exponentially large number of exponentially-unlikely outcomes. As it turns out, the distribution over the outcomes of a Haar-random interferometer is very flat, which will also play a role in the criticism that we will address shortly \cite{Gogolin2013}.

The actual construction of a Haar-random interferometer can be done in terms of $m(m-1)$ random linear-optical elements (i.e.\ beam splitters and phase shifters). In fact, by a procedure due to Reck \textit{et al.}\ \cite{Reck1994}, \emph{any} linear-optical device can be decomposed in terms of $m(m-1)$ two-mode transformations. We will defer an explicit construction to \chap{bosonnew}, when we also show the equivalent decomposition layout for integrated photonic chips, used in the experiments. Besides this, we will also consider other ensembles of random interferometers. For example, some of the experiments reported in \chap{bosonnew} use interferometers built out of sequential layers of balanced beam splitters alternated with layers of random phase shifters. We will present simulations suggesting that this alternative design approximates well the uniform ensemble for several figures of merit of interest. Such simplified interferometers are relevant, from an experimental point of view, as they are easier to construct, allowing us to benchmark several experimental techniques, while still not presenting any obvious reason to suggest that they are easy BosonSampling instances.

\subsection{Bosonic birthday paradox} \label{sec:bosonreview_c_BBP}

One result that shows up as an intermediate step for the main BosonSampling result of \cite{Aaronson2013a}, but which is also of independent interest, is the so-called bosonic birthday paradox \cite{Aaronson2013a, Arkhipov2012}. The ``classical'' birthday paradox relates to the following question: given $n$ people distributed uniformly and independently into $m$ birth dates, what is the probability of observing a collision, i.e., of observing two or more people sharing the same birthday? Note that the question can be analogously phrased in terms of balls thrown uniformly into bins, or distinguishable photons into output modes of a uniformly random interferometer. It is called a paradox due to people's natural tendency to overestimate this probability---as an example, in a calendar of $m=365$ days, 99\% probability is obtained with $n=57$ people, and 50\% is reached with $n=23$ people. A simple combinatorics problem gives the probability of observing at least one collision in the classical case, $P_c$, as
\begin{equation} \label{eq:classicalBP}
P_c = 1 - \prod_{k=1}^{m-1}(1 - k/m) \approx 1 - e^{-n(n-1)/2m}.
\end{equation}
This expression only holds for $n \leq m$, evidently, as otherwise the probability is simply 1. It also shows that, in the asymptotic ($m \rightarrow \infty$, $n \rightarrow \infty$) limit, if $m > n^2 / 2 \ln 2$ then the probability of observing \textbf{no} collisions, which is the regime of greater interest to us, is greater than 1/2. 

The bosonic birthday paradox, then, is the analogous result for indistinguishable bosons. That is, by inputting $n$ photons into a uniformly random interferometer, initially in some standard state (say, $\ket{1, \ldots, 1, 0, \ldots, 0}$), what is the expected probability of observing a collision? Recall, from \sec{introduction_b}, that bosons have a natural tendency to bunch together, and thus one could expect them to behave very differently from classical particles. Curiously, however, it was shown \cite{Aaronson2013a, Arkhipov2012} that their behavior is not that different. More specifically, any input Fock state when evolved according to an ensemble of uniform-random interferometers is taken into the maximally mixed state of all possible configurations. To clarify, we are not referring to a Fock state evolving according to one specific randomly-sampled interferometer---that is always a pure state---but rather the description of a Fock state evolving according to an ensemble of interferometers weighed according to the Haar measure.

Then, a combinatorics problem, similar to the one done in the classical case, gives the probability $p_q$ of observing a collision in the quantum case as
\begin{equation}
p_q = 1 - \prod_{k=1}^{m-1}\frac{1-k/m}{1 + k/m} \approx 1 - e^{-n^2/m}.
\end{equation}
From this expression, we see that, in the asymptotic limit, the no-collision outcomes occur with greater than 1/2 probability if $m > n^2 / \ln 2$. That is, while identical photons tend to bunch more than distinguishable particles, the asymptotic behavior is the same, differing only by a constant factor.

One interesting consequence of this result is that, in the limit $m >> n^2$, the output of the BosonSampling experiment is dominated by $n$-photon coincidence outcomes, that is, outcomes in $\Phi_{m,n}^{*}$ (cf. \sec{bosonreview_c_motiv}). The relevance of this fact is twofold: first, the probabilities of collision outcomes involve calculating permanents of matrices with repeated columns, which are easier to compute than those with no repetitions\footnote{As a limiting case, consider a matrix consisting of one column repeated $m$ times. Its permanent is simply the product of elements in one column, times a constant factor.}. Thus, in this regime the outputs that dominate are precisely those corresponding to the hardest permanents to compute. Second, coincidence outputs correspond to a much simpler experimental setting. Whereas measuring a general output requires detectors that discriminate $0, 1, 2, \ldots, n$ photons, for no-collision outputs it suffices to use only ``bucket detectors'', that is, detectors that only distinguish between states with zero or nonzero photons. Thus, in this regime, one can simply measure the outputs with bucket detectors, and ignore any outcome that does not contain the $n$ photons in $n$ different modes---the bosonic birthday paradox guarantees, in this case, that only a negligible portion of the outcomes will have been discarded.

In \chap{bosonnew}, we report an experimental demonstration of the bosonic birthday paradox, with $3$ photons interfering in photonic chips of up to $20$ modes. We also provide a more refined theoretical result concerning bosonic bunching in these devices, and observe it experimentally.

\subsection{BosonSampling certification} \label{sec:bosonreview_c_certif}

One of the major open questions for BosonSampling, since the original paper, is the matter of certification of the device. More specifically, if one is given a black-box that allegedly samples from distribution $\mathcal{D}^{'}_U$ satisfying the conditions of \conj{BSishard}, is it possible to test that claim? Note that, so far, there is no known practical problem solvable in this model, and thus its motivation is in great part academic, as a demonstration that quantum devices can outperform classical ones in \emph{some} task of ``simple'' experimental implementation. Thus, if for some reason the certification of the device turned out to be impossible, this would deliver a blow to the motivation of model as it stands---a skeptic could simply claim that the laws of physics change for $> 100$ photons, say, in a subtle way as to be impossible to detect but sufficient to make the BosonSampling task easy. Note that arbitrary quantum computing does suffer from this problem, to some extent: any problem in BQP $\cap$ NP can obviously be classically verified, but it is widely believed that BQP is not strictly contained in NP \cite{Aharonov2008a}, and there is no known way to classically verify the solution to an arbitrary BQP-complete problem. An interesting partial solution to this problem was given \cite{Aharonov2008a, Broadbent2009}, in terms of what is known as an (quantum) interactive proof protocol. It was shown that, if the verifier (e.g.\ us) wants to check whether an untrusted prover (e.g.\ the experimental device) is actually performing the correct quantum computation, and if the verifier has access to a very small trusted quantum device, there is a protocol where the two parties exchange rounds of (classical and quantum) information such that the prover convinces the verifier of its authenticity. It is an interesting open question whether the small trusted quantum device can be eliminated and the certification done purely by classical means. It would also be interesting to investigate whether, for BosonSampling, some analogous result holds.

This hardness of certification is the grounds for one criticism made to the model by Gogolin \textit{et al.}\ \cite{Gogolin2013}, that ruled out this possibility for a certain class of classical algorithms known as symmetric algorithms. More specifically, a symmetric algorithm is one that does not make use of the outcomes' labels, only their multiplicities. The authors justify this by saying that, since the output probabilities are $\#$P-hard to calculate, the verifier does not have any information about the output distribution $\mathcal{D}_U$, and thus the outcome's labels (and, with it, the actual expression of $U$) should mean nothing to them. They proceed to argue that, if the interferometer is sampled from the Haar distribution then, with overwhelming probability, the distribution $\mathcal{D}_U$ of the corresponding experiment is $\epsilon$-flat (that is, all of its probabilities are smaller than $\epsilon$), for some $\epsilon$ that is exponentially small in $n$. This reiterates what we mentioned before, that all permanents encoding outcome probabilities are typically exponentially small. Finally, they show that no symmetric algorithm can distinguish whether an output was sampled from $\mathcal{D}_U$, the BosonSampling distribution, or from $\mathcal{U}$, the uniform distribution over $\Phi_{m,n}^{*}$\footnote{In this result, the distributions are only taken over the no-collision outcomes ($\Phi_{m,n}^{*}$), rather than all possible outcomes ($\Phi_{m,n}$), but in the regime of interest of BosonSampling this distinction is not too important, as these are the outcomes that dominate anyway.}, with only a polynomial number of samples. This is, in fact, a quite natural result: if the distribution is exponentially flat and the experiment is only repeated a polynomial number of times, no outcome is expected to happen more than once, and the table of multiplicities consists only of a list of observed outcomes whose labels are ignored by the symmetric algorithm. In this setting, the verifier obviously cannot distinguish the output from a uniform distribution. Note that, after an exponential number of samples, the table of multiplicities will start to resemble the ``shape'' of the BosonSampling distribution, and will start to deviate from the uniform distribution.

However, symmetric algorithms are overly restrictive. The fact that the probabilities in $\mathcal{D}_U$ cannot be computed efficiently does not imply, in any way, that no useful information can be efficiently obtained from the unitary $U$ itself. Furthermore, assuming that the verifier always has information about the matrix elements of $U$ is an arguably realistic modeling of current experimental efforts since, as we will describe in \chap{bosonnew}, there is a procedure to experimentally reconstruct $U$ from a polynomial number of single- and two-photon measurements \cite{Laing2012}. In fact, in a follow-up response to the criticism of \cite{Gogolin2013}, the authors of the original paper reported a procedure \cite{Aaronson2013b} to distinguish the output of a BosonSampling device from the uniform distribution in polynomial time, by using information efficiently computable from $U$. The procedure works as follows:

\begin{itemize}
\item[(i)] Perform a run of the experiment. Suppose that output $\ket{S}$ was observed, for a fixed input $\ket{T}$\footnote{Recall that the distribution $\mathcal{D}_U$ is defined also by the choice of fixed input $\ket{T}$.}.
\item[(ii)] Compute the product of the squared row-norms of $U_{S,T}$. That is, compute the quantity $R(U_{S,T})$, where, for an $n \times n$ matrix $X$, $R(X)$ is defined as
\begin{equation}
R(X) = \prod_{i=1}^{n} \sum_{j=1}^n |x_{ij}|^2.
\end{equation}
\item[(iii)] If $R(U_{S,T})$ is greater than a threshold $t := (n/m)^n$, ``guess'' that the sample was taken from $\mathcal{D}_U$. Otherwise, guess that the sample was taken from $\mathcal{U}$.
\end{itemize}

Note that $R(X)$ can be in fact classically computed in polynomial time. What the authors of \cite{Aaronson2013b} showed is that, if $U$ is taken from the Haar ensemble, and if $m$ is sufficiently larger than $n$ (roughly in the same regime necessary for the sub-matrices of $U$ to be approximately Gaussian, as explained in \sec{bosonreview_c_interf}), then, with high probability, the above procedure only needs a constant number of samples to convince the verifier that the underlying distribution is not $\mathcal{U}$. More specifically, they show that, for any experimental outcome $S$,
\begin{equation}
\underset{S \sim \mathcal{D}_U}{\textrm{Pr}}[R(U_{S,T}) \geq t] - \underset{S \sim \mathcal{U}}{\textrm{Pr}}[R(U_{S,T}) \geq t] \geq \frac{1}{9},
\end{equation}
with high probability over the uniform distribution of $U$'s. This also implies that $\mathcal{D}_U$ and $\mathcal{U}$ have constant total variation distance. The intuition behind this result is that the quantity $R(U_{S,T})$ is mildly correlated to the permanent of $U_{S,T}$. This correlation is not sufficient for one to approximate the actual value of the permanent (nor should we expect it to be), but it is still enough to convince the verifier, after a constant number of samples, that the distribution of outcomes is carrying \emph{some} information about $U$, rather than just some gibberish produced by a uniform distribution. 

This does not prove that one can completely certify the BosonSampling device. In fact, the use of the words certify and validate are a certain abuse of terminology, as there are several different things they can mean. The procedure above only validates the device against the uniform distribution as a null hypothesis, but not against \emph{every} other classically-samplable distribution. One example of ``classical'' distribution for which $R$ cannot be used is $\mathcal{M}_U$ (i.e.\ the corresponding distribution observed for distinguishable photons), as shown in \cite{Aaronson2013b}. This should be expected, as $R$ only depends on the squared absolute value of the matrix elements of $U$, and thus is insensitive to photon distinguishability (cf.\ discussion on \sec{bosonreview_c_simulation}). A very important open question, posited in the original BosonSampling paper, is whether, for every authentic BosonSampling distribution $\mathcal{D}_U$, one can construct a classically-samplable distribution that is indistinguishable from it in polynomial time.

Finally, we point out that the authors of \cite{Gogolin2013,Aaronson2013b} are mostly interested in certification in the asymptotic limit---that is, whether specific certification tasks are in BPP or not. However, experimental BosonSampling will only be interesting in a limited range anyway, of no more than roughly $\approx 50$ photons, where the individual probabilities can still be computed. In this case, a standard likelihood ratio test \cite{LivroCover}  might provide a better rate of convergence for convincing the verifier, especially given the constant total variation distance between $\mathcal{D}_U$ and $\mathcal{U}$, and can in principle be used to rule out any ``classical'' hypothesis. This does not have the same theoretical standing as the results of \cite{Aaronson2013b}, since the likelihood ratio test requires computing the probabilities, but might be of greater interest to experimentalists in small-scale implementations. In \chap{bosonnew} we report an experimental validation of a BosonSampling device using both the estimator $R$ (against the uniform distribution) and likelihood ratio test (against both the uniform distribution and $\mathcal{M}_U$).

\newpage
\chapter{New results: Matchgates} \label{chapter:fermnew}

In this chapter, I present our results for the computational power of matchgates. It is based almost exclusively on my co-authored papers Refs.\ \cite{Brod2011,Brod2012}, published in \textit{Physical Review A}, and Ref.\ \cite{Brod2013}, recently accepted for publication at \textit{Quantum Information and Computation}. 

Recall from \chap{fermreview} that matchgates are unitaries of the form
\begin{equation*} 
G(A,B) = \begin{pmatrix}
A_{11} & 0 & 0 & A_{12} \\
0 & B_{11} & B_{12} & 0 \\
0 & B_{21} & B_{22} & 0 \\
A_{21} & 0 & 0 & A_{22}
\end{pmatrix},
\end{equation*}
satisfying the additional constraint that det$A=$det$B$ (which we call the \emph{determinant condition}, for future use). Recall also that matchgates display the following two computational regimes:

\begin{itemize}
\item[(i)] If the input is an $n$-qubit product state, followed by a circuit of a poly$(n)$ nearest-neighbor matchgates, and the final measurement is a computational basis measurement on any qubit $k$, the output of this measurement can be computed efficiently by \eq{expectedZ}, as shown in \sec{fermreview_a}.

\item[(ii)] An arbitrary quantum circuit can be simulated by a circuit consisting of a computational basis input, a sequence of matchgates on nearest and next-nearest neighbors, and a final computational basis measurement on the first qubit. This simulation induces only a linear overhead in the number of qubits and operations, as shown in \sec{fermreview_b}.
\end{itemize}

Throughout this chapter I show some ways in which the underlying (and often implicit) restrictions delimiting these regimes can be modified, and what the corresponding change in computational power of the model is.

First, in \sec{fermnew_a}, I show how quantum universality can be recovered in (i) by adding certain other unitary operations. Often, the gap from classical to quantum computational power is bridged by adding some particular gate to the set of allowed operations---for example, it is well-known that circuits of single qubit gates\footnote{Acting on input product states and with final measurements on the computational basis.} are classically simulable, but become universal for quantum computation when complemented with any entangling gate \cite{Bremner2002}. Another example is the Toffoli gate, which is universal for classical computation \cite{Fredkin1982}, but becomes universal for quantum computation by adding any non-basis-preserving gate \cite{Shi2003}. In both cases, there is clearly a quantum resource missing (entanglement and quantum superposition, respectively) that is provided by the added operation. It is curious that, for matchgates, this gap is bridged by a limited use of the $\swap$ gate, just enough to allow for interactions between next-nearest neighbors. In contrast to, say, entanglement, the $\swap$ gate does not seem to provide a particularly ``quantum'' resource, and it is often taken for granted, especially in physical implementations of quantum computing with flying qubits (that is, the particles which carry the information can be moved around freely, as is the case e.g.\ with photons). This is even more curious if we recall that $\swap = G(I,X)$---it is strikingly similar to a matchgate, only failing to satisfy the determinant condition. This raises the question of whether the resource provided by the $\swap$ gate is not, in fact, ``hidden'' in this difference between the determinants of the inner and outer matrices. In \sec{fermnew_a} I show that this is indeed the case, and any gate of the type $G(A,B)$ where det$A\neq$det$B$ is sufficient to extend matchgates to quantum universality \cite{Brod2011}. I will analyze this determinant condition from the point of view of nonlocal invariants of unitary gates \cite{Makhlin2002, Zhang2003} and, using the Jordan-Wigner transformation, show how it can be understood in the fermionic formalism.

In \sec{fermnew_b} I study this gap in computational power in terms of the interaction geometry. Note that, as usual in quantum computation, all results mentioned so far in \chap{fermreview} implicitly assume that the qubits of the circuit are arranged on a path\footnote{A path is a graph where each vertex has two neighbors, except for two vertices at the endpoints. It is also often called a 1D array, a linear graph, or a 1D chain. Throughout this thesis, we just call it a path.}. This is often a natural assumption, especially since the $\swap$ gate can be implemented using a sequence of three nearest-neighbor $\textsc{cnot}$ gates, thus there is no loss of generality in restricting a \emph{universal} set of operations to act only on nearest neighbors\footnote{To be fair, the effect of the geometrical arrangement of the qubits has also been studied in the context of distributed quantum computation (see e.g.\ \cite{Hirata2011,Beals2013}), but the goal there is to reduce the overhead induced by having to swap the qubits around, whereas here the role of the geometry is to actually bridge the gap between classical and quantum computational power.}. However, for nearest-neighbor matchgates this clearly is no longer true, as they cannot implement the $\swap$ gate on their own. Hence, we can ask whether the computational power changes if we modify the underlying geometrical arrangement of the qubits. In other words, we know that nearest-neighbor matchgates are classically simulable on a path, but does this remain true if the qubits are arranged on a cycle, a binary tree, a square or hexagonal lattice, and so on? In \sec{fermnew_b} I show that, in fact, matchgates are universal acting on any graph that is not a path or a cycle and, furthermore, that they are classically simulable on a cycle \cite{Brod2012,Brod2013}, thus completely closing this particular question. I also show that, if we restrict the allowed operations only to the XY interaction (cf.\ \sec{ferm_reviewXY}), the exact same dichotomy holds.

Finally, \sec{fermnew_c} is devoted to some concluding remarks, in particular on the relevance of the results proved in this chapter to experimental efforts, as well as possible questions that remain open.

\section{Matchgates vs. \PP\ gates} \label{sec:fermnew_a}

In this section, I address the computational power of matchgates when supplemented by other single- or two-qubit gates, and it is mostly based on the results published in \cite{Brod2011}. Since we want to single out the effect of complementing the set with other gates, our starting assumptions are those under which matchgates are classically simulable---more specifically, throughout this section we only consider computational basis inputs, computational basis measurements, and gates acting on nearest-neighbor qubits arranged on a path.

Recall, from \sec{fermreview_b}, that one possible two-qubit gate enabling universal quantum computation with matchgates is the $\swap=G(I,X)$ gate, that is remarkably similar to a matchgate, failing only to satisfy the determinant condition. From this point on, we refer to general $G(A,B)$ gates as \emph{parity-preserving} (\PP) gates whether they satisfy the determinant condition or not, with matchgates corresponding to those restricted \PP\ gates that do satisfy it. It is clear that the set of \PP\ gates as a whole is, in fact, universal for quantum computation.

In \sec{fermreview_a}, we gave an interpretation of the $\swap$ operation as allowing matchgates to act between more distant qubits, rather than just nearest-neighbors. We can then ask whether this feature can be simulated by other gates similar to the $\swap$, such as the $\fswap = G(Z,X)$\footnote{The $\fswap$ gate already appeared in \chap{fermreview}, in \eq{fswaponXX}.} or the $\iswap = G(I,iX)$. Both gates are similar to the $\swap$ in the sense that they can also be seen as swapping operations when acting on the computational basis (in fact, the $\fswap$ was already used in \sec{fermreview_a} to swap the roles of two fermionic modes). However, they also differ from the $\swap$ in that they are entangling operations, and they are matchgates. This last fact immediately tells us that they cannot replace the $\swap$, as otherwise it would contradict the classical simulability of nearest-neighbors matchgates (we return to this point in \sec{fermnew_b} where we will show that we can replace the $\swap$ by the $\fswap$ or $\iswap$ gates in a limited sense, if we additionally change the underlying arrangement of the qubits). 

To better characterize the set of \PP\ gates and the role of the determinant condition for their computational power, we now review some results about nonlocal parameters of two-qubit gates. 

\subsection{Nonlocal parameters} \label{sec:nlpar}

We begin this section by stating a result fundamental for the discussion that follows:
\begin{theorem} \label{thm:NonlocalTheorem}\cite{Kraus2001, Khaneja2001}
Any two-qubit gate $U \in \SU{(4)}$ can be written as 
\begin{align} \label{eq:NonlocalPar}
U & =  (U_1 \otimes U_2)  U_{NL} (V_1 \otimes V_2) \notag \\
& = (U_1 \otimes U_2) e^{i \left ( a X \otimes X + b Y \otimes Y + c Z \otimes Z \right ) } (V_1 \otimes V_2) 
\end{align}
where $U_i$ and $V_i$ are single-qubit gates on qubit $i$, and $a$, $b$ and $c$ are real parameters in the interval $\left[ 0,\frac{\pi}{2} \right)$.
\end{theorem}

In other words \thm{NonlocalTheorem} states that, out of the 15 parameters that define a general $4 \times 4$ unitary matrix (up to a global phase), only 3 are truly nonlocal in the sense of requiring some interaction between the two qubits, while the remaining 12 can be implemented by two sets of single-qubit gates, before and after this nonlocal ``core''. \eqbeg{NonlocalPar} suggests a connection to the representation of matchgates developed in \sec{fermreview_a}, but this discussion will be postponed until after we obtain an explicit parameterization of general \PP\ gates similar to \eq{NonlocalPar}.

This set of nonlocal parameters is convenient because it gives an explicit way to parameterize two-qubit gates, but we must be careful before ascribing special physical significance to any of them. One reason is the non-uniqueness of the decomposition: we can implement permutations of the type $a \leftrightarrow c$ using only single-qubit gates (in this case, $H \otimes H$). This ambiguity can be eliminated by imposing some restriction such as $a \geq b \geq c$ (alternatively, for a geometrical approach in terms of what is known as the Weyl chamber, see \cite{Zhang2003}). For our purposes this ambiguity will not be relevant, and thus we make no such restriction. This ambiguity also means that $\{ a,b,c \}$ are not, strictly speaking, local invariants\footnote{A local invariant is defined precisely as any property of a two-qubit gate that is invariant under local operations.}. However, any true local invariant can be written as a (symmetric) function of these parameters---one such example is the set of invariants defined by Makhlin in \cite{Makhlin2002}, as shown in \cite{Zhang2003}. Another example of particular interest in the discussion that follows is the entangling power, introduced by Zanardi \textit{et al.} in \cite{Zanardi2000}, which we now define.

Consider a two-qubit Hilbert space $\mathcal{H}=\mathcal{H}_1 \otimes \mathcal{H}_2$, and a state $\ket{\Psi} \in \mathcal{H}$. Also consider the \emph{linear entropy}, an entanglement measure of $\ket{\Psi}$ defined by
\begin{align*}
E(\ket{\Psi}):= & 1-\textrm{tr}_1 \rho^2, \\
\rho := & \textrm{tr}_2 \ket{\Psi} \bra{\Psi}
\end{align*} 
Then, the entangling power $e_p(U)$ of a two-qubit unitary $U$ is defined as the average linear entropy generated by the action of $U$ on all product states $\ket{\psi_1} \otimes \ket{\psi_2}$ \cite{Zanardi2000}:
\begin{equation*}
e_p(U) := \overline{E(U \ket{\psi_1} \otimes \ket{\psi_2})}^{(\psi_1, \psi_2)},
\end{equation*}
where the bar denotes average with respect to some probability distribution $p(\psi_1, \psi_2)$. By using the identity
\begin{equation} \label{eq:tridentity}
\textrm{tr} (AB) = \textrm{tr} (A \otimes B \; \swap)
\end{equation}
we obtain, after some simple manipulations, that
\begin{equation*}
e_p(U) = 2 \; \textrm{tr} \left[ U^{\otimes 2} \Omega_p U^{\dagger \otimes 2} P_{13}^{-} \right ], 
\end{equation*}
where $P_{13}^{-}$ is the projector on the antisymmetric subspace of qubits 1 and 3 [the trace is taken over a doubled Hilbert space as a consequence of \eq{tridentity}], while
\begin{equation*}
\Omega_p = \int d \mu(\psi_1, \psi_2) (\ket{\psi_1} \bra{\psi_1} \otimes \ket{\psi_2} \bra{\psi_2})^{\otimes 2},
\end{equation*}
and $d \mu$ denotes the measure over the space of product states induced by probability distribution $p(\psi_1,\psi_2)$.

It can be shown that, if the average is taken over the uniform distribution, the entangling power is both local invariant (i.e., it remains the same if $U$ is multiplied by single-qubit gates on either side), suggesting that it must be written as some function of the nonlocal parameters $\{ a,b,c \}$, and $\swap$ invariant (i.e., it is the same for $U$ and $U \cdot \swap$). If fact, simply by plugging \eq{NonlocalPar} in the above definition (and rescaling the value of $e_p(U)$ so it goes from 0, for local gates and $\swap$, up to 1, for perfect entanglers) we find that
\begin{equation} \label{eq:entanglement}
e_p(U) = 1 - \textrm{cos}^2 2a \; \textrm{cos}^2 2b \; \textrm{cos}^2 2c - \textrm{sin}^2 2a \; \textrm{sin}^2 2b \; \textrm{sin}^2 2c.
\end{equation}

It is clear from the expression above that any local gate has $e_p(U_1 \otimes U_2) = 0$. It can also be seen that $e_p(\cnot)=1$, since for the $\cnot$, $\{ a,b,c \}=\{\frac{\pi}{4}, 0, 0\}$, and that $e_p(\swap)=0$, since for the $\swap$ we have $\{\frac{\pi}{4}, \frac{\pi}{4}, \frac{\pi}{4}\}$, all results which were expected. The above expression is invariant under permutations of $\{a,b,c\}$, another previously anticipated feature.

Recall that arbitrary single-qubit gates together with any entangling gate form a universal set, and we can check that a gate $U$ is entangling by calculating $e_p(U)$. Can a similar characterization be given for matchgates? The $\swap$ gate, which seems to be responsible for boosting their computational power, cannot be implemented only by single-qubit gates, and thus must clearly have some nonlocal property (even though it is not entangling). We will show in the next section that this property can be characterized in terms of the $\{ a,b,c \}$ parameters above, and that it is central for the computational power of \PP\ gates.

\subsection{Extending matchgates with parity-preserving unitaries}

In this section I present our first result, regarding \PP\ unitaries that extend nearest-neighbor matchgates to universality. I start with an example to build intuition by analyzing a family of gates that interpolates between the $\swap$ and the $\iswap$. I then provide a characterization of general \PP\ gates in terms of nonlocal parameters, identifying their relation to the determinant condition and thus their role in the computational speedup provided by \PP\ gates relative to matchgates.

\begin{example} \label{ex:swapinterpol}

Consider the following family of \PP\ gates:

\begin{equation}
G(I, e^{i \tau} X) =
\begin{pmatrix}
1 & 0 & 0 & 0 \\
0 & 0 & e^{i \tau} & 0 \\
0 & e^{i \tau} & 0 & 0 \\
0 & 0 & 0 & 1
\end{pmatrix},
\end{equation}
where $\tau$ goes from 0 ($\swap$) to $\pi / 2$ ($\iswap$). Besides being a \PP\ gate, this gate is also a matchgate if (and only if) $\tau$ = $\pi/2$. Now recall the universality scheme described in \sec{fermreview_b} (cf. \fig{JozsaDem}): (i) encode each logical qubit state $\ket{0_L}$ ($\ket{1_L}$) into two physical qubits as $\ket{00}$ ($\ket{11}$), as per \eq{evenencoding}; (ii) implement any encoded single-qubit gate $A_L$ by a nearest-neighbor matchgate $G(A,A)$; (iii) implement any encoded $\cz$ gate using a $\swap$ followed by a $\fswap$; and finally (iv) measure the output in the computational basis. The only non-matchgate unitary in this scheme is a $\swap$ in \sfig{JozsaDem}{b}---let us replace it by $G(I, e^{i \tau} X)$. The resulting encoded two-qubit gate is
\begin{equation}
G(I, e^{i \tau} X) \cdot G(Z, X) = \begin{pmatrix}
1 & 0 & 0 & 0 \\
0 & e^{i \tau} & 0 & 0 \\
0 & 0 & e^{i \tau} & 0 \\
0 & 0 & 0 & -1
\end{pmatrix}.
\end{equation}
This matrix is diagonal, and thus preserves the encoding [cf.\ \eq{evenencoding}], as it should. We can also check that the nonlocal parameters [cf.\ \eq{NonlocalPar}] for this gate are $\{ 0 , 0, \frac{\pi}{4} - \frac{\tau}{2} \}$ and so, by \eq{entanglement}, its entangling power is $\cos^2 \tau$. This is consistent with what we expect: the entangling power is maximum when we use the \swap, but goes to 0 as the gate approaches the \iswap, in which case any circuit built out of these building blocks must be classically simulable---since then the \emph{physical} circuit is composed only of matchgates (alternatively, the \emph{logical} circuit is composed only of non-entangling gates).

However, perfect entanglers are not necessary for an universal set \cite{Bremner2002}. If arbitrary single-qubit gates are available, a gate that creates any (nonzero) amount of entanglement is sufficient for universal quantum computation. In view of this, we conclude that $G(I, e^{i \tau} X)$ is sufficient to achieve universal quantum computation, together with matchgates, if and only if $\tau \neq \pi/2$---or, equivalently, \emph{if and only if it is not a matchgate itself}. \qed
\end{example}

We have thus obtained a continuous family of gates akin to the $\swap$ that can replace it. While there is an inverse relation between entangling power of the physical gate and entangling power of the resulting logical gate (as evident by the extreme cases of the $\swap$ and $\iswap$), it will soon become clear that this is just an attribute of this particular gate family, and there is no such relation for \PP\ gates in general.

\ex{swapinterpol} suggests that the parameter $\tau$ in $G(I, e^{i \tau} X)$ is fundamental in obtaining a universal set. We will now formalize and generalize this intuition. We start with a convenient parameterization of general \PP\ matrices. It is well-known \cite{LivroNielsen} that any single-qubit unitary matrix can be parametrized as
\begin{equation*}
\begin{pmatrix}
\cos \theta e^{i (\beta + \alpha )} & i \sin{\theta} e^{i (\beta + \gamma )} \\
i \sin{\theta} e^{i (\beta - \gamma )} & \cos{\theta} e^{i (\beta - \alpha )}
\end{pmatrix},
\end{equation*}
where $\theta$, $\beta$, $\alpha$ and $\gamma$ are real parameters in the interval $[0,2 \pi )$. Note that the determinant of this matrix depends only on $\beta$. By parameterizing $A$ and $B$ in this way, the most general \PP\ matrix can be written as 
\begin{equation} \label{eq:GenPP}
G(A,B) = \begin{pmatrix}
\cos \theta \; e^{i (\beta + \alpha )} & 0 & 0 & i \sin{\theta} \; e^{i (\beta + \mu )} \\
0 & \cos \phi \; e^{i (-\beta + \gamma )} & i \sin{\phi} \; e^{i (-\beta + \nu )} & 0 \\
0 & i \sin{\phi} \; e^{i (-\beta - \nu )} & \cos \phi \; e^{i (-\beta - \gamma )} & 0 \\
i \sin{\theta} \; e^{i (\beta - \mu )} & 0 & 0 & \cos \theta \; e^{i (\beta - \alpha )}
\end{pmatrix}
\end{equation}
in terms of 7 real parameters: $\{ \theta, \alpha, \gamma, \phi, \mu, \nu, \beta \}$. Although $A$ and $B$ each have 4 free parameters, $G(A,B)$ is defined only up to global phase, and we have fixed the determinant such that det$A = 1/\textrm{det}B = e^{2 i \beta}$.

How much ``nonlocal freedom'' remains in the 7-parameter family above? In other words, how does the \PP\ restriction constrain the nonlocal parameters of $G(A,B)$? To answer this, let us explicitly write the nonlocal ``core'' defined in \eq{NonlocalPar}:
\begin{equation} \label{eq:NonlocalCore}
U_{NL} =
\begin{pmatrix}
\cos (a-b) \; e^{i c} & 0 & 0 & i \sin{(a-b)} \; e^{i c} \\
0 & \cos (a+b) \; e^{-i c} & i \sin{(a+b)} \; e^{-i c} & 0 \\
0 & i \sin{(a+b)} \; e^{-i c} & \cos{(a+b)} \; e^{-i c} & 0 \\
i \sin{(a-b)} \; e^{i c} & 0 & 0 & \cos (a-b) e^{i c}
\end{pmatrix}.
\end{equation}

Comparing \eq{GenPP} and \eq{NonlocalCore} we see that $U_{NL}$ is, itself, a \PP\ gate, with nonlocal parameters $\{ \frac{\phi + \theta}{2}, \frac{\phi - \theta}{2}, \beta \}$, which points us to two important facts. First, that general \PP\ gates can have any value of the nonlocal parameters (and, consequently, of any local invariant derived from them), hence any $\SU{(4)}$ matrix is locally equivalent to a \PP\ gate. Second, that the determinant condition defining matchgates ($\beta=0$) is equivalent to fixing one of these nonlocal parameters ($c=0$). Note that in this case there is no ambiguity in the definition of $c$---any permutation between $c$ and either $a$ or $b$ must be implemented by some non-\PP\ single-qubit gate, such as $H$. Thus, if we restrict our attention only to \PP\ gates, parameter $c$ indeed becomes a local invariant, and our characterization of matchgates as the subset of \PP\ gates with $c=0$ is meaningful.

Now compare \eq{GenPP} and \eq{NonlocalCore} again. Three of the independent parameters of $G(A,B)$ have been accounted for, in terms of $\{ a,b,c \}$. The remaining four parameters are given by individual phases, and can be obtained by multiplying whole rows and columns by phases (i.e., by applying $Z$-rotations to the left and to the right). Recalling from \sec{fermreview_a} that single-qubit $Z$-rotations are also matchgates, we can write \eq{GenPP} as
\begin{equation*}
G(A,B) = G \bigl ( R_Z(\tau_1), R_Z(\tau_2) \bigr ) \; U_{NL}  \; G \bigl ( R_Z(\tau_3), R_Z(\tau_4) \bigr ),
\end{equation*}
where $\{\tau_1,\tau_2,\tau_3,\tau_4\}$ are defined as $\{\alpha + \mu, \gamma + \nu , \alpha - \mu , \gamma - \nu \}$. To put this equation in an even more convenient form, recall from \eq{NonlocalCore} that $U_{NL}$ is an exponential involving Hermitian operators $X \otimes X$, $Y \otimes Y$, and $Z \otimes Z$. Since these operators commute, the $Z \otimes Z$ term can factor out in its own exponential. Furthermore, since  $Z \otimes Z$ also commutes with single-qubit $Z$ rotations, and since $X \otimes X$ and $Y \otimes Y$ generate matchgates [cf.\ \eq{JWquad2}], we can write the equation above as
\begin{align}
G(A,B) & = G \bigl ( R_Z(\tau_1), R_Z(\tau_2) \bigr ) \; G \bigl ( R_X(\theta), R_X(\phi) \bigr ) \; G \bigl ( R_Z(\tau_3), R_Z(\tau_4) \bigr ) \; e^{i c Z \otimes Z} \notag \\
& = G \bigl ( R_Z(\tau_1) R_X(\theta) R_Z(\tau_3), R_Z(\tau_2) R_X(\phi) R_Z(\tau_4) \bigr ) \; e^{i c Z \otimes Z}  \label{eq:PPdecomposition}
\end{align}

This decomposition for a general \PP\ gate connects nicely with our initial description of the generators of matchgates in \sec{fermreview_a}. It shows that any \PP\ gate $G(A,B)$ can be written as a matchgate multiplied by $e^{i c Z \otimes Z}$, where $c$ is directly connected to the relative phase between $A$ and $B$, and furthermore these two unitaries commute. It also reiterates our previous claim, that matchgates correspond exactly to the set of \PP\ gates where $c=0$ since, if we impose this condition on \eq{PPdecomposition}, what remains is a completely arbitrary matchgate\footnote{Again, up to an irrelevant global phase.}.

We can now return to the universality scheme of \sec{fermreview_b}. Recall that, there, the only non-matchgate was the $\swap$, and it only appeared in the following sequence:
\begin{equation} \tag{\ref{eq:CZL}}
\cz = \fswap \cdot \swap.
\end{equation}
That is, the $\swap$ was used in conjunction with a matchgate ($\fswap$) to implement an entangling two-qubit gate. This $\cz$ gate on the physical qubits corresponded precisely to a $\cz$ gate between the \emph{logical} qubits, hence supplying the entanglement necessary for the encoded universality. This gate is also diagonal, which is required so as to not disrupt the encoding of \eq{evenencoding}. At that point in \sec{fermreview_b} we interpreted the role of the $\swap$ as undoing an unwanted interchange of the qubit states induced by the $\fswap$. We will now reverse that intuition: in fact, the role of $\swap$ is to induce a relative phase between the odd and even parity subspaces (i.e.\ parameter $c$), but it also induces an unwanted exchange of the qubit states, which the $\fswap$ undoes. This intuition is motivated by the following Theorem (rephrased from \cite{Brod2011}):

\begin{theorem} \label{thm:PPTheorem}
Let $G(A,B)$ be any parity-preserving unitary, and $\mathcal{M}$ be the set of all matchgates together with $G(A,B)$. If $G(A,B)$ is not a matchgate (i.e.\ if det$A \neq$  det$B$), gates from $\mathcal{M}$ acting only on nearest-neighbor qubits are universal for quantum computation.
\end{theorem}

\begin{proof}
The encoding of the qubits, implementation of single-qubit gates and measurement of final output are all the same as in the proof of universality of \sec{fermreview_b}. We only provide the explicit construction for an entangling two-qubit gate. To that end, begin with the decomposition of \eq{PPdecomposition}, that we write as
\begin{equation*}
G(A,B) = G(A',B') e^{i c Z \otimes Z},
\end{equation*}
where $c$ is such that det$A=1/ \textrm{det}B$ = $e^{2 i c}$ [cf.\ \eq{GenPP}], and $A'$ and $B'$ are the ``$\SU(2)$ versions'' of $A$ and $B$, respectively (i.e., with global phases fixed such that the determinant is 1). Then a suitable entangling gate $E$ can be obtained simply by canceling out the unwanted part of $G(A,B)$
\begin{equation*}
E = G(A,B) \cdot G(A',B')^{\dagger} = e^{i c Z \otimes Z}.
\end{equation*}
From \eq{entanglement}, it is straightforward to see that $e_p(E)=\sin^2{2c}$. Hence, $E$ is (partially) entangling as long as $c \neq 0$, i.e., as long as $G(A,B)$ is not a matchgate, and it is well-known that partially entangling two-qubit gates suffice for universal quantum computation if arbitrary single-qubit gates are available \cite{Bremner2002}. This completes a universal set of encoded operations built only of nearest-neighbor gates in $\mathcal{M}$.
\end{proof}

We now justify our intuition described previously: the role of the $\swap$ in \eq{CZL} is to provide a relative phase between the odd and even parity subspaces, and the role of the $\fswap$ is that of undoing the unwanted exchange in the qubits. As \thm{PPTheorem} shows, this same reasoning applies to any \PP\ gate that we may want---if $G(A,B)$ is a \PP\ gate, its role in the universal scheme is to provide some nonzero value of the parameter $c$, and we must apply a $G(A',B')^{\dagger}$ operation to cancel out any unwanted effect it may have (that could, for example, disrupt the encoding).

\subsection{Discussion} \label{sec:fermnew_a_disc}

\thm{PPTheorem} shows that any non-matchgate \PP\ gate suffices to extend nearest-neighbor matchgates to universality, and the $\swap$ gate is just one particular example. Furthermore, it establishes a certain dichotomy---matchgates jump from classically simulable to quantum universal by relaxing the determinant condition by any amount\footnote{Obviously, any \emph{constant} amount. If the gate $G(A,B)$ has a parameter $c$ that somehow decreases with the size of the input, our results might not hold, but this seems a very unnatural assumption, both in theoretical and experimental terms.}, for any gate, and they do not display any intermediate power similar to that of IQP, BosonSampling, or constant-depth quantum circuits (cf.\ \sec{introduction_c}). In the next section, when we consider changes in the underlying graph, the jump in computational power will also be abrupt in this sense. 

So far, we have restricted our attention to \PP\ gates, as they are a natural generalization of matchgates. Are there other single- or two-qubit gates that can also fulfill this role? It is easy to see that matchgates together with arbitrary single-qubit gates form a universal set, since matchgates include many perfect entanglers \cite{Terhal2002}. The Hadamard gate alone is sufficient for this purpose, since, as mentioned in \sec{nlpar}, it can be used to implement the permutation $a \leftrightarrow c$. Thus, a matchgate with $a \neq 0$ conjugated by $H^{\otimes 2}$ is a non-matchgate \PP\ gate, satisfying the conditions of \thm{PPTheorem}. Other $H$-like gates that implement similar permutations (i.e. $a \leftrightarrow c$ and $b \leftrightarrow c$, but not $a \leftrightarrow b$) can be used in the same fashion. While this use of the $H$ to extend the power of matchgates is natural in our context, it is not particularly surprising: matchgates contain all single-qubit $Z$-rotations which, when added to the $H$ gate, generate arbitrary single-qubit gates, and we obtain a universal set even without encoding. This section's focus on \PP\ gates is natural as we encoded the logical qubits on the even parity subspace of two qubits, but other encodings could lead to different universal constructions, which in turn may help identify other unitaries that can complement the set of matchgates in this manner.

In this section, we have also pointed out that the parameter $c$, in the context of \PP\ gates, is qualitatively different from the other two, since we cannot implement the permutations $a \leftrightarrow c$ and $b \leftrightarrow c$. To further make this point, note that the subset of \PP\ gates with $a=0$ is already universal, since we can use single-qubit $Z$ rotations to implement the permutation $a \leftrightarrow b$, and so this subset trivially generates the whole set of \PP\ gates. This point can be understood easily if we return to the fermionic formalism. As mentioned in \sec{fermreview_a}, if we map matchgates to fermionic operators using the Jordan-Wigner transformations [cf.\ \eqs{JWquad1}{JWquad2}], we obtain Hamiltonians quadratic in the fermionic operators (i.e.\ noninteracting fermions). On the other hand, from \eq{JWc} we see that $Z_i Z_{i+1} = - c_{2i-1} c_{2i} c_{2i+1} c_{2i+2}$, which is a quartic operator, and so describes an interaction between the fermions. Thus, the break in the permutation symmetry among the three nonlocal parameters, which is not obvious in the algebraic formalism, turns out to have a very physical meaning in the fermionic formalism. In fact, in \cite{Bravyi2002}, Bravyi and Kitaev had already showed that the exp$\{ i \frac{\pi}{4} Z \otimes Z \}$ interaction can be used to implement a $\cz$ gate, which is then used to generate a universal set with fermionic modes. 

Our result is closely related to the one in \cite{Bravyi2002}, but also more general as we provide an explicit construction where \textit{any} non-matchgate $G(A,B)$ matrix can be used. This alternative approach can be useful whenever circuits of \PP\ gates arise in physical systems other than fermionic linear optics. For example, in the fermionic picture the $\fswap$ corresponds to a simple exchange of the fermionic modes, whereas the $Z \otimes Z$ Hamiltonian corresponds to an interaction between the fermions. In contrast, in the algebraic formalism followed throughout most of this chapter, these roles are reversed: the $\fswap$ is a maximally entangling operation, and so requires some interaction between the qubits, whereas the $\swap$ gate may require no interaction between the qubits at all, especially if the information carriers are flying qubits (as, for example, in a recent linear-optical implementation of matchgates \cite{Ramelow2010}). This interpretation of $c$ in terms of fermionic interaction is also not very natural in the context of the Heisenberg interactions described in \sec{ferm_reviewXY}. In that context, the $XY$ interaction has $c=0$ [cf.\ \eq{HeisenHamil}] and, if applied only on nearest neighbors, is classically simulable, whereas the exchange interaction has $c \neq 0$ and is universal\footnote{Note however that these two results neither imply nor are implied by our own.}. Just from looking at the Hamiltonians that generate the two interactions, it is not obvious what feature the exchange interaction has that is qualitatively different from the $XY$ interaction and that could grant this computational power---of course, by using the Jordan-Wigner transformation we can see how this relates to fermionic interactions, but \thm{PPTheorem} provides a much more natural and straightforward understanding of this distinction.

\section{Matchgates and XY interaction in graphs} \label{sec:fermnew_b}

So far, in this chapter, we have been concerned with the change in the computational power of matchgates when extra gates are added to the set, with all other conditions kept the same. We now consider an alternative scenario, and ask what the power of matchgates is when acting on other graphs, rather than the path. In the circuit model, qubits are often implicitly assumed to be aligned on a path. As we mentioned before, this is a natural assumption, given that the $\swap$ gate can be implemented using the $\cnot$ gate, and the $\cnot$ is one of the most standard gates one aims to implement when constructing a universal set. However, for matchgates the $\swap$ gate itself plays a central role for their computational power, and variations of this underlying arrangement of the qubits become meaningful, as they provide qualitatively different regimes. This question of interaction graphs and universal quantum computation was already posited in \cite{Kempe2001b}, in the context of the XY interaction, and we will answer it here fully both for general matchgates and for the XY interaction. The results from this section are taken from \cite{Brod2012} and \cite{Brod2013}.

Throughout the next few sections, we always consider matchgates to act on some graph. More specifically, we consider a graph such that its vertices correspond to the qubits, and we can act with a matchgate on a pair of qubits if the graph contains the corresponding edge. Without loss of generality we consider only connected graphs, as qubits in different components of a general graph cannot interact, so the components can be treated separately. Of course, changes in the underlying graph can be alternatively viewed as a very particular set of allowed long-range interactions on qubits aligned on a path. However, in order to provide a complete characterization, it will be much more convenient to consider, from hereon, that matchgates act exclusively on qubits arranged according to some graph.

\begin{figure}[t]
\capstart
\centering
\subfloat[]{\centering \raisebox{0.45in}{\includegraphics[width=0.5\textwidth]{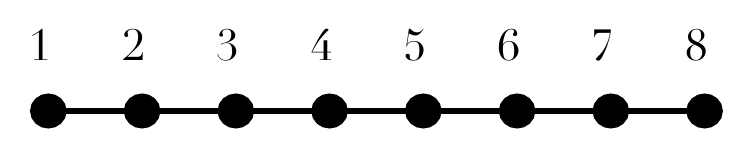}}} \qquad
\subfloat[]{\centering \includegraphics[width=0.4\textwidth]{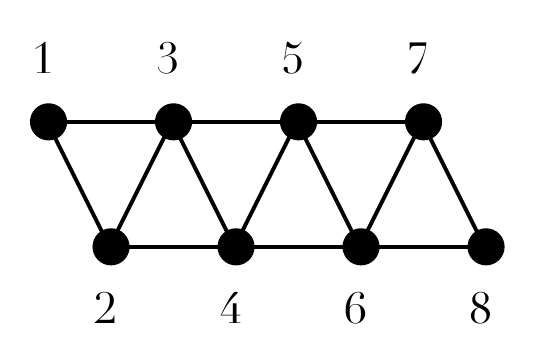}}
\caption[Correspondence between triangular ladder graph and path]{In a triangular ladder graph, vertices have a one-to-one correspondence to vertices of a path such that nearest neighbors on the triangular ladder correspond to nearest and next-nearest neighbors on the path.}
\label{fig:triangladder}
\end{figure} 

In order to further clarify this formalism, let us recast the results from the beginning of the section in terms of these interactions graphs. Recall from \chap{fermreview} that matchgates are classically simulable if allowed to act only on nearest neighbors, but become universal if allowed to act also on next-nearest neighbors. In terms of graphs, this means that matchgates are classically simulable if acting on the graph of \sfig{triangladder}{a}, which is just a path, and universal if acting on the graph of \sfig{triangladder}{b}, which we call a triangular ladder---it is easy to see that neighbors on the latter correspond to nearest and next-nearest neighbors on the former. This correspondence was already observed explicitly in \cite{Kempe2002}.

In this section, we proceed as follows. In \sec{match_arbit} we begin with two instructive examples, which culminate on our proof that matchgates are universal on any connected graph other than a path or a cycle. In \sec{simul_cyc} we build upon the simulability proof of \sec{fermreview_a} to prove the classical simulability of matchgates acting on a cycle. Finally, in \sec{XY} we specialize the result of \sec{match_arbit} and show that the XY interaction is also universal on any graph other than a path or cycle. Although this latter result implies the first, we present them separately as the first proof is easier and develops tools that are useful later on, while the simulation using the XY interaction is less explicit. The basic definitions and nomenclature concerning graphs can be found in any standard textbook on Graph Theory \cite{LivroBollobas}.

\subsection{Universality of matchgates on arbitrary graphs} \label{sec:match_arbit}

Recall, from \sec{fermreview_b}, that matchgates are universal when supplemented by the $\swap$ gate, a result that translates to their universality on the triangular ladder of \sfig{triangladder}{b}. Recall also, from \sec{fermnew_a}, that there are other gates, such as the $\fswap$ or the $\iswap$, that closely resemble the $\swap$ but fail in replacing it in the universal set, given that they are matchgates. We now show that, if we change the underlying graph, there is a sense in which these gates can replace the $\swap$. Before giving the proof for the most general case, it is instructive to work through two cases that exemplify the main ideas. 

\begin{example} \label{ex:path}
Suppose the qubits are arranged according to a graph of the form shown in \fig{appendedline}, obtained by joining a new vertex to some degree-2 vertex of a path. To prove that such a graph is universal, we use the following two tricks.

\begin{figure}
\capstart
\centering
\includegraphics[width=0.5\textwidth]{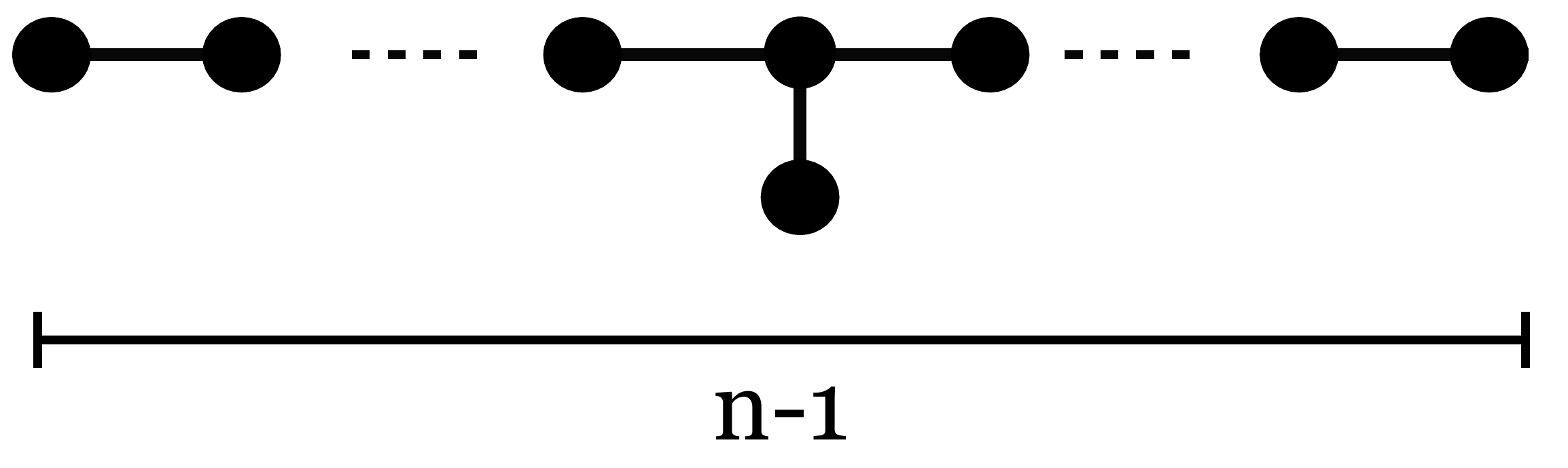}  
\caption[An $n$-vertex graph joining a new vertex to some $(n-1)$-vertex path.]{An $n$-vertex graph obtained from an $(n-1)$-vertex path by joining a new vertex to some degree-$2$ vertex of the original path.}	
\label{fig:appendedline}
\end{figure}

First, suppose we have a logical qubit in an arbitrary state $\ket{\Psi}_L=\alpha \ket{00} + \beta \ket{11}$ and a third physical qubit in any state $\ket{\phi}$. We then have the identity
\begin{equation} \label{eq:logicswap}
\fs_{12} \fs_{23} \ket{\Psi}_L \ket{\phi} = \ket{\phi} \ket{\Psi}_L,
\end{equation}
where $\fs$ is shorthand for the $\fswap$ gate, and subscripts denote the pair being acted on. The above identity follows from the trivial observation that the logical qubit is always a superposition of $\ket{00}$ and $\ket{11}$, so the $\fswap$ gate either does not induce a minus sign, or does so twice. Thus, we can use the $\fswap$ as a swapping operation provided it always exchanges the two qubits that form a logical qubit at a time. Note that, by linearity, this holds even if the logical state of qubits $1$ and $2$ is entangled with other logical qubits, as long as it is a physical state of even parity. 

The second trick is the identity
\begin{equation} \label{eq:0swap}
\fs \ket{0} \ket{\psi} = \ket{\psi} \ket{0}
\end{equation}
where $\ket{\psi}$ is the state of any physical qubit. This can be easily seen from the explicit form of the $\fswap$ gate:
\begin{equation} \label{eq:fS}
\fs = G(Z,X) = \begin{pmatrix}
1 & 0 & 0 & 0 \\
0 & 0 & 1 & 0 \\
0 & 1 & 0 & 0 \\
0 & 0 & 0 & -1
\end{pmatrix}.
\end{equation}
This follows simply because when either of the qubits is in the $\ket{0}$ state, the $\fswap$ does not induce a minus sign, behaving exactly as the $\swap$. We will use this fact to initialize some ancilla qubits in the $\ket{0}$ state and move them around as necessary. 

These two tricks consist essentially of special situations in which the $\fswap$ behaves as the $\swap$. We can now prove universality for \ex{path}. First note that the graph of \fig{branching} is guaranteed to appear as a subgraph of the one in \fig{appendedline} if the number of vertices is greater than $6$. We refer to the degree-$3$ vertex in that graph---and more generally, to any vertex of degree greater than $2$ in a tree\footnote{A tree is a graph with no cycles.}---as a branching point. We initialize two ancilla qubits near the branching point (specifically, at vertices $\alpha$ and $\beta$ in \fig{branching}) as $\ket{0}$ and encode the logical qubits using pairs of adjacent qubits as in \eq{evenencoding}. Depending on the number of vertices and the location of the branching point, some physical qubits might be unpaired, in which case one or two qubits at the endpoints may not be used.

\begin{figure}
\capstart
\centering
\includegraphics[width=0.4\textwidth]{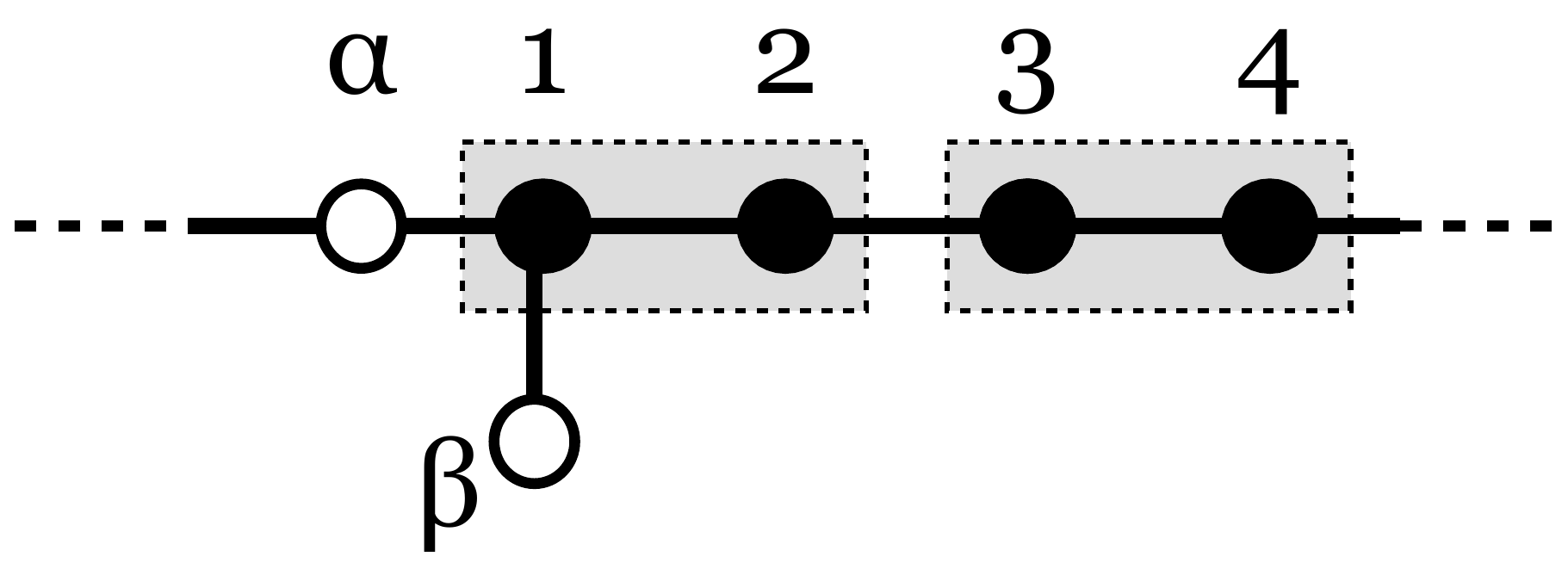}
\caption[Close-up view of the degree-$3$ vertex of the graph in \fig{appendedline}]{Close-up view of the degree-$3$ vertex of the graph in \fig{appendedline}. Vertices labeled $\alpha$ and $\beta$ correspond to ancillas initialized in the $\ket{0}$ state. Vertex pairs $\{1,2\}$ and $\{3,4\}$ correspond to the two logical qubits on which we want to implement a logical $\cz$ gate. The $\alpha$ and $\beta$ ancillas are used to change the order of the state of the other qubits, as per \eq{switch}.}
\label{fig:branching}
\end{figure} 

As discussed in \sec{fermreview_b}, a logical single-qubit gate $A$ can be implemented simply by a $G(A,A)$ matchgate between adjacent qubits. Now say we want to implement a logical $\cz$ gate between two (not necessarily adjacent) logical qubits. We first use the identity of \eq{logicswap} to place the two desired pairs near the branching point, as in \fig{branching}. Recall that the logical $\cz$ can be implemented by a physical $\cz$ between two of the four qubits (e.g., $1$ and $3$, as labeled in \fig{branching}), which in turn is equal to $\swap$ followed by $\fswap$. We can implement this sequence by swapping qubit $2$ with both qubits of pair $\{3,4\}$, which is possible by \eq{logicswap}, and then using the fact that $\alpha$ and $\beta$ are ancillas in the $\ket{0}$ state to switch the order of the qubits placed in vertices $1$ and $2$. This effectively implements the $\swap$ of \eq{CZL}. If we follow this with an $\fswap$ again between qubits $1$ and $2$, the final result is the desired $\cz$ gate. We can then use \eq{logicswap} to return all qubits to their original places. The explicit sequence is
\begin{equation} \label{eq:switch}
\fs_{23} \, \fs_{34} \, \fs_{12} \, \fs_{\beta 1} \, \fs_{12} \, \fs_{\alpha 1} \, \fs_{\beta 1} \, \fs_{12} \, \fs_{\alpha 1} \, \fs_{34} \, \fs_{23}.
\end{equation} 

This sequence uses only matchgates to implement a $\cz$ between the logical qubits which, together with the single-qubit gates mentioned previously, provides a universal set. Since any logical qubit can be moved to any desired location using $O(n)$ $\fswap$ gates, the overhead in the number of such gates grows polynomially with the number of 2-qubit gates in the original circuit. \qed
\end{example}

\begin{example} \label{ex:leaves}
Now suppose the qubits are arranged on a complete binary tree of $m$ levels, as in \fig{binarytree}\footnote{This example also shows why this formalism in terms of graphs is convenient---the description of the graph of \fig{binarytree} in terms of long-range connectivity restrictions of qubits on a path is both unnatural and cumbersome.}. This graph has  $n = 2^{m+1}-1$ vertices. Since the longest path contains only $2m-1=O(\log n)$ vertices, the strategy of \ex{path} cannot be trivially adapted to this case: the number of available logical qubits would not be sufficient.

\begin{figure}
\capstart
\centering
\includegraphics[width=0.5\textwidth]{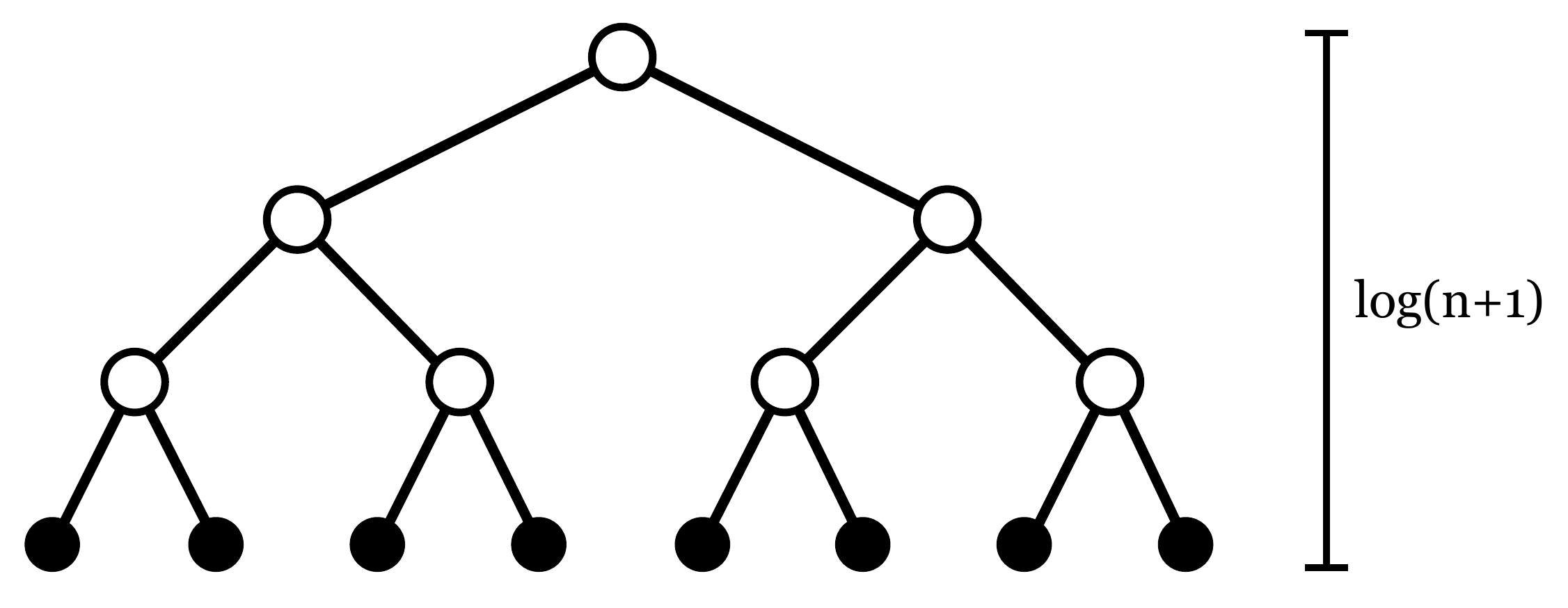}  
\caption[An $n$-vertex complete binary tree]{An $n$-vertex complete binary tree. White vertices represent $\ket{0}$ ancillas and black vertices are used in pairs to store computational qubits. This is one particular arrangement that enables universal computing with matchgates.}
\label{fig:binarytree}
\end{figure}

Instead, we store logical qubits using the $2^m=(n+1)/2$ leaves as shown in \fig{binarytree}. By using the leaves as the computational qubits and filling the paths that connect them with $\ket{0}$ ancillas, we can use the identity of \eq{0swap} to move the state of any qubit to a vertex adjacent to any other desired qubit in less than $2 \log(n/2)$ steps, apply the desired matchgate between them, and return them to their initial positions. This means we can use the $\fswap$ to implement an effective interaction between any pair among the $(n+1)/2$ computational qubits, which clearly is sufficient for universal computation, as per the construction of \sec{fermreview_b}. The overhead of this approach is modest: it requires twice the number of qubits and uses $2 \log(n)$ $\fswap$ operations per 2-qubit gate in the original circuit. Note that this approach works for any pairing of physical into logical qubits. \qed
\end{example}

\begin{figure}
\capstart
\centering
\subfloat[Square lattice]{\includegraphics[width=0.25\textwidth]{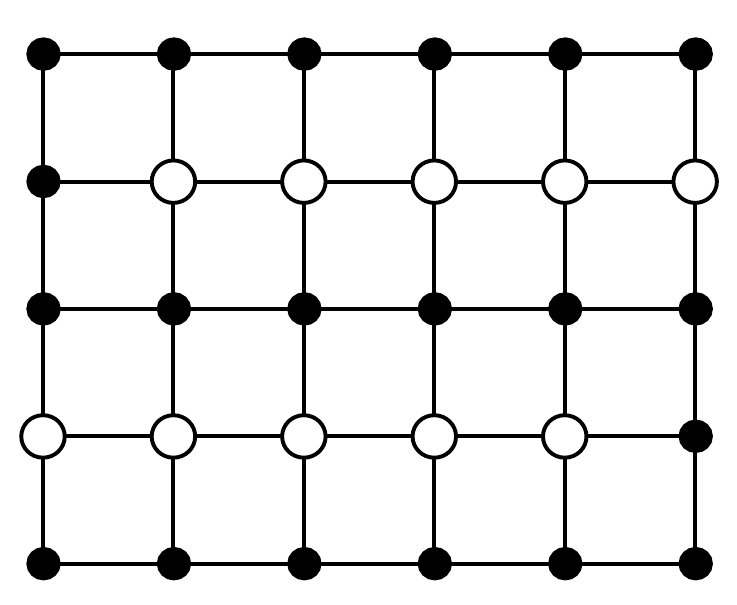}} \qquad
\subfloat[Cycle with extra vertex]{\includegraphics[width=0.25\textwidth]{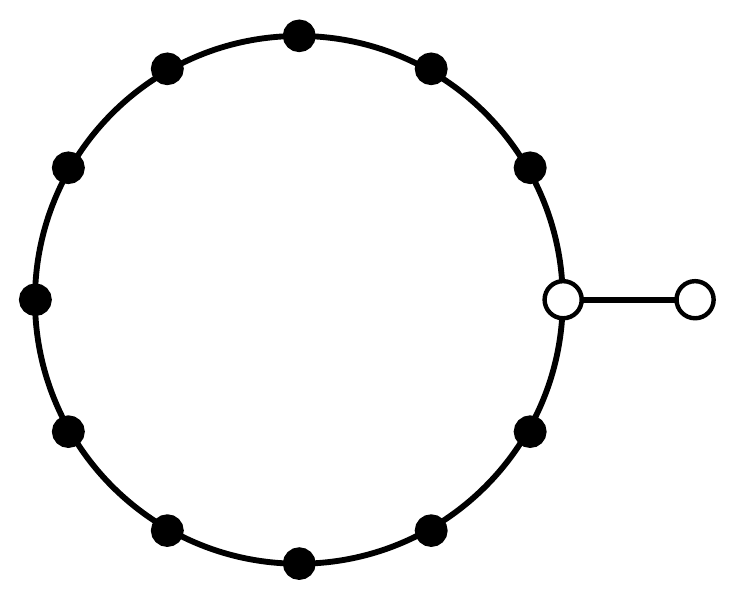}} \\
\subfloat[A $3$-regular graph]{\raisebox{0.45in}{\includegraphics[width=0.25\textwidth]{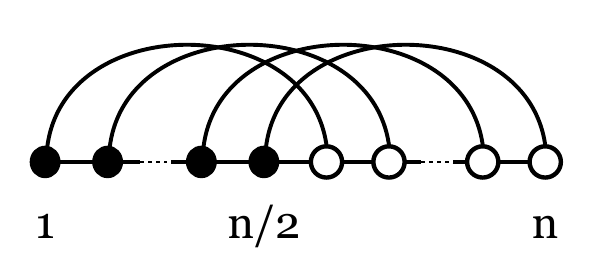}}} \qquad
\subfloat[Star graph]{\includegraphics[width=0.25\textwidth]{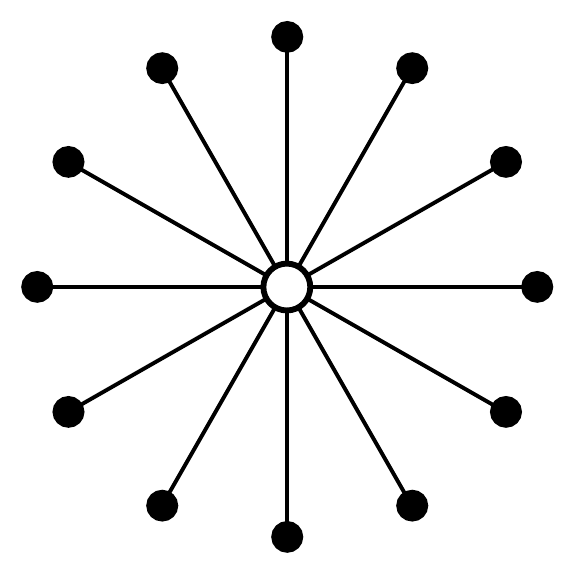}}
\caption[Several universal graphs for quantum computation with matchgates]{Several graphs that are universal for quantum computation with matchgates. White circles represent one possible placement of the $\ket{0}$ ancillas that makes the universality explicit by the arguments in Examples \ref{ex:path} and \ref{ex:leaves}.}
\label{fig:graphexamples}
\end{figure} 

These two examples prove the universality of matchgates in two extremal situations. In \ex{path}, we have a graph that has a very long path as a backbone, whereas in \ex{leaves} the longest path is too short, and we must turn to the many leaves of the graph to store the logical qubits. In fact, in \cite{Brod2012} we used adapted versions of these two ideas to prove the universality in many different graphs of interest, such as those shown in \fig{graphexamples}. Examples include regular lattices, that often arise naturally in solid state systems, the $n$-vertex star graph, that relates to the formalism of ancilla-driven quantum computation (i.e., where the computation is implemented by interaction of every qubit with one single ancilla \cite{Anders2010}), and more exotic cases. However, Examples \ref{ex:path} and \ref{ex:leaves} also raise a natural question: is it always possible to find either a sufficiently large number of leaves or a sufficiently long path, on any given graph, to enable universal computation? In \cite{Brod2013} we showed that the answer to this question is indeed yes. We begin by proving the following Lemma:

\begin{lemma} \label{lem:graph}
Let $T$ be an $n$-vertex tree with $l$ leaves and a longest path of length $p$. Then either (i) $l > \sqrt{n}$ or (ii) $p > \sqrt{n}$.
\end{lemma}

\begin{proof}
Choose any leaf $v$ of $T$. Delete every vertex on the path from $v$ to the nearest branching point, not including the branching point (see \fig{arbitree}). Since, by hypothesis, this path has length smaller than $p$, the result is a subtree of $T$ where one leaf and at most $p-1$ vertices are removed. Repeat this procedure until only a path remains (i.e., $l-2$ times). Finally, delete the remaining path, removing the last two leaves and at most $p$ vertices. This process deletes every vertex in $T$. Therefore $n \leq (l-2)(p-1)+p < lp$, so $\max\{l, p\} > \sqrt{n}$ as claimed.
\end{proof}

\begin{figure}
\capstart
\centering
\includegraphics[width=0.5\textwidth]{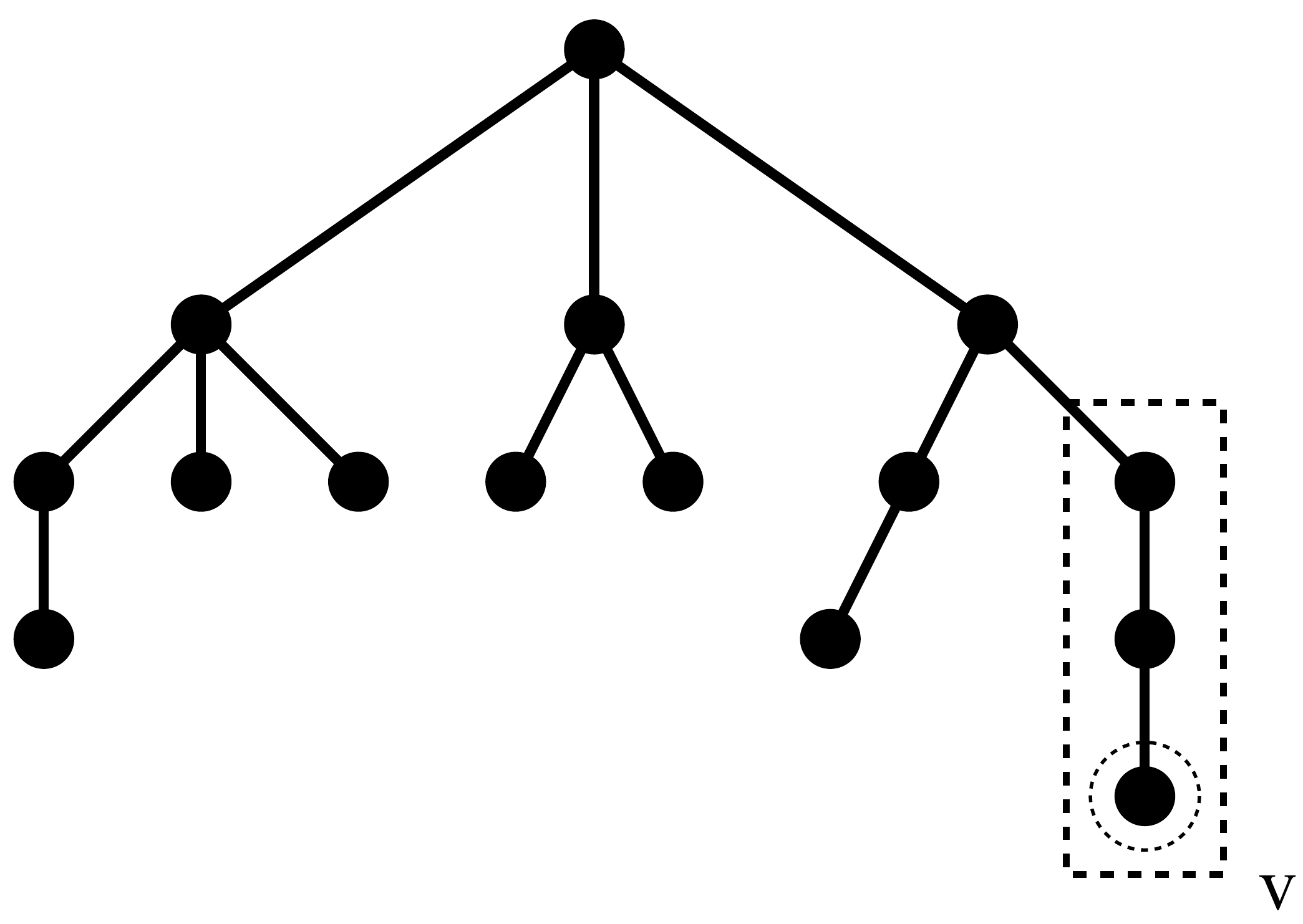}  
\caption[A tree]{A tree. The dashed rectangle indicates the vertices in the path from $v$ to the nearest branching point, which are deleted in the proof of \lem{graph}. Upon deletion of these vertices, the remaining tree has one fewer leaf, and at most $p-1$ vertices have been removed.}	
\label{fig:arbitree}
\end{figure}

The main result follows straightforwardly from \lem{graph} and the previous examples:

\begin{theorem} \label{thm:maintheoarbgraph}
Let $G$ be any $n$-vertex connected graph, other than a path or a cycle, where every vertex represents a qubit and we can implement arbitrary matchgates between neighbors in $G$. Then it is possible to efficiently simulate (i.e., with polynomial overhead in the number of operations) any quantum circuit on $\Omega(\sqrt{n})$ qubits.
\end{theorem}

\begin{proof}
Since $G$ is not a path or a cycle, it has some spanning tree $T$ that is not a path\footnote{A spanning tree of $G$ is a tree that contains all vertices of $G$.}. This holds because $G$ necessarily contains a vertex of degree more than 2 and one can construct a spanning tree that includes all edges adjacent to this vertex. It suffices to show that universal computation can be implemented in $T$, since all edges of $T$ are edges of $G$. By \lem{graph}, either (i) the longest path of $T$ or (ii) the set of all its leaves must have more than $\sqrt{n}$ vertices. 

First, suppose (i) holds. Assign each qubit of a longest path of $T$ as a computational qubit, with the exception of one qubit at a branching point. We also use one qubit adjacent to the branching point and not in the path as an ancilla. All other qubits are ignored. We implement the circuit as shown in \ex{path} of the previous section. Since the longest path has more than $\sqrt{n}$ vertices by hypothesis, this allows the simulation of an arbitrary quantum circuit on $\lfloor (\sqrt{n}-1)/2 \rfloor$ qubits. This simulation uses $O(n)$ $\fswap$ operations for each two-qubit gate.

Otherwise (ii) holds, so $T$ has more than $\sqrt{n}$ leaves. Proceed by assigning every qubit at a leaf as a computational qubit and initializing every other qubit as a $\ket{0}$ ancilla. The intermediate vertices on the (unique) path between any two leaves represent qubits in the $\ket{0}$ state. As in \ex{leaves}, we can use the identity of \eq{0swap} to move the state of any qubit to a vertex adjacent to any other, implement a matchgate, and move it back. Thus we can effectively implement any matchgate between any pair of logical qubits. Since the longest path has length less than $\sqrt{n}$, this simulation uses $O(\sqrt{n})$ $\fswap$ operations for each gate in the original circuit.
\end{proof}

\thm{maintheoarbgraph} proves the universality of matchgates on any graph that is not a path or a cycle, neatly tying all the particular cases from \cite{Brod2012} in one single result. However, the fact that it does not work for the cycle suggests that a degree-3 vertex is necessary for the universality. The proof method also indicates this---both cases (i) and (ii) of the Theorem rely on the existence of branching points to maneuver the states of the qubits around according to \eqs{0swap}{logicswap}. In the next section we show that a branching point, besides being a sufficient condition, is indeed also necessary.

\subsection{Classical simulation of matchgates on the cycle} \label{sec:simul_cyc}

I will now prove that matchgates are classically simulable on the cycle \cite{Brod2013}. In \sec{fermreview_a}, I reproduced a proof from \cite{Knill2001a,Terhal2002,Jozsa2008b} that matchgates are classically simulable on a path, if the input of the circuit is a product state and the output is the measurement of some qubit $k$ in the computational basis (i.e., this corresponds to strong simulation, as described in \sec{introsimul}). The proof was along the following lines. I started by showing that any unitary $U$ corresponding to a circuit of nearest-neighbor matchgates corresponds to a linear transformation on the space of the fermionic operators $c_i$. More specifically, there is some rotation $R \in\SO(2n)$ such that
\begin{equation*}
U^{\dagger} c_i U = \sum_{j=1}^{2n} R_{i,j} c_j,
\end{equation*}
where the $c_i$ are the operators defined in \eq{JWc}. Furthermore, recall that to compute the probability $p_0$ that qubit $k$ will be measured in the $\ket{0}$ state, it suffices to compute the expectation value $\langle Z_k \rangle = p_0 - p_1 = 2p_0 -1$. Since $Z_k = c_{2k-1} c_{2k}$, for any input product state $\ket{\psi}=\ket{\psi_1} \ket{\psi_2}\ldots\ket{\psi_n}$ we can write
\begin{equation*} \tag{\ref{eq:expectedZ}}
\langle Z_k \rangle = -i \bra{\psi} U^{\dagger} c_{2k-1} c_{2k} U \ket{\psi} = -i \sum_{a,b=1}^{n} R_{2k-1, a} R_{2k, b} \bra{\psi} c_a c_b \ket{\psi}.
\end{equation*}
Finally, since any monomial $c_a c_b$ is a tensor product of Pauli operations, each expectation value in the right-hand side of expression above factors into a product $\prod_{i=1}^{n} \bra{\psi_i} \sigma_i \ket{\psi_i}$. Hence, the desired probability can be calculated from a polynomial number of efficiently computable terms, so the action of the circuit is classically simulable.

However, this result does not immediately apply to the case of a cycle, which corresponds to applying periodic boundary conditions to a path, because a matchgate between the first and last qubits does not translate into a Hamiltonian that is quadratic in the $c_i$s, and vice versa. For example,
\begin{equation} \label{eq:notnnmatch}
c_1 c_{2n} = i X_1 X_n \prod_{i=1}^{n} Z_i,
\end{equation}
which is clearly not a matchgate, as it is a unitary operation acting on every qubit in the circuit. 

Note that \thm{quadratic} still applies to the Hamiltonian in \eq{notnnmatch} even though it does not correspond to a matchgate. However, we do not have a straightforward way of writing the operators we need, such as $X_1 X_n$, in terms of these quadratic operators. 

To show that matchgates are simulable in this case nonetheless, first consider the case where the input state $\ket{\psi}$ is a computational basis state. Suppose that $\ket{\psi}$ has even parity (e.g., $\ket{000\ldots0}$). Matchgates preserve parity, so the state at any point in the computation has a well-defined (even) parity. Now notice that $\prod_{i=1}^{n} Z_i$ is the operator that measures overall parity, so it acts as the identity on the even-parity subspace. This means that for any even-parity input we have a correspondence similar to \eq{JWquad2}:
\begin{equation*} 
X_1 X_n = X_1 X_n \prod_{i=1}^{n} Z_i = - i c_1 c_{2n} \quad \text{(even parity)},
\end{equation*}
where the second equality is just \eq{notnnmatch}. The equivalent equations for $Y_1 Y_n$, $X_1 Y_n$, and $Y_1 X_n$ are straightforward. Since we have recovered a correspondence between matchgates on qubits $1$ and $n$ and quadratic Hamiltonians, the simulation can be carried out exactly as in \sec{fermreview_a}. The case of an odd-parity input state (e.g., $\ket{100\ldots0}$) is analogous, except that the operator $\prod_{i=1}^{n} Z_i$ now acts as minus the identity, and we write
\begin{equation*} 
X_1 X_n = - X_1 X_n \prod_{i=1}^{n} Z_i =  i c_1 c_{2n} \quad \text{(odd parity)}
\end{equation*}
and its equivalents for $Y_1 Y_n$, $X_1 Y_n$, and $Y_1 X_n$. 

Now consider a general product input state $\ket{\psi}$. Let $\ket{\psi_{\pm}}$ denote the projections of $\ket{\psi}$ onto the even- and odd-parity subspaces, respectively. The expectation value $\langle Z_K \rangle$, analogous to \eq{expectedZ}, is
\begin{subequations}  \label{eq:simulationcycle}
\begin{align}
\langle Z_k \rangle = & -i \bra{\psi} U^{\dagger} c_{2k-1} c_{2k} U \ket{\psi} \notag \\ 
= & -i \sum_{a,b=1}^{n} ( R_{2k-1, a} R_{2k, b} \bra{\psi_{+}} c_a c_b \ket{\psi_{+}} \notag \\ 
& \qquad\quad + R'_{2k-1, a} R'_{2k, b} \bra{\psi_{-}} c_a c_b \ket{\psi_{-}} ).
\end{align}
\end{subequations} 
Here $R$ and $R'$ correspond to two sets of rotations, where $R'$ includes an extra minus sign for every matchgate applied between qubits $1$ and $n$. The expression above does not contain cross terms such as $\bra{\psi_{-}} c_{a} c_{b} \ket{\psi_{+}}$ because $c_a c_b$ preserves parity.

The sum in \eq{simulationcycle} contains a polynomial number of terms, just as in \eq{expectedZ}, but now each term may not be easy to compute, since $\ket{\psi_{\pm}}$ are not product states in general. However, we have
\begin{align*}
\bra{\psi} c_{a} c_{b} \ket{\psi} & = \bra{\psi_{+}} c_{a} c_{b} \ket{\psi_{+}} + \bra{\psi_{-}} c_{a} c_{b} \ket{\psi_{-}}, \\
\bra{\psi} c_{a} c_{b} \prod_{i=1}^{n} Z_i \ket{\psi} & = \bra{\psi_{+}} c_{a} c_{b} \ket{\psi_{+}} - \bra{\psi_{-}} c_{a} c_{b} \ket{\psi_{-}}.
\end{align*}
We can invert these equations to obtain
\begin{align}
\bra{\psi_{+}} c_{a} c_{b} \ket{\psi_{+}} & = \frac{1}{2} \left[ \bra{\psi} c_{a} c_{b} \ket{\psi}+\bra{\psi} c_{a} c_{b} \prod_{i=1}^{n} Z_i \ket{\psi} \right ], \notag \\
\bra{\psi_{-}} c_{a} c_{b} \ket{\psi_{-}} & = \frac{1}{2} \left[ \bra{\psi} c_{a} c_{b} \ket{\psi}-\bra{\psi} c_{a} c_{b} \prod_{i=1}^{n} Z_i \ket{\psi} \right ]. \label{eq:expectedparity}
\end{align}

The left-hand sides are precisely the two terms of $\langle Z_k \rangle$ that we need, while the right-hand sides are combinations of terms that can be efficiently computed, as both are expected values of products of Pauli operators on product states. Explicitly, if $\ket{\psi}=\ket{\psi_1}\ket{\psi_2}\ldots\ket{\psi_n}$ and $c_a c_b = \sigma_1 \sigma_2 \ldots \sigma_n$, we have
\begin{align*}
\bra{\psi} c_{a} c_{b} \ket{\psi} &= \prod_{i=1}^{n} \bra{\psi_i} \sigma_i \ket{\psi_i}, \\
\bra{\psi} c_{a} c_{b} \prod_{i=1}^{n} Z_i \ket{\psi} &= \prod_{i=1}^{n} \bra{\psi_i} \sigma_i Z_i \ket{\psi_i}.
\end{align*}

Plugging \eq{expectedparity} into \eq{simulationcycle}, we recover an expression that can be efficiently computed in the same manner as \eq{expectedZ}, with only four times as many terms. This gives an efficient classical simulation for matchgates acting on a cycle, as claimed.

The simulability proof of matchgates on a path, shown in \sec{fermreview_a}, was recently exploited \cite{Jozsa2010} to show that circuits of nearest-neighbor matchgates on $n$ qubits (on a path) are equivalent to general quantum circuits on $O(\log n)$ qubits, and subsequently \cite{Kraus2011, Boyajian2013} to show a protocol for ``compressed'' simulations (i.e., with quantum circuits on $O(\log n)$ qubits) of the Ising and XY models of spin systems with open boundary conditions. It is an open question whether the observations made in this section lead to analogous results for systems with periodic boundary conditions.

\subsection{Universality of the XY interaction on arbitrary graphs} \label{sec:XY}

In \sec{match_arbit} and \sec{simul_cyc}, I addressed the computational power of the set of all matchgates on arbitrary graphs. We now consider the computational power of a restricted set of matchgates corresponding to the XY (or anisotropic Heisenberg) interaction on arbitrary graphs, which we introduced in \sec{ferm_reviewXY}. This interaction corresponds to a subset of matchgates generated by the Hamiltonian $H_A := X \otimes X + Y \otimes Y$ (recall from \sec{fermreview_a} that matchgates are generated by the two-qubit Hamiltonians $X \otimes X$, $X \otimes Y$, $Y \otimes X$, $Y \otimes Y$ together with the single-qubit Hamiltonian $Z$). It is easy to see that these interactions form a proper subset of matchgates as, e.g., they act nontrivially only on the odd-parity subspace of the $2$-qubit Hilbert space.

Despite being a proper subset of matchgates, the XY interaction is also known \cite{Kempe2002} to be universal for quantum computation when acting on the graph of \fig{triangladder} (i.e., nearest and next-nearest neighbor interactions on a path). It also follows trivially from \sec{fermreview_a} and \sec{simul_cyc} that the XY interaction is classically simulable on paths and cycles. This prompts the question of whether our results from \sec{match_arbit} can be adapted for the XY interaction on arbitrary graphs.

We now show that the XY interaction alone is universal for quantum computation on any connected graph that is not a path or a cycle \cite{Brod2013} . Since these operations are a subset of matchgates, this result subsumes the one of \sec{match_arbit}. However, the argument given for the XY interaction is less explicit, and the simulation is less efficient in general.

First observe that the XY interaction acts trivially on the even-parity subspace, so the encoding of \eq{evenencoding} cannot be used. A suitable alternative (as used in \cite{Kempe2002}) is
\begin{align}
\ket{0}_L & = \ket{01}, \notag \\
\ket{1}_L & = \ket{10}, \label{eq:oddencoding}
\end{align}
which is simply the corresponding encoding on the odd-parity subspace. 

We also need to adapt some of the identities used in \sec{match_arbit}. The fermionic $\swap$ gate is not available, but we can instead use $\iswap$ (denoted by the shorthand $\is$):
\begin{equation} \label{eq:i-SWAP}
\is := \exp( i \tfrac{\pi}{4} H_A ) = G(I, iX) = \begin{pmatrix}
1 & 0 & 0 & 0 \\
0 & 0 & i & 0 \\
0 & i & 0 & 0 \\
0 & 0 & 0 & 1
\end{pmatrix}.
\end{equation}

For an arbitrary logical state $\ket{\Psi}_L = \alpha \ket{10} + \beta \ket{01}$ in the encoding of \eq{oddencoding}, and for any physical qubit in an arbitrary state $\ket{\phi}$, we have the following identity (already used implicitly in \cite{Kempe2002}):
\begin{equation} \label{eq:logicswap2}
\is_{12} \, \is_{23} \ket{\Psi}_L \ket{\phi} = i \ket{\phi} \ket{\Psi}_L.
\end{equation}
Thus these states can be swapped up to an irrelevant global phase. 

Another useful identity, akin to \eq{0swap}, is given by 
\begin{equation} \label{eq:0iswap}
\is_{12}\ket{0} \ket{\psi} = \left ( P \ket{\psi} \right ) \ket{0},
\end{equation}
where $\ket{\psi}$ is any state and $P := \diag(1,i)$. This identity has a familiar operational interpretation: once more any state can be ``swapped through'' a $\ket{0}$ ancilla, but now with the caveat that the state suffers an unwanted $P$ gate. We must take this into account when using \eq{0iswap} in a simulation, but one can already see that if we only need to swap states through an even number of ancillas at a time, we can cancel out the $P$ gates by alternating $\iswap$ and $\iswap^\dagger$ swapping operations. In fact, a trivial adaptation of \thm{maintheoarbgraph} gives a proof of universality for those graphs that have an odd cycle (i.e., non-bipartite graphs), since then there is always an even-length path between any two vertices. We state this without proof, as the details are not instructive and the result is implied by the general case. Note however that for non-bipartite graphs, one can obtain a universal set of unitary matrices, whereas for general graphs we will only obtain a universal set of orthogonal matrices.

\begin{figure}[t]
\capstart
\centering
\subfloat[]{\includegraphics[width=0.25\textwidth]{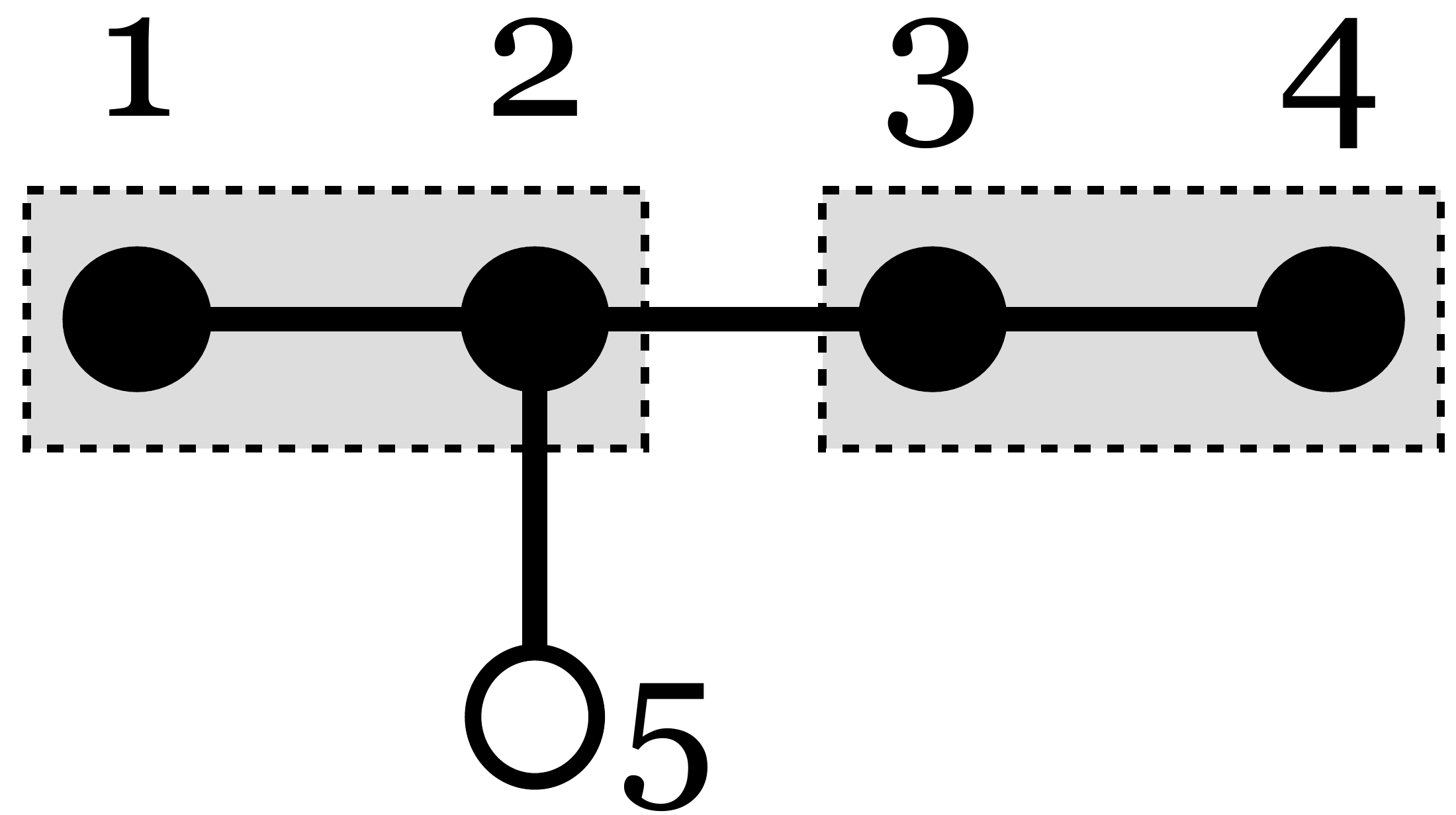}} \qquad
\subfloat[]{\includegraphics[width=0.25\textwidth]{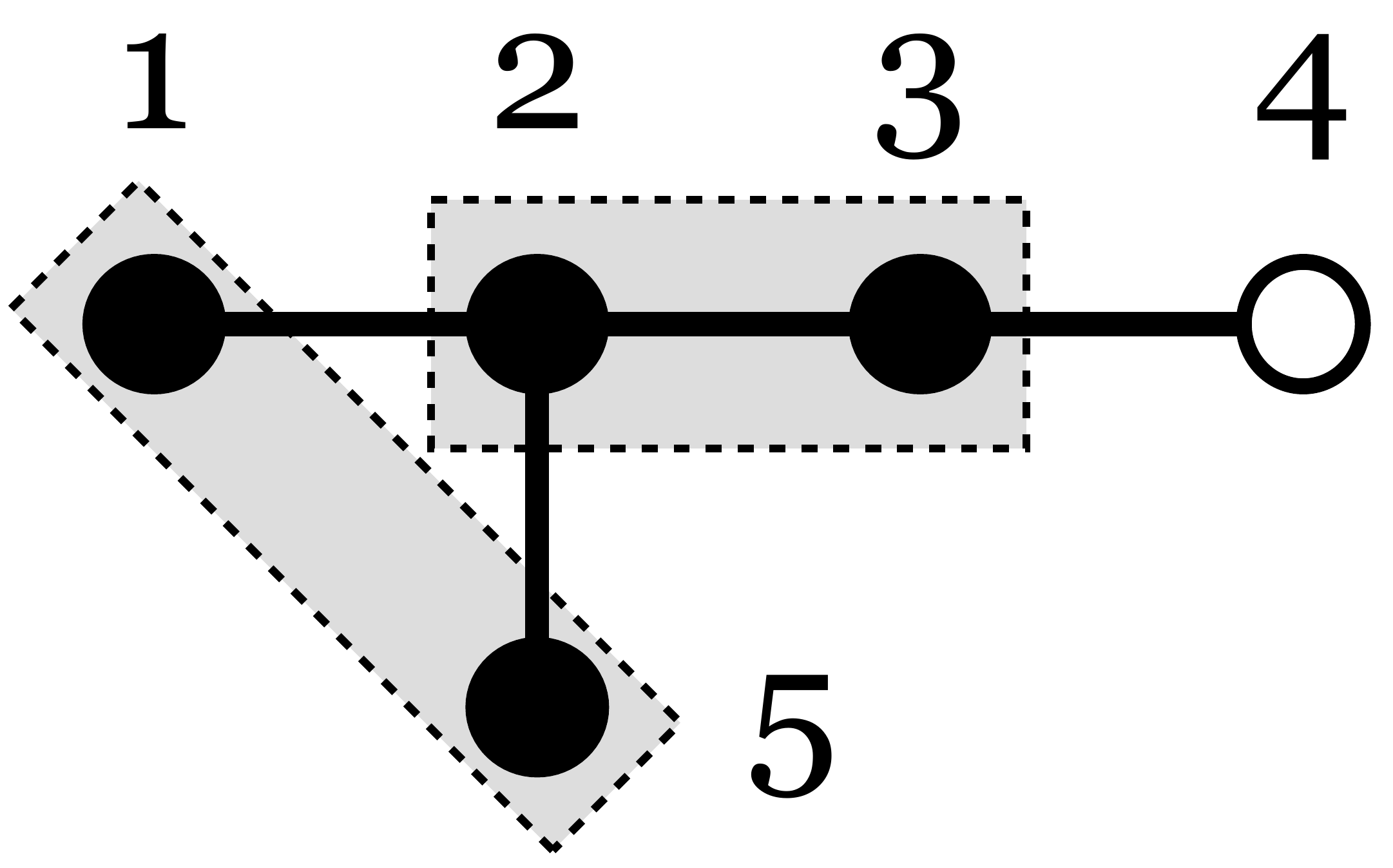}} \qquad
\subfloat[]{\includegraphics[width=0.25\textwidth]{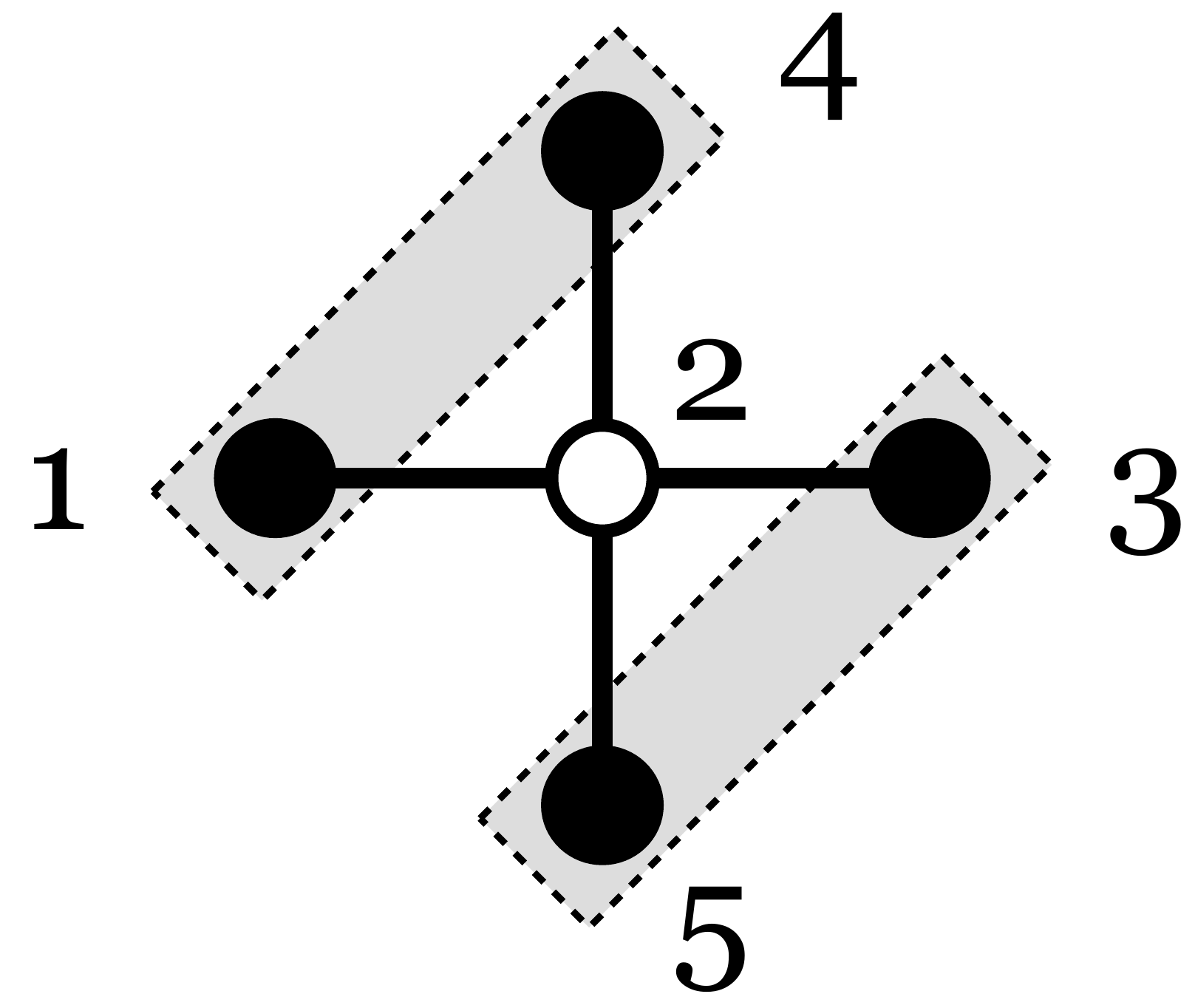}} \\
\caption[5-vertex graphs for implementing a universal set of logical two-qubit gates with the XY interaction.]{5-vertex graphs for implementing a universal set of logical two-qubit gates with the XY interaction. In all figures, gray boxes identify pairs of physical qubits that make up a logical qubit and white vertices represent ancillas initialized as $\ket{0}$.}
\label{fig:5vertex}
\end{figure} 

We first show how to implement a particular set of one- and two-qubit gates on the two 5-vertex graphs of \fig{5vertex}, similar to the simulation in \ex{path} (cf.\ \fig{branching}). Suppose the two logical qubits can be initialized as in \sfig{5vertex}{a}, according to the encoding of \eq{oddencoding}, together with one $\ket{0}$ ancilla. 

Since
\begin{equation*}
H_A = \begin{pmatrix}
0 & 0 & 0 & 0 \\
0 & 0 & 2 & 0 \\
0 & 2 & 0 & 0 \\
0 & 0 & 0 & 0
\end{pmatrix},
\end{equation*}
a logical $X$ rotation on the logical qubit stored in physical qubits $\{1,2\}$ can be implemented by a simple XY interaction:
\begin{equation*}
\exp( i a X_L ) = \exp( i \tfrac{a}{2} H_{A,12} ) = \begin{pmatrix}
1 & 0 & 0 & 0 \\
0 & \cos{a} & i \sin{a} & 0 \\
0 & i \sin{a} & \cos{a} & 0 \\
0 & 0 & 0 & 1
\end{pmatrix}.
\end{equation*}

We can also implement the two-qubit gate $R_{XZ}(a) := \exp(i a \, X \otimes Z)$ on the logical qubits $\{1,2\}$ and $\{3,4\}$ by the following sequence:
\begin{equation} \label{eq:2qubitg}
\is_{25} \, \is_{23} \, \is_{34} \, \left[ \is_{25}^\dagger \, \exp( i \tfrac{a}{2} H_{A,12} ) \, \is_{25} \right ] \, \is_{34}^\dagger \, \is_{23}^\dagger \, \is_{25}^\dagger.
\end{equation}
This sequence works as follows. The first three $\iswap$ gates use \eq{0iswap} to swap the qubits and place them as in \sfig{5vertex}{b}. Notice that the first logical qubit suffers a $P$ gate during this operation. The sequence inside the square brackets implements an effective unitary with Hamiltonian $Y \otimes Z$. This can be verified by explicit multiplication, but can also be understood as follows: the $\is_{25}$ and $\is_{25}^{\dagger}$ swap qubits $2$ and $5$, leaving the first logical qubit encoded in pair $\{1,2\}$, up to some phases that depend upon the states of both qubits. The $H_{A,12}$ Hamiltonian then acts as a logical $X$ rotation on the first qubit. Keeping track of the dependence of the relative phases on the states of both qubits, we see that the overall operation is $Y \otimes Z$. Finally, the last three $\iswap$ gates return the states of all qubits to their original positions, while inducing a $P^{\dagger}$ gate on the first logical qubit. Since $P^{\dagger} Y P=X$, the overall operation on the encoded states is $X \otimes Z$, as claimed.

We now make a brief digression to explain why the set of Hamiltonians 
\begin{equation*}
  \A := \{X \otimes I, I \otimes X, X \otimes Z, Z \otimes X, X \otimes Y, Y \otimes X \}
\end{equation*}
is universal for quantum computation in the usual circuit model. First notice that the Hamiltonians $X \otimes Y$ and $Y \otimes X$ are included; this is without loss of generality, as they can be obtained as simple sequences of the remaining interactions, e.g., $X \otimes Y = U(X \otimes Z)U^{\dagger}$ where $U=\exp[ i \tfrac{\pi}{4} (I \otimes X)]$. By conjugating every element in $\A$ by $P$, we obtain the set 
\begin{equation*}
  \B := \{Y \otimes I, I \otimes Y, Y \otimes Z, Z \otimes Y, X \otimes Y, Y \otimes X \}.
\end{equation*}
These are exactly the generators of the special orthogonal group $\SO(4)$. This can be seen by writing them down explicitly, but also understood by a counting argument, as $\B$ contains six linearly independent, purely imaginary $4 \times 4$ matrices.

Now we recall the well-known fact (see, e.g., \cite{Bernstein1993} and \cite{Rudolph2002}) that universal quantum computation is possible using only orthogonal, rather than general unitary, matrices, with the overhead of one extra ancilla qubit and a polynomial number of operations. Furthermore, any special orthogonal matrix on $n$ qubits [i.e., in $\SO(2^n)$] can be decomposed in terms of $\SO(4)$ gates acting nontrivially only on pairs of qubits, so the set $\B$ is universal for quantum computation. But this means that the set $\A$ is also universal, since we can assume that initialization and measurements are done in the computational basis, so the initial and final single-qubit $\{P,P^{\dagger}\}$ gates do not affect the outcomes. 

While the graph in \sfig{5vertex}{a} may not appear as a subgraph of the given graph, the sequence \eq{2qubitg} can be easily adapted to the graph of \sfig{5vertex}{c}. In that case, we can just use \eq{0iswap} to swap the ancilla with any of the other qubits and obtain a similar arrangement to that of \sfig{5vertex}{b}. The corresponding sequence is
\begin{equation} \label{eq:2qubitgb}
\is_{24} \, \left[ \is_{25}^\dagger \, \exp( i \tfrac{a}{2} H_{12} ) \, \is_{25} \right ] \, \is_{24}^\dagger.
\end{equation}
In this case, every operation described before is obtained up to conjugation by $P$, and the set of available operations is $\B$, rather than $\A$. However, as described above, this still suffices for universal computation. 

It remains to show that, for any graph other than a path or cycle, we can assign sufficiently many vertices as computational qubits and swap them around to one of the arrangements of \fig{5vertex} with a polynomial number of operations. 

\begin{theorem} \label{thm:maintheoarbgraph2}
Let $G$ be any $n$-vertex connected graph, other than a path or a cycle, where every vertex represents a qubit and we can implement the interaction $H=X \otimes X + Y \otimes Y$ between any nearest neighbors in $G$. Then it is possible to efficiently simulate any quantum circuit on $\Omega(\sqrt{n})$ qubits.
\end{theorem}

\begin{proof}
As in \thm{maintheoarbgraph}, it suffices to prove the universality of $H$ on any $n$-vertex tree $T$ that is not a path.

By \lem{graph}, either (i) the longest path of $T$ or (ii) the set of all its leaves must have more than $\sqrt{n}$ vertices. Suppose first that (i) holds. Then the universal construction is directly analogous to case (i) of \thm{maintheoarbgraph}. Simply assign pairs of adjacent vertices on the longest path as logical qubits, and every other as a $\ket{0}$ ancilla. Then, by using \eq{logicswap2}, we can swap any two logical qubits to the closest degree-3 vertex, where we use sequence \eq{2qubitg} to implement the $X \otimes Z$ Hamiltonian as per \sfig{5vertex}{a}. As explained previously, this together with the logical $X$ Hamiltonian on any qubit (given by $H$ on adjacent qubits) enables universal computation with overhead of at most $O(n)$ $\iswap$ operations per orthogonal matrix in the original circuit of \cite{Rudolph2002}. 

Otherwise, (ii) holds. Then, first suppose that $T$ is not a star. Any such $T$ contains the graph of \sfig{5vertex}{a} as a subgraph, so we assign those $5$ vertices as $\ket{0}$ ancillas, together with all non-leaves, and pair the remaining leaves arbitrarily into computational qubits. We can now use \eq{0iswap} to bring the states of any two logical qubits to the structure of \sfig{5vertex}{a}, but with one caveat: this process may induce an overall $P$ gate on some logical qubits, depending on whether an odd or even number of $\ket{0}$ ancillas is traversed. This separates the logical qubits into two disjoint sets, namely those that suffer an overall $P$ gate and those that do not (there is no need to single out the case where the qubits suffer an overall $P^{\dagger}$, as this can be prevented by using $\iswap$$^{\dagger}$, rather than $\iswap$, as the swapping operation). We then take the larger of these two sets, which has at least $\sqrt{n}/4$ logical qubits, and for simplicity we disregard the rest. On the remaining qubits, as argued previously, we can either implement the set of operations $\A$ or its conjugated-by-$P$ version $\B$. Since either set is universal, this gives an universal construction with an overhead of $O(\sqrt{n})$ operations for each gate in the original circuit.

Finally, for the star graph, we replace sequence \eq{2qubitg}, corresponding to \sfig{5vertex}{a}, by the equivalent sequence \eq{2qubitgb} corresponding to \sfig{5vertex}{c}. This enables us to implement the set of Hamiltonians mentioned in the previous paragraph, and concludes the proof.
\end{proof}

\subsection{Discussion} \label{sec:fermnew_b_disc}

We completely characterized the computational power of nearest-neighbor matchgates when the qubits are arranged on an arbitrary graph---the only connected graphs for which matchgates are classically simulable are paths and cycles, whereas on any other connected graph they are universal for quantum computation. Furthermore, the same dichotomy holds when we restrict matchgates to the proper subset described by the XY interaction. 

Once again, as in \sec{fermnew_a}, this dichotomy excludes the possibility that these two sets of interactions (general matchgates and the XY interaction), acting on graphs, could exhibit intermediate computational power such as that displayed by circuits of commuting observables (IQP) \cite{Bremner2011} or noninteracting bosons \cite{Aaronson2013a}. However, this does not rule out such a result for other subsets of matchgates. As one example, consider the set generated by the $X \otimes X$ Hamiltonian acting on some graph. All such operations commute, and this set corresponds to a proper subclass of IQP. 
Furthermore, the set of two-qubit $X \otimes X$ and single-qubit $X$ Hamiltonians is sufficient to implement the non-adaptive measurement-based quantum computation model (cf.\ \sec{introduction_c}) that was also shown \cite{Hoban2013} to be hard to simulate classically, in the same fashion as IQP and noninteracting bosons. It is an open question whether an analogous result can be obtained by further restricting the operation to only the $X \otimes X$ Hamiltonian, or possibly some other proper subset of matchgates, and how the power of such a model depends on the underlying interaction graph.

While the results of this section establish the universality of matchgates on any connected graph that is not a path or a cycle, it should be possible to improve the efficiency of these constructions. These results take an operational approach, where each $\ket{0}$ is seen as an ``empty space'' through which we can move logical qubits, allowing for a simple and unified proof of universality for all graphs. In some cases, such as for the star graph, where all vertices but one are leaves, this construction is optimal. But in many others, our construction could ignore many vertices and/or edges, making it far from optimal. One such case is the binary tree of \fig{binarytree}, where we could have filled most of the non-leaves with logical qubits and used \eq{logicswap} rather than \eq{0swap} whenever it was necessary to ``move'' two logical qubits through each other. Since the bounds of \lem{graph} are tight (e.g., consider the graph obtained from a $\sqrt{n}$-leaf star by subdividing each edge $\sqrt{n}$ times), an optimal simulation presumably requires a more efficient assignment of logical qubits than in \thm{maintheoarbgraph}. While being markedly non-optimal in some cases, the constructions presented in this section nevertheless provide powerful tools for case-by-case optimization. It remains an open question whether there is a way to systematically obtain a more efficient construction, and in particular, whether in every case only a constant fraction of the qubits must be discarded as non-computational.

\section{Conclusions, open questions and relations to other work} \label{sec:fermnew_c}

Throughout this chapter, we have considered variations of the known results regarding matchgates reviewed in \chap{fermreview}. In every case considered, we have observed the jump in computational power to be abrupt. We can start with matchgates on a path, which are classically simulable, and obtain quantum universality simply by adding any parity-preserving gate that is not a matchgate (\sec{fermnew_a}), or by having at least one qubit of the circuit interact with more than two other qubits (\sec{fermnew_b}), and the same is true if we start only with the XY interaction. In what other directions could these results be extended? Could they provide other computational regimes?

One possibility is that of modifying the state preparation or measurement stages of the circuit. Throughout this chapter, we have not considered the use of nontrivial measurements to implement other unitary operations---it has been shown, for example, that noninteracting fermions (i.e., matchgates on a path) become universal if nondestructive charge measurements are allowed \cite{Beenakker2004}. These charge measurements clearly cannot be implemented by combining matchgates and computational basis measurements. It is plausible that a similar result could be obtained by allowing the use of more general two-qubit input states, that somehow require a non-matchgate unitary for their preparation. 

Another approach, already mentioned previously, would be to consider other non-matchgate unitaries, together with a change in the encoding. The parity-preserving gates of \sec{fermnew_a} arise naturally given that we are encoding qubits on a space of well-defined parity, which in turn is a particularly suitable choice for matchgates, but not the only one (for example, see \cite{Kempe2001b} for a general formalism of encoded universality, and a procedure for obtaining valid encodings for universal computation with the exchange interaction). 

Note also that all of these variations can in principle be combined. One example of this was discussed in \sec{fermnew_b_disc}, where it was left as an open question whether other notable restricted subsets of matchgates, such as e.g.\ only the $X \otimes X$ interaction, could have their computational power affected by changing the underlying graphs. Another known example is that of combining preparation of special ancillas (in the $\ket{+}$ state) with changes in the graph. This possibility, which was omitted here, was analyzed in \cite{Brod2012} for several examples. Although preparation of these special input states was not strictly necessary for universality, it did remove the need of encoding in several cases, thus suggesting that these combinations may be relevant for the matter of overall efficiency of the computation, both in terms of time and spatial resources.

Other interesting open questions we can point out relate to the robustness of our results against experimental noise and imperfections. For the results of \sec{fermnew_a}, one could also ask how much noise can be added to the non-matchgate operation before it stops being useful. Similar results are known in the context, for example, of Clifford gates \cite{Dam2009}. There, the authors considered noisy non-Clifford operations (for a particular noise model) and showed that, if the level of noise is too high, that operation falls inside the ``Clifford polytope''. At that point, the operation becomes a convex sum of Clifford gates, and its action can be simulated by a classical probabilistic mixture of Clifford gates, thus providing no quantum speedup. It would be interesting, although it seems mathematically daunting, to characterize the equivalent ``matchgate polytope''\footnote{Not necessarily a polytope, since matchgates are a continuous set of operations, in contrast to Clifford gates.} of convex sums of matchgates. This would allow one to investigate whether a very noisy $Z \otimes Z$ interaction can be simulated in a similar fashion.

For the robustness of the results of \sec{fermnew_b}, it would be interesting to see whether the technique of encoded selective recoupling (cf.\ \sec{ferm_reviewXY}) can be adapted for the XY interaction on arbitrary graphs. Recall that the XY interaction is an idealized model of real-world physical interactions \cite{Imamoglu1999, Quiroga1999, Zheng2000, Mozyrsky2001}. Techniques have been developed to deal with spurious terms in the Hamiltonian of the XY interaction or background fields \cite{Vala2002,Lidar2001}, but these techniques seem to need second-neighbor interactions throughout the whole circuit. It remains an open question whether these realistic models retain the universality of the XY interaction on arbitrary graphs. Also, during the writing of this thesis a new pre-print was released \cite{Herrera2014} with a proposal for matchgate quantum computing using polar molecules trapped in optical lattices. This proposal has a set of tunable interactions consisting of matchgates which is larger than just the XY interaction, and so it might present an interesting first candidate for a more robust version of the results presented here.

Besides the broader question addressed in \sec{fermnew_b} of characterizing the power of matchgates on \emph{any} graph, there are interesting connections between the results proven for \emph{particular} graphs and other results in the literature, obtained by giving the connectivity determined by the graph a new interpretation. There are two particular cases I would like to mention: that of quantum control of spin chains \cite{Burgarth2010}, and that of ancilla-controlled quantum computation \cite{Proctor2013}.

In the formalism of quantum control of spin chains of \cite{Burgarth2010}, the system consists of a path of qubits with an always-on nearest-neighbor interaction (such as e.g.\ an XY interaction), together with arbitrary control of a particular subset of the qubits. Specifically in \cite{Burgarth2010}, one assumes the ability to do arbitrary single-qubit gates on two ancillas at one of the end-points (i.e., the first and second qubits of the path). The computation is done, then, by a careful pulsing of the single-qubit gates on the controllable qubits such that the information is propagated to the remainder of the path by the always-on interaction. This formalism has some analogy to our results for the universality of matchgates (or just the XY interaction) on the graph of \fig{appendedline}, where the extra appended vertex is on the second qubit. In both cases an otherwise-simulable set of operations (the XY interaction) on a path is taken to quantum universality by extra operations at the endpoints, which is where the information is effectively processed (recall that our result is based on bringing the logical qubits to the branching point in order to implement the universal gate set). 

In a recent paper \cite{Proctor2013}, the authors introduce the formalism of ancilla-controlled quantum computation (inspired by the setting of ancilla-driven quantum computation of \cite{Anders2010}), where the system consists of one ancilla that can interact via one particular interaction with all the other computational qubits, which cannot interact between themselves, together with arbitrary single-qubit gates on the ancilla. This is similar to our results of universality for the star graph, which can also be viewed as a computation where the only allowed interactions are between one ancilla and the remaining qubits. In fact, in \cite{Proctor2013} the universality is proven by initializing the ancilla in the $\ket{0}$ state, using the trick of \eq{0swap} to exchange the states of the ancilla and the computational qubits and then implementing the desired operations between the computational qubits, which is a very similar construction to the one we used in \sec{match_arbit}. Besides being completely flexible in terms of changes in the graph, our result is also more general in that universality is achieved either by general matchgates or just by the XY Hamiltonian, with no need for single-qubit gates. However, our result is also more restrictive in the sense that it uses an encoding, which doubles the number of qubits used (although that is not strictly necessary, as discussed in our paper \cite{Brod2012}), and that we, for simplicity, consider a continuous family, rather than a discrete set, of operations. 

Finally, I'd like to address the relationship between the results of this chapter and the fermionic formalism. The results in this chapter take a more algebraic approach, with little concern of the underlying fermionic nature of matchgates. This is done for two main reasons: (i) the questions are much more naturally posed and investigated (e.g., consider the fermionic equivalent of matchgates acting on the complete binary tree), and the results and more readily accessible in various other contexts. That is not to say, however, that some of these results do not have a fermionic parallel. The results of section \sec{fermnew_a}, as already discussed, relate to the fact that $Z \otimes Z$ is the only generator of \PP\ gates that translates, via the Jordan-Wigner transformation, to an interaction. As for the results of \sec{fermnew_b}, consider first \ex{path}. One way of stating this result is saying that universality can be obtained by a circuit of nearest-neighbor matchgates, if additionally (any) one particular qubit can interact with one of its second neighbors. This second-neighbor interaction corresponds to some Hamiltonian of the type $X_k X_{k+2}$ which, by using \eq{JWc}, translates to a quartic operator of the type $c_{2k} c_{2k+1} c_{2k+2} c_{2k+3}$, again a fermionic interaction. In the fermionic formalism, then, the construction of \ex{path} corresponds simply to having a sequence of fermionic modes, occupied by a certain number of fermions, where just two of the modes can interact. In order to interact two distant fermions, we can move them around, place them in these special modes, interact them, then move them back (remark the similarity between this and the use of the $\fswap$ to swap the states of logical qubits close to the branching point in \ex{path}). In fact our claim that a branching point is necessary in the graph may be related precisely to this. However, the claim that a branching point is also \emph{sufficient} has a less clear interpretation. The action of matchgates on the binary tree or on the star graph, for instance, would translate to a very peculiar set of polynomials of the fermionic operators, including polynomials of very high degrees. The question of \ex{leaves} in this case would be very unnatural and unmotivated, and its answer cumbersome.

\newpage
\chapter{New results: Linear optics and BosonSampling} \label{chapter:bosonnew}

In this chapter, I present our new theoretical and experimental results on the BosonSampling model. The experiments reported were done in collaboration with the quantum optics groups of Roberto Osellame (Milan), and Paolo Mataloni and Fabio Sciarrino (Rome), and the results reported here are drawn from Ref.\ \cite{Crespi2013b}, published in \textit{Nature Photonics}, Ref.\ \cite{Spagnolo2013b}, published in \textit{Physical Review Letters}, and Ref. \cite{Spagnolo2013c}, currently submitted for publication. The theoretical results consist mostly of yet-unpublished material, and a more detailed description of the simulations and data analysis done in support of the experiments, which were omitted from \cite{Crespi2013b,Spagnolo2013b, Spagnolo2013c}. This chapter is organized as follows.

First, in \sec{bosonnew_a}, I show that exact BosonSampling remains hard to simulate even if the linear optical circuit is restricted to have constant depth. More specifically, a linear optical device consisting of an input Fock state, followed by four rounds of arbitrary two-mode transformations and a final number-resolving measurement, produces a distribution over the possible outputs that should be hard to simulate classically (either exactly or up to a multiplicative error, as discussed in \sec{bosonreview_b}), unless the polynomial hierarchy collapses to its third level. This result is similar in spirit to the one of Terhal and DiVincenzo \cite{Terhal2004} which shows that arbitrary quantum circuits of depth 4 should be hard to simulate. Curiously, this value is tight for the case of arbitrary quantum circuits---it was also shown that arbitrary quantum circuits of depth 3 can be simulated classically---but not for the case of linear optics, since I show that optical circuits with two layers of beam splitters can be simulated classically, whereas the situation for circuits with three layers remains unresolved. Besides being conceptually interesting, this result should also be of interest to experimentalists since all recent implementations of BosonSampling are based on integrated optics and, at least in the current state of the technology, the number of layers of these devices is heavily responsible for photon losses. One drawback to this result is that, a priori, the constant-depth construction requires beam splitters to act between arbitrary pairs of modes, something that may be hard to do with current experimental designs.

In \sec{bosonnew_b} I present several experimental and theoretical aspects that permeate our experiments. I give an overview of the experimental setup, with emphasis on the workings of the integrated interferometers. This is followed by discussions of the numerical analysis performed during different stages of the experiment: simulation of different unitary ensembles, sampling of the interferometers to be fabricated, tomography of the device, and comparison between experimental and theoretical data.  I reiterate that my contribution, and this thesis, are of a mostly theoretical nature, and thus I will only give brief nontechnical descriptions of the experimental apparatus, focusing on the theoretical and computational efforts.

In \sec{bosonnew_c} I report the data and conclusions drawn from the three experimental papers: (i) a small-scale implementation of the full BosonSampling model, using 3 photons in a uniformly-random 5-port interferometer \cite{Crespi2013b}, (ii) an investigation of the bosonic bunching behavior of 3 photons in several interferometers of increasing size (2--16 modes), including a comparison with behavior expected from the bosonic birthday paradox (cf.\ \sec{bosonreview_c_BBP}) and a new bunching law that was observed experimentally and proven theoretically, and finally (iii) a larger demonstration of BosonSampling on interferometers of 7 and 9 modes, with emphasis on certification of the device, as described in \sec{bosonreview_c_certif}.

Finally, \sec{bosonnew_d} is devoted to concluding remarks.  A discussion on the scalability issues that will affect future experiments is provided, as well as a discussion of several open questions of interest stemming both from experimental and complexity-theoretical aspects of the model. This chapter is also complemented by several appendices at the end of the thesis, consisting of additional tables of numerical and experimental data, and copies of the Mathematica notebooks used in the simulations and device tomography.

\section{Constant-depth (exact) BosonSampling} \label{sec:bosonnew_a}

In this section, I prove that the exact BosonSampling result of \sec{bosonreview_b} holds even if the linear-optical circuits only have a constant number of beam splitter layers. This will follow the general recipe, described in that section, whereby if a certain restricted model of quantum computation, when imbued with post-selection, becomes as powerful as arbitrary quantum computation with post-selection (postBQP), then a weak simulation of that model (even within multiplicative error) cannot be efficiently performed only by classical means unless the polynomial hierarchy collapses to its third level. We begin by showing how this reasoning applies to a family of circuits that is in the intersection of constant-depth quantum circuits \cite{Terhal2004}, IQP \cite{Bremner2011}, and non-adaptive measurement-based quantum computation (cf.\ \sec{introduction_c}), in \sec{constQC}. We believe such a unified proof may also be of independent pedagogical interest. We then proceed to show, in \sec{constLO}, how to adapt it for constant-depth linear optics via the KLM scheme reviewed in \sec{bosonreview_a}.

\subsection{Constant-depth quantum computing} \label{sec:constQC}

Let us begin with an alternative proof of the result found in \cite{Terhal2004}. In that paper, the authors prove that efficient classical simulation of quantum circuits of depth $4$ would imply a collapse of the polynomial hierarchy. They also show that this depth is the smallest for which this holds by giving an explicit simulation for circuits of depth $\leq 3$. For the present purposes, the depth of a quantum circuit is defined as the number of layers of arbitrary two-qubit gates, plus one final round of measurements. We consider that gates acting on qubit pairs that share a common qubit (e.g.,  $\{1,2\}$ and $\{2,3\}$) cannot be done simultaneously, even if they commute. Single-qubit gates do not contribute to the depth count since we can consider them as part of the nearest two-qubit gate which means, in particular, that we can choose initialization and measurements in any single-qubit basis without loss of generality.

Our starting point is an universal construction in the measurement-based model of quantum computing (MBQC) \cite{Raussendorf2001, Raussendorf2012}. We will restrict ourselves to a description of how the computation is performed, with no discussion on how its universality is proved---that would lie much beyond the scope of this thesis, and we refer to \cite{Broadbent2009} and references therein.

Let $G$ be the family of brickwork graphs, such as in \sfig{brickwork}{a}, parameterized by some size $n$ (i.e.\ the graph has poly$(n)$ vertices). Consider, then, the corresponding graph state $\ket{G}$, built as follows: (i) for each vertex in $G$ prepare a qubit in the $\ket{+}$ state and (ii) for each edge in $G$ apply a $\cz$ gate between the two corresponding qubits. This generates a highly-entangled multi-qubit state. The computation then proceeds by a sequence of single-qubit measurements on the qubits of $\ket{G}$. Each measurement is done in one of a discrete set of bases, and their outcomes determine the bases of future measurements. The complete set of measurements, including the order in which they are performed and the dependence of some measurements on the results of others, is known as a \emph{measurement pattern}. 

\begin{figure}[t]
\capstart
\centering
\subfloat[]{\centering \includegraphics[width=0.55\textwidth]{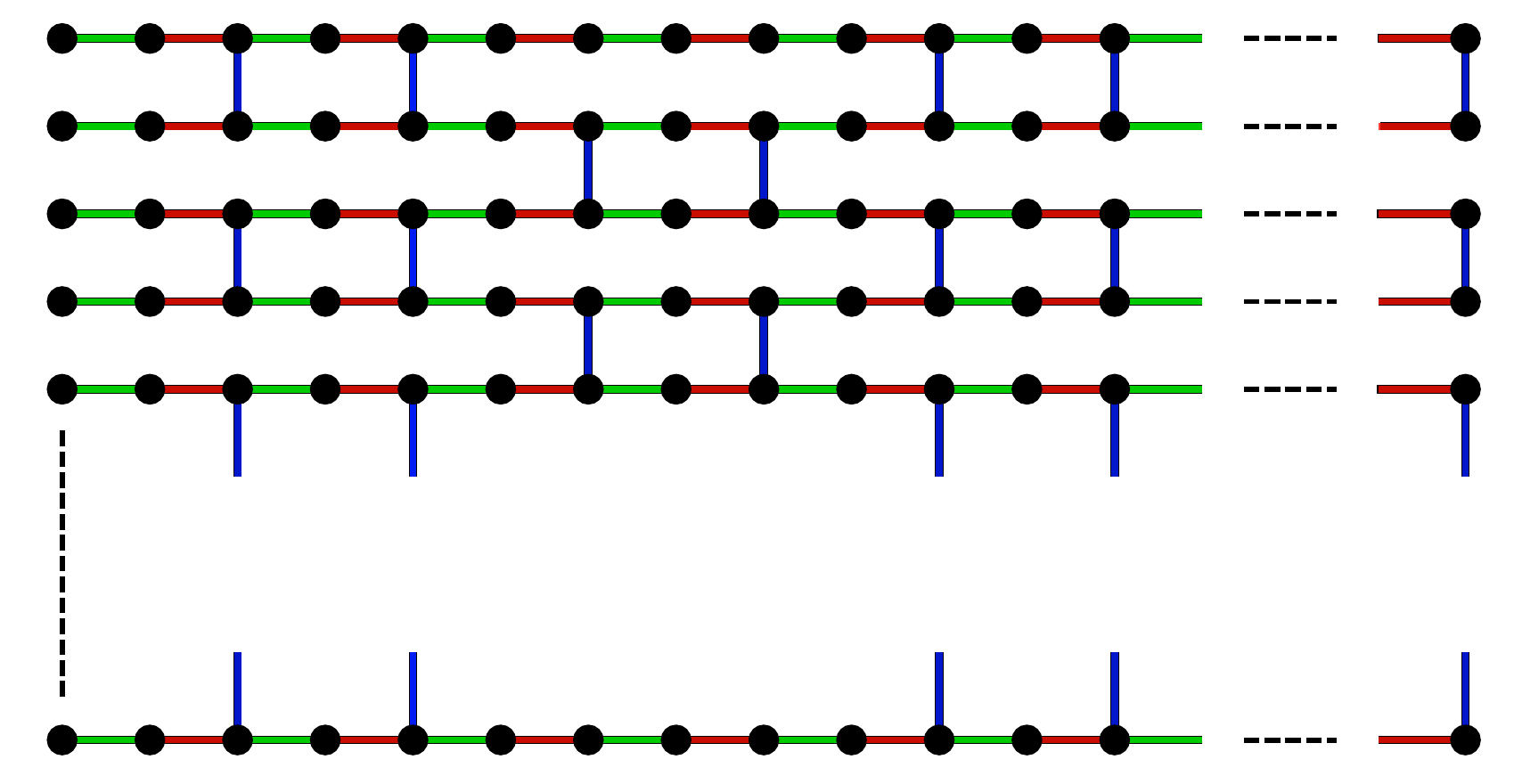}} \qquad
\subfloat[]{\centering \includegraphics[width=0.3\textwidth]{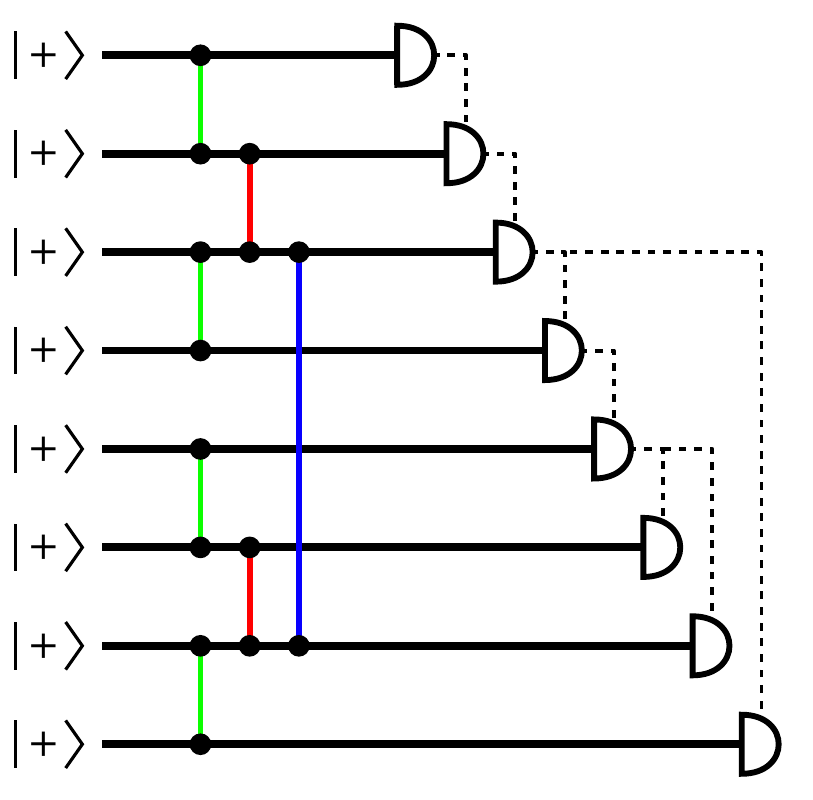}}
\caption[The brickwork graph state.]{(a) The brickwork graph state. (b) Representation of the MBQC protocol as a circuit. Dashed lines conceptually represent the dependence of some measurements bases on the others. The color coding shows how, in the translation from (a) to (b), preparation of the state corresponds to only three rounds of $\cz$ gates, while almost all of the temporal structure resides on the measurement adaptation.}
\label{fig:brickwork}
\end{figure} 

We now restate the following Theorem, which is taken from \cite{Broadbent2009}\footnote{Although we follow the result of \cite{Broadbent2009} for convenience, the brickwork state was known to be universal for MBQC before that, see e.g.\ \cite{Childs2005}.}:

\begin{theorem} \label{thm:MBQC}
Let $G$ be the brickwork graph shown in \fig{brickwork}. The corresponding graph state is universal for measurement-based quantum computation. The computation is performed as follows:
\begin{itemize}
\item[(i)] Initialize a qubit in the $\ket{+}$ state for every vertex of $G$;
\item[(ii)] Apply $\cz$ gates according to the edges of $G$;
\item[(iii)] Sequentially measure each qubit of $G$ in one the bases 
\begin{equation*}
\left \{\ket{\pm_\theta}=\ket{0} \pm \textrm{e}^{i \theta} \ket{1} | \theta = 0, \pm \frac{\pi}{4}, \pm \frac{\pi}{2}, \pm \frac{3\pi}{4}, \pi \right \}.
\end{equation*}
\end{itemize}
Any poly-sized quantum circuit on $n$ qubits can be simulated in this way using some graph $G$ of poly$(n)$ size, and by a suitable choice of the measurement pattern which, furthermore, can be computed efficiently classically.
\end{theorem}

A formal proof of this result, found in \cite{Broadbent2009} and references therein, gives explicit measurement patterns corresponding to gates in the universal set $\{ \cnot, H, T \}$, and thus an explicit procedure to simulate any quantum computation. Omitting further details from this proof, we would just like to point out that this protocol, described as a quantum circuit (see e.g.\ \sfig{brickwork}{b}), inherits a temporal structure almost exclusively from the ``classical'' adaptation of measurements. In other words, preparation of the graph state only takes a few rounds of two-qubit gates, and the depth of the circuit stems from the fact that some qubits \emph{must} be measured prior to others.

Now consider what happens if we replace adaptive measurements by post-selection. That is, rather than making a measurement and conditioning future measurements on its outcome, we just \emph{postselect} each measurement to a given outcome such that adaptation is unnecessary---this is akin to what we did with the KLM scheme in \sec{bosonreview_b}. As an example, suppose we measure one qubit in basis $M_1$, with two possible outcomes labeled $+$ and $-$, and must measure a second qubit on either basis $M_2^{+}$ or basis $M_2^{-}$ depending on the outcome of $M_1$. This can be replaced simply by post-selecting $M_1$ to the outcome $+$ and simultaneously measuring the second qubit in basis $M_2^{+}$. Doing this for every measurement effectively flattens out the temporal structure of the protocol, allowing us to perform all measurements in a single round\footnote{If this sounds too good to be true, recall from the discussion of \sec{introduction_d} that post-selection is indeed an extremely unrealistic ``power''.}. 

Our goal is now basically reached, as the procedure described already constitutes a quantum circuit of depth 4. To see this, consider the circuit description of the MBQC protocol, as shown in \fig{brickwork}. In this circuit, the depth count goes as follows\footnote{Recall that the $H$ gates necessary to initialize the qubits in the $\ket{+}$ state do not count for the depth.}: (i) a first round of $\cz$ gates (green edges of \sfig{brickwork}{a}); (ii) a second round of $\cz$ gates (red); (iii) a third round of $\cz$ gates (blue), followed by a round of single-qubit gates to prepare the measurement bases; and (iv) final measurement in the computational basis. Note that only three rounds of $\cz$ gates suffice because vertices of the brickwork state have degree at most 3, as show in \fig{brickwork}. It should be clear by our arguments and \thm{MBQC} that circuits of this form, when imbued with the power of post-selection, can implement any computation in postBQP. We then conclude, by the previous discussion, that an efficient classical simulation of its output would imply collapse of the polynomial hierarchy.

It is interesting that the MBQC approach is very well suited for this proof, since almost all the temporal structure lies in the adaptive measurements, which is precisely what we replace by post-selection. Curiously, all information about the computation \emph{also} lies in the measurements, since the brickwork state does not depend at all on the underlying quantum computation (except, of course, in its size). Thus, it seems that all ``computational power'' (to abuse the terminology) of the constant-depth quantum circuit resides on the combinations of possible choices of measurement bases. An interesting open question is whether there is a natural way to randomly choose the measurement bases such that the final simulation is a provably hard instance, similar to the worst-case/average-case equivalence conjectured for the approximate BosonSampling scenario, as discussed in \sec{bosonreview_c_interf}.

Another curious aspect of this proof is that it encompasses several similar proofs for different models. The resulting circuit obtained by flattening out the adaptivity of MBQC has depth 4, is a non-adaptive measurement-based protocol (by definition), and is in IQP, since it only uses gates diagonal in the $X$ basis. This is interesting, but maybe not specially surprising---the concept of gate teleportation is present, in one form or another, in all of these results, and is closely related to the historical origin of MBQC. We should also point out that, while this proof unifies several results, it does not provide concrete relations between these models---the resulting circuit lies at their intersection, but each model may have circuits that perform tasks outside of this intersection.

\subsubsection{Quantum circuits of depth 3} \label{sec:constQC3}

Besides proving that depth 4 is sufficient for the hardness result, the authors in \cite{Terhal2004} also prove that it is necessary. This is done by giving an explicit simulation of any depth 3 circuit, which we briefly reproduce here.

Consider an arbitrary depth-3 quantum circuit. That is, the qubits are initialized in the computational basis, undergo a first round of arbitrary two-qubit gates, followed by a second round of arbitrary two-qubit gates, and a final round of measurements. The simulation becomes straightforward if we reinterpret this as a two-round computation: first the qubits are prepared in arbitrary two-qubit states, and then they are measured in arbitrary two-qubit bases. The classical simulation then proceeds as follows: 

\begin{itemize}
\item[(i)] Choose one particular measurement $M_i$. This is a two-qubit measurement where each qubit may be entangled with some other qubit. Thus, its outcome probabilities only depend on a four-qubit state and are trivially easy to compute in constant time.
\item[(ii)] Compute the outcome probabilities and simulate $M_i$ by classically sampling a two-bit string from the corresponding distribution and fixing the output of the measurement accordingly.
\item[(iii)] Project the two measured qubits onto the fixed bit string, and update the description of the two qubits entangled to them accordingly. These two qubits are now in an arbitrary known two-qubit state.
\item[(iv)] The next measurement involving one of the qubits updated in (iii) is again a two-qubit measurement that depends only on an arbitrary four-qubit state. 
\item[(v)] Repeat this procedure until all measurements have been fixed.
\end{itemize}

It is easy to see that this simulation samples from the same distribution as the corresponding quantum circuit. It also consists of a weak simulation (cf.\ \sec{introsimul}), rather than strong, since at no point do we compute the probabilities associated with the full output of the circuit, only small two-bit strings at a time. As a consistency check, one could attempt to generalize this simulation for circuits of depth 4. This would naturally fail, because the first measurement in step (i) would depend on an 8-qubit state, hence step (iii) would collapse the unmeasured qubits to an arbitrary 6-qubit state. But then, the next time we returned to step (i) the measurement might depend on the state of 14 qubits, and the iteration is clearly disrupted.

\subsection{Constant-depth linear-optical circuits} \label{sec:constLO}

We now combine the previous results into one---the computational complexity of exact BosonSampling with a constant-depth linear optical circuit. 

The first step is a rather straightforward concatenation of previously reviewed results. Consider a constant-depth circuit obtained from some universal MBQC protocol by replacing adaptive measurements with post-selection, as in \sec{constQC}. Now map the resulting quantum circuit to a linear-optical circuit using the constructions of \fig{KLMscheme}. Finally, replace the adaptive measurements of the KLM protocol by post-selection. The resulting circuit is naturally as strong as postBQP, and thus it cannot have an efficient classical simulation unless the polynomial hierarchy collapses. The fact that there are actually three rounds of post-selection is immaterial.

Let us now count the depth of the resulting optical circuit. The single-qubit gates cannot be ``absorbed'' into the two-qubit gates anymore, as was done in \sec{constQC}, due to the dual-rail encoding. More specifically, in this encoding the single qubit gates involve the two modes of a single qubit, while two-qubit gates involve two modes of different qubits. We will, however, absorb phase shifters into the closest beam splitter, and count layers of arbitrary two-mode transformations. The total count then, goes as follows: one layer of balanced beam splitters for the preparation of the $\ket{+}$ states, followed by six layers for the preparation of the entangled state (i.e.\ two for each layer of $\cz$ gates of the brickwork state, using the $\cz$ of \sfig{KLMscheme}{c}). Finally, one last layer of beam splitters and phase shifters for the preparation of measurement bases, and a round of measurements. This amounts to a total depth of 9.

We can now use two tricks to reduce this depth count. First note that, according to \sfig{KLMscheme}{c}, when we perform a $\cz$ gate, only one mode of each encoded qubit is involved, while the other remains idle. By a simple adaptation of the circuit of \sfig{KLMscheme}{c}, we can use the second mode of each qubit to simultaneously implement the second round of $\cz$s. To do this, we include two $\pi$ phase shifters before the circuit of \sfig{KLMscheme}{c}, resulting in a gate which, up to a global phase, only adds a minus sign to the photonic $\ket{00}$ state. It is clear that, if we act with this gate on the modes encoding the $\ket{0}$ state of two qubits, we obtain precisely a logical $\cz$ gate. By a clever alternation of which mode we use for each $\cz$ gate, we can build the long paths of the brickwork state of \fig{brickwork} in one round of $\cz$s (i.e., two rounds of beam splitters). This reduces the overall depth to 7.

The second trick is to use a procedure similar to the usual gate teleportation \cite{Gottesman1999b}, which can be found in \cite{Knill2001b}, to parallelize the third layer of $\cz$s. The original idea is that we want to implement some faulty (i.e.\ probabilistic) gate $U$ on a computational qubit, but without endangering the information stored in that qubit. We can then do the following: we prepare an auxiliary two-qubit state $\ket{\phi^+}=\frac{1}{\sqrt{2}}(\ket{01}+\ket{10})$, implement $U$ on the second qubit of this auxiliary state, and, if $U$ succeeds, project the first qubit of the auxiliary state together with our original qubit on the $\bra{\phi^+}$ state (see \sfig{teleportation}{a}). Of course, in standard quantum computation (i.e.\ BQP) we cannot guarantee that a two-qubit measurement will project our qubits on the $\bra{\phi^+}$ state, but for our purposes we can just post-select on observing this outcome. Now notice that, if rather than being a two-qubit state, $\ket{\phi^+}$ actually represents a state of one photon in two modes, the mathematical structure is exactly the same, and so it allows us to use an equivalent scheme to teleport the states of the modes in the KLM scheme. Since each $\cz$ gate only involves one mode of each qubit, as in \sfig{KLMscheme}{c}, this allows us to implement the gate teleportation by teleporting only one mode, rather than a complete qubit, as shown in \sfig{teleportation}{a}. Thus, we can perform all the $\cz$ gates for the brickwork state in parallel, and just teleport the computational states around, in a very similar spirit to the original result of \cite{Terhal2004}. This scheme, and how it can used to reduce the depth, is illustrated in \sfig{teleportation}{b}. This reduces the total depth count to 5.

\begin{figure}[t]
\capstart
\centering
\subfloat[]{\centering \includegraphics[width=0.3\textwidth]{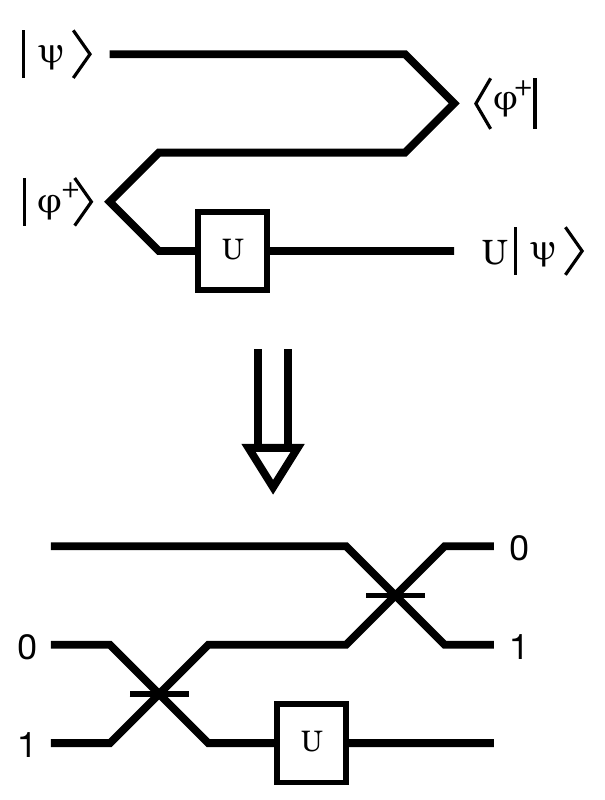}}
\subfloat[]{\centering \includegraphics[width=0.65\textwidth]{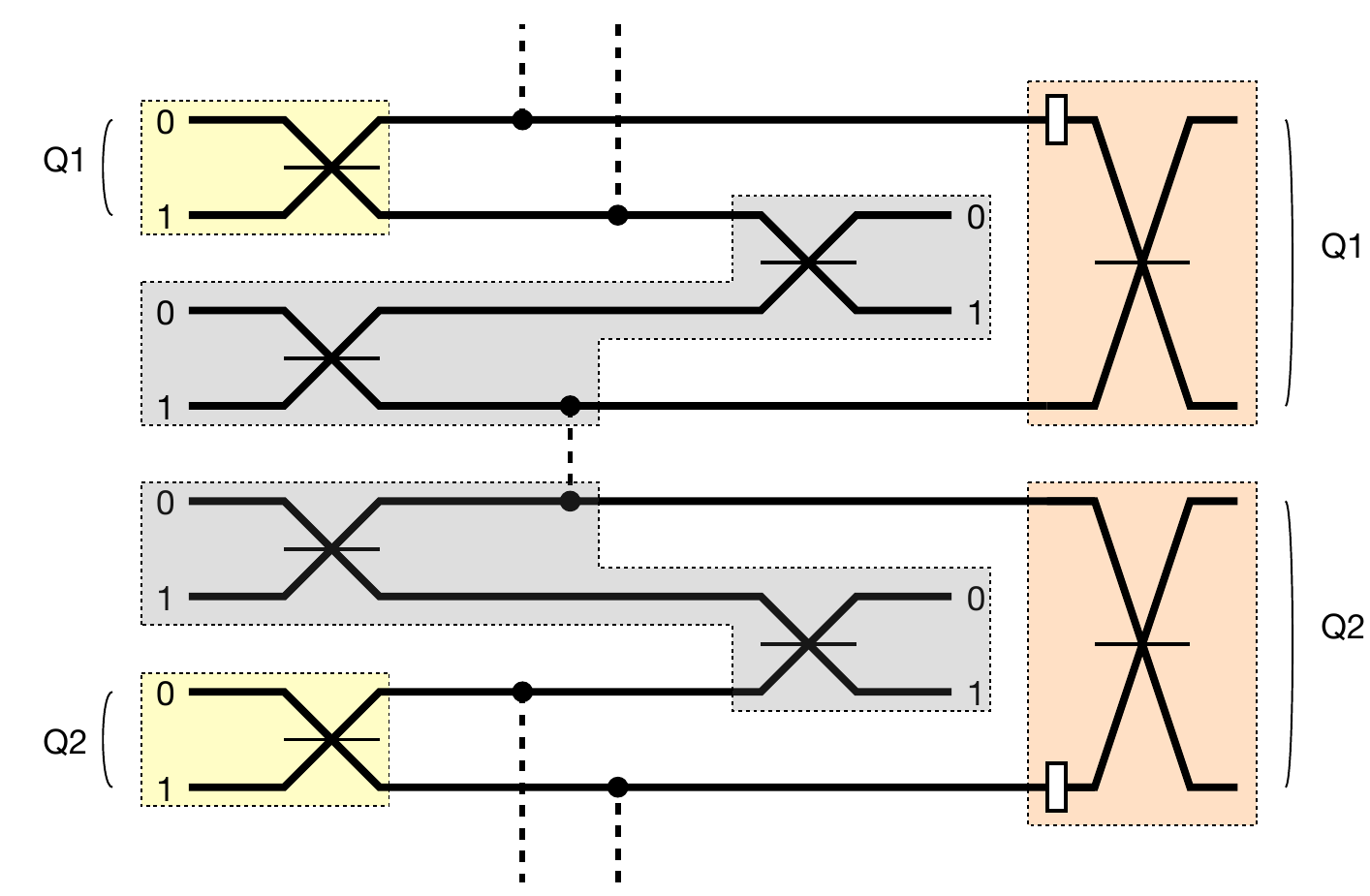}}
\caption[Constant-depth BosonSampling]{(a) The standard gate teleportation scheme in the quantum circuit model, and its linear-optical version, replacing the entangled two-qubit state $\ket{\phi^+}$ by the corresponding photonic state $\frac{1}{\sqrt{2}}(\ket{01}+\ket{10})$, which can be mapped to the Fock basis by a balanced beam splitter. (b) The resulting scheme for constant-depth (exact) BosonSampling. Every beam splitter in this figure is balanced, and white rectangles represent phase shifters. Dashed lines represent the probabilistic $\cz$ gates of the usual KLM scheme, constructed as in \sfig{KLMscheme}{c}, shown implicitly so as to not clutter the figure. The color coding is intended only to aid in the mapping from the construction of \sec{constQC}: yellow regions represent preparation of the $\ket{+}$ states, gray regions represent the gate teleportation scheme shown in (a), and orange regions represent the single-qubit measurements. The depth bottleneck of the model is represented here as the two central modes. Each is involved in four beam splitters: one for state preparation (explicit), two in the $\cz$ gate (implicit), and one for the arbitrary single-qubit measurement (explicit), resulting in a circuit of depth 5.}
\label{fig:teleportation}
\end{figure} 

Note that, of the two tricks mentioned, the second would already suffice to reduce the depth count to 5. We presented the first as it may be of independent interest for the parallelization of linear-optical circuits even in future implementations of the full KLM scheme. Also, since here we do not have to worry about the gates failing anyway, it would be wasteful to do all the $\cz$'s in fresh ancilla modes rather than directly on the qubit themselves whenever possible.

\subsubsection{Depth-3 linear optics} \label{sec:constLO3}

We now ask whether the two depth bounds meet. That is, can we prove that 5 layers are not only sufficient, but also necessary? Alas, unfortunately so far the answer is no. A direct adaptation of the simulation of \sec{constQC3} only proves that depth-3 linear optical circuits are classically simulable.

The main difference between the simulation of \sec{constQC3} for quantum circuits and the analogous for linear optics is the following: assuming that each mode is initialized with either one or zero photons, the state after the first layer of beam splitters can have some modes occupied by two photons. Thus, one could argue that, when we do the iterative step of the simulation, the sample space changes from one step to the next, possibly disrupting the iteration. But, if we observe that no mode can ever have more than four photons, this can easily be accommodated into the simulation scheme as follows:

\begin{itemize}
\item[(i)] Choose a measurement $M_i$. This is a two-mode measurement, where each mode has up to two photons and may be ``entangled'' with some other mode. Thus, this measurement only depends on the states of four modes, each with up to two photons, and its probabilities can be computed in constant time.
\item[(ii)] Compute the outcome probabilities and simulate $M_i$ by sampling from the corresponding distribution and fixing the output of the measurement accordingly. Note that the space of outcomes can have from zero to four photons distributed in any way among the two measured modes.
 \item[(iii)] Project the two measured modes onto the sampled state, and update the description of the two unmeasured modes accordingly. They are now in an arbitrary, but known, two-mode state. By photon-number preservation, the outcome fixed in (ii) determines the number of photons in the collapsed state, but it is easy to check that it can have at most two photons per mode.
\item[(iv)] The next measurement involving one of the modes collapsed in (iii) is again a two-mode measurement that depends only on a four-mode state. One caveat is that these four modes may now contain up to eight photons overall (up to two per mode) depending on previously fixed outcomes. Nonetheless, any subsequent two-mode measurement can only detect 4 of these photons, so step (ii) can always be applied.
\item[(v)] Iterate this procedure until all measurements have been fixed.
\end{itemize}

As should be clear, since we have only two layers of beam splitters, no mode can ever have more than $2$ photons, except immediately prior to a measurement, in which it can have up to $4$ photons. This guarantees that, even though the first measurement differ from subsequent measurements in the set of allowed states, this does not disrupt the iterative procedure. Also note that this simulation seems to break down for more layers. If the optical circuit has depth 4, say, then a first measurement could depend on the state of eight modes, and collapse the unmeasured modes into an arbitrary state of six modes. But then the next measurement involving these modes could depend on 14 modes, and so forth. Since each new choice of measurement may depend on a larger number of modes, the iteration step fails.

\subsubsection{Discussion}

In this section, I showed that the exact BosonSampling result reviewed in \sec{bosonreview_b} already holds if the linear optical circuit has depth 5. I also showed that it cannot hold for depth 3 or lower. This seems to be a consequence of the fact that a probabilistic linear-optical $\cz$ gate must be constructed with at least two layers of beam splitters. This is natural, since it is known that ancilla modes are necessary \cite{Knill2002} for the construction of the gate (even probabilistically), and so a single layer of beam splitters would clearly not suffice to interact both the two computational modes and the ancilla modes in question. It cannot be ruled out, of course, that an ingenious use of post-selection might provide a construction for a hard-to-simulate depth-4 exact BosonSampling instance, via completely different means.

Besides the corresponding result for depth 4, we leave another open question: is it possible to somehow close the depth gap between this result and the approximate BosonSampling result reviewed in \sec{bosonreview_c}? For approximate BosonSampling, the minimal known depth is $m$log$(n)$, for $n$ photons in $m$ modes, as shown in \cite{Aaronson2013a}, and for exact BosonSampling the minimal known depth, as we showed, is 5. There are two natural routes to close this gap: either prove that the approximate BosonSampling result can be simplified, using a further parallelized scheme or possibly a matrix ensemble other than the Haar ensemble, or give an efficient classical algorithm for approximate constant-depth BosonSampling (thus showing that constant-depth and regular BosonSampling are indeed fundamentally different). The latter would not be too surprising: there exist problems for which the exact solution is hard (e.g.\ in $\#$P), while an approximate solution is easy (i.e.\ in P), such as the permanent of positive matrices, as discussed in \sec{bosonreview_c_simulation}. Another possibility would be to explicitly investigate the permanents of constant-depth interferometers. However, the unitary of a depth-5 interferometer has at most $32$ nonzero entries in each row or column, simply because a photon that enters one mode cannot leave in more than $32$ different modes, and similarly a photon that exits in a particular mode cannot have entered in more than $32$ different input modes. Thus, these matrices obey a very special structure, very different from the Haar ensemble, and a worst-case/average-case equivalence allowing us to condition the hardness of their permanents on some plausible conjecture seems quite unlikely.

An approximate result for constant-depth BosonSampling would also have important consequences for the corresponding complexity classes. For example, it can contribute to the more general program of relating the different restricted models described so far. Recall, from \sec{bosonreview_b}, that an arbitrary $n$-photon, $m$-mode linear-optical circuit can be simulated in BQP simply by treating each mode as an $n$-level system, and using a standard mapping from $m$ qudits to O$(m$ log $n)$ qubits \cite{Aaronson2013a,LivroNielsen}. As discussed in \cite{Terhal2004}, this translation induces an O$($log $n)$ depth increase, simply because an arbitrary single-qudit gate will translate to an arbitrary $($log $n)$-qubit circuit. This logarithmic depth increase suggests that arbitrary linear-optics might not be simulable by constant-depth quantum circuits. However, the same is not true for constant-depth linear optics. In an optical circuit consisting of a no-collision Fock state input (e.g.\ $\ket{1 1 1 \ldots 1 0 \ldots 0}$) followed by a constant number of beam splitter layers, it is easy to see that no mode can have more than a constant number of photons. For example, if the circuit has depth 5, no mode ever has more than 32 photons, which means that this system can be viewed simply as a collection of $m$ 32-level systems. Consequentially, simulating it on a quantum computer does not induce a logarithmic depth increase and it can, in fact, be simulated in constant depth. So we see that constant-depth linear optics is provably contained within constant-depth quantum computing (although the constant will probably not be 5), and so an approximate result for the former, along the lines of \sec{bosonreview_c}, might suggest an approximate result for the latter. This would be the first evidence that the other hardness results are in fact as robust as BosonSampling.

Finally, we point out that this result might also have important implications for experiments. As we will discuss shortly, all BosonSampling experiments so far have been performed using integrated photonic circuits. In these devices, the depth of the circuit is one of the leading causes of photon loss and attenuation of the signal, so any optimization of the model in terms of depth should reduce the experimental efforts. A drawback of the proof given here, however, is that it requires beam splitters acting between arbitrarily distant modes, a requirement that is not well-suited for these devices. Nevertheless, we believe our proof may pave the way for more practical simplifications of the proposed experiments.

\section{Experimental design and chip tomography} \label{sec:bosonnew_b}

In this Section, we describe in more detail several aspects relevant to the experiments reported in \sec{bosonnew_c}. In \sec{bosonnew_b_integr} we describe the general concept of integrated photonic chips, including the physical principle of the device, the fabrication procedure, as well as a novel technique to control independently the phase shifter and beam splitter parameters, necessary for the construction of an arbitrary interferometer. In \sec{bosonnew_b_experiment} we give a conceptual description of the experimental procedure, from the photon sources to the detectors. In \sec{bosonnew_b_ensembles} we describe the numerical analysis behind the sampling of the random interferometers---namely, the sampling of a uniformly-random unitary and the numerical simulations of the behavior of other (non-uniform and of easier fabrication) ensembles, in order to investigate how well they reproduce the uniform ensemble for several figures of merit of interest. Finally, \sec{bosonnew_b_tomography} is dedicated to the reconstruction of the unitary matrix from single- and two-photon data. We describe a tomography algorithm due to Laing and O'Brien \cite{Laing2012}, and how we refined it to improve the agreement between the reconstructed unitary and the experimental data. 

These tools and techniques, both theoretical and experimental, are present in one way or another in all the experiments, which is why we concentrated them on this section, and will devote \sec{bosonnew_c} exclusively to the results and analysis of each experiment.

\subsection{Integrated photonics}  \label{sec:bosonnew_b_integr}

The most conventional approach for implementing an arbitrary linear-optical transformation is via the setup of optical elements---such as beam splitters, phase shifters, wave plates, etc---on an optical table, and propagation of photons through free space. However, this suffers from severe scalability issues, mainly due to mechanical instabilities. As mentioned in \sec{bosonreview_a}, to obtain a a $\cz$ gate that works with probability greater than $95\%$, one would require thousands of beam splitters and phase shifters. The idea of aligning thousands of beam splitters (per two-qubit gate!) on an optical table raised natural skepticism over whether linear-optical quantum computing would ever be really feasible. Of course, this drove a theoretical effort to simplify the scheme, reducing this requirement by up to two orders of magnitude. On the experimental side, several groups adopted a more promising approach, based on \emph{integrated} photonics. 

The physical workings of integrated photonic devices are based on the concept of optical confinement. The well-known phenomenon of total internal reflection happens when light, whilst propagating in a medium of refractive index $n_1$, strikes the boundary with a medium of lower refractive index $n_2<n_1$ at an angle larger than $\textrm{arcsin}(n_2/n_1)$, causing it to be completely reflected. In turn, an optical waveguide consists of a strip or cylinder of material embedded in another of lower refractive index, in which case the phenomenon of total internal reflection can be exploited to create a conduit for light: the light is led along the waveguide, propagating by repeated reflections at the boundaries\footnote{This is the principle behind, for example, the optical fiber.}. 

Consider now what happens when two of these waveguides are embedded in the medium. Although light suffers total internal reflection in each waveguide, there is an evanescent wave that reaches into the external medium with an exponential decay. If the two waveguides are brought sufficiently close to each other, such that the evanescent waves have non-negligible overlap, a fraction of the light can ``leak'' from one waveguide to another. More specifically, let the waveguide separation be $s$, and let the power in waveguides 1 and 2 be $P_1(z)$ and $P_2(z)$, respectively, as a function of the interaction length $z$ (i.e.\ the distance through which they interact). Then, it can be shown (see \cite{LivroSaleh,Szameit2007,Yariv1973}) that, if $P_1(0)=P_0$ and $P_2(0)=0$,
\begin{align}\label{eq:waveguideshifting}
P_1(z)=&P_0 \cos^2 (\kappa z), \\
P_2(z)=&P_0 \sin^2 (\kappa z),
\end{align}
where $\kappa$ is a coupling coefficient that decays exponentially with $s$ as
\begin{equation}\label{eq:expodecay}
\kappa = \kappa_0 e^{-s/s_0},
\end{equation}
in which $\kappa_0$ and $s_0$ are constants. The values $\kappa_0$ and $s_0$ may have a complicated dependency on the shape of the waveguides, the difference between the refractive indices $n_1$ and $n_2$, the light frequency, etc, and it is often more practical to obtain them experimentally, as we will discuss in \sec{parametercontrol}. If the two waveguides are infinite in length, the power periodically flows from one to the other, as \eq{waveguideshifting} shows. However, if they only interact for a given length $Z$, a specific fraction of the light will flow---this configuration is known as a directional coupler. In the quantum regime, this coupling translates to an amplitude that a photon will hop from one waveguide to another---in other words, for our purposes, it will behave precisely as a beam splitter\footnote{In fact, from now on we interchangeably refer to these devices as beam splitters and directional couplers.}, and we will identify $\sin{\kappa z}$ as the beam splitter's transmissivity (recall, from \sec{twomode}, that we use transmissivity to refer to the square root of the transmission probability). For a comprehensive introduction to the theory of waveguides, see \cite{LivroSaleh}.

As a generalization of a directional coupler, one can construct larger optical circuits, consisting of several waveguides whose relative distances are varied so as to apply sequences of beam splitters. This is precisely the setting in which several recent experiments, including those reported later in this chapter, were performed. The use of integrated devices, rather than free-space optics, has led to miniaturization of the devices\footnote{Not unlike the shift from vacuum tubes to transistors in modern day electronics.}, which not only reduces errors from misalignment of optical elements, vibrations from external sources, etc, but also allows packing of a larger numbers of beam splitters than would be possible in a typical optical table (in some cases, up to 50 beam splitters in a few centimeters of glass). For a (noncomprehensive) list of recent experiments using these techniques, see \cite{Sansoni2012,Peruzzo2011,Meany2012,Spagnolo2013a, Broome2013, Crespi2013b,Spring2013,Tillmann2013,Spagnolo2013b,Carolan2013,Spagnolo2013c}.

\subsubsection{Femtosecond laser micromachining of integrated devices}

There are several techniques for the fabrication of these integrated devices, the most conventional ones consisting of lithographic methods. However, the chips constructed for the experiments reported here were built using the femtosecond laser micromachining technique, which has recently gathered increasing attention (for an extensive review of this method and how it compares to others, see \cite{Valle2009}). This writing technique consists of focusing femtosecond laser pulses into the volume of a transparent material, such as borosilicate glass. By focusing the ultrashort pulses in a very small region of the material, the field intensities become high enough to allow the medium to absorb energy, in a localized volume near the focus, via several nonlinear processes. This localized energy absorption, in turn, induces a permanent change in the refractive index of the material. By translating the sample relative to the laser beam at a uniform velocity, one can inscribe a variety of configurations of optical waveguides directly into the sample.

Writing waveguides with this femtosecond laser technique has several advantages over more conventional methods. Since the writing is done directly in the interior of the material by the laser, it does not require a mask (necessary for lithographic methods), and it is an inherently 3D technology, since the changes in refractive index can be induced at different depths of the material (typically 10 $\mu$m to 10 mm). As such, this technique allows a very fast and cost-effective prototyping of devices, and its intrinsic 3D nature allows implementation of novel layouts (impossible with conventional lithography) such as a tritter, where three waveguides are coupled at once \cite{Spagnolo2013a}, or the independent control of linear optical parameters we will describe in \sec{parametercontrol}. On the down side, devices produced by this technique also have lower waveguide quality, in terms of propagation loss, than other methods, which is why it is expected to be a complementary, rather than alternative, technology.

\subsubsection{Arbitrary integrated interferometers} \label{sec:parametercontrol}

Suppose now that we want to implement an arbitrary $m$-mode linear transformation $U$ on a photonic chip. This has three main requirements, namely (i) the ability to decompose $U$ in terms of beam splitters and phase shifters, and the ability to independently implement (ii) arbitrary phase shifters and (iii) arbitrary beam splitters. I will now describe how each of these is done, in turn.

To decompose $U$ into two-mode optical elements we use a procedure due to Reck \textit{et al.}\ \cite{Reck1994}. An example of the circuit that results from such a decomposition, and its implementation as an integrated circuit, can be seen in \fig{Reck}, for $m=5$. This procedure is akin to Gaussian elimination, where we start with unitary $U$ and sequentially obtain (unitary) matrices acting nontrivially only on two levels that ``cancel out'' off-diagonal matrix elements of $U$. By doing this for every off-diagonal element of $U$, we are left with an equation of the type $U_1 U_2 U_3 \ldots U_f U = I$, where $f=m(m-1)/2$ and I is the $m \times m$ identity. By inverting this equation, we obtain a decomposition of $U$ in terms of the $U_i$'s, that are two-level matrices. The explicit form of this decomposition can be found in \cite{Reck1994}, but also on \cite{LivroNielsen} where it is used to decompose an $n$-qubit unitary into two-qubit gates. The resulting optical circuit is also analogous to normal form given for matchgate circuits in \sec{fermreview_a}---in fact, the circuit of \fig{normalform} is nothing more than the ``fermionic'' version of \sfig{Reck}{a}. 

\begin{figure}[t]
\capstart
\centering
\subfloat[]{\centering \includegraphics[width=0.45\textwidth]{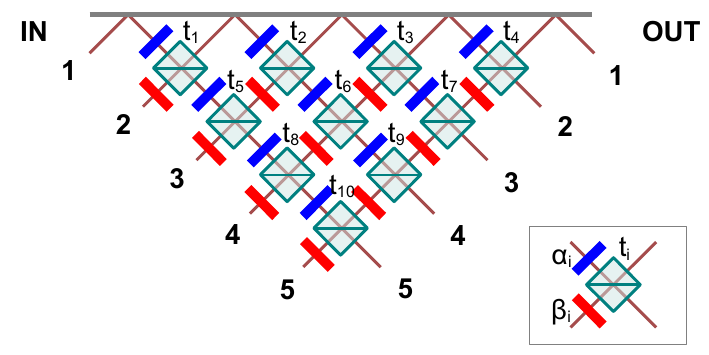}} \qquad
\subfloat[]{\centering \includegraphics[width=0.45\textwidth]{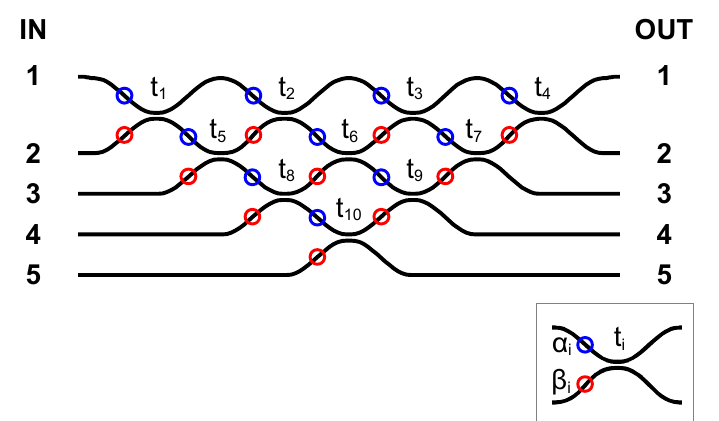}}
\caption[Arbitrary 5-port interferometer]{(a) Schematics of an arbitrary 5-mode interferometer composed by a circuit of 10 beam splitters with transmissivities $t_i$ and 20 phase shifters with phases $\alpha_i$ and $\beta_i$ (blue and red rectangles). (b) Realization of the same scheme using integrated photonics, where beam splitters are replaced by directional couplers, and phase shifters are replaced by S-bends (see main text).}
\label{fig:Reck}
\end{figure} 

It is also necessary to control the beam splitters and phase shifters. Ideally, we want to control the two types of parameters independently. However, recall that the effective transmissivity of a directional coupler depends on parameters such as distance between the waveguides and coupling length. And it is clear that attempting to change either value for one particular beam splitter, e.g.\ $t_2$ in \sfig{Reck}{b}, would change the relative optical path between the two modes that are involved in that coupler (1 and 2) and the three that are not (3, 4, and 5). Unfortunately, this would also result in an effective phase shift between modes 1 and 2 relative to modes 3, 4, and 5, not to mention it might misalign the circuit as a whole. Fortunately, there is a much more convenient way to control these parameters, which exploits a 3D architecture. Our paper \cite{Crespi2013b} is the first to use this idea to independently control all parameters needed for a completely arbitrary linear interferometer (although similar ideas were used in previous work by the same group, see e.g.\ \cite{Sansoni2012, Spagnolo2013a}).

\begin{figure}
\capstart
\centering
\includegraphics[width=0.7\textwidth]{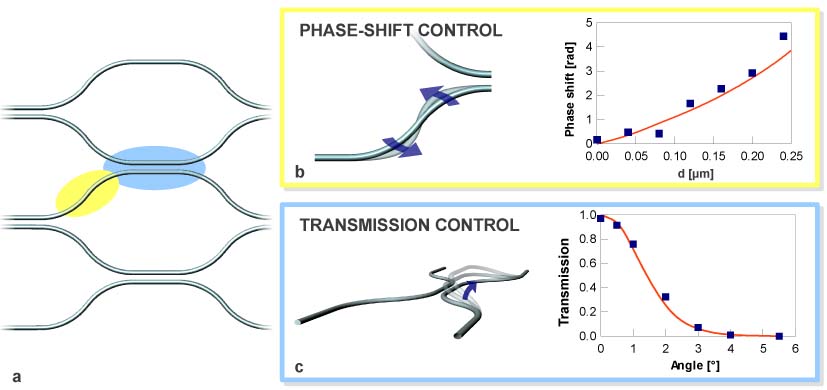}  
\caption[Independent control of integrated beam splitters and phase shifters]{(a) Controlled deformation of the S-bend at the inputs of each directional coupler and the 3D geometry allow independent control over the phase shift and transmissivity. (b) The drawing shows an undeformed S-bend together with a deformed one. The deformation is carefully chosen to increase the optical path without inducing kinks or abrupt bends. The plot shows experimental (squares) and theoretical (solid line) dependence of the induced phase shift on the deformation parameter $d$ of the S-bend, for light of wavelength $\lambda=806$ nm. (c) Control over the transmission probability ($T$) of the directional coupler is performed by changing the waveguide distance. This is obtained via a rotation of one arm of the directional coupler out of the main circuit plane. The plot shows the experimental (squares) and theoretical (solid line) transmission dependence on the rotation angle at $\lambda = 806$ nm.}	
\label{fig:LOelements}
\end{figure}

The phase shifters are controlled by deforming the S-bent waveguides at the input of each beam splitter, which has the effect of changing the path length (see \sfig{LOelements}{a-b}). The profile of the S-bend is chosen carefully so as to stretch the curve smoothly, otherwise abrupt bends would increase losses in the device. More specifically, an undeformed S-bend is described by a sinusoidal function of the kind
\begin{equation}
y = -\frac{h}{2} \cos \left( \frac{2 \pi}{L} x \right),
\end{equation}
where $L$ and $h$ correspond to longitudinal and transversal (relative to the waveguide direction) extensions of the waveguide, such as shown in \fig{cell}. The deformation of the S-bend is then obtained by a coordinate transformation
\begin{equation}
x' = x + d \sin  \left( \frac{2 \pi}{L} x\right),
\end{equation}
where $d$ parameterizes the transformation. This procedure ensures that the deformation has all the necessary smoothness properties, and the resulting length of the S-bend (as a function of $d$) can be obtained numerically. The calibration of this technique showed that it is possible to introduce a phase shift of up to $\pi$ in this way without causing additional losses, and is shown in \sfig{LOelements}{b}. The root mean square deviation of the phase shift measured experimentally relative to the predicted theoretically is of 0.25 rad.

\begin{figure}
\capstart
\centering
\includegraphics[width=0.7\textwidth]{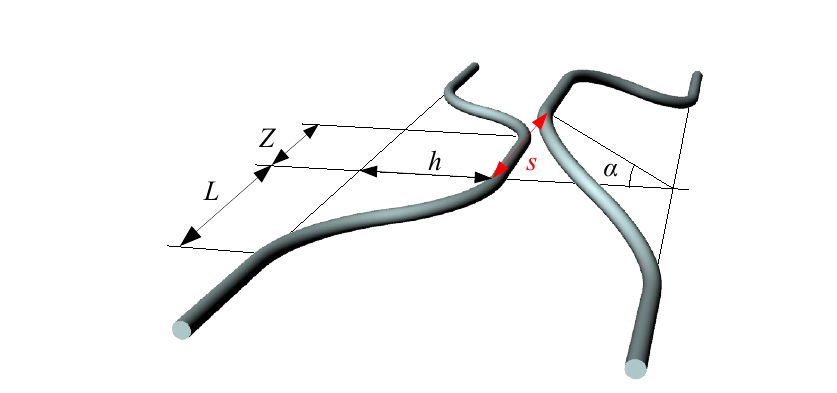}  
\caption[A unitary waveguide cell.]{Scheme of the smallest cell of a directional coupler, displaying all the relevant geometric parameters.}	
\label{fig:cell}
\end{figure}

The beam splitters, on the other hand, are controlled by exploiting the 3D fabrication capability of the femtosecond laser writing technique to rotate the waveguide out of the main plane of the circuit (see \sfig{LOelements}{c}). This provides a way to change the distance between the waveguides without changing the path length and introducing unwanted phase shifts. Recall from \eqs{waveguideshifting}{expodecay} that the coupling coefficient decays exponentially with the distance between the two waveguides, with two parameters $\kappa_0$ and $s_0$ that depend on particularities of the waveguides. The dependence of the angle $\alpha$ on the desired transmission probability $T=t^2$ can be obtained analytically (see the Supplementary Material of \cite{Crespi2013b}) and is plotted in \sfig{LOelements}{c}, although we omit its explicit expression as it is not very enlightening. Calibration of the technique allows fitting of the parameters, providing $\kappa_0 = 42$ mm$^{-1}$ and $s_0 = 2.4$  $\mu$m, as well as showing that that a rotation by a few degrees already allows for the whole range of transmissivities.

The calibrations of both types of optical elements were done by fabricating several single instances of each, using the same fabrication technique of the larger interferometers. One might naturally worry whether the concatenation of several of these devices might induce additional errors beyond those observed during calibration. We will return to this question in \sec{bosonnew_c}, where we will show that the agreement with the single-, two- and three-photon distributions observed in larger devices are, in fact, compatible with an error-per-parameter rate equivalent to observed in the single instances.

\subsection{Experimental setup}  \label{sec:bosonnew_b_experiment}

All the experiments we will discuss consist, in essence, on the observation of one, two or three photons interfering in a multimode integrated chip followed by a detection of coincidence outcomes at the output. The general setup is exemplified in \fig{setup} for the case of a 5-port interferometer. Let us now discuss separately aspects of the photon preparation and measurements stages.

\begin{figure}
\capstart
\centering
\includegraphics[width=0.7\textwidth]{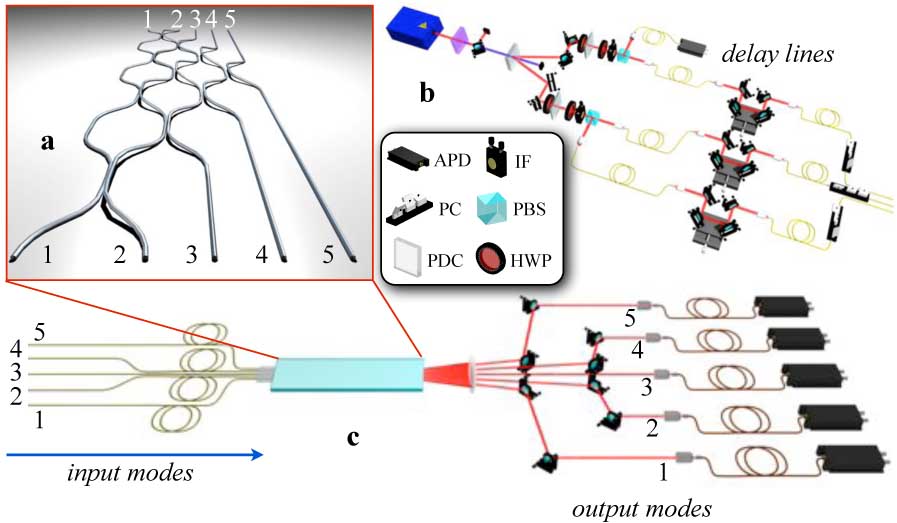}  
\caption[Experimental setup for integrated interferometry.]{General experimental set-up for most of the reported experiments, exemplified with a 5-mode chip. (a) Schematic representation of a 5-mode chip, fabricated by the procedure described in \sec{bosonnew_b_integr}. (b) One-photon, two-photon and three-photon states, generated by parametric downconversion, are injected into the interferometer. Spatial delay lines are used to synchronize the three photons. Legend: APD, avalanche photodiode; IF, interferential filter; PC, polarization controller; PBS, polarizing beamsplitter; PDC, parametric downconversion; HWP, half wave plate. (c) Single-, two- and threefold coincidence detection is performed at the output ports of the chip to reconstruct the probability of obtaining a given output state.}	
\label{fig:setup}
\end{figure}

The single photons used in the experiment are generated at 785nm by the second order process of a parametric downconversion source \cite{Kwiat1995,Crespi2013b}. Four photons are generated, two from each parametric downconversion event, and are routed to the experiment. One is coupled into a single-mode fiber and detected by a single-photon counting detector, and acts as a trigger for the coincidence events. The other three photons are coupled into single-mode fibers, and propagate through different delay lines before being coupled into the chip. The delay lines control the time overlap between the photons, and thus allow for a continuous control over their relative distinguishability. Finally, the output of the chip is coupled into multimode fibers and detected using single-photon avalanche photodiodes. The three-photon outcomes are signaled by a fourfold coincidence between the trigger and three of the detectors.

\subsubsection{Photon distinguishability} \label{sec:sources}

The photon distinguishability can be controlled by the use of delay lines, as depicted in \sfig{setup}{b}. Ideally, this would allow for a continuous variation between the regimes where the photons are perfectly distinguishable (which we call the classical regime) to that where they are perfectly identical (the quantum regime). However, the photons also present an inherent level of partial distinguishability, due to a nonperfect overlap between the spectral functions of photons generated in different downconversion events. This can be modeled by describing the input state as a mixture of the two regimes. For example, an input state of the type $\ket{10101}$, used in the experiments with the 5-mode chip, is actually described as follows:
\begin{equation}
\rho = r \ket{10101}\bra{10101} + (1-r) \ket{1_a \, 0 \, 1_b \, 0 \, 1_a}\bra{1_a \, 0 \, 1_b \, 0 \, 1_a}.
\end{equation}
Here $r$ encodes the partial distinguishability factor, which is proportional to the overlap between the the spectral functions of the three photons, while indexes $a$ and $b$ label the fact that photons input on modes 1 and 5 are generated by the same source, whereas the photon input on mode 3 is generated by a different source, and is twin to the trigger. The factor $r$ can be written as $r=p^2$, where $p$ is the indistinguishability factor of two photons belonging to different pairs. Recall, as mentioned in \sec{HOMeffect}, that two perfectly identical photons impinging on a 50:50 beam splitter always exit  together, in the same mode, which is the Hong-Ou-Mandel (HOM) effect. Thus, the fraction of coincidence outcomes observed in a HOM experiment when photons are \emph{not} perfectly distinguishable gives a measure of their partial distinguishability. By this method, it was possible to obtain $p=0.63 \pm 0.03$.

It is also worthy of note that the parametric downconversion has a nonzero probability of actually producing a six-photon event. The final data reported in \sec{bosonnew_c} takes into account both this (implicitly) and the partial distinguishability (explicitly).

\subsubsection{Output measurements} \label{sec:measurements}

All the experiments involved mainly three types of output detection: single-photon, two-photon visibilities and three-photon coincidences.

For a single photon input on mode $i$ of an $m$-mode interferometer $U$, the probability that it will exit in output mode $L$\footnote{Throughout this section we use lower case letters to represent inputs and upper case to represent outputs.} is simply determined by the corresponding matrix element of $U$ as
\begin{equation}
P^1_{i;L} = |U_{Li}|^2.
\end{equation}
This means that, conversely, the absolute values of $U$ can be determined using only single-photon measurements. By inputting single photons in the chip it was possible to obtain $P^1_{i;L}$ for all combinations of $i$ and $L$, and these values were used: (i) as a first test of the fabrication quality; (ii) in conjunction with two-photon data for an algorithmic reconstruction \cite{Laing2012} of $U$, which we will explain in \sec{bosonnew_b_tomography}; and (iii) to compute the expected output probabilities for two and three photons in the \emph{classical} regime, using \eq{MUclassical}.

For a two-photon input on modes $(i,j)$ and output on modes $(L,M)$, detection consisted on performing HOM interferences to obtain the corresponding visibilities. The visibility is defined as 
\begin{equation} \label{eq:visibs}
V_{i,j;L,M}=\frac{P^{2,c}_{i,j;L,M}-P^{2,q}_{i,j;L,M}}{P^{2,c}_{i,j;L,M}}
\end{equation}
where $P^{2,c}_{i,j;L,M}$ and $P^{2,q}_{i,j;L,M}$ correspond to the classical and quantum two-photon probabilities, respectively. The visibility can be directly obtained from a HOM curve as the one depicted in \fig{HOMvisib}. Note that, while the ``standard'' HOM curve described in \sec{HOMeffect} consists of observing two photons in a 50:50 beam splitters, a general HOM curve is actually a transition between the probabilities associated with the two regimes, and can be a dip or a peak depending on which probability is larger. Also note that the visibilities are independent of photon losses in the device, assuming that losses affect both classical and quantum regimes equally.

\begin{figure}
\capstart
\centering
\includegraphics[width=0.7\textwidth]{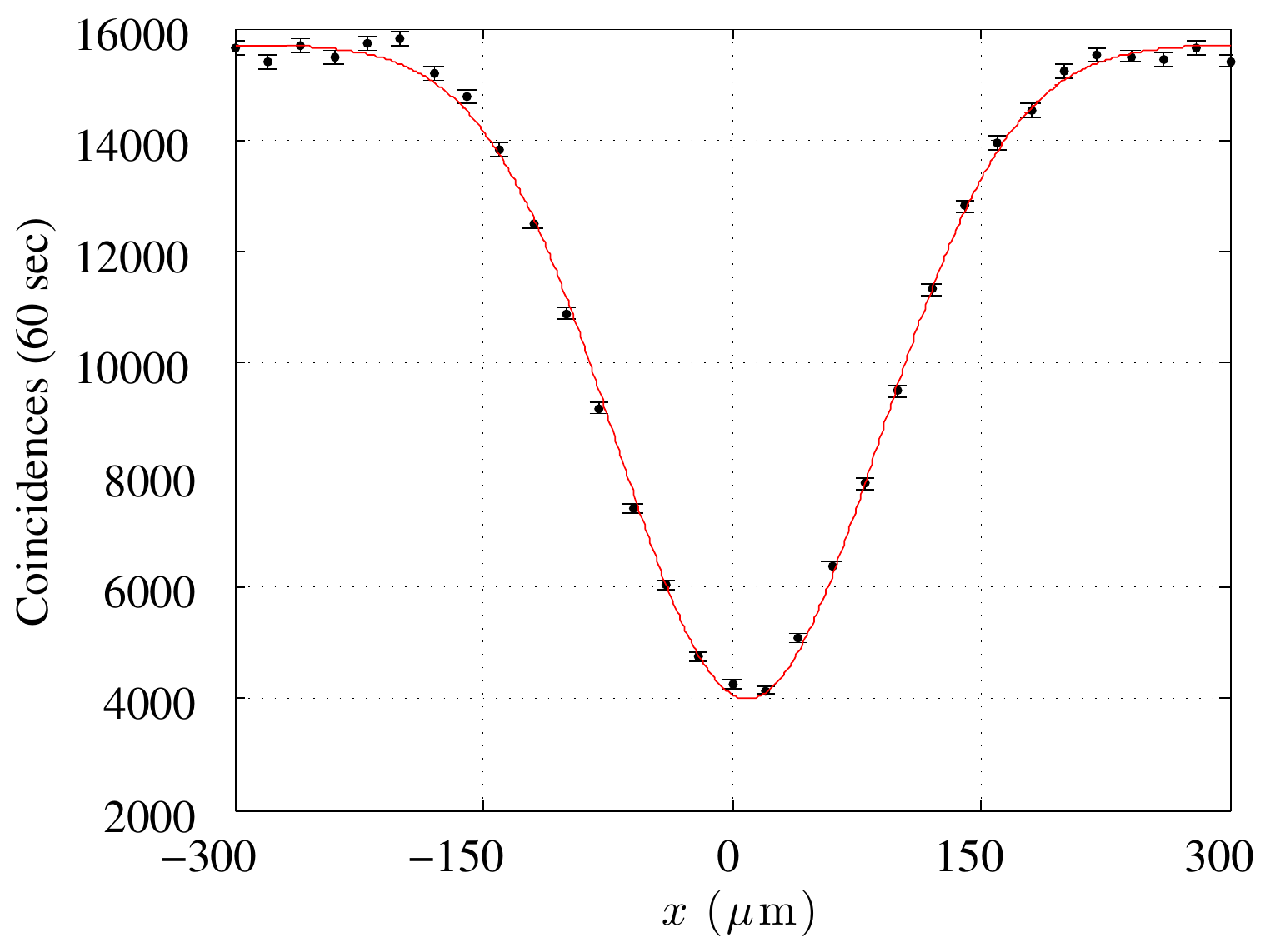}  
\caption[HOM curve for measurement of two-photon visibilities.]{Typical HOM curve for experimental measurement of two-photon visibilities. Note that the visibility corresponds only to the \emph{relative} height of the dip (or peak).}	
\label{fig:HOMvisib}
\end{figure}

The two-photon probabilities $P^{2,q}_{i,j;L,M}$ can be obtained from the visibilities by inverting \eq{visibs} and computing $P^{2,c}_{i,j;L,M}$ from the corresponding single-photon probabilities. Interestingly, the tomography algorithm described in \sec{bosonnew_b_tomography} relies directly on the visibilities, rather than on the probabilities themselves, and thus has the convenient feature of reconstructing a unitary matrix independently of photon losses. As we will see, this algorithm requires the measurement of only a subset of all visibilities, however the measurement of the complete set helps to avoid those values of $V_{i,j;L,M}$ with larger experimental errors. For a 5-mode chip, the set of all visibilities requires the measurement of 100 HOM curves, whereas for a 7-mode chip this number is raised to 441---measuring these many HOM curves is no simple task, which is why these full sets of measurements were only performed on chips of size 5 and 7. For larger chips, either an indirect certification of the quality of the device was performed by using the validation algorithm of \sec{bosonreview_c_certif}, as well as standard statistical tests, to compare ideal and experimental three-photon distributions, or else other properties were being investigated that did not heavily depend on the actual expression of $U$.

Both single- and two-photon probabilities were obtained using twin photons from a single downconversion event, but for the three-photon measurements two sources and a trigger were required, as described in the previous section. The corresponding probabilities for input $\{i,j,k\}$ and output $\{L,M,N\}$ are denoted $P^{3,q}_{i,j,k;L,M,N}$, and in a sense are the crux of the experiments. The comparison is made between the corresponding probabilities corrected for the $r$ distinguishability factor, which are denoted as $P^{3,r}_{i,j,k;L,M,N}$ and which can be computed with the aid of single- and two-photon data, and the experimental probabilities $P^{3,e}_{i,j,k;L,M,N}$. This makes up a large part of the results reported in \sec{bosonnew_c}.

Finally, another important quantity that was measured in one of the experiments was the bunching fraction, corresponding to the probability $p_q$ that \emph{at least} two photons will exit in the same output mode. This is the figure of merit of the bosonic birthday paradox, described in \sec{bosonreview_c_BBP}. To measure this quantity, a setup was arranged with the 3 input photons and 3-fold coincidence measurements in which both classical and quantum regimes were alternated, each running for the same time interval, and the overall number of detected events per interval was recorded. Since the bunching fraction is the complement to the total fraction of coincidence outcomes, this gives a direct estimate for $(1-p_q)/(1-p_c)$, where $p_c$ is the corresponding classical bunching fraction. Since $p_c$ can be computed directly from the single-photon measurements, this allows a good estimate for $p_q$, without the much-less-feasible requirement of identifying multiphoton outcomes. 

\subsection{Random ensembles of unitary matrices}  \label{sec:bosonnew_b_ensembles}

Given the ability to implement an arbitrary $m$-mode integrated interferometer, the question arises of what should be its specifications. As discussed in \sec{bosonreview_c_interf}, the choice that most faithfully reproduces the setting where BosonSampling was defined is a uniformly random unitary matrix. An $m \times m$ Haar-random unitary $U$ can be obtained by the following procedure:

\begin{itemize}
\item[(i)] Sample an $m \times m$ matrix $Z$ from the standard complex Gaussian ensemble.
\item[(ii)] Apply some QR decomposition to $Z$. In other words, obtain a decomposition $Z=QR$ for some orthogonal matrix $Q$ and upper-triangular matrix $R$. Notice that this decomposition is not unique.
\item[(iii)] Define the matrix $R' = \textrm{diag}\left(\frac{R_{11}}{|R_{11}|}, \frac{R_{22}}{|R_{22}|} \ldots \frac{R_{nn}}{|R_{nn}|}\right)$. This fixes the non-uniqueness of the QR decomposition.
\item[(iv)] The matrix $U=QR'$ is randomly distributed according to the Haar measure.
\end{itemize}

For a discussion on why this method works, see \cite{Mezzadri2007}. Alternatively, the Haar ensemble can also be simulated by a choice of decomposition such as \sfig{Reck}{a}, together with a suitable random sampling of each individual parameter. However, the distribution satisfied by each reflectivity depends on the beam splitters' position within the circuit, and the resulting procedure is considerably more cumbersome than simply using the method above and then the decomposition of \cite{Reck1994}.

Although it is closer in spirit to the original BosonSampling proposal, the uniform ensemble has some drawbacks, especially in terms of losses at the S-bends. At each S-bend, there is a non-negligible probability that a photon will escape tangentially (these can be measured experimentally, and are typically around $7\%$ per S-bend), which has two main consequences for the experiment. The first is that these losses induce an overall attenuation of the signal, related to the smallest number of S-bends a photon must transverse. Although simulations suggest that lossy BosonSampling devices may retain their hard-to-simulate properties \cite{Rohde2012}, this nonetheless presents one of the major scalability obstacles to these techniques since the overall signal actually decays exponentially with the circuit depth\footnote{This was the primary motivation for the constant-depth result proved at the beginning of the chapter.}, and we will return to this subject at the end-of-chapter discussion. Secondly, the asymmetry in the layout of \sfig{Reck}{b} causes losses to be biased relative to different paths within the chip. More specifically, if we factor out the overall losses, each output mode is subject to a different loss factor. To exemplify, consider the 5-mode circuit of \sfig{Reck}{b}. A photon will transverse two more S-bends if it leaves in mode 3 than if it leaves in mode 5, independently of its input. This can be corrected by multiplying the count of each output mode by a fixed factor, but is nonetheless a source of extra biased error in the experimental data.

\begin{figure}[t]
\capstart
\centering
\includegraphics[width=0.7\textwidth]{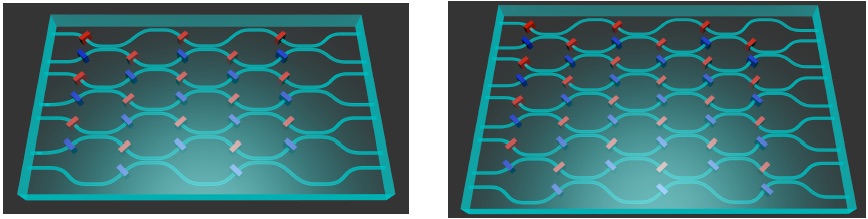}  
\caption[Random phases ensemble.]{Schematic representation of the random phases ensemble. The first example has $m=7$ modes and $L=5$ layers of 50:50 beam splitters, and the second has $m=9$ and $L=6$. Red and blue rectangles represent phase shifts sampled uniformly from the $\left [ 0, \pi \right ]$ interval.}	
\label{fig:walkinterf}
\end{figure}

With these two concerns in mind, we numerically investigated an alternative ensemble, described by $m$-mode chips consisting of $L$ layers of alternated 50:50 beam splitters\footnote{As recently shown in \cite{Bouland2013}, any unitary linear transformation can be efficiently approximated by a sequence of beam splitters of any fixed reflectivity. Thus, contrary to what one might expect, this restriction on the beam splitters does not make the ensemble trivial.} between neighboring modes, interspersed with random phase shifts sampled uniformly from the $\left [ 0, \pi \right ]$ interval. Examples of these interferometers for $m=7$ and $m=9$, with $L=5$ and $L=6$ respectively, are depicted in \fig{walkinterf}. Henceforth, we denominate this the random phases ensemble.

The random phases ensemble was defined for practical, rather than mathematical, reasons. Its properties are not as well-studied as the uniform ensemble, and a conjecture that their permanents are $\#$P-hard does not seem as natural as \conj{BSishard}. Nonetheless, chips built according to this ensemble have the nice feature of being symmetric, in order to avoid biased losses. Also, if we restrict ourselves to using only the central 3 input modes and if $L \geq (m+3)/2$, the chip is ``fully connected'', in the sense that any of the photons has a nonzero probability of exiting at any output mode. This is a more feasible depth range for those chips with $m \geq 7$, if compared with arbitrary unitaries, where generally $L=m-1$. Since the first experiment that was performed, with $m=5$, demonstrated a good control over the technique for fabricating arbitrary unitaries \cite{Crespi2013b}, in subsequent experiments \cite{Spagnolo2013b, Spagnolo2013c} we decided to focus on other technical aspects such as certification of the device, for which the random phase ensemble suffices. Thus, this ensemble is well-motivated for its physical, if not complexity-theoretical, properties.

\begin{figure}[p!]
\capstart
\centering
\subfloat[]{\centering \includegraphics[width=1\textwidth]{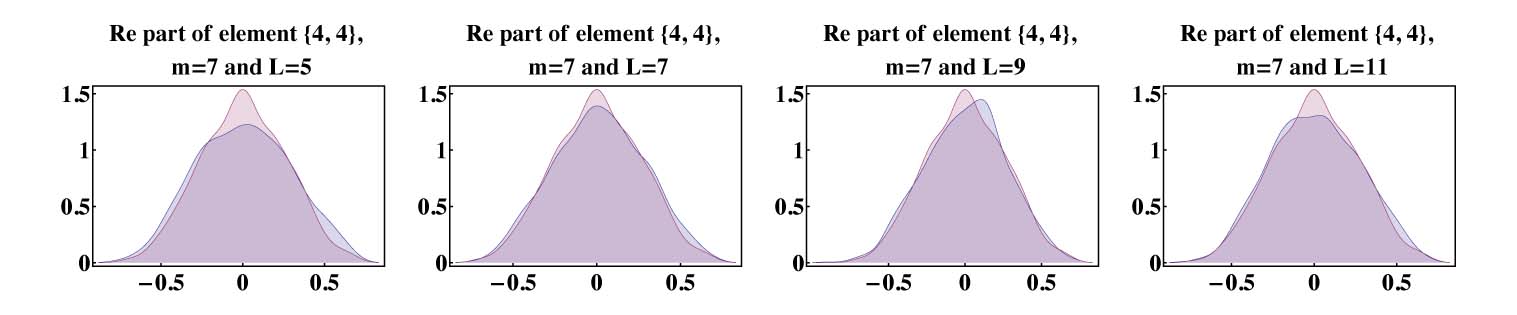}} \\
\subfloat[]{\centering \includegraphics[width=1\textwidth]{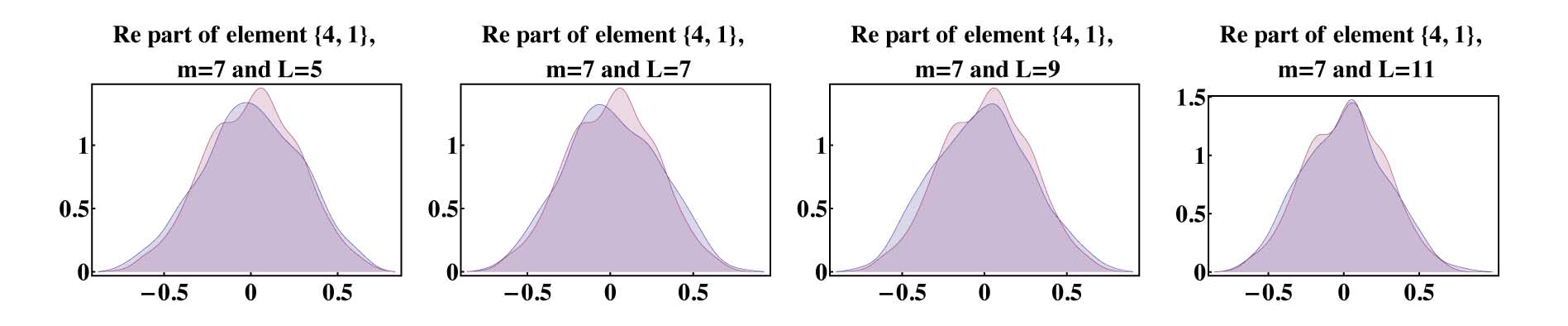}} \\
\subfloat[]{\centering \includegraphics[width=1\textwidth]{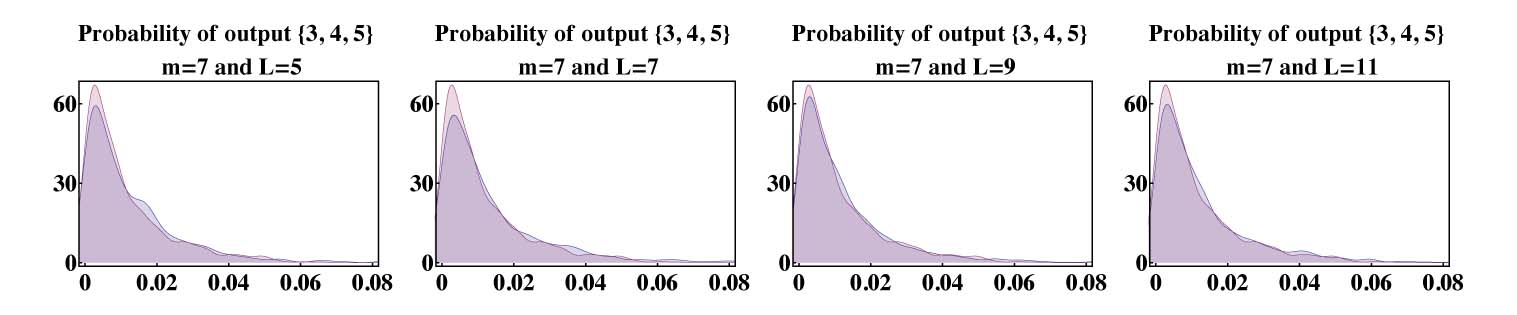}} \\
\subfloat[]{\centering \includegraphics[width=1\textwidth]{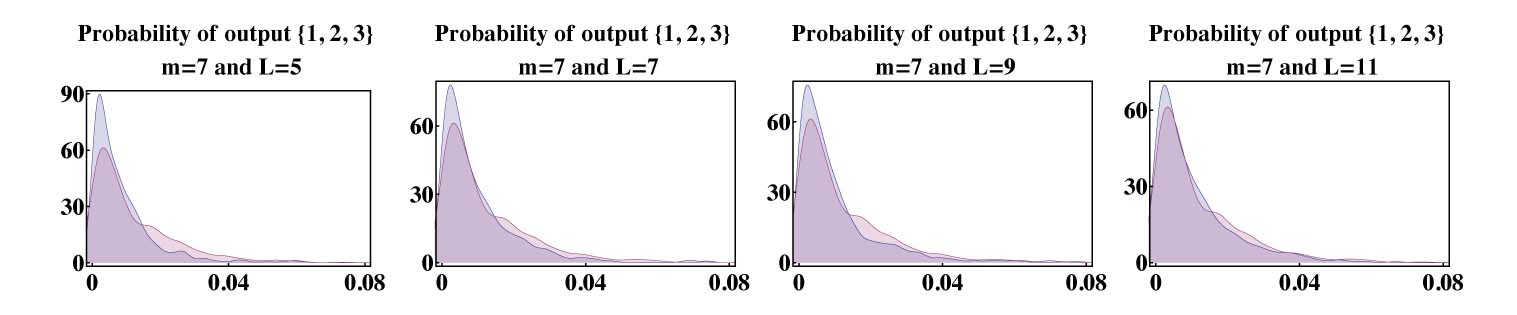}} \\
\subfloat[]{\centering \includegraphics[width=1\textwidth]{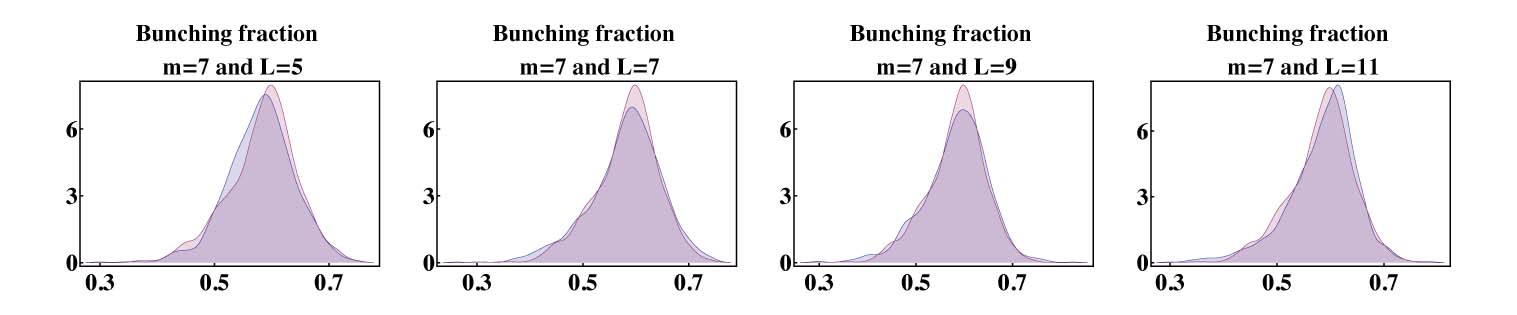}}
\caption[Simulation of random phase ensemble for $m=7$]{Simulations for the random phase (blue) and Haar ensembles (magenta), for $m=7$ and several values of $L$. (a) Real part of the amplitude connecting the central input to the central output. (b) Same as (a) but for an endpoint output. (c) Probability connecting three photons input on central modes to the central output modes. (d) Same as (c), but for the upper three output modes. (e) Bunching fraction.}
\label{fig:simuls7}
\end{figure} 

\begin{figure}[p!]
\capstart
\centering
\subfloat[]{\centering \includegraphics[width=1\textwidth]{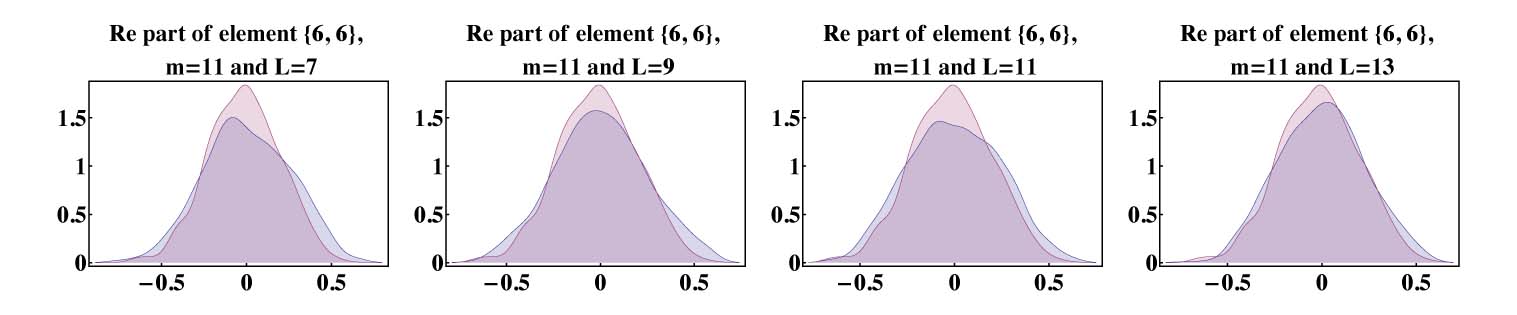}} \\
\subfloat[]{\centering \includegraphics[width=1\textwidth]{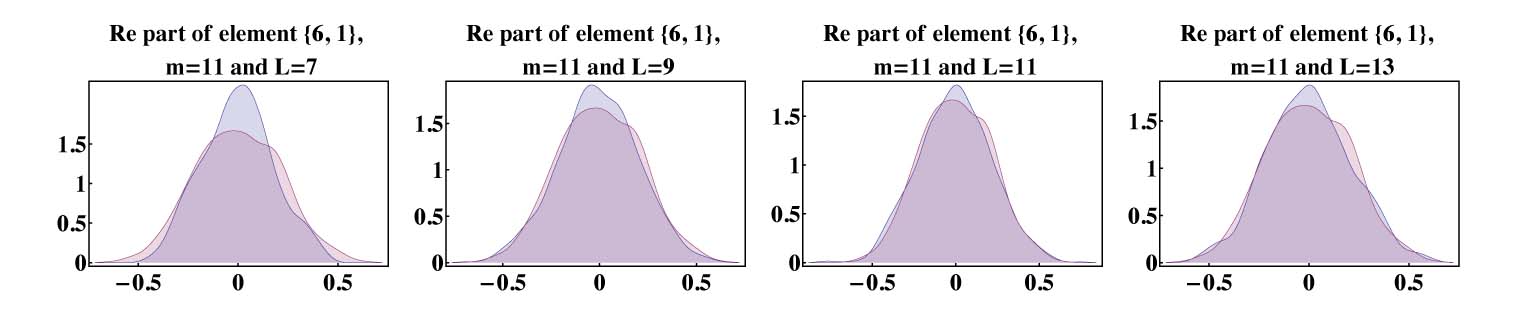}} \\
\subfloat[]{\centering \includegraphics[width=1\textwidth]{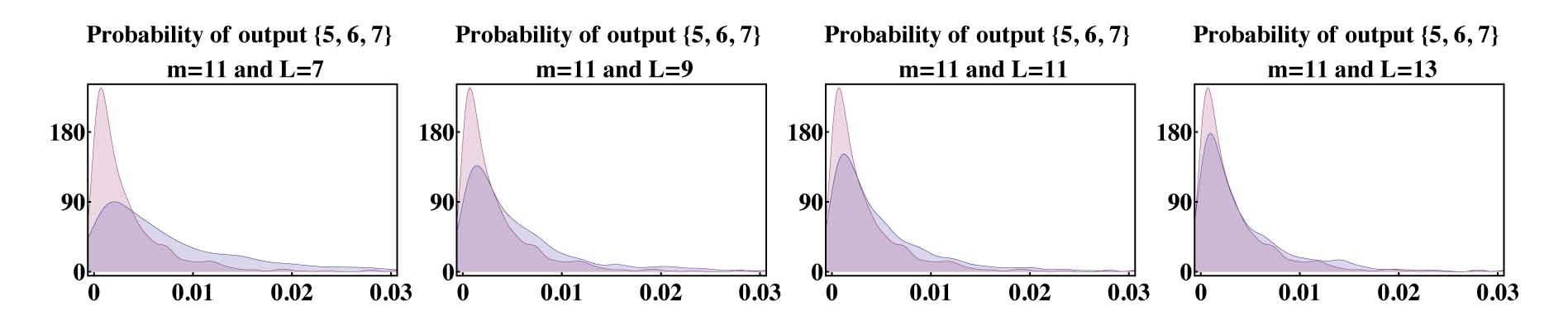}} \\
\subfloat[]{\centering \includegraphics[width=1\textwidth]{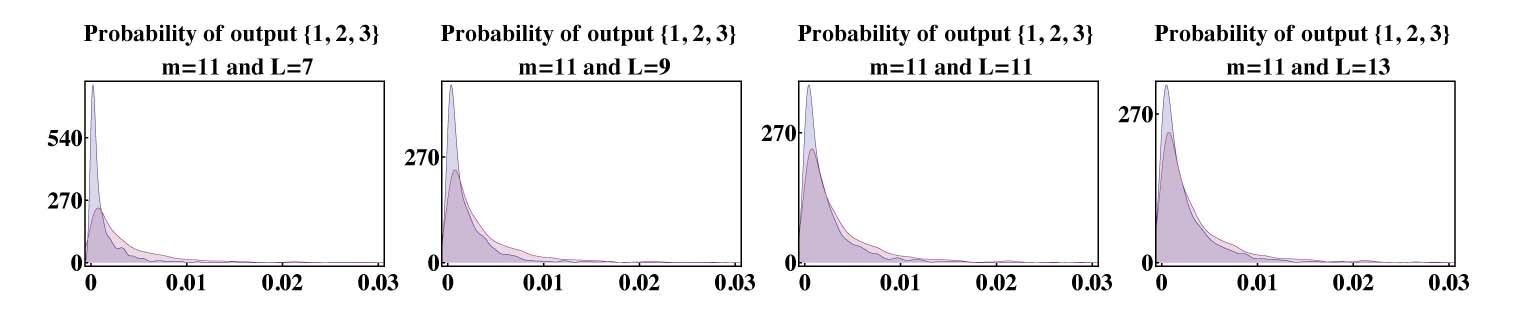}} \\
\subfloat[]{\centering \includegraphics[width=1\textwidth]{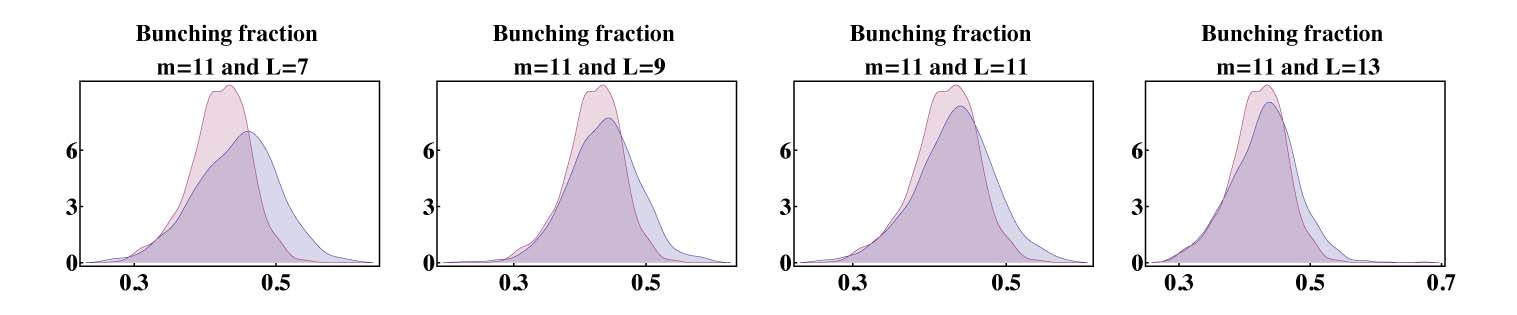}}
\caption[Simulation of random phase ensemble for $m=11$]{(a)-(e) Same as \fig{simuls7}, but for interferometers with $m=11$.}
\label{fig:simuls11}
\end{figure} 

In \fig{simuls7} and \fig{simuls11} we show the distributions of several figures of merit of interest over 10000 samples from the random phase and from the uniform ensembles, for $m=7$ and $m=11$ and several values of $L$. We plotted the real part of matrix elements $\{\frac{m+1}{2}, \frac{m+1}{2}\}$ and $\{\frac{m+1}{2}, 1\}$ of $U$. These matrix elements correspond to the transition amplitudes for photons entering in the central mode to exit in the central and endpoints modes, respectively. The plots for the imaginary parts of these elements have a similar appearance and are not reported here. We also plotted the permanents of $3 \times 3$ matrices associated with the transition of photons entering on the three central inputs and exiting either in the three central or three endpoint modes. We chose these two sets of outputs since they display the extremal influence of the chip's constraints on the transition probabilities---if the depth is smaller than $(m+3)/2$ some photons simply cannot transition from a central mode to one of the endpoint modes, and the corresponding matrix elements are zero, which is why one might expect these matrix elements to be very biased at $L$ close to $(m+3)/2$. Finally, we also simulated the behavior of the bunching fraction (cf.\ \sec{bosonreview_c_BBP}). These simulations were done on Mathematica$^{\copyright}$ software, and the corresponding files can be found in \appdx{app1}. 

All simulations suggest that relevant figures of merit approach the corresponding distributions for the Haar ensemble surprisingly fast as $L$ becomes larger than $(m+3)/2$. This provides evidence that this ensemble, despite its limitations, still has ``enough randomness'' to motivate its use for small-scale implementations of BosonSampling in order to benchmark the several experimental techniques involved.

\subsection{Device tomography}  \label{sec:bosonnew_b_tomography}

A final important aspect of data analysis that warrants a separate discussion, due to its ubiquity in analysis of the experiments, is \emph{chip tomography}. The implementation of a previously-known tomography algorithm \cite{Laing2012} on Mathematica$^{\copyright}$ software, as well as several refinements of this algorithm to obtain better agreement with experimental data, were one of my main contributions to these experimental papers. 

The starting point is an algorithm due to Laing and O'Brien \cite{Laing2012}, which we will describe shortly, that allows the reconstruction of the chip's unitary $U$ from the set of single-photon probabilities and two-photon visibilities (cf.\ \sec{measurements}). Rigorously speaking, this does not consist of complete process tomography (see e.g.\ \cite{Mohseni2008} and references therein). Recall that the Fock space of an $n$-photon, $m$-mode system has dimension $\binom {m+n-1} {n}$, and that $U$ induces an unitary $U_F$ in this space, in the manner described in \sec{intro_secondquant}. However, the real-world matrix $U_F$ is not completely determined by $U$, but rather involves several other spurious nonlinear and non-unitary processes. In this sense, complete process tomography would consist on reconstructing $U_F$, and this would require an unfeasible amount of experimental effort and technology. Instead, the tomography we will describe consists simply of reconstructing $U$, under the idealized assumption that it completely determines $U_F$. 

Under this assumption that all processes are linear, one may wonder why is it necessary to use both single- and two-photon data. In principle, one could perform the reconstruction by describing a single photon in an $m$-mode chip as an $m$-mode system, and performing some suitable well-known single-qudit process tomography. The reason is that single-qudit tomography would require inputting the photons in several well-controlled superposition states, corresponding, for example, to some set of $m-1$ mutually unbiased bases \cite{Wootters1989}, which in turn requires a very precise control on the relative phases between different input modes. However this would be experimentally challenging, since random input phases are introduced e.g.\ by the coupling between the single-mode fibers and the chip. The algorithm we will describe, in contrast, only depends on the transition probabilities between Fock states, at the cost of requiring some two-photon measurements, but which are nonetheless more feasible than  preparing the arbitrary single-photon superpositions. As we will see, a drawback of this algorithm, foreshadowed by this discussion, is that we will only obtain $U$ up to a round of phase shifters at the beginning and at the end. This disadvantage, however, is compensated by the fact that BosonSampling itself depends on matrix permanents only via such transition probabilities, and thus is also insensitive to input and output phases. In other words: we cannot experimentally determine what phases are induced at the input and output of the chip due to the coupling with the fibers and other factors, but this is irrelevant since they have no observable effect on the subsequent experiments.

We should point out that there are other methods for reconstructing $U$. Most notably, it is possible to reconstruct $U$ using only coherent states, as shown in \cite{Rahimi-Keshari2013}. Since in our experiments the measurement of single- and two-photon data was required to take into account the partial distinguishability of photon sources (cf.\ \sec{sources} and \sec{measurements}), we opted for the simpler choice of using this data to also perform the tomography. It has been claimed that determining the expected theoretical behavior of a 3-photon experiment using a reconstructed matrix obtained from single- and two-photon experiments consist of circular reasoning. To illustrate this claim, suppose we perform an experiment with $n=50$ photons. The experiment would verify that the corresponding laws of Physics hold for $n=50$, but only assuming that they already hold for $n=2$ since we would be using the reconstructed $U$ as the theoretical prediction. However, I do not agree that this reasoning is circular: the HOM curves (such as depicted in \fig{HOMvisib}) have been perfected through the last three decades, and the laws of Physics have been very well verified for $n=2$, and it makes perfect sense to consider this data trustworthy enough to base our theoretical predictions for $n=50$.

We will now give a brief description of the Laing-O'Brien algorithm, as first described in \cite{Laing2012}, followed by a discussion on its features and shortcomings and how it can be refined.

\subsubsection{The Laing-O'Brien algorithm}

Let each matrix element of $U$ be written as $\tau_{k;J} e^{i \alpha_{k;J}}$, where lower case letters label columns and upper case letters label rows, to maintain consistency with our previously adopted notation. Let also $P^{1,e}_{k;J}$ correspond to the experimental single-photon transition probability connecting input mode $k$ to output mode $J$, and $V^{e}_{k,h;J,G}$ be the two-photon visibility connecting input pair $\{k,h\}$ to output pair $\{J,G\}$. Let us also define the following auxiliary quantities:
\begin{align}
x_{k,h;J,G}&:=\sqrt{\frac{P^{1,e}_{k;J} P^{1,e}_{h;G}}{P^{1,e}_{h;J}P^{1,e}_{k;G}}}, \\
y_{k,h;J,G}&:=x_{k,h;J,G}+x^{-1}_{k,h;J,G}.
\end{align}
Note that $x_{k,h;J,G}$ and $y_{k,h;J,G}$ are completely determined by experimental data. It can be shown, then, that the matrix parameters $\tau_{i;J}$ and $\alpha_{i;J}$ are related to experimentally accessible quantities ($x_{k,h;J,G}$, $y_{k,h;J,G}$ and $V^{e}_{k,h;J,G}$) as follows \cite{Laing2012}:
\begin{subequations} \begin{align}
x_{k,h;J,G}&=\frac{\tau_{k;J} \tau_{h;G}}{\tau_{h;J} \tau_{k;G}}, \label{eq:constraintsa} \\
V^{e}_{k,h;J,G} \; y_{k,h;J,G}&=-2 \cos(\alpha_{k;J}-\alpha_{h;J}-\alpha_{k;G}+\alpha_{h;G}). \label{eq:constraintsb}
\end{align} \end{subequations}
The algorithm, as we will describe shortly, consists of inverting these two equations to sequentially obtain every parameter $\tau_{k;J}$ and $\alpha_{k;J}$, for $k$ and $J$ greater than 2, in terms of the experimental values and the parameters of the first row and column (i.e.\ $\tau_{1;J}$, $\tau_{k;1}$, $\alpha_{1;J}$, and $\alpha_{k;1}$), which we will call ``reference values'', and then impose unitarity of $U$ to obtain equations that completely determine these reference values. Note that \eq{constraintsb} does not allow one to completely obtain $\alpha_{h;G}$, since the sign for the argument is unknown. However, if the absolute value $|\alpha_{h;G}|$ is known, and the other three reference angles are completely known, the sign of $\alpha_{h;G}$ can be determined by
\begin{align} \label{eq:signalpha}
\textrm{sgn} [\alpha_{h;G}] = & \text{sgn} [ |\beta_{k,h;J,G} - |\alpha_{k;J} - \alpha_{h;J} - \alpha_{k;G}-|\alpha_{h,G}||| \notag \\
& - |\beta_{k,h;J,G} - |\alpha_{k;J} - \alpha_{h;J} - \alpha_{k;G}+|\alpha_{h,G}||| ],
\end{align}
where $\beta_{k,h;J,G}:=\alpha_{k;J} - \alpha_{h;J} - \alpha_{k;G}+\alpha_{h,G}$.

In order to make the sequential replacement possible, we must make the following assumptions. First, the elements of the first column and row are real (i.e.\ $\alpha_{1;J}=\alpha_{k;1}=0$ for all $J$ and $k$). This can be done without loss of generality, since these can be fixed by a round of phase shifters at the input and output which, as we discussed, the algorithm is insensitive to. Second, the element $\{2,2\}$ has positive imaginary part (i.e.\ $\alpha_{2;2} \geq 0$). This can be done since the photon probabilities are also insensitive to the symmetry $U \rightarrow U^{*}$.

We now start from the initially unknown matrix $M_1$, parameterized as
\begin{equation}
M=
\begin{pmatrix}
\tau_{1;1} & \, \tau_{1;2} & \, \cdots & \, \tau_{1,m} \\[8pt] 
\tau_{2;1} & \, \tau_{2;2} \, e^{ i \alpha_{2;2}} & \, \cdots & \, \tau_{2;m} e^{ i \alpha_{2;m}} \\[8pt]
\vdots & \, \vdots & \, & \, \vdots \\[8pt]
\tau_{m;1} & \, \tau_{m;2} e^{ i \alpha_{m;2}} & \, \cdots & \, \tau_{m;m} e^{ i \alpha_{m;m}}
\end{pmatrix},
\label{eqM1}
\end{equation}
 and proceed to replace the parameters as as follows:
\begin{itemize}
\item[(i)] Fix $\{j,K\}=\{1,1\}$ in \eq{constraintsa} and rewrite every $\tau_{h;G}$ (for $h,G \geq 2$) in terms of $\tau_{1;1}$, $\tau_{1;G}$ and $\tau_{h;1}$ and $x_{1,h;1,G}$.
\item[(ii)] Do the same as (i), but in \eq{constraintsb}. Since the $\alpha_{h;1}=\alpha_{1;G}=0$, this allows directly obtaining the absolute value of $\alpha_{h,G}$ as arccos$\left( \frac{1}{2} V^{e}_{1,h;1,G} \; y_{1,h;1,G} \right)$.
\item[(iii)] Finally, recall that the sign of $\alpha_{2,2}$ is defined to be positive. Thus, by fixing $\{k,J\}=\{1,2\}$ in \eq{signalpha} one can determine the sign of all $\alpha$ in the second column, by fixing $\{k,J\}=\{2,1\}$ one can determine the signs of $\alpha$ in the second row, and fixing $\{k,J\}=\{2,2\}$ determines the signs of the remaining phases.
\end{itemize}
After this replacement, the matrix is completely determined up to the values of $\tau_{k;J}$ on the first row and column. Finally, these values are determined by imposing that the first column and row are normalized, and that they are orthogonal to the remaining rows and columns, respectively\footnote{Note that these do not correspond to all constraints due to unitarity, since orthogonality between remaining columns or rows is not imposed.}. If the experimental data corresponds to a perfect unitary matrix, the resulting matrix obtained after this step is already the correct one. However, experimental data typically does not correspond perfectly to some unitary matrix, both because the real-world process is not really unitary and because even then the measurements themselves would be noisy, and a polar decomposition must be performed to obtain $U$. The polar decomposition consists in writing $M$ as $M = U P$, where $U$ is unitary and $P$ is a positive-semidefinite Hermitian matrix. Intuitively, this decomposition resembles the polar decomposition for complex numbers: $U$ represents a ``direction'' in space, or a rotation, while $P$ represents a ``stretching'' along that direction. It can be shown that $U$ is in fact the closest unitary matrix to the reconstructed $M$, which justifies this step.

The resulting unitary matrix, $U$, is the best obtained by this procedure that approximates the experimental data. Of course, the experimental data cannot be completely consistent with some particular unitary in any realistic setting, and the authors of \cite{Laing2012} also report a numerical investigation on the effect of noisy data on the reconstruction. More specifically, sets of noisy data were simulated starting from the distribution expected from some chosen unitary, and the gate fidelity [$=\frac{1}{m} \text{Tr}(U V^{\dagger})$] between the original and the reconstructed matrices was plotted as function of the noise level and the interferometer size. As they report, a noise of up to $5\%$ in a $4$-mode experiment causes the reconstructed unitary to have an average $95\%$ fidelity with the original one, and this value drops to $85\%$ in a 20-mode experiment with noise rate of $0.25\%$. 

\subsubsection{Refining the Laing-O'Brien algorithm}

We now describe two refinements made to the Laing-O'Brien algorithm that provided us with a reconstructed unitary displaying considerably better agreement with experimental data.

To quantify the quality of the reconstructed matrices, we use two figures of merit for the distance between reconstructed and experimental data points: (i) the total variation distance between the distributions (cf.\ definition in \sec{introsimul}), both for single- and two-photon measurements, and each averaged over all possible choices of input; and (ii) the $\chi^2$ between the complete data sets, defined as 
\begin{equation}
\chi^2 = \sum_{i} \left( \frac{x_i^r - x_i^e}{\sigma x_i} \right)^2,
\end{equation}
where $x_i^r$ and $x_i^e$ are the reconstructed and experimental values, respectively, $\sigma x_i$ is the corresponding error bar, and the sum runs over the complete data set. In these experiments, we cannot ascribe the usual significance to the value of the $\chi^2$. Typically, a model that gives a good description for a certain data set should present a $\chi^2$ of roughly 1 per data point---this represents the fact that each data point is roughly at most one deviation from the expected theoretical value. In our data set, with the best reconstructed unitaries, $\chi^2$ ranges around values of 30 per data point. However, the usual interpretation of this value is only meaningful when the data describes independent variables, and our data most certainly does not. There are correlations both on the fact that the data corresponds to probability distributions (and thus are normalized), but also that they come from a same unitary matrix $U$, and there is a great deal of correlations across data from different inputs. As an example, note that the absolute values of $U$, which determine single-photon probabilities, must be normalized along both rows \textbf{and} columns. The correlations between the various two-photon data points are even more intricate. In simpler terms: a variation of one standard deviation in one data point must induce a corresponding variation on several of the other points just to maintain normalization, and this causes the $\chi^2$ to blow up. Nonetheless, it still provides a measure of distance between experimental and theoretical data, moreover one that takes into account our confidence on the experimental values via their estimated errors.

We also compute the gate fidelity $F$ between the ideal ($U^t$) and reconstructed ($U^r$) unitaries, given by
\begin{equation}
F = \frac{1}{m} |\textrm{Tr} (U^{t \, \dagger} U^r)|.
\end{equation}
Given that the algorithm is insensitive to a round of arbitrary input and output phase shifters and to the symmetry $U \rightarrow U^{*}$, but $F$ is not, we fix these symmetries so as to maximize the value of $F$. Of course, this fidelity measures the distance from the reconstructed unitary to the theoretical one, not to the experimental data, and as such it is not a measure of quality of the algorithm itself, but rather of the quality of the fabrication. A matrix with smaller value of $\chi^2$ or total variation distance, but greater $F$, has better agreement with experimental data, and if it does not agree with the original theoretical unitary it is most likely because the fabrication procedure did not implement the desired parameters perfectly.

The first way we refine the Laing-O'Brien algorithm is simply by varying the reference row and column used in the original algorithm. Notice that, for an $m \times m$ unitary, the complete set of experimental data consists of $m^2$ single-photon and $\frac{1}{4} m^2(m-1)^2$ two-photon data points. However, the Laing-O'Brien algorithm only uses $(m-1)^2$ single-photon and $(m-1)^2$ two-photon data points\footnote{Plus an additional $m(m+4)$ points to fix the signs of the phases. We do not include these into the count since each only determines a discrete information, i.e.\ a sign, so the final reconstructed matrix typically only depends very mildly on their values.}, a consequence of the choice of the first row and column as the references. In the original paper the authors raise the question of whether the algorithm can be adapted to produce a matrix that uses all the available data set for a more robust result, or whether it is better to focus experimental efforts on obtaining the highest possible precision on a minimal subset of data points.  In two of the experiments (for 5- and 7-mode chips) we will report in the next section, the complete data set was measured, and this allowed us to use different subsets in the reconstruction algorithm. More specifically, what we did was obtain $m^2$ different reconstructed matrices, each obtained from a different choice of reference row and column. This was achieved by permuting the data set labels prior to implementing the Laing-O'Brien algorithm as is.

By running the reconstruction over all possible permutations for the available experimental data, we observed that the quantities of interest vary greatly. In \tabl{permutation5} and \tabl{permutation7} we report the best and worst permutation from the point of view of each of these figures of merit. Note that the records for each figure of merit are not held all by the same choice of permutation, so one may choose different permutations depending on what they wish to optimize. This argument does not answer the authors' open question, as it does not really use the redundant data for increased robustness, but nonetheless it allows us to sidestep certain pathological data sets that provide below-average reconstructed matrices.

\begin{table}[h]\centering
\begin{tabular}{|c||c|c|c|c|}
\hline 
$\{k,J\}$ & $F$ & $\chi^2$ & Single-photon TVD & Two-photon TVD \\
\hline
$\{2,4\}$ & 0.945 & \textcolor{blue}{\textbf{4952}} & 0.0714 & 0.0744 \\
\hline
$\{0,4\}$ & \textcolor{blue}{\textbf{0.966}} & 6758 & 0.087 & 0.082 \\
\hline
$\{3,4\}$ & 0.953 & 5701 & \textcolor{blue}{\textbf{0.0711}} & \textcolor{blue}{\textbf{0.069}} \\
\hline
$\{3,2\}$ & 0.833 & \textcolor{red}{\textbf{77510}} & 0.22 & 0.18 \\
\hline
$\{0,3\}$ & \textcolor{red}{\textbf{0.804}} & 58546 & 0.224 & 0.201 \\
\hline
$\{1,1\}$ & 0.811 & 74894 & \textcolor{red}{\textbf{0.238}} & \textcolor{red}{\textbf{0.228}} \\
\hline
Refined from $\{2,4\}$ & 0.953 & 3556 & 0.066 & 0.067 \\
\hline
\end{tabular}
\caption{Illustrative values of $F$, $\chi^2$, and total variation distance relative to single and two photon data, for several matrices obtained by permutations of the Laing-O'Brien algorithm, for the 5-mode experiment. The first column describes the permutation: specifically, a value of $\{k,J\}$ means that the reference row and column were $J+1$ and $k+1$ respectively. The $\chi^2$ are reported for the complete set of 125 data points. Values in blue and red are the best values of the corresponding figure of merit among the 25 possible permutations. The last row represents the best value output by a 15 min run of the numerical search on a standard 2.20 GHz quadcore laptop, with 8Gb RAM, running the algorithm shown in \appdx{app2} on Mathematica$^{\copyright}$ software.}
\label{tab:permutation5}
\end{table}

\begin{table}[h]\centering
\begin{tabular}{|c||c|c|c|c|}
\hline 
$\{k,J\}$ & $F$ & $\chi^2$ & Single-photon TVD & Two-photon TVD \\
\hline
$\{3,4\}$ & 0.96 & \textcolor{blue}{\textbf{$3.6 \times 10^6$}} & 0.09 & 0.124 \\
\hline
$\{2,2\}$ & \textcolor{blue}{\textbf{0.974}} & $4.9 \times 10^6$ & 0.080 & 0.102 \\
\hline
$\{2,3\}$ & 0.967 & $7.4 \times 10^6$ & \textcolor{blue}{\textbf{0.073}} & \textcolor{blue}{\textbf{0.095}} \\
\hline
$\{5,3\}$ & 0.789 & \textcolor{red}{\textbf{$6 \times 10^8$}} & 0.204 & 0.256 \\
\hline
$\{5,6\}$ & \textcolor{red}{\textbf{0.788}} & $6.9 \times 10^6$ & 0.232 & 0.257 \\
\hline
$\{0,5\}$ & 0.844 & $4.1 \times 10^7$ & \textcolor{red}{\textbf{0.221}} & \textcolor{red}{\textbf{0.276}} \\
\hline
Refined from $\{3,4\}$ & 0.975 & $2.38 \times 10^4$ & 0.074 & 0.086 \\
\hline
\end{tabular}
\caption{Similar values as for \tabl{permutation5}, for the 7-mode chip. The $\chi^2$ are reported for the complete set of 490 data points, and there are 49 possible permutations. The last row represents the best value output by a 8 hr run of the numerical search.}
\label{tab:permutation7}
\end{table}

The second refinement was a numerical procedure to search for better fits, in the region close to the best matrix obtained from the standard Laing-O'Brien algorithm. Among the several techniques attempted, including genetic algorithms, gradient searches, brute force searches, the method that produced the best and fastest results was the following:

\begin{itemize}
\item[(i)] Start from $U_0$, the matrix with best $\chi^2$ from the Laing-O'Brien algorithm.
\item[(ii)] Compute the expected set of all data (single- and two-photon), $\{x_i\}$.
\item[(iii)] Replace each $x_i$ by $x^{'}_i = x_i + \delta_i$, where $\delta_i$ is sampled from a Gaussian distribution of mean 0 and variance equal to the experimental error associated with $x_i$.
\item[(iv)] From the new simulated data set $\{x^{'}_i\}$, obtain a new set of $m^2$ matrices using all permutations of the Laing-O'Brien algorithm.
\item[(v)] If any of the new reconstructed matrices, say $U_1$, has a better value of $\chi^2$ to the experimental data than $U_0$, replace $U_0$ with $U_1$ and start from (i).
\item[(vi)] If no matrix has a better $\chi^2$, discard them and restart from (i).
\end{itemize}

The intuition behind this algorithm is the following. Given that each experimental data point is expected to lie within a certain range of values due to its associated error, even if the real process was perfectly unitary (which it is not) its distributions probably would not correspond exactly to the center of each uncertainty interval, but rather some set of values close to it. Thus, the experimental error bars offer a natural measure of how close we expect the real-world unitary to be from the reconstructed one obtained by the Laing-O'Brien algorithm, and serve as a guide for a random search in the space of unitaries. 

This algorithm was run for the 5- and 7-mode chips, since in those experiments the complete data set was available. The Mathematica$^{\copyright}$ files that implement this algorithm can be found in \appdx{app2}. In \tabl{permutation5} and \tabl{permutation7} we show a comparison of how this algorithm improves on the best obtained matrix from the Laing-O'Brien algorithm. For the 5-mode chip, a 15 min run of the algorithm improved the $\chi^2$ by 1396 (11 per data point), while improving the gate fidelity and variation distances for single and two photon data, upon which the algorithm saturated. For the 7-mode chip the improvement is more impressive: an 8 hr run of the algorithm improved the $\chi^2$ by two orders of magnitude, while considerably improving the single- and two-photon variation distances, upon which the algorithm saturated. The algorithm was run several times, starting from different permutations, displaying a consistent performance.

These results suggest that the algorithm has a practical application as a refinement of the Laing-O'Brien algorithm. Considering that any experimental imperfections will only be amplified on experiments with three (and more) photons, a reliable method for estimating the fabricated unitary matrix is essential for the scalability of the proposal. In this sense, a complete brute-force search is completely unfeasible---the $m$-mode chip has typically $m^2$ free parameters, which quickly makes the parameter space too large for this approach. Thus, an efficient tomography algorithm, of which the Laing-O'Brien is an example, is an indispensable tool for experiments on progressively larger systems. However, this algorithm, while having the good feature of being immune to device losses, still suffers considerable variations from noise in experimental data stemming from non-unitarity of the device and statistical errors. As such, I believe that the refinement proposed helps to mitigate these variations and, furthermore, receives additional feedback from the experiments, in the form of error estimates that indicate the more reliable subset of data, which the algorithm should strive to approximate more closely. At present I do not have a rigorous theoretical analysis of the performance of the algorithm, nor why it seems to work much more efficiently than other well-known algorithms such as gradient search, only the empirical observation that it is so. Nonetheless, I believe they provide compelling evidence for the usefulness of the algorithm for any experiment that may be performed in the near future.

Finally, I would like to point out that the output of the numerical searches reported in \tabl{permutation5} and \tabl{permutation7} are not the same as those I will report on the next section. This is because I opted for presenting the experimental results that follow in a manner faithful to the papers they were originally published in, but in which an outdated version of this algorithm was used, which is why the results will be slightly worse.

The Mathematica$^{\copyright}$ files for these algorithms can be found in \appdx{app2}.

\section{Experimental results} \label{sec:bosonnew_c}

In this section, we reproduce the results reported in references \cite{Crespi2013b, Spagnolo2013b,Spagnolo2013c}. The first of these was published in 2013 in \textit{Nature Photonics} \cite{Crespi2013b} and consists of one of the first demonstrations of a small-scale implementation of BosonSampling, where we observed one, two and three photons interfering in a $5$-mode unitary sampled from the uniform ensemble. The second paper, published in 2013 in \textit{Physical Review Letters} \cite{Spagnolo2013b}, consists of the observation of bosonic bunching in larger interferometers. Besides observing a decrease in the fraction of bunching events as the size of the chips increases, which is compatible with the bosonic birthday paradox described in \sec{bosonreview_c_BBP}, we also observed experimentally and proved theoretically a new bunching law, relating the classical and quantum probabilities of full-bunching events (i.e.\ those where all photons exit at the same modes). Finally, the third manuscript, which is still undergoing the review process as of the writing of this thesis, focuses on the certification of BosonSampling devices \cite{Spagnolo2013c}. More specifically, we sampled two interferometers from the random phase ensemble described in \sec{bosonnew_b_ensembles}, of 7 and 9 modes respectively, and applied the certification algorithm due to Aaronson and Arkhipov to efficiently distinguish the distribution generated by the device from the uniform distribution, as discussed in \sec{bosonreview_c_certif}.

\textbf{Notation:} The notation used throughout this section differs slightly from that used in the rest of the chapter. This is to conform with the published versions, specially of the figures. However, the meaning of each notation will be clear from context.

\subsection{BosonSampling in arbitrary integrated photonic chips} \label{sec:BosonSamplingExp}

In this section we report on the experimental implementation of a small instance of the BosonSampling proposal, using up to three photons interfering in a randomly chosen, 5-mode integrated photonic chip, first reported in \cite{Crespi2013b}. The unitary describing the chip, $U_{5}^t$, was sampled from the uniform distribution (cf.\ \sec{bosonnew_b_ensembles}), and decomposed in phase shifters and beam splitters according to the method of Reck \textit{et al.}\ described in \cite{Reck1994}. The resulting circuit schematics is depicted in \sfig{Reck}{b}, and a table with the corresponding parameters $t_i$, $\alpha_i$ and $\beta_i$, as well as $U_{5}^t$, can be found in \appdx{app3}. Each parameter was independently controlled using the novel three-dimensional technique depicted in \fig{LOelements}. This was the first demonstration of such a random chip where the unitary was sampled numerically and the chip was fabricated as specified, rather than delegating the randomness to lack of fabrication control. The chip was used in single, two, and three photon experiments that made use of the sources and detection apparatuses described in \fig{setup}. The input-output transmission of the device, in other words the fraction of the overall signal that is not lost inside the device, is 30\%. 

As a first step we characterized the 5-mode chip by injecting single photons in each input port $i$ and measuring the probability $P^{1}_{exp}(i,K)$ of detecting it in output mode $K$. The probability distribution obtained experimentally is shown in \sfig{ExpData1}{a}, together with the theoretical prediction $P^1_{t}(i,K)$ of the sampled unitary $U_{5}^t$. Each output probability corresponds to the absolute value squared of one matrix element of $U_{5}^t$. To quantify the agreement between theory and experiment we calculated the variation distance between the distributions, as defined in \sec{introsimul}. We obtained $d^1_{exp,t}=0.158\pm0.003$ by averaging over the values corresponding to all the different inputs; this small distance provides a first confirmation of the proper functioning of the device. 

\begin{figure}[t]
\capstart
\centering
\includegraphics[width=1\textwidth]{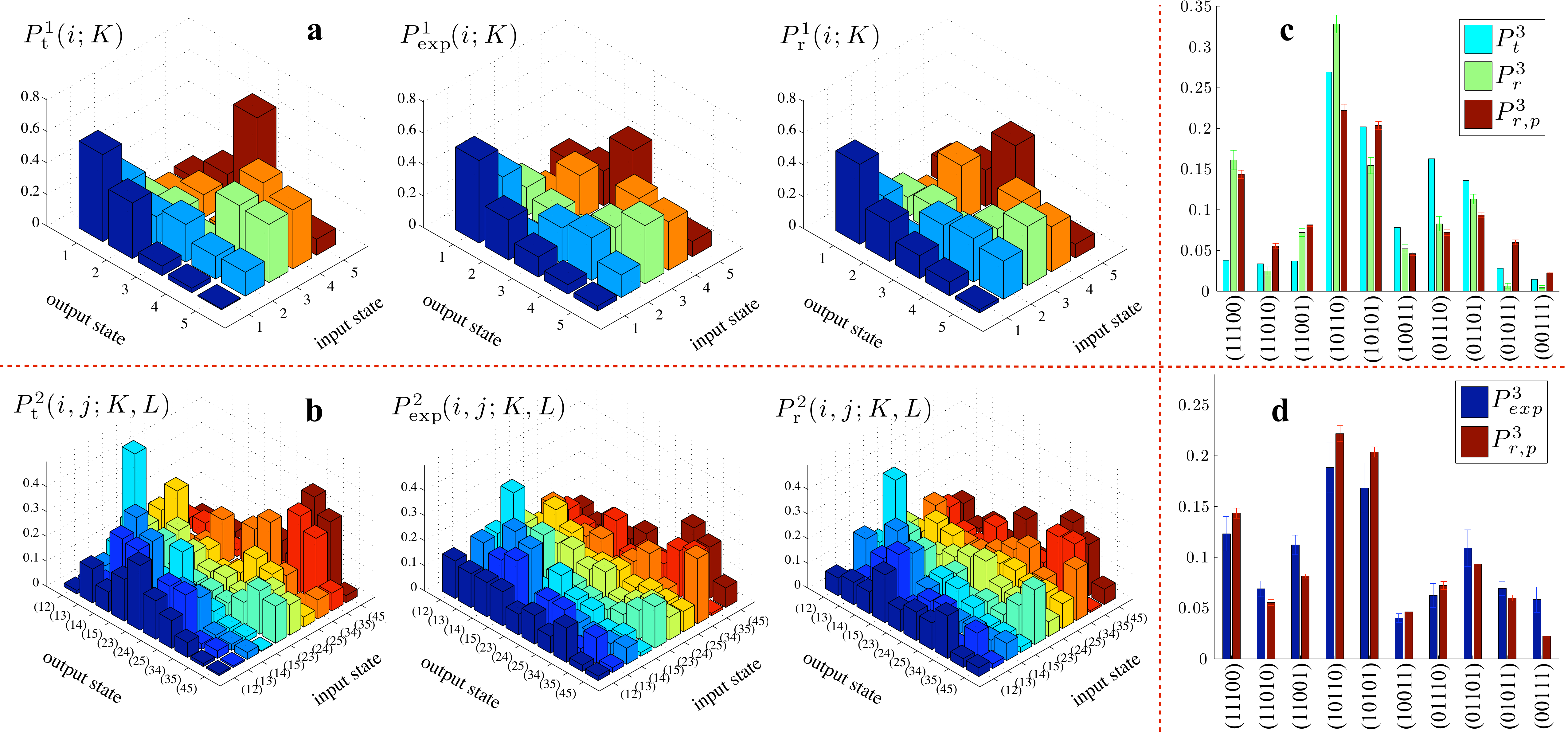}  
\caption[Experimental results for 5-mode chip.]{(a) One-photon probability distribution: theoretical distribution $P^{1}_{t}(i,K)$, experimental distribution $P^{1}_{exp}(i,K)$ and reconstructed distribution $P^{1}_{r}(i,K)$. (b) Two-photon probability distributions: theoretical distribution $P^{2}_{t}(i,j,K,L)$, experimental distribution $P^{2}_{exp}(i,j,K,L)$ and reconstructed distribution $P^{2}_{r}(i,j,K,L)$. Expected three-photon probability distribution for input state $\ket{10101}$: (c) theoretical distribution $P^{3}_{t}$, reconstructed distribution $P^{3}_{r}$ and reconstructed distribution with partial distinguishability $P^{3}_{r,p}$, and (d) experimental distribution $P^{3}_{exp}$ and reconstructed distribution with partial distinguishability $P^{3}_{r,p}$. Error bars in the experimental data are due to Poissonian statistics of the measures events, while error bars on theoretical predictions were obtained from a Monte Carlo simulation.}
\label{fig:ExpData1}
\end{figure}

A complete characterization of the implemented interferometer was performed by simultaneously injecting two single photons on all ten combinations of two different input modes $\{i,j\}$. For each input combination we measured the ten corresponding visibilities from the HOM curves, as described in \sec{measurements}, for the photons exiting in two distinct modes $\{K,L\})$. From the visibilities we were able to compute the experimental two-photon probability distributions $P^2_{exp}(i,j,K,L)$, reported in \sfig{ExpData1}{b} together with the theoretical distribution $P^2_{t}(i,j,K,L)$ expected from $U_{5}^t$. We observe a good agreement between the experimentally and theoretical probabilities, as evidenced by the variation distance $d^2_{exp,t}=0.221\pm0.013$ (averaged over all the inputs), thus confirming good control over the chip's fabrication parameters.

Having measured the two-photon visibilities, we also applied the refined version of the Laing-O'Brien algorithm described in \sec{bosonnew_b_tomography} to obtain a reconstructed unitary $U_{5}^r$. The average variation distance between the predictions of $U_{5}^r$ and our experimental data is $d^1_{exp,r}=0.065\pm0.003$ (single photon experiments) and $d^2_{exp,r}=0.103\pm0.013$ (two-photon experiments), which indicates a good characterization of the unitary implemented experimentally. We also obtained a gate fidelity between of $U_{5}^t$ and $U_{5}^r$ of $F=0.950\pm0.002$. As explained in \appdx{app3}, we performed a Monte Carlo simulation to estimate how errors in the phase shifts and transmissivities affect the overall unitary. Our simulations show the gate fidelity observed is consistent with the average error rate observed in the calibrations shown in \fig{LOelements}. This, in turn, provides evidence that the precision over each parameter is maintained even during the fabrication of larger circuits.

Finally, we have also probed the chip's behavior in the multi-photon regime, by injecting three single photons into modes 1, 3 and 5 of our interferometer (this choice was random) and measuring all probabilities of coincidence outcomes. In \sfig{ExpData1}{c} we compare three distributions: the ideal distribution $P^3_t$ obtained from $U_{5}^t$; the distribution $P^3_r$ arising from our reconstructed $U_{5}^r$ and the one $P^3_{r,p}$ taking into consideration the partial indistinguishability $p$ of the photons (cf.\ \sec{sources}). \sfig{ExpData1}{d} shows a good agreement between the distribution $P^3_{r,p}$ and our experimental results $P^3_{exp}$ as quantified by the variation distance between these two distributions $d^3_{exp,rp}=0.105\pm0.024$. This is an experimental confirmation of the permanent formula [cf.\ \eqbeg{DUquantum}] in the three-photon, five-mode regime. 

\subsection{Bosonic bunching in multimode interferometers} \label{sec:BosonicBunchingExp}

In \sec{introduction_b}, we discussed how bosons and fermions exhibit distinctly different statistical behaviors. For fermions, the required wave-function anti-symmetrization results in the Pauli exclusion principle. Bosons, on the other hand, tend to occupy the same state more often than fermions or classical particles do, which is reflected in phenomena such as Bose-Einstein condensation, and the HOM effect described in \sec{HOMeffect}. The HOM effect consists in a suppression of the probability that two identical photons incident on separate ports of a 50:50 beam splitter will exit on separate modes, and is observed as a curve (e.g.\ \fig{HOMdip} and \fig{HOMvisib}) describing the dependency of this probability on the partial distinguishability between the photons. However, as discussed in \sec{measurements}, in larger interferometers a HOM curve connecting input $\{i,j\}$ and output $\{K,L\}$ may be a dip or a peak, since there is obviously no constraint on which regime, classical or quantum, must present a larger probability of transition between two arbitrary states. Nonetheless, there are signatures of photonic bunching in systems with larger interferometers and more particles. Three examples are: (i) recent suppression laws demonstrated for $m$-photon interference in particular $m$-mode interferometers \cite{Tichy2010}; (ii) the bosonic birthday paradox\footnote{Interestingly, while the bosonic birthday paradox paper \cite{Arkhipov2012} shows that, on average, there is an enhancement in the bunching fraction in the quantum relative to the classical regime, the authors' main point is that this behavior is not as strong as one could expect, and the bunching fraction become negligible in the limit $m \rightarrow \infty$.}, describing the average behavior of bunching outcomes in uniformly random interferometers in the regime $m \gg n$ (cf.\ \sec{bosonreview_c_BBP}); and (iii) a rule showing that the transition to a full-bunching output (e.g.\ $\ket{n,0,\ldots,0}$), if nonzero, is always enhanced in the quantum relative to the classical regime, which we proved theoretically and verified experimentally in \cite{Spagnolo2013b}. The experiments reported in this section consist on a comprehensive observation of bosonic bunching, most notably effects (ii) and (iii), in interferometers of sizes ranging from $m=2$ to $m=16$, and designs as depicted in \fig{BBPchips}. 

\begin{figure}[t]
\capstart
\centering
\includegraphics[width=0.5\textwidth]{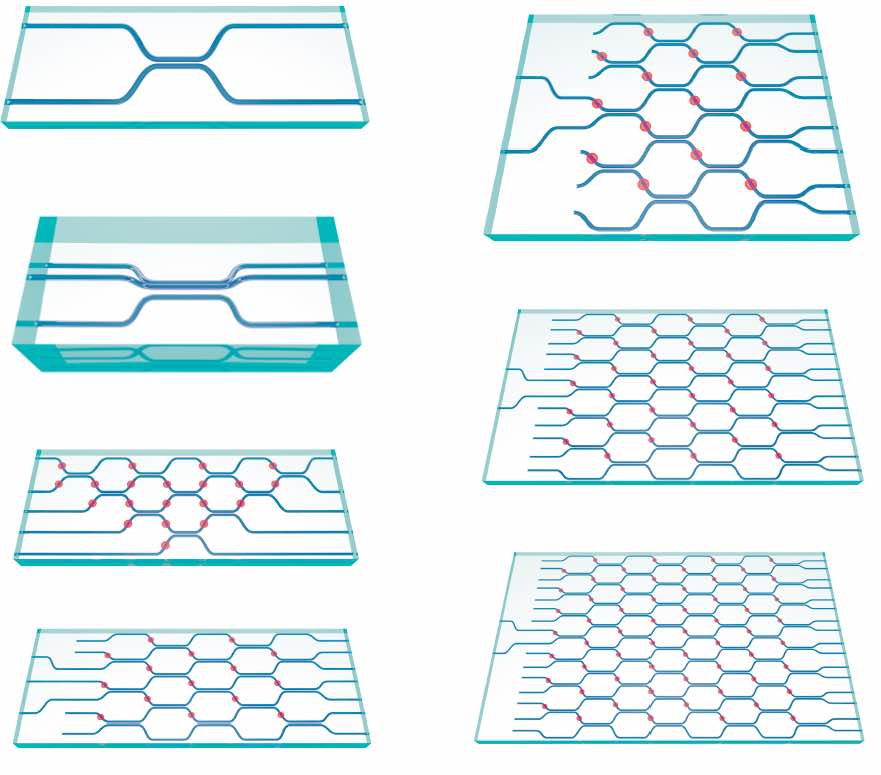}
\caption[Design schematics for chips of 2-16 modes.]{Layouts of the chips used in the experiment. Red spots represent phase shifters, and directional couplers perform as described in \sec{parametercontrol}. For more details on the chips' design, see the Supplementary Material of \cite{Spagnolo2013b}.}
\label{fig:BBPchips}
\end{figure}

Let us now state and prove the aforementioned rule governing bosonic full-bunching probabilities\footnote{During the writing of this thesis, it was brought to our attention that this result was also reported independently in the PhD thesis of Malte Tichy \cite{Tichy2011}.} of $n$ photon in arbitrary $m$-mode interferometers. 

\begin{theorem} 
Let $t_k$ denote the occupation number  of input mode $k$. Let us denote the probabilities that all $n$ bosons leave the interferometer in mode $j$ by $q_c(j)$ (distinguishable bosons) and  $q_q(j)$ (indistinguishable bosons). Then the ratio of full-bunching probabilities $r_{fb}=q_q(j)/q_c(j)=n!/\prod_k{t_k!}$, independently of $U,m$, and $j$.
\end{theorem}

\begin{proof} Recall from \lem{bosonperm} in \sec{twomode} that, for $n$ photons interfering in an $m$-mode interferometer described by unitary $U$, the probability amplitude associated with input $\ket{T}=\ket{t_1 t_2 \ldots t_m}$ and output $\ket{S} = \ket{s_1 s_2 \ldots s_m}$ is given by
\begin{equation} \label{permanent}
\bra{S} U_F\ket{T} = \frac{\textrm{perm}(U_{S,T})}{\sqrt{t_1! \ldots t_m! s_1! \ldots s_m!}},
\end{equation}
where $U_{S,T}$ is the matrix obtained by repeating $t_i$ times the $i^{th}$ column of $U$, and $s_j$ times its $j^{th}$ row. Recall from \sec{bosonreview_c_simulation} that, for distinguishable bosons, the analogous equation holds:
\begin{equation}\label{dist}
p_{S,T}=\frac{\textrm{perm}(|U_{S,T}|^2)}{\prod_i{s_i!}},
\end{equation}
where $|U_{S,T}|^2$ is the matrix obtained by taking the absolute value squared of each corresponding element of $U_{S,T}$.

Let us now introduce an alternative, convenient way of representing the input occupation numbers. Define a $n$-tuple of $m$ integers $r_i$ so that the first $t_1$ integers are 1, followed by a sequence of $t_2$ 2's, and so on until we have $t_m$ $m$'s. As an example, input occupation numbers $t_1=2, t_2=1, t_3=0, t_4=3$ would give $r=(1,1,2,4,4,4)$. 
Using Eq. (\ref{permanent}) we can evaluate the probability $q_q(j)$ that the $n$ indistinguishable bosons will all exit in mode $j$:
\begin{equation}
q_q(j)=\frac{|\textrm{per}(A)|^2}{n!\prod_k{t_k!}},
\end{equation}
where $A$ is a $n \times n$ matrix with elements $A_{i,k}=U_{j,r_k}$. Since all rows of $A$ are equal, $\textrm{perm}(A)$ is a sum of $n!$ identical terms, each equal to $\prod_k U_{j,r_k}$. Hence
\begin{equation}
q_q(j)=\frac{|n! \prod_k U_{j,r_k}|^2}{n!\prod_k{t_k!}} = n! \frac{|\prod_k U_{j,r_k}|^2}{\prod_k{t_k!}}.
\end{equation}

Using Eq. (\ref{dist}), we can calculate the probability $q_c(j)$ that $n$ distinguishable bosons will leave the interferometer in mode $j$: $q_c(j)=\textrm{perm}(B)/n!$, where $B$ has elements $B_{i,k}=|A_{i,k}|^2=|U_{j,r_k}|^2$. Hence
\begin{equation}
q_c(j)=n! \frac{\prod_k |U_{j,r_k}|^2}{n!}= \prod_k |U_{j,r_k}|^2.
\end{equation}
Our new bosonic full-bunching rule establishes the value of the quantum/classical full-bunching ratio, which we can now calculate to be
\begin{equation}
r_{fb}=\frac{q_q(j)}{q_c(j)}=\frac{n!}{\prod_k{t_k!}}.
\end{equation}
\end{proof}

This bosonic full-bunching rule generalizes the HOM effect into a universal law, now applicable to any interferometer and any number of photons $n$. Despite becoming exponentially rare as $n$ increases, as discussed in \sec{bosonreview_c}, full-bunching events are enhanced by a factor as high as $n!$ when at most one boson is injected into each input mode, as in our experiments.

The apparatus for this experiment consists mostly on that described throughout \sec{bosonnew_b}, including the sources described in \sec{sources} and the single, two and three photon measurements described in \sec{measurements}. Randomness was purposefully incorporated in the designs of the chips in various ways, including sampling from the Haar and the random phase ensembles of \sec{bosonnew_b_ensembles}. In particular, the chip with $m=5$ was the same as used in the experimented reported in \sec{BosonSamplingExp}, taken from the Haar ensemble. The $7$-mode chip was taken from the random phase ensemble. The remaining chips were taken from experiments performed previously by the same quantum optics group which I was not a part of. For more details see the Supplementary Material of \cite{Spagnolo2013b}, together with the provided references for the chip with $m=2$ \cite{Sansoni2010}, the two chips with $m=3$ \cite{Spagnolo2013a}, one chip with $m=8$ \cite{Sansoni2012}, and the remaining chips with $m=8,12,16$ \cite{Crespi2013a}.

\begin{figure}[t]
\capstart
\centering
\includegraphics[width=1.\textwidth]{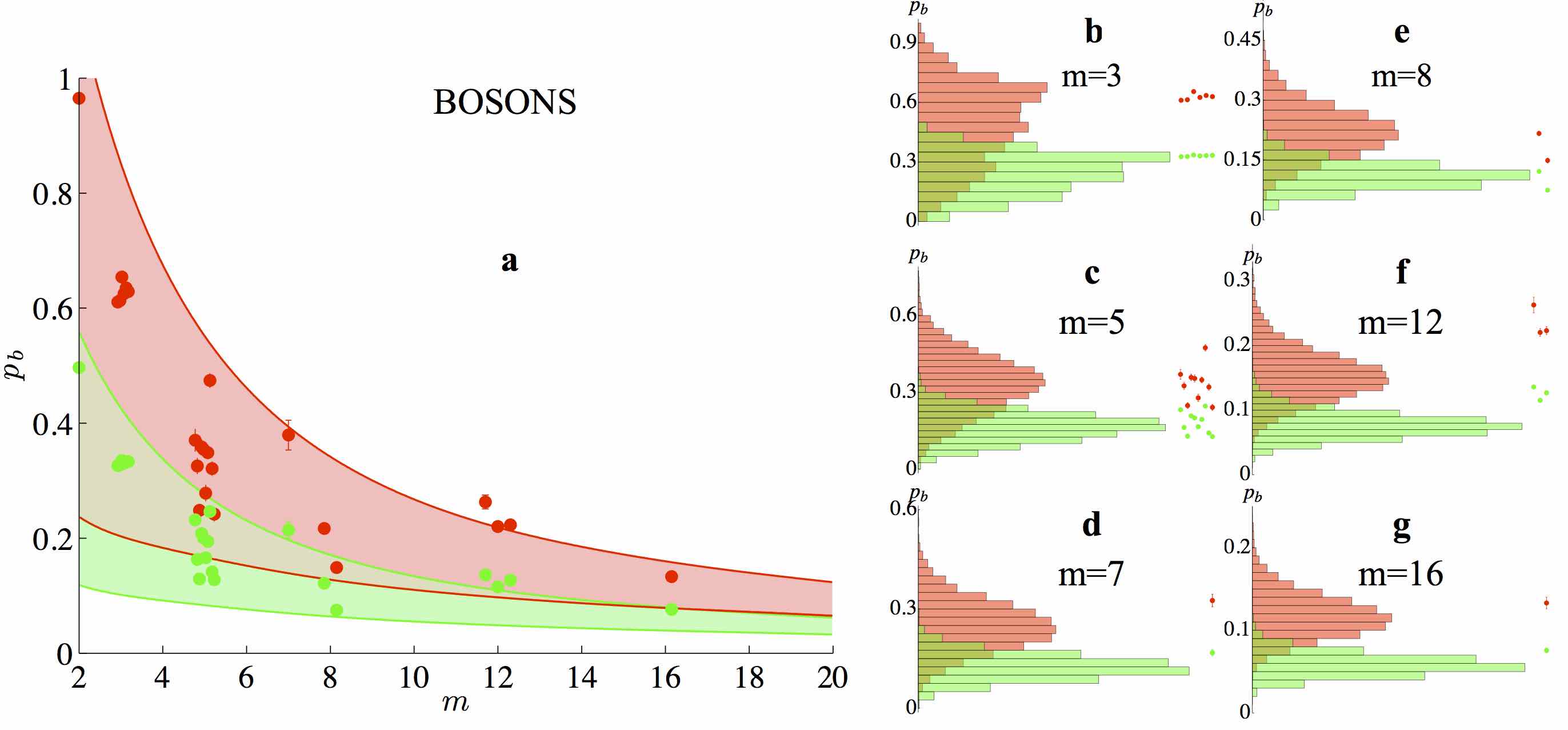}
\caption[Two-photon bunching data.] {(a) Bunching fraction for two indistinguishable photons ($p_{b}^{(q)}$, red points) and two distinguishable photons ($p_{b}^{(c)}$, green points). Experiments were performed on different unitaries ($m=3$, $m=8$, $m=12$) or different input states ($m=3$, $m=5$). Shaded areas correspond to of 1.5 standard deviation around the mean, obtained from sampling over $10000$ uniformly random unitaries. (b)-(g), Experimental results (points) together with histograms obtained from the numerical simulations, for several values of $m$. Error bars in the experimental data are due to the Poissonian statistics of the measured events, and where not visible are smaller than the symbol.}
\label{fig:twophotons}
\end{figure}

A first set of experiments aimed at measuring the bunching probabilities $p_b^{(q)}$ and  $p_b^{(c)}$ respectively of quantum (i.e.\ indistinguishable) and classical (i.e.\ distinguishable) photons. We note that these probabilities depend both on the interferometer's design and the input state used. A bunching event involves, by definition, the overlap of at least two photons in a single output mode. The classical bunching probability $p_b^{(c)}$ is obtained from single-photon experiments that characterize the transition probabilities between each input/output combination. An estimate of $p_b^{(q)}$ was obtained from the procedure described at the end of \sec{measurements}, where the ratio $t:=(1-p_b^{(q)})/(1-p_b^{(c)})$ is estimated by running the experiment on the classical and quantum regimes for equal periods of time, without the need for distinguishing different multi-photon states at each output.

We summarize the experimental results for a number of different photonic chips in \fig{twophotons} (two-photon experiments) and \fig{threephotons} (three-photon experiments). The results are in good agreement with theory, taking into account the partial indistinguishability of the photon source. For all the employed interferometers we find that indistinguishable photons display a higher coincidence rate than distinguishable photons do ($p_b^{(q)}>p_b^{(c)}$); this is known to be true for averages \cite{Arkhipov2012}. Furthermore, $p_b^{(q)}$ falls as $m$ increases, in the manner predicted by the bosonic birthday paradox. It is important to point out that the distinguishable-photon distribution plotted in \fig{twophotons} and \fig{threephotons} does not correspond to the classical birthday paradox described in \sec{bosonreview_c_BBP}, more specifically \eq{classicalBP}. This is because the usual classical birthday paradox consists of sampling $n$ people independently from the uniform distribution, whereas here all distinguishable photons in the experiment are subject to the \emph{same} random unitary. The correlations between matrix elements due to unitarity actually make a distinguishable-photon experiment have typically a lower bunching fraction than that described by the classical birthday paradox.

\begin{figure}[t]
\capstart
\centering
\includegraphics[width=0.6\textwidth]{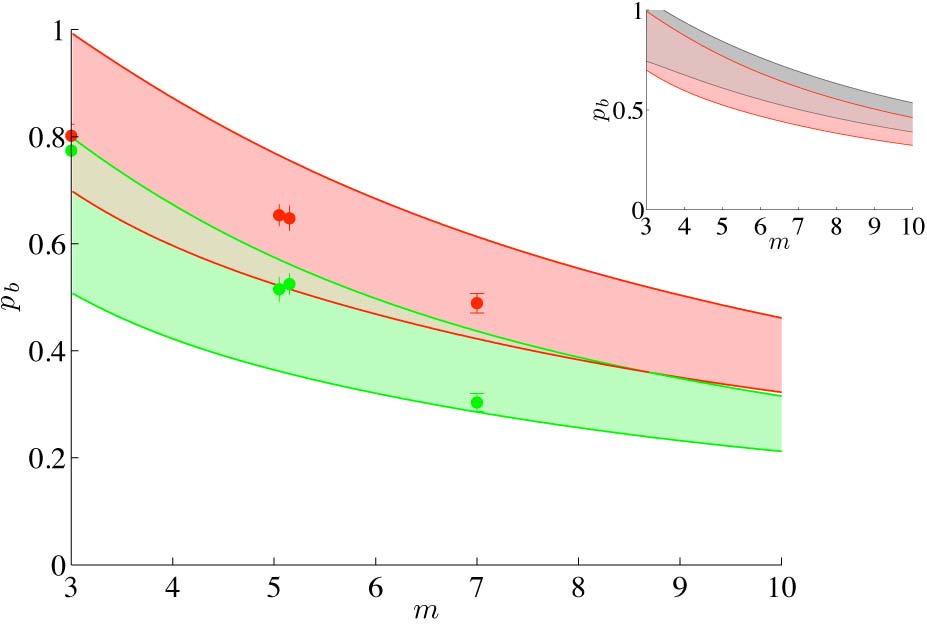}
\caption[Three-photon photonic bunching data.] {Experimental results for the three-photon photonic bunching in the quantum ($p_{b}^{(q)}$, red points) and classical ($p_{b}^{(c)}$, green points) regimes. Shaded areas correspond to 1.5 standard deviation around the mean, obtained from sampling over $10000$ uniformly random unitaries. Red area: simulation taking into account partial photon indistinguishability. Error bars in the experimental data are due to the Poissonian statistics of the measured events, and where not visible are smaller than the symbol. Inset: numerical simulation of the effect of partial photon distinguishability on the bunching probability $p_{b}$, where the gray area represents perfectly indistinguishable photons.}
\label{fig:threephotons}
\end{figure}

\begin{figure}[t]
\capstart
\centering
\includegraphics[width=0.6\textwidth]{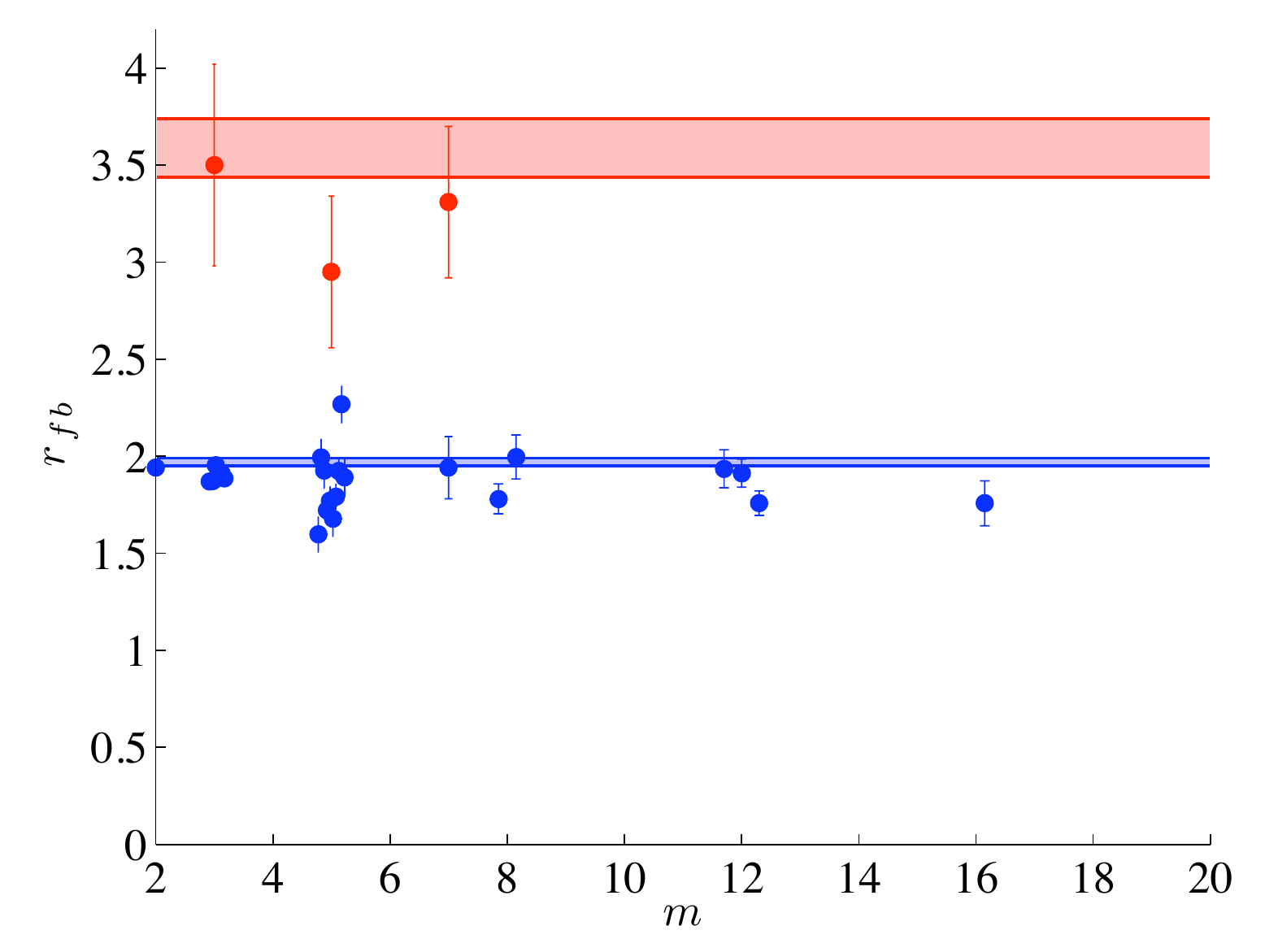}
\caption[Ratio between quantum and classical full-bunching probabilities]{Here we report $r_{fb}$ for two-photon (blue) and three-photon (red) experiments on a number of photonic chips. Shaded regions: expected values for two-photon (blue) and three-photon (red) taking into account partial photon distinguishability. Error bars in the experimental data are due to the Poissonian statistics of the measured events, and where not visible are smaller than the symbol.}
\label{fig:nfactorial}
\end{figure}

We now turn to experiments that test our bosonic full-bunching rule. We estimated the quantum to classical full-bunching probability ratio $r_{fb}$ by introducing delays to change the distinguishability regime, and performing photon counting measurements in selected output ports, using fiber beam-splitters and multiple single-photon detectors. In \fig{nfactorial} (blue data) we plot the full-bunching ratio for all two-photon experiments referred to in \fig{twophotons}, and find good agreement with the predicted quantum enhancement factor of $2!=2$. Note that in two-photon experiments every bunching event is also a full-bunching event, which means that when $n=2$ the ratio $r_{fb}=r_b=2$, independently of the number of modes $m$. We have also measured three-photon, full-bunching probabilities in random interferometers with number of modes $m=3,5,7$. Perfectly indistinguishable photons would result in the predicted $3!=6$-fold quantum enhancement for full-bunching probabilities. The partial indistinguishability  $\alpha=0.63 \pm 0.03$ of our three injected photons reduces this quantum enhancement to a factor $r_{fb}=\alpha^2 \, 3! + (1-\alpha^2) \, (3-1)!=3.59 \pm 0.15$.  The results can be seen in \fig{nfactorial} (red data), showing good agreement with the predicted value.

In conclusion, these experiments characterize the bunching behavior of up to three photons evolving in a variety of integrated multimode circuits. The results are in agreement with the predictions of the bosonic birthday paradox (c.f.\ \sec{bosonreview_c_BBP}). We have also proved a new rule that sharply discriminates quantum and classical behavior, by focusing on events in which all photons exit the interferometer bunched in a single mode, and have obtained experimental confirmation of this new full-bunching law. 

\subsection{Experimental validation of BosonSampling} \label{sec:validation}

In this section, we report an experiment focused on validation of BosonSampling devices. Recall, from \sec{bosonreview_c_certif}, that recent criticism \cite{Gogolin2013} of the BosonSampling model showed that no symmetric algorithm can efficiently distinguish a BosonSampling device from one that samples from the trivial uniform distribution over all outcomes. This prompted a response from the authors of the original BosonSampling paper, arguing that symmetric algorithms are overly restrictive and proving that an efficient non-symmetric algorithm exists to distinguish between these two hypotheses \cite{Aaronson2013b}. Here we report new BosonSampling experiments on photonic chips of 5, 7, and 9 modes, and analyze the data using the algorithm of \cite{Aaronson2013b}, which we described in \sec{bosonreview_c_certif}. We show that the test successfully validates small experimental data samples against the hypothesis that they are uniformly distributed. We also show how to discriminate data arising either from indistinguishable or distinguishable photons, albeit in a non-scalable manner. 

The experiments were performed using the same techniques described throughout this chapter. The 5-mode chip is the same used for the BosonSampling experiment described in \sec{BosonSamplingExp}, while the 7- and 9-mode chips were drawn from the random phases ensemble described in \sec{bosonnew_b_ensembles}. Their parameter specifications are reported in \appdx{app3}. The 7-mode chip was also reconstructed using the algorithm described in \sec{bosonnew_b_tomography}, and the result of this reconstruction is also reported in \appdx{app3}.

Let us now discuss how a certifier can validate small sets of BosonSampling data generated by an agent we call the BosonSampler (BS), against the hypothesis that they were generated by Uniform Sampler (US), an agent that samples from the uniform distribution. The verifier succeeds by using the efficiently-computable estimator $R$ described in \sec{bosonreview_c_certif} on small experimental data sets. $R$ is correlated to the outcome probabilities just enough to reflect some of the structure of the BosonSampling distribution, allowing us to determine with great confidence that it is not simple gibberish produced by US (but not correlated enough to approximate the values of the probabilities themselves). We applied this test to multiple, random experimental data sets of varying sizes. This enabled us to gauge the trade-off between set size and success rate, which has been theoretically studied in the asymptotic limit \cite{Aaronson2013b}. The results are shown in \fig{analysis}. For experiments with the $5$-, $7$-, and $9$-mode chips, the verifier reaches a $95\%$ average success rate with very modest set sizes of just $\sim 100$ events. This establishes experimentally the usefulness of the R estimator of \cite{Aaronson2013b} for the analysis of small-scale experiments.

\begin{figure}[p]
\capstart
\centering
\includegraphics[width=0.9\textwidth]{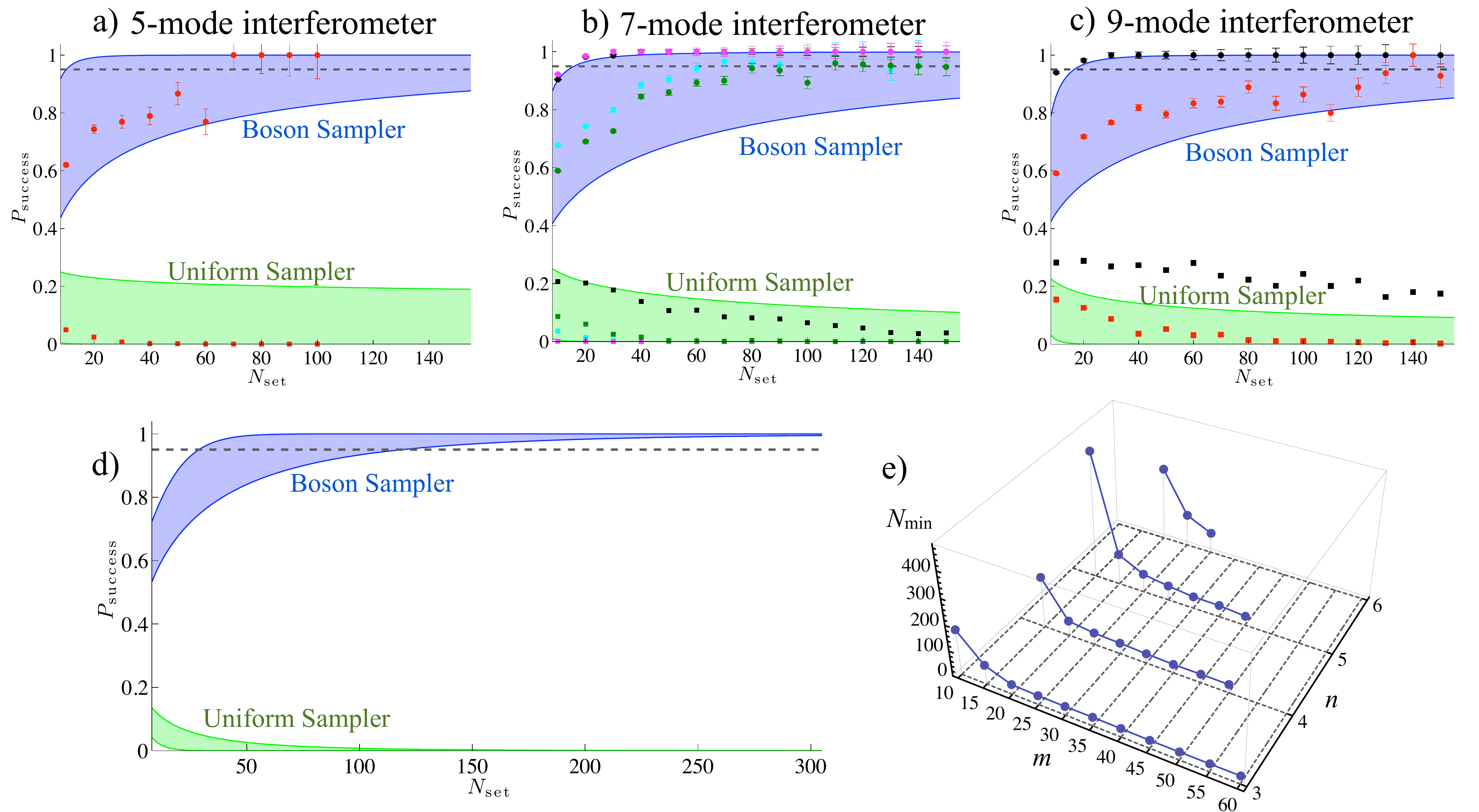}
\caption[Experimental validation of BosonSampling.]{Efficient validation test using experimental data sets of varying sizes. We show the success rate as a function of sample set size $N_{\mathrm{set}}$, in experiments using (a) one Haar-random 5-mode interferometer (red circles: input state $\ket{01110}$); (b) a random 7-mode interferometer with input states: $\ket{0011100}$ (green), $\ket{0001110}$ (cyan), $\ket{0101010}$ (black), and $\ket{1110000}$ (magenta); (c) a random 9-mode interferometer with input states: $\ket{000111000}$ (red) and $\ket{001100010}$ (black). Grey dashed line: level for $95\%$ success probability. Squares: numerical simulations, averaged over 1000 data sets, of the validation test using data generated by US. In all plots, blue shaded regions correspond to theoretical prediction for validation of BS, shown as $1.5$ standard deviation over 1000 Haar-random unitaries. Simulations exclude cases where success rate does not reach $95\%$ even with $N_{\mathrm{set}}=5000$. The number of cases with the (asymptotically proven) correct behavior was $434$ ($m=5$), 573 ($m=7$), and 822 ($m=9$). Green shaded regions correspond to theoretical prediction for the validation of US. (d) Simulated performance for BosonSampling experiments with $n=5$ photons and $m=25$, averaged over in 100 Haar-random interferometers. (e) Minimum data set size for both $>95\%$ success probability using BS data and $<5\%$ using US data, as a function of $n$ and $m$ obtained through numerical simulation. For each point, the simulation is averaged over 50 or 100 Haar-uniform unitaries.}
\label{fig:analysis}
\end{figure}

To show that the test will also work in as-yet unperformed, larger-scale experiments, in \sfig{analysis}{d} we simulate the test's success rate for BosonSampling experiments with $n=5$ and $m=20$. Additionally, in \sfig{analysis}{e} we numerically determine the minimum data set size $N_{\mathrm{min}}$ for which the test based on the $R$ estimator discriminates BosonSampling data from the uniform distribution (and vice-versa) with a success rate $>95\%$. Not only is $N_{\mathrm{min}}$ small for all experiments we simulated, it actually decreases as we increase $m$. Despite proving successful for all the interferometers we implemented experimentally, our numerical simulations revealed that the test fails for some interferometers if the ratio $m/n$ is too low.

\begin{figure}[p]
\capstart
\centering
\includegraphics[width=0.9\textwidth]{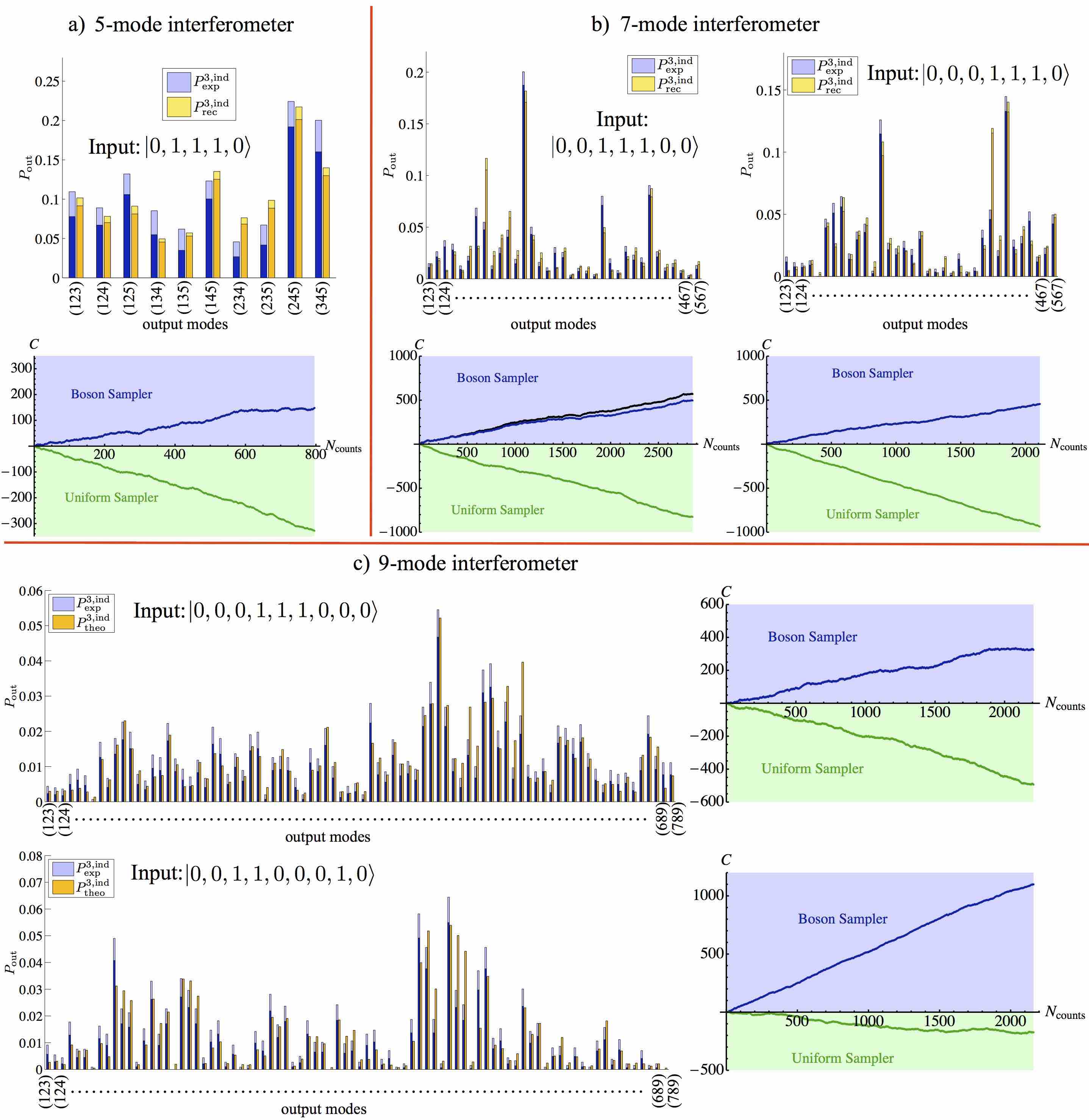}
\caption[Full validation of the BosonSampling experiments.]{Experimentally measured (blue) and theoretical (yellow) probabilities for (a) Haar-random 5-mode chip with input state $\ket{01110}$; (b) random 7-mode chip with input states $\ket{0011100}$ and $\ket{0001110}$; (c) random 9-mode chip with input states $\ket{000111000}$ and $\ket{001100010}$. Lighter regions of the blue bars correspond to experimental error due to Poissonian statistics. Lighter regions of the yellow bars correspond to error in the reconstruction process, retrieved by Monte Carlo simulation. Bottom figures of panels (a) and (b), right figures of panel (c): application of the $R$-estimator test to the full set of experimental data. $C$ is a counting variable that is increased by 1 for each event assigned to BS, and decreased by 1 for each event assigned to US. Blue points: test applied on the experimental data by exploiting the ideal unitaries. Black points: test applied on the experimental data by exploiting the reconstructed unitary (for the $m=5,7$ chips only). Green points: test applied on simulated data generated by US. For states $\ket{01110}$ ($m=5$) and $\ket{0001110}$ ($m=7$) blue and black points present a large overlap and superimpose in the figures.}
\label{fig:histogram}
\end{figure}

In the probed regime with $n=3$ photons and interferometers with up to $m=9$ modes it is possible to perform a full validation of the BosonSampling experiments by reconstructing all probabilities associated with no-collision events. This requires recording experimental data sets of a larger size; for the $m=7$ chip, for example, we recorded $\sim 2100$ events. The experimentally reconstructed probabilities are then compared with the theoretical prediction. For the chips with $m=5,7$, we compared the experimental data with the theoretical obtained from the reconstructed corresponding unitaries, while for $m=9$ we used the ideal, theoretical unitary for comparison. The results are shown in \fig{histogram}, and the good agreement between the experiments and the predictions is quantified by the variation distance $d=1/2 \sum_{k} \vert p_k - q_k \vert$, which reaches values $d_{\mathrm{exp,r}}^{(2,3,4)}= 0.104 \pm 0.022$ ($m=5$), $d_{\mathrm{exp,r}}^{(3,4,5)}=0.168 \pm 0.016$ and $d_{\mathrm{exp,r}}^{(4,5,6)}=0.133 \pm 0.017$ ($m=7$), $d_{\mathrm{exp,t}}^{(4,5,6)}=0.113 \pm 0.017$  and $d_{\mathrm{exp,t}}^{(3,4,8)}=0.167 \pm 0.020$ ($m=9$). Furthermore, we have applied the $R$-estimator test to the full data set. We remark that the high fidelity of our devices relative to the designed ones allows us to apply this test successfully by using either the reconstructed or the ideal unitaries.

In addition, for small-scale experiments we can perform simple statistical tests to validate BosonSampling data against probability distributions which are more natural in experimental settings than the uniform distribution, such as that produced by the corresponding BosonSampling experiments with distinguishable photons. These tests are inefficient, in the sense that they require computing the probabilities, and thus the $\#P$-hard permanents which are the cornerstone of the original BosonSampling result. Nonetheless, at least until theoretical advances are made concerning efficient validation of BosonSampling against these alternative distributions, these inefficient statistical tests will be the best way to validate the physical device and will remain feasible for experiments involving up to $30-40$ photons. Since, as we will see, these statistical tests only require computing a small number of permanents (more specifically, one for each outcome that is observed in the experiment), they presumably will be considerably more efficient than performing a classical simulation as described in \sec{bosonreview_c_simulation}.

\begin{figure}[p]
\capstart
\centering	
\includegraphics[width=0.9\textwidth]{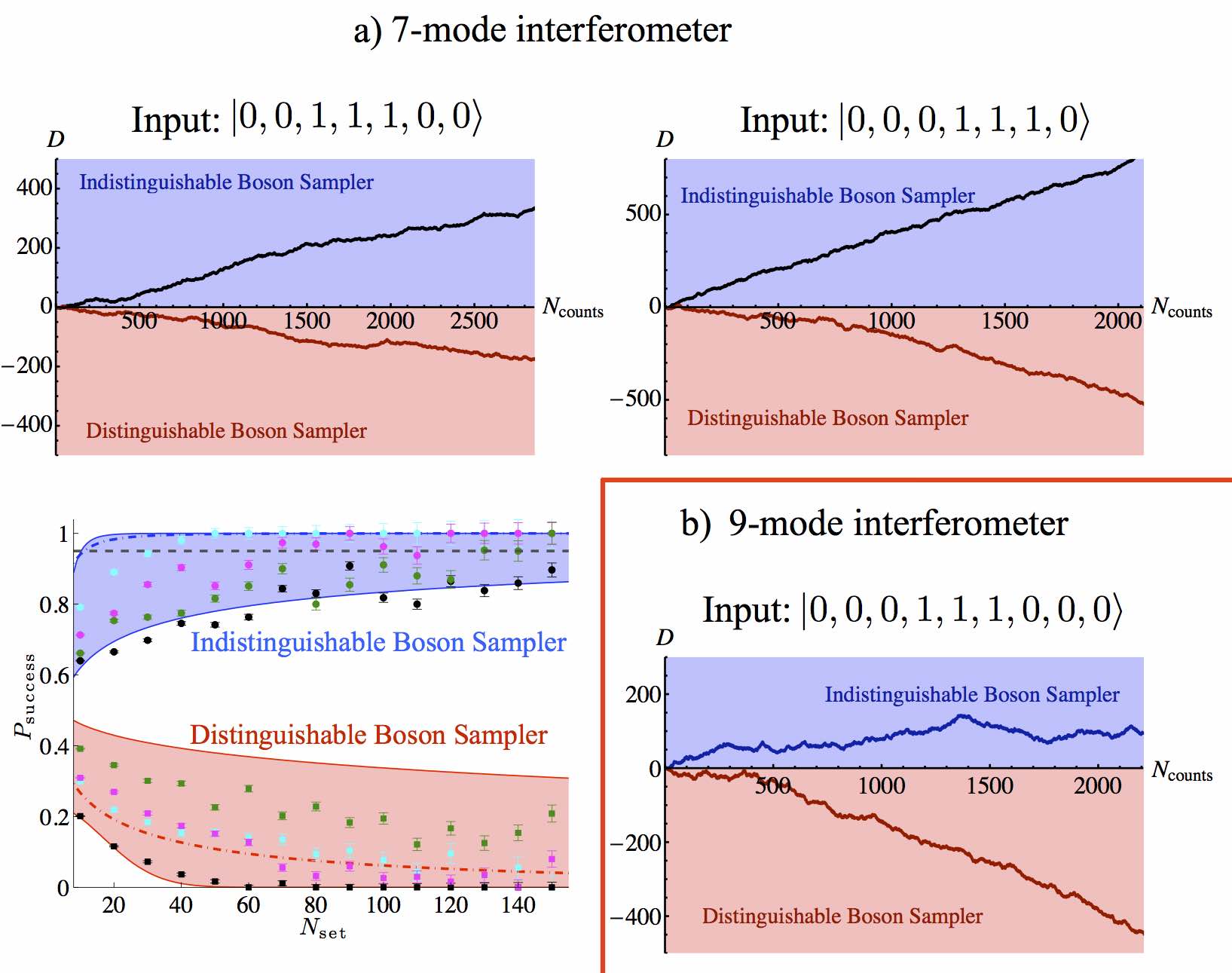}
\caption[Discrimination between alternative distributions]{(a) Experimental results of the discrimination between BosonSampling distributions in the classical and quantum regimes for the 7-mode chip. The protocol is applied by using the reconstructed unitary. Top figures: evaluation of the $D$ parameter for two different input states. Black data: indistinguishable photons. Red data: distinguishable photons. Bottom figure: success probability $P_{\mathrm{success}}$ of the discrimination protocol as a function of the data set size $N_{\mathrm{set}}$. Circles correspond to input states $\ket{0011100}$ (green), $\ket{0101010}$ (blue), $\ket{0001110}$ (cyan), $\ket{1110000}$ (magenta). Squares: corresponding success probability for ``false positive'' events with distinguishable photons. Blue shaded region: numerical simulation of the success probability of discrimination test for indistinguishable photons, taking into account the partial photon distinguishability. Blue dash-dotted line: average behaviour for perfectly indistinguishable photons. Red shaded region: numerical simulation for distinguishable photons. Red dash-dotted line: average behaviour without taking into account partial photon-distinguishability. (b) Experimental results of the discrimination between BosonSampling distributions in the classical and quantum regimes for the 9-mode chip with input state $\ket{000111000}$. The protocol is applied by using the probability distributions obtained from the ideal unitary. Blue data: indistinguishable photons. Red data: distinguishable photons.}
\label{fig:indvsdis}
\end{figure}

In our experiment, the statistical test used was an adaptation of a standard likelihood ratio test \cite{LivroCover}, with added thresholds to better take into consideration experimental imperfections. More specifically, let $p^{\mathrm{ind}}_{i}$ and $q^{\mathrm{dis}}_{i}$ be the probabilities associated with indistinguishable and distinguishable photons for measured outcome $i$, and let $D$ be a discrimination parameter, initialized to the value $D=0$. For each experimental outcome, we calculate the ratio of the expected probabilities for indistinguishable and distinguishable photons. If the ratio is close to one, up to to a threshold  $k_{1} < p^{\mathrm{ind}}_{i}/q^{\mathrm{dis}}_{i} < 1/k_{1}$, the event is considered to be inconclusive and $D$ is left unchanged. These inconclusive events, however, are still counted as a resource and do contribute to the effective number of events required to discriminate the two distributions. If $1/k_{1} \leq p^{\mathrm{ind}}_{i}/q^{\mathrm{dis}}_{i} < k_{2}$, the event is assigned to the Indistinguishable BosonSampler by adding $+1$ to $D$. If the ratio between the two probabilities is high, $p^{\mathrm{ind}}_{i}/q^{\mathrm{dis}}_{i} \geq k_{2}$, the event is assigned to the Indistinguishable BosonSampler by adding $+2$ to $D$, thus reflecting the higher level of confidence in this case. Conversely, if $1/k_{2} < p^{\mathrm{ind}}_{i}/q^{\mathrm{dis}}_{i} \leq k_{1}$ and $p^{\mathrm{ind}}_{i}/q^{\mathrm{dis}}_{i} \leq 1/k_{2}$ the event is assigned to the Distinguishable BosonSampler by adding $-1$ and $-2$ to $D$ respectively. Finally, after $N$ experimental outcomes, if $D>0$ the whole data set is assigned to the Indistinguishable BosonSampler, and conversely if $D<0$. In our analysis we set $k_{1} = 0.9$ and $k_{2} = 1.5$.

In conclusion, our experiments have shown how to leverage available information about the BosonSampling experiment to distinguish experimental data from the uniform distribution, using the scalable test proposed by Aaronson and Arkhipov \cite{Aaronson2013b}. Our analysis shows that this test works even for small instances of BosonSampling experiments, and provides experimental support for the recent theoretical refutation \cite{Aaronson2013b} of a recent criticism on BosonSampling experiments \cite{Gogolin2013}. We have also certified the experimental data using a test that distinguishes them from a similar experiment done with distinguishable photons. An independent paper \cite{Carolan2013} with some results similar to ours was reported online shortly after ours.

\section{Discussion and concluding remarks} \label{sec:bosonnew_d}

In this chapter, we presented both theoretical and experimental results on BosonSampling. On the theoretical side we showed that exact BosonSampling is hard, under the standard complexity assumptions, even if the optical circuits consists only of a constant number of beam splitter layers (\sec{bosonnew_a}). We also proved a new rule governing bosonic bunching in multi-photon multimode experiments, stating that there is always an enhancement for the probability of full-bunching outcomes in the quantum relative to classical regime (\sec{BosonicBunchingExp}). Numerically, we performed simulations to better characterize the random unitary ensemble described by layers of 50:50 beam splitters alternated with uniformly random phase shifters, and provided evidence that this ensemble displays a sufficiently rich behavior to make it a physically-motivated alternative to the Haar ensemble (\sec{bosonnew_b_ensembles}). Finally, we showed how the Laing-O'Brien algorithm for reconstruction of the unitaries from single- and two-photon data can be further refined by the use of numerical searches, consistently providing better agreement with experimental data (\sec{bosonnew_b_tomography}).

On the experimental side, we reported work performed in collaboration with the Quantum Optics groups of Rome and Milan, where interference of one, two and three photons was observed in integrated photonic chips. We experimentally verified that the formula describing multiphoton transitions in terms of matrix permanents holds well in the 3-photon, 5-mode regime (\sec{BosonSamplingExp}). We investigated effects of bosonic bunching of 3 photons in chips of up to 16 modes, observing a behavior consistent with the bosonic birthday paradox, and with our new full-bunching rule (\sec{BosonicBunchingExp}). Finally, we performed 3-photon experiments in chips of 5, 7 and 9 modes in order to apply a recent algorithm that distinguishes a BosonSampling distribution from the uniform distribution (\sec{validation}), and also applied standard statistical tests which, albeit asymptotically inefficient, distinguish BosonSampling from the corresponding classical-photon distribution using a small number of experimental samples.

Much like standard quantum computing, BosonSampling still has a lot of ground to cover before it provides concrete evidence that quantum systems are in fact outperforming classical computers in a given task. The most prominent advantage of BosonSampling is that this evidence will, presumably, be less demanding to reach---an experiment with $50-100$ photons already consists of something expected to be unfeasible with classical computers, whereas a useful quantum computation (such as Shor's algorithm), might require systems several orders of magnitude larger. However, in some sense the BosonSampling model is still in its infancy, and there are several theoretical and experimental obstacles that must be tackled if it is to stand on ground as solid as standard quantum computation. 

The experimental challenges consist mostly in showing that a given experiment can be scaled efficiently, at least in principle. For the sake of discussion, let us consider only integrated interferometers---free space linear optics has its own obvious scalability problems, while BosonSampling with other types of bosons is a much less compelling scenario \cite{Shen2014}. For a scalable integrated optical experiment, the main technological advances that should be pursued are (i) near-deterministic sources of nearly-indistinguishable photons, (ii) scalable chips, and (iii) better detectors. For (i) and (iii), it might help to have integrated sources and detectors, i.e., some method for generating and measuring the photons from within the chip itself, to reduce losses in the coupling with fibers. However, in my view, point (i) still has a serious fundamental problem: as long as sources only produce photons probabilistically, the overall coincidence signal will decay exponentially with the desired number of photons, not only because an $n$-fold coincidence event is exponentially unlikely, but even when it happens photons will likely have a large distinguishability factor\footnote{Recently, in the blog of Scott Aaronson (\url{http://www.scottaaronson.com/blog/?p=1579}), an idea was proposed, credited to Steve Kolthammer, called Scattershot BosonSampling, which seems to circumvent (i) somewhat. It is based on the following argument: rather than having $n$ photon sources and waiting for all $n$ of them to produce a photon at the same time, we can use $N > n$ sources and wait for any subset of $n$ of them to produce a photon each. In this way, the probability of an $n$-fold coincidence is boosted arbitrarily high using only a polynomial overhead in the number of sources. This, however, is still a preliminary idea, and furthermore it is unclear how this would help with the partial photon distinguishability factor.}. 

The major obstacle for point (ii) currently seems to be the fact that the signal decays exponentially with the depth of the circuit. This could be resolved with the development of novel fabrication techniques where losses scale more favorably, or by showing that the result of \sec{bosonnew_a} can be generalized for approximate BosonSampling with sub-polynomial-depth circuits. I do not consider that errors in the control of beam splitter transformations pose a greater obstacle to (ii) than losses, mainly because experimentalists have efficient reconstruction algorithms at their disposal, such as described in \sec{bosonnew_b_tomography}, which allow them to always know with very good precision what the implemented unitary is. Of course, from a complexity-theoretic point-of-view, this argument does not hold---$U$ is the input to the BosonSampling task, and any deviation from $U$ in the experiment must be considered an error. However, from a physical point of view, the unitary $U^{'}$ obtained from Haar-random $U$ by shifting each internal parameter by a small random amount is still typically random.

The main theoretical challenges of BosonSampling are those related to simplifying experimental efforts, such as a depth-wise optimization of the model, and those of a more fundamental nature, such as device certification---or, in other words, an answer to the question ``I have a BosonSampling device, now how do I prove it?''. Note that universal quantum computation also presents similar certification issues: any problem in BQP $\cap$ NP of course can be verified in polynomial time, but it is still an open problem whether any task outside this intersection---including any BQP-complete problem---has an efficient procedure that convinces a purely classical verifier of its correctness. Even leaving aside the usual ``paranoia'' of complexity theory, where a BosonSampling device would be asked to prove its authenticity against a miriad of conceivable ``classical'' distributions, it would already be an impressive result to show that BosonSampling can be distinguished, in polynomial time, from any classical distribution natural to the experimental apparatus, such as e.g.\ the distributions generated by distinguishable photons. This might also serve as a partial replacement for tomography if the experimental data showed that, despite imperfections in the fabrication of the device, the output distribution is still close enough to the ideal one as to make it distinguishable from other suitable candidates. This, in a sense, is what we did with the 9-mode interferometer of \sec{validation}: regardless of the fact that we did not perform the reconstruction of the unitary since we did not have all two-photon data, we can deduce that our fabricated chip must have been sufficiently close to the ideal one since the statistical tests were successful.

To conclude this chapter, I would like to say that even standard quantum computation has its skeptics. However, since universal quantum computers have the power of fault-tolerance and error correction, a serious skeptic must come up with quite artificial models to explain why some fundamental physical principle disallows quantum computing whilst allowing the plethora of quantum-mechanical phenomena observed in the last several decades. BosonSampling, unfortunately, does not share this feature, as it does not have (yet) a notion of fault-tolerance. It will be exciting to see, in the future, whether BosonSampling follows in the footsteps of universal quantum computers, where increasing theoretical progress makes skeptics' jobs harder every day, or if it will follow the unsuccessful path of classical analog computers, where physical imperfections necessarily wash out any supra-classical computational speedup.

\newpage
\chapter{Final Remarks} \label{chapter:conclusions}

There are several fascinating aspects that drive much of the current research in quantum computation and information. On the practical side, quantum computers hold the promise to outperform classical computers in several tasks of interest, the most notable of which is integer factoring. Given that the hardness-of-factoring assumption is the basis for one of the most widely used public-key cryptosystems of our days, this is naturally a very strong driving force for development of the field, both theoretically and experimentally. However, practical large-scale quantum computing remains safely in the distant future, and so another strong motivation is of a more fundamental nature. Quantum computing is not the first case of a connection between Physics and Computer Science, but it is one of the most compelling, and that has gathered the greatest amount of attention. Physics and Computer Science originated within different kinds of human activities and developed with very distinct purposes, and the whole concept of quantum computing seems to suggest an \textit{a posteriori} connection between them that is surprisingly robust: there are several different routes (i.e.\ models) that provide universal quantum computation, but any mechanism that reduces a universal quantum system to a classical regime (defined on reasonable physical grounds), such as e.g.\ decoherence, seems to enact a corresponding reduction in the computational power. This did not need to be the case: if a universal quantum computer retained its power when driven to a classical regime by some physical mechanism, this would provide a ``quantum proof'' that certain computational problems, such as factoring, are in fact classically easy. In my opinion, one of the most notable examples of such a connection is the difference in the computational complexity of simulating the dynamics of free fermions and free bosons. An experimental setup corresponding to arguably the simplest nonclassical bosonic/fermionic behavior---the particles are only required to evolve according to identical particle statistics defined by the symmetry of the wave function, with no need for controllable interactions---provides natural computational tasks that have a sub-classical (fermions) and a supra-classical (bosons) simulation complexity.

In this thesis, I reported several results that extend the known regimes of computational power for the two restricted models corresponding to free-particle dynamics: fermions (matchgates) and bosons (BosonSampling). 

In \chap{fermnew} I studied the computational power of matchgates. Matchgates, when acting on nearest-neighboring qubits aligned on a path, correspond to free fermions and are classically simulable. I investigated how their computational power changes when we modify the underlying restrictions, moving the model away from a fermionic description. I showed that adding any parity-preserving two qubit gate (that is not a matchgate) to the set suffices to bridge the gap between (sub-)classical and quantum computational power. Alternatively, I showed that this can also be achieved by almost any change in the connectivity: matchgates acting on nearest neighbors on a path or cycle correspond to free fermion dynamics and are classically simulable, whereas matchgates acting on any other graph are universal for quantum computing and no longer correspond to free fermions. Notice how this echoes a claim made in the previous paragraph: it was conceivable that matchgates would become universal for quantum computation acting on a given graph, whilst not disrupting their efficient classical simulation scheme. The fact that this does not happen provides additional evidence that the distinction between quantum and classical regimes, in terms of computational power, is meaningful. I also showed that the exact same dichotomy holds for the XY interaction. In some sense, the XY interaction already captures all the computational power of the complete set of matchgates, despite explicitly being a proper subset. 

A list of more technical open questions on this subject and relations to other work can be found in \sec{fermnew_c}. At a high level, the main open questions we leave are more quantitative than qualitative. We completely characterized the computational regimes of matchgates in arbitrary graphs, but we have not directed our efforts to determining the more efficient way to perform a given quantum computation (i.e., that is more robust to noise, uses less additional operations and/or qubits, etc). Such an effort might be of interest to experimentalists working with physical systems where matchgates arise naturally. Most proposals in this category that I am aware of are over a decade old \cite{Imamoglu1999, Quiroga1999, Zheng2000, Mozyrsky2001}, and I am not aware how they compare, in terms of feasibility, with other well-known implementations. Nonetheless, during the writing of this thesis a new proposal arose \cite{Herrera2014} for matchgate quantum computing using polar molecules trapped in optical lattices. I hope that the results contained here may provide an extensive flexibility to these constructions, especially as they show that only a rudimentary control over the geometrical arrangement of the qubits might already suffice for a nontrivial computation. Finally, our results also have some relation to qualitatively different quantum computing models, such as e.g.\ ancilla-controlled quantum computing \cite{Proctor2013}, and they might suggest improvements for these known models, or even lead to completely new ones.

In \chap{bosonnew} I presented several results, both experimental and theoretical, on the BosonSampling model of quantum computation. BosonSampling is a restricted model that is not expected to be universal for quantum computation---as a matter of fact, it possesses no practical application so far, and there are no known decision problems solvable by this model (that are not also solvable classically). Nonetheless, it has been gathering increasing attention, mostly due to how it lends itself so naturally to experimental implementation. Rather than imposing a (quite artificial) qubit structure on the photon states, followed by a yet-prohibitively-large complex network of optical elements, auxiliary photons, and adaptive measurements, which make up the KLM scheme, a BosonSampling device consists of nothing more than observing photons in their ``natural habitat''. By inputting an $n$-photon, $m$-mode Fock state into a random garden-variety linear interferometer and measuring the output distribution over the Fock basis, one is already performing a task that is strongly conjectured to be classically hard \cite{Aaronson2013a}. Furthermore, the running time of the best known classical algorithms to perform this task grows very fast with $n$---an experiment with $n$ between 50 and 100 would already surpass the capacities of modern day classical computers. As such, it is extremely natural that great efforts be focused on this model, since interesting computational regimes will be experimentally feasible for BosonSampling much earlier than for other quantum computing tasks, such as factoring. 

My main theoretical contribution to this area is a proof that exact BosonSampling is hard even if the optical circuit has only constant depth, under similar complexity assumptions as standard BosonSampling. I believe this may be an important first step in further simplifying the experimental requirements for a convincing implementation of the model. For this to actually unfold, it will be necessary to show that the hardness of the \emph{approximate} BosonSampling in the constant-depth regime also reduces to some natural hardness conjecture\footnote{This can turn out to not be the case: constant-depth BosonSampling might be hard to perform exactly, but easy to perform approximately. This would also be a very interesting result.}. From a more conceptual point of view, this result also connects two different restricted models, namely constant-depth quantum computing and constant-depth BosonSampling, by showing that the latter is contained in the former. This containment is not known to be true for standard BosonSampling, and could be used to translate results between the models (e.g., if approximate constant-depth BosonSampling is also hard, this might suggest a similar proof for approximate simulation of constant-depth quantum circuits, providing the first proof that other restricted models are as robust as BosonSampling).

I have also contributed to experiments that were among the first small-scale implementations of BosonSampling using integrated photonics, as part of an ongoing collaboration with quantum optics groups in Rome and Milan. In these experiments, we observed three-photon interference in chips of 5 to 16 modes, and a good agreement with the description of bosonic dynamics in terms of matrix permanents, which is at the heart of the BosonSampling result. These experiments also showed that fabrication techniques allow a high level of control over the chips' parameters (transmissivities and phase shifts). This is fundamental not only for progressively larger BosonSampling experiments, but also benchmarks the experimental techniques for use in other quantum information processing tasks, and ultimately in a large-scale implementation of the KLM scheme. My main contribution to these experiments was essentially theoretical: I helped bridge the gap between the highly-technical BosonSampling paper with the real-world available experimental setups; I provided simulations of the behavior of interferometer ensembles inspired on experimental limitations; and I helped to numerically improve known techniques for interferometer tomography. 

A more detailed discussion on these experiments and the main theoretical and experimental open questions regarding BosonSampling can be found in \sec{bosonnew_d}. In the big picture, however, I believe that BosonSampling has, today, a similar status that quantum computing had in the beginning of the '90s. It is not very clear what practical applications it may have, and it still has several ``loopholes'', so to speak, that raise a lot of skepticism\footnote{With the added disadvantage that techniques that have been extensively developed to deal with similar loopholes of quantum computing, such as error correction, cannot be applied directly to BosonSampling.}. The most serious issues, in my opinion, are: (i) the serious scalability issues of current experimental implementations, (ii) the lack of a practical application for BosonSampling, other than the (arguably) purely academic one of providing evidence that quantum systems can outperform classical ones, (iii) a lack of a purely linear-optical fault-tolerant construction, or a satisfying argument that one is unnecessary, and (iv) the elusive matter of efficient verification of the device. 

Issue (i) can be closed with a combination of theoretical and experimental efforts. Experimentally, one might perfect techniques for deterministic photon generation and detection, develop fabrication procedures that allow integration of sources and detectors into the photonic device, reduce losses due to circuit depth, etc. Theoretically, one could improve the known results to show that they hold even in simpler experimental setups. Examples might include the constant-depth regime mentioned previously, or considering the complexity of BosonSampling with variable inputs\footnote{We mentioned an example in \sec{bosonnew_d}, which was called Scattershot BosonSampling. This proposal would require many more photon sources, which might mean that it will not be feasible in the near future, at least in its current preliminary form, but it shows that BosonSampling can be scaled up efficiently using probabilistic sources.}. Issues (ii) to (iv), on the other hand, involve mostly theoretical research. Of greater importance (and corresponding challenge) is (iv), since the lack of an efficient validation procedure allows a skeptic to permanently question whether a real-world device is in fact performing any nontrivial task (in fact, this issue might turn out to be fundamentally unresolvable, as it is not ruled out that, for any BosonSampling distribution, there exist some classical distribution that is indistinguishable from it in polynomial time). Note that standard quantum computing does suffer from this problem, but to a lesser extent. Any quantum algorithm that solves some NP problem, such as factoring, can be efficiently verified (by definition), but any task that is outside of NP, in principle, cannot. An interesting partial solution to this problem was given by \cite{Aharonov2008a, Broadbent2009}, who showed that an arbitrary untrustworthy quantum computer can be verified efficiently if the verifier has access to a small trusted quantum resource. A similar result for BosonSampling is not ruled out, for example in a setting where the verifier uses a small linear-optical resource to verify the functioning of an arbitrarily larger device. However, it is unclear how such a protocol would work within the BosonSampling model itself, since introducing intermediate measurements and adaption already takes the model all the way to universal quantum computing via the KLM scheme.

Finally, there are other interesting questions that are suggested by our work, which however are not direct extensions of the questions answered here. For example, it would be interesting to investigate the computational power associated with noninteracting exotic particles. In our real (three-dimensional) world, bosons and fermions are the only two kinds of fundamental particles, corresponding to a symmetry/antisymmetry of the wave function, respectively. However, other possibilities arise as quasi-particles in two-dimensional systems, such as Abelian anyons (where an interchange of two particles induces a complex phase, rather than just a $\pm 1$), which are the basis for the model of topological quantum computation, and the more abstract generalization known as quons \cite{Greenberg1991, Bardek1994}, which provide several ways to continuously interpolate between bosons and fermions. Although this sounds like an extremely daunting challenge, it would be interesting to see how the inherent simulation complexity of these systems change as we interpolate between fermions (classically easy) and bosons (classically hard), assuming that this question is even completely well-defined. I conjecture that fermions are a pathological case of trivial computational power, since the experience with matchgates showed that very small deviations from pure free-fermionic behavior along several different directions already disrupts the classical simulation.

As a final thought, I point out that all aspects discussed here follow the connection from Physics to Computer Science by using the former to shed light on the latter. However, if there really is a connection that runs deeply between the two fields, one may also expect Computer Science to provide insights into physical phenomena. Quantum computing skeptics may suggest that any connection is illusory: Physics studies natural phenomena, while Computer Science is concerned with artificial constructs that arise mostly from the human mind, and it is arrogant to believe that any apparent connection between them is not simply an artifact of the way humans describe the world around them. The verdict for these philosophical lines is still open, but there is no question that quantum computing has deeply impacted our way of describing and thinking about the physical world. Quantum mechanics is a theory of a highly counterintuitive nature, so much so that, regardless of impressive experimental confirmations, its interpretation remains subject of a heated debate that has lasted almost a century. In this sense, new ways of describing the quantum world in terms of computational and information-theoretic constructs might just be a welcome addition.
\newpage

\begin{appendices}
\chapter{Simulation of random phases ensemble} \label{appdx:app1}

This Appendix contains the Mathematica$^\copyright$ code for the simulation of the random phases ensemble and comparison with the Haar ensemble, described in \sec{bosonnew_b_ensembles}. This code was written in Mathematica 9.0, Student Edition. 

\noindent \textbf{Notebook for simulation of random phases ensemble}

{\small \noindent \textbf{Preliminary definitions}}
\begin{lstlisting}
(* Definitions of Permanent, Beam splitters and phase shifters *)
Permanent3x3[U_]:=U[[1,1]] U[[2,2]] U[[3,3]]+U[[1,2]]U[[2,3]]U[[3,1]]+U[[1,3]]U[[2,1]]U[[3,2]]+U[[1,3]] U[[2,2]] U[[3,1]]+U[[1,2]]U[[2,1]]U[[3,3]]+U[[1,1]]U[[2,3]]U[[3,2]];
beamSplitter[k_,m_]:=Normal[SparseArray[{{k,k}->1/Sqrt[2],{m,k}->I/Sqrt[2],{k,m}->I/Sqrt[2],{m,m}->1/Sqrt[2] ,{l_,l_}-> 1}, {size,size},0]];
phaseShifter[\[Phi]_,m_]:=Normal[SparseArray[{{m,m}->Exp[I \[Phi]],{l_,l_}-> 1}, {size,size},0]];

(* Defining optical elements as pairs of balanced beamsplitter plus two random phase shifters in [0,\[Pi]] range. *)
optEle[{k_,m_}]:=beamSplitter[k,m].phaseShifter[RandomReal[\[Pi]],k].phaseShifter[RandomReal[\[Pi]],m];
optEle[{k_}]:=phaseShifter[RandomReal[\[Pi]],k];

(* Routine to sample Haar-random unitary RU[n] *)
RR:=RandomReal[NormalDistribution[0,1]];
RC:=RR+I*RR;
RG[n_]:=Table[RC,{n},{n}];
RU[n_]:=Module[{Q,R,r,L},{Q,R}=QRDecomposition[RG[n]];
r=Diagonal[R];
L=DiagonalMatrix[r/Abs[r]];
Q.L]; 

(* UPS[S,L] = Unitary of random phases with S modes and L layers *)
PSLayer1[S_]:=Block[{size},size=S;Dot@@Map[optEle,Partition[Range[size,1,-1],2,2,-1,{}]]]
PSLayer2[S_]:=Block[{size},size=S;Dot@@Map[optEle,Partition[Range[size,1,-1],2,2,1,{}]]]
UPS[S_,L_]:=Dot@@(Extract[{A,B},Reverse[Partition[-(-1)^Range[L],1]]])/.{A:>PSLayer1[S],B:>PSLayer2[S]}
\end{lstlisting}
{\small \noindent \textbf{Output probabilities}}
\begin{lstlisting}
(* Definitions for calculation of probabilities. *)
subU[U_,out_List,IN_List]:=U[[out,IN]]
probCl[U_,out_List,IN_List]:=Permanent3x3[Abs[subU[U,out,IN]]^2]/(Times@@((Tally[out]/.{_,x_}->x)!))
probQu[U_,out_List,IN_List]:=Abs[Permanent3x3[subU[U,out,IN]]]^2/(Times@@((Tally[out]/.{_,x_}->x)!))
\end{lstlisting}
{\small \noindent \textbf{Simulations}}
\begin{lstlisting}
(* runs determines number of sampled matrices each size and depth *)
(* {dmin, dmax, dstep}, define depth values to be simulated, {smin, smax, sstep} does the same for sizes *)
(* Note that size and Size are two different variables used in internal working of the code. *)
runs=10000;
{dmin,dmax,dstep}={(Size+3)/2,(Size+3)/2+7,2};{smin,smax,sstep}={7,11,4};
SeedRandom[]

(* Samples 'runs' Haar-random matrices of sizes {smin,smax,sstep}. *)
(* For each, samples one matrix from random phase ensemble of each depth {dmin,dmax,dstep}. *)
(* Computed quantities as follows
RePartHaarmidmid and RePartPSmidmid: matrix element at position {central, central};
RePartHaar1mid and RePartPS1mid: matrix elements at position {central, endpoint};
Per3x3Haarmid and Per3x3PSmid: permanent of 3x3 submatrix of 3 central inputs and 3 central outputs;
Per3x3Haarend and Per3x3PSend: permanent of 3x3 submatrix of 3 central inputs and 3 endpoint outputs;
BunchHaar and BunchWalk: Bunching fractions. *)

Do[Do[

Haar=RU[Size];

Haarmid=subU[Haar,{(Size-1)/2,(Size+1)/2,(Size+3)/2},{(Size-1)/2,(Size+1)/2,(Size+3)/2}];
Haarend=subU[Haar,{1,2,3},{(Size-1)/2,(Size+1)/2,(Size+3)/2}];
RePartHaarmidmid[Size,t]=Re[Haar[[(Size+1)/2,(Size+1)/2]]];
RePartHaar1mid[Size,t]=Re[Haar[[1,(Size+1)/2]]];
Per3x3Haarmid[Size,t]=Abs[Permanent3x3[Haarmid]]^2;
Per3x3Haarend[Size,t]=Abs[Permanent3x3[Haarend]]^2;
BunchHaar[Size,t]=1-Total[Flatten[Table[probQu[Haar,{i,j,k},{(Size-1)/2,(Size+1)/2,(Size+3)/2}],{i,1,Size},{j,i+1,Size},{k,j+1,Size}]]];

Do[
WalkPS=UPS[Size,Dept];

WalkPSmid=subU[WalkPS,{(Size-1)/2,(Size+1)/2,(Size+3)/2},{(Size-1)/2,(Size+1)/2,(Size+3)/2}];
WalkPSend=subU[WalkPS,{1,2,3},{(Size-1)/2,(Size+1)/2,(Size+3)/2}];
RePartPSmidmid[Size,Dept,t]=Re[WalkPS[[(Size+1)/2,(Size+1)/2]]];
RePartPS1mid[Size,Dept,t]=Re[WalkPS[[1,(Size+1)/2]]];
Per3x3PSmid[Size,Dept,t]=Abs[Permanent3x3[WalkPSmid]]^2;
Per3x3PSend[Size,Dept,t]=Abs[Permanent3x3[WalkPSend]]^2;
BunchWalk[Size,Dept,t]=1-Total[Flatten[Table[probQu[WalkPS,{i,j,k},{(Size-1)/2,(Size+1)/2,(Size+3)/2}],{i,1,Size},{j,i+1,Size},{k,j+1,Size}]]];

,{Dept,dmin,dmax,dstep}]

,{Size,smin,smax,sstep}],{t,1,runs}]
\end{lstlisting}
{\small \noindent \textbf{Histograms}}
\begin{lstlisting}
(* Exports histograms in pdf files as found in Thesis *)
SetDirectory[NotebookDirectory[]]
Do[
Print[Style["***"<>ToString[Size]<>"MODES***","Subtitle",TextAlignment->Center]];
Export["walk"<>ToString[Size]<>"a.pdf",GraphicsRow[Table[SmoothHistogram[{Table[RePartPSmidmid[Size,Dept,t],{t,1,runs}],Table[RePartHaarmidmid[Size,t],{t,1,runs}]},PlotLabel->Style["Re part of element "<>ToString[{Size/2+1/2,Size/2+1/2}]<>", \n m="<>ToString[m=Size]<>" and L="<>ToString[Dept],30,Bold],ImageSize->600,Frame->True,FrameTicks-> {{{0,0.5,1,1.5},None},{{-0.5,0,0.5},None}},LabelStyle->{30,Bold},Filling->Axis,AxesOrigin->{-1,0},FrameStyle->Thick],{Dept,dmin,dmax,dstep}],-120]];
Export["walk"<>ToString[Size]<>"b.pdf",GraphicsRow[Table[SmoothHistogram[{Table[RePartPS1mid[Size,Dept,t],{t,1,runs}],Table[RePartHaar1mid[Size,t],{t,1,runs}]},PlotLabel->Style["Re part of element "<>ToString[{Size/2+1/2,1}]<>", \n m="<>ToString[m=Size]<>" and L="<>ToString[Dept],30,Bold],ImageSize->600,Frame->True,FrameTicks-> {{{0,0.5,1,1.5},None},{{-0.5,0,0.5},None}},LabelStyle->{30,Bold},Filling->Axis,AxesOrigin->{-1,0},FrameStyle->Thick],{Dept,dmin,dmax,dstep}],-120]];
Export["walk"<>ToString[Size]<>"c.pdf",GraphicsRow[Table[SmoothHistogram[{Table[Per3x3PSmid[Size,Dept,t],{t,1,runs}],Table[Per3x3Haarmid[Size,t],{t,1,runs}]},PlotLabel->Style["Probability of output "<>ToString[{(Size-1)/2,(Size+1)/2,(Size+3)/2}]<>" \n m="<>ToString[m=Size]<>" and L="<>ToString[Dept],30,Bold],ImageSize->600,FrameTicks-> {{(1-(7-Size)/2){0,30,60,90},None},{(1/2-(Size-7)/16){0,0.04,0.08,0.12,0.16},None}},Frame->True,FrameTicks->Automatic,LabelStyle->{30,Bold},Filling->Axis,AxesOrigin->{-1,0},PlotRange->{{0,0.03+(11-Size)*0.05/4},All},FrameStyle->Thick ],{Dept,dmin,dmax,dstep}],-120]];
Export["walk"<>ToString[Size]<>"d.pdf",
GraphicsRow[Table[SmoothHistogram[{Table[Per3x3PSend[Size,Dept,t],{t,1,runs}],Table[Per3x3Haarend[Size,t],{t,1,runs}]},PlotLabel->Style["Probability of output "<>ToString[{1,2,3}]<>" \n m="<>ToString[m=Size]<>" and L="<>ToString[Dept],30,Bold],ImageSize->600,Frame->True,LabelStyle->{30,Bold},FrameTicks-> {{(2+4(Size-7)){0,15,30,45,120,150,180},None},{(1-3(Size-7)/16){0,0.04,0.08,0.12,0.16},None}},Filling->Axis,AxesOrigin->{-1,0},PlotRange->{{0,0.03+(11-Size)*0.05/4},All},FrameStyle->Thick ],{Dept,dmin,dmax,dstep}],-120]];
Export["walk"<>ToString[Size]<>"e.pdf",GraphicsRow[Table[SmoothHistogram[{Table[BunchWalk[Size,Dept,t],{t,1,runs}],Table[BunchHaar[Size,t],{t,1,runs}]},FrameTicks-> {{{0,3,6},None},{{0,0.3,0.5,0.7},None}},PlotLabel->Style["Bunching fraction \n m="<>ToString[m=Size]<>" and L="<>ToString[Dept],30,Bold],ImageSize->600,Frame->True,LabelStyle->{30,Bold},Filling->Axis,AxesOrigin->{-1,0},FrameStyle->Thick],{Dept,dmin,dmax,dstep}],-120]];
,{Size,smin,smax,sstep}]
\end{lstlisting}
\chapter{Reconstruction algorithm} \label{appdx:app2}

This Appendix contains the Mathematica$^\copyright$ code for the numerical refinement of the Laing-O'Brien algorithm, described in \sec{bosonnew_b_tomography}. This code was written in Mathematica 9.0, Student Edition. 

\noindent \textbf{Notebook for Reconstruction Algorithm}

{\small \noindent \textbf{Preliminary definitions}}
\begin{lstlisting}
(* Number of modes of the chip *)
Size=7;

(* Definition of mxm permanent for m\[GreaterEqual]3 *)
Permanent[m_List]:=With[{v=Array[x,Length[m]]},Coefficient[Times@@(m.v),Times@@v]] 

(* More efficient definition of 2x2 permanent *)
Permanent2x2[U_]:=U[[1,1]]*U[[2,2]]+U[[2,1]]*U[[1,2]] 

(* Preliminary definitions for calculation of gate fidelity *)
(* maximized over round of input and output phase shifters and symmetry U \[Rule] U\[Conjugate] *)
PhaseShifter[A_,m_]:=Normal[SparseArray[{{m,m}->Exp[I A],{l_,l_}-> 1}, {Size,Size},0]]; 

ArbPhase[U0_,\[Alpha]_List,\[Beta]_List]/;(Length[\[Alpha]]==Size&&Length[\[Beta]]==Size):=Chop[(Dot@@(PhaseShifter[\[Alpha][[#]],#]&/@Range[Size])).U0.(Dot@@(PhaseShifter[\[Beta][[#]],#]&/@Range[Size]))];
GateFid[U_,V_]:=Chop[Abs[Tr[U.V\[ConjugateTranspose]]/Size]];
MaxGateFid[U_,V_]:=Max[(First[FindMaximum[GateFid[U,ArbPhase[#,\[Alpha]/@Range[Size],\[Beta]/@Range[Size]]],Flatten[{\[Alpha]/@Range[Size],\[Beta]/@Range[Size]}]]]&)/@{V,V\[Conjugate]}]//Quiet 

(* Definition of \[Chi]2 *) 
Chi2[U_]:=Total[(TeoSingle[U]-ExpSingle)^2/ErrorsSingle^2]+Total[(TeoVisib[U]-ExpVisib)^2((1)/(ErrorsVisib)^2)];
Chi2Single[U_]:=Total[(TeoSingle[U]-ExpSingle)^2/ErrorsSingle^2];
Chi2Visib[U_]:=Total[(TeoVisib[U]-ExpVisib)^2((1)/(ErrorsVisib)^2)];

(* Definition of Total Variation Distances *)
TVDSingle[U_]:=1/2 Total[Flatten[Abs[RExp-RTeo[U]]]]/Size
TVDVisib[U_]:=1/2 Total[Flatten[Table[Total[Abs[ExpProb2[i,j]
-TeoProb2[U,i,j]]],{i,2,Size},{j,1,i-1}]]]/(Size*(Size-1)/2)
\end{lstlisting}
{\small \noindent \textbf{Importing External data}}
\begin{lstlisting}
SetDirectory[StringJoin[NotebookDirectory[],"\\****"]];

(* Importing of ideal unitary from separate .dat files with real and imaginary parts*)
ImUIdeal=ToExpression[Import["****.dat"]]; 
ReUIdeal=ToExpression[Import["****.dat"]];
UIdeal=ReUIdeal+I ImUIdeal;

(* Importing of experimental data from .dat files *)
RExp=Import["****.dat","Table"]; (* Single-photon data. *)
\[Sigma]RExp=Import["****.dat","Table"]; (* Single-photon error bars. *)
VIMPORT=Import[StringJoin["****.dat"],"Table"]; (* Two-photon data. *)
\[Sigma]VIMPORT=Import[StringJoin["****.dat"],"Table"]; (* Two-photon error bars. *)
VExp=Table[0,{j,1,Size},{g,1,Size},{h,1,Size},{k,1,Size}];
\[Sigma]VExp=Table[0,{j,1,Size},{g,1,Size},{h,1,Size},{k,1,Size}];

(* Redefinition of visibilities indexed as for use in Laing-O'Brien algorithm *)
ROW=0;
COLUMN=0;
Do[Do[++COLUMN;ROW=0;
Do[Do[++ROW;
VExp[[k,h,j,g]]=VExp[[h,k,j,g]]=VExp[[k,h,g,j]]=VExp[[h,k,g,j]]=VIMPORT[[ROW,COLUMN]];
\[Sigma]VExp[[k,h,j,g]]=\[Sigma]VExp[[h,k,j,g]]=\[Sigma]VExp[[k,h,g,j]]=\[Sigma]VExp[[h,k,g,j]]=\[Sigma]VIMPORT[[ROW,COLUMN]];
,{k,h+1,Size}],{h,1,Size-1}],{g,j+1,Size}],{j,1,Size-1}];

(* Definition of theoretical variables from unitary matrix U *)
RTeo[U_]:=Abs[U]^2 (* Matrix that stores classical light amplitudes  *)
QTeo[U_,k_,h_,j_,g_]:=Abs[Det[{{1,-1},{1,1}}*U[[{j,g},{k,h}]]]]^2  (* Quantum probabilities for 2 photons entering at k,h and exiting at g,j *)
ClTeo[U_,k_,h_,j_,g_]:=Det[{{1,-1},{1,1}}*Abs[U[[{j,g},{k,h}]]]^2 ];   (* Classical probabilities for 2 photons entering at k,h and exiting at g,j *)
VTeo[U_,k_,h_,j_,g_]:=Chop[1-(Abs[Det[{{1,-1},{1,1}}*U[[{j,g},{k,h}]]]]^2)/(Det[{{1,-1},{1,1}}*Abs[U[[{j,g},{k,h}]]]^2 ]+10^(-20))];(* Two-photon visibilities *)

(* Definition of auxiliary vectors containing the complete data set. *)
ExpSingle=Flatten[RExp]; (* 1D Vector with experimental single-photon data *)
ErrorsSingle=Flatten[\[Sigma]RExp];  (* 1D Vector with experimental single-photon error bars *)
ExpVisib=Flatten[Table[Table[Table[Table[VExp[[k,h,j,g]],{k,h+1,Size}],{h,1,Size-1}],{g,j+1,Size}],{j,1,Size-1}]]; 

(* 1D Vector with experimental visibilities *)
ErrorsVisib=Flatten[Table[Table[Table[Table[\[Sigma]VExp[[k,h,j,g]],{k,h+1,Size}],{h,1,Size-1}],{g,j+1,Size}],{j,1,Size-1}]]; 

(* 1D Vector with visibilitiy error bars *)
ExpProb2[k_,h_]:=Flatten[Table[Table[(1-VExp[[k,h,j,g]])*(RExp[[j,k]]*RExp[[g,h]]+RExp[[j,h]]*RExp[[g,k]]),{g,j+1,Size}],{j,1,Size-1}]]; (* 1D Vector with experimental 2-photon probabilities *)
TeoProb2[U_,k_,h_]:=Flatten[Table[Table[QTeo[U,k,h,j,g],{g,j+1,Size}],{j,1,Size-1}]];  (* 1D Vector with theoretical 2-photon probabilities *)
TeoSingle[U_]:=Flatten[RTeo[U]];  (* 1D Vector with theoretical single-photon probabilities *)
TeoVisibSuppt=Table[Table[Table[Table[{{j,g},{k,h}},{h,k+1,Size}],{k,1,Size-1}],{g,j+1,Size}],{j,1,Size-1}];
TeoVisib[U_]:=Flatten[Map[1-(Abs[Permanent2x2[U[[Sequence@@#]]]]^2)/(10^(-20)+Permanent2x2[Abs[U[[Sequence@@#]]]^2 ])&,TeoVisibSuppt,{-3}]]; (* 1D Vector with theoretical visibilities *)
\end{lstlisting}
{\small \noindent \textbf{Laing-O'Brien Subroutine}}
\begin{lstlisting}
(* Vis and Coup are the matrices with visibilities and single-photon data.*)
(* Set REFj and REFk to use row REFj+1 and column REFk+1 as references *)

LaingOBrien[REFj_,REFk_,Vis_,Coup_]:=Block[{V,R,x,y,\[Beta],SN,T,A,A0,\[Tau],\[Alpha],\[Alpha]0,M2,M1,M,u,v,w,M0},

(* Definition of auxiliary variables *)
V[k_,h_,j_,g_]:=Vis[[Mod[k+REFk,Size,1],Mod[h+REFk,Size,1],Mod[j+REFj,Size,1],Mod[g+REFj,Size,1]]];
 R[j_,k_]:=Coup[[Mod[j+REFj,Size,1],Mod[k+REFk,Size,1]]];
x[k_,h_,j_,g_]:=Sqrt[R[j,k]R[g,h]/(R[j,h]R[g,k])];
y[k_,h_,j_,g_]:=x[k,h,j,g]+1/(x[k,h,j,g]);
\[Beta][k_,h_,j_,g_]:=Abs[ArcCos[-V[k,h,j,g]y[k,h,j,g]/2]];
SN[k_,h_,j_,g_]:=Sign[Abs[\[Beta][k,h,j,g]-Abs[Mod[A[[j,k]]-A[[j,h]]-A[[g,k]]-A0[[g,h]],2\[Pi],-\[Pi]]]]-Abs[\[Beta][k,h,j,g]-Abs[Mod[A[[j,k]]-A[[j,h]]-A[[g,k]]+A0[[g,h]],2\[Pi],-\[Pi]]]]];
T=Table[\[Tau][i,j],{i,1,Size},{j,1,Size}];
A=Table[\[Alpha][i,j],{i,1,Size},{j,1,Size}];
A0=Table[\[Alpha]0[i,j],{i,1,Size},{j,1,Size}];
A[[All,1]]=A[[1,All]]=0^Range[Size];

(* Sequential replacement of absolute values and phases of matrix elements *)
A0[[2;;Size,2;;Size]]=Table[Abs[ArcCos[-V[1,h,1,g]y[1,h,1,g]/2]],{g,2,Size},{h,2,Size}];
A[[2,2]]=A0[[2,2]];
A[[3;;Size,2]]=Table[SN[1,2,2,m]A0[[m,2]],{m,3,Size}];
A[[2,3;;Size]]=Table[SN[2,m,1,2]A0[[2,m]],{m,3,Size}];
A[[3;;Size,3;;Size]]=Table[SN[2,h,2,g]A0[[g,h]],{g,3,Size},{h,3,Size}];
M2=Simplify[Table[x[1,j,1,i],{i,1,Size},{j,1,Size}]]*Exp[I A];

(* Linear system to solve for absolute values of elements in reference row and column *)
T[[All,1]]=Sqrt[LinearSolve[M2\[ConjugateTranspose],UnitVector[Size,1]]];
T[[1,All]]=Sqrt[LinearSolve[M2,UnitVector[Size,1]]];
T[[2;;Size,2;;Size]]=Table[x[1,h,1,g]T[[1,h]]T[[g,1]]/T[[1,1]],{g,2,Size},{h,2,Size}];
M1=Simplify[T*Exp[I A]];
M0=Chop[M1[[Mod[#1-REFj,Size,1]&/@Range[Size],Mod[#1-REFk,Size,1]&/@Range[Size]]]];

(* Singular Value Decomposition *)
{u,w,v}=SingularValueDecomposition[M0];
MTomo[REFj,REFk]=u.ConjugateTranspose[v];

(* Outputs a vetor with \[Chi]2, Gate fidelity and permutation coordinates. Also stores reconstructed matrix as MTomo[REFj,REFk]. *)
{Chi2[MTomo[REFj,REFk]],MaxGateFid[MTomo[REFj,REFk],UIdeal],{REFj,REFk}}]
\end{lstlisting}
{\small \noindent \textbf{Permutations of Laing-O'Brien}}
\begin{lstlisting}
(* All matrices output by the Laing-O'Brien algorithm, and outputs the best one in terms of \[Chi]2. *) 
Recordist=Sort[Flatten[Table[LaingOBrien[i,j,VExp,RExp],{i,0,Size},{j,0,Size}],1]][[1]]

(* Defines the starting point for the numerical search. *)
URecon=MTomo[Sequence@@(Recordist[[3]])]; 
BestChiSq=Chi2[URecon]
BestGFid=MaxGateFid[URecon,UIdeal]
\end{lstlisting}
{\small \noindent \textbf{Numerical Search}}
\begin{lstlisting}
(* RUNS encodes the number of runs of the algorithm. *)
(* STEP enables fine-tuning when the algorithm saturates. *)
RUNS=1000;
STEP=1;
Do[
Clear[MTomo,RSam,VSam,BestIndex,BestTomo];

(* Samples new set of single- and two-qubit data from Gaussian distribution centered on previous values. *)
SeedRandom[];
RSam=((#/Total[#]&)/@((RTeo[URecon]+Map[RandomReal[NormalDistribution[0,#/STEP]]&,\[Sigma]RExp,{2}])\[Transpose]))\[Transpose];
VSam=Table[0,{j,1,Size},{g,1,Size},{h,1,Size},{k,1,Size}];
Do[Do[Do[Do[
VSam[[k,h,j,g]]=VSam[[h,k,j,g]]=VSam[[k,h,g,j]]=VSam[[h,k,g,j]]=RandomReal[NormalDistribution[VTeo[URecon,k,h,j,g],\[Sigma]VExp[[k,h,j,g]]/STEP]];
,{k,h+1,Size}],{h,1,Size-1}],{g,j+1,Size}],{j,1,Size-1}];

(* Laing-O'Brien algorithm, with permutations, on simulated data set. *)
BestIndex=Sort[Flatten[Table[LaingOBrien[i,j,VSam,RSam],{i,0,Size},{j,0,Size}],1]][[1]];

(* If the \[Chi]2 improved, redefine all variables, print updated values as {\[Chi]2, Fidelity} and start again.*) 
(* Otherwise just start again. *)
If[BestIndex[[1]]<BestChiSq,

BestTomo=MTomo[Sequence@@BestIndex[[3]]];
BestGFid=MaxGateFid[BestTomo,UIdeal];
BestChiSq=BestIndex[[1]];
URecon=BestTomo;
Print[{BestChiSq,BestGFid}]];
,{RUN,1,RUNS}]
\end{lstlisting}
\chapter{Sampled and reconstructed unitaries} \label{appdx:app3}

\section{5-mode chip}

The matrix $U_{5}^{t}$ was sampled from the uniform, Haar distribution over $5 \times 5$ unitary matrices: 
\scriptsize
\begin{equation*}
U_{5}^{t} = \begin{pmatrix}
 0.212 & -0.018+0.165 i & -0.238-0.18 i & -0.429+0.32 i & -0.715+0.2 i \\
 -0.193-0.388 i & -0.045-0.379 i & 0.19 +0.311 i & 0.328 -0.269 i & -0.594+0.03 i \\
 -0.723+0.363 i & 0.087 -0.09 i & -0.076-0.155 i & 0.206 +0.443 i & -0.153-0.193 i \\
 -0.092+0.045 i & -0.148-0.645 i & -0.588+0.184 i & -0.369-0.086 i & 0.167 +0.025 i \\
 0.318 -0.009 i & -0.144-0.594 i & 0.452 -0.405 i & 0.037 +0.387 i & 0.071 +0.025 i
\end{pmatrix},
\end{equation*}
\normalsize
where the global phase was fixed so as to make the upper-left element of $U_{5}^{t}$ real. We then decomposed $U_{5}^{t}$ as a product of matrices that act nontrivially on two modes only, each representing a set of one beam splitter and two phase shifters in the range $[0,\pi]$, as described in \cite{Reck1994}. The use of $[0,\pi]$ phase shifters was an adaptation of the decomposition of \cite{Reck1994}, which originally used phase shifters in the $[0,2\pi]$ range. This was done to limit the phase shift (and waveguide deformation) introduced by each element, as this may lead to losses. In this new decomposition, only one of the two phase shifters at each beam splitter input branch needs to introduce a nonzero phase shift. \tabl{Parameters5} reports the parameters obtained in this decomposition.

\begin{table}[ht]\centering
\begin{tabular}{|c||c|c|c|}
\hline 
$i$ & $t_{i}$ & $\alpha_{i}$ [rad] & $\beta_{i}$ [rad] \\
\hline
1 & 0.19 & 0 & 0 \\
\hline
2 & 0.40 & 0.64 & 0 \\
\hline
3 & 0.48 & 0 & 1.37 \\
\hline
4 & 0.44 & 0 & 1.10 \\
\hline
5 & 0.55 & 2.21 & 0 \\
\hline
6 & 0.54 & 0 & 1.02 \\
\hline
7 & 0.51 & 2.93 & 0 \\
\hline
8 & 0.76 & 1.08 & 0 \\
\hline
9 & 0.99 & 2.58 & 0 \\
\hline
10 & 0.95 & 0 & 0\\
\hline
\end{tabular}
\caption{Transmissivities, $t_i$, and phases, $\alpha_i$ and $\beta_i$, related to the layouts in \sfig{Reck}{b}, that result from the decomposition of the sampled unitary matrix.}
\label{tab:Parameters5}
\end{table}

Since all experimental outcomes are invariant under multiplication of $U_{5}^{t}$ by a phase shifter at each input and output port, we have set $\alpha_{1}, \beta_{1}, \beta_{5},\beta_{8}$ and $\beta_{10}$ equal to zero in \tabl{Parameters5}. 

We have also found a unitary that fits well the single- and two-photon data, following the refinement of the Laing-O'Brien algorithm described in \sec{bosonnew_b_tomography}.  Since the reconstruction method only obtains the unitary up to a round of arbitrary phases at the input and output modes, we have multiplied each row and column of the optimized unitary so as to obtain the highest gate fidelity with $U_{5}^{t}$, making the matrices easier to compare. The reconstructed unitary $U_{5}^{r}$ is found to be:

\scriptsize
\begin{equation*}
U_{5}^{r} = \begin{pmatrix}
 0.370 	 & 0.007 +0.151 i & -0.164-0.31 i & -0.442+0.138 i & -0.702+0.099 i \\
 -0.109-0.465 i & -0.013-0.585 i & 0.121 +0.381 i & 0.076 -0.134 i & -0.474-0.147 i \\
 -0.677+0.180 i & 0.134 -0.027 i & -0.283-0.133 i & 0.036 +0.498 i & -0.206-0.319 i \\
 -0.039+0.240 i & -0.080 -0.572 i & -0.496-0.046 i & -0.475-0.220 i & 0.265 +0.125 i \\
 0.262 +0.133 i & 0.090 -0.524 i & 0.479 -0.377 i & 0.055 +0.486 i & 0.143 +0.007 i
\end{pmatrix}.
\end{equation*}
\normalsize

We have implemented a Monte Carlo simulation to check the consistency of our reconstruction method. We used $U^r$ to simulate 1000 complete, new data sets for single- and two-photon experiments, with error bars compatible with those of our real experimental data. We then applied the Laing-O'Brien method to obtain a reconstructed unitary $U_{5}^{r\,\prime}$ for each data set. The standard deviation in the results serves as our estimated error bars for any quantity of interest. 

The error bars for the real and for the imaginary parts or $U_{5}^r$ are respectively:
\scriptsize
\begin{equation*}
\Delta \textrm{Re}(U_5^{r}) = \begin{pmatrix}
0.011 &	0.010 & 0.010	& 0.008	& 0.007 \\
0.015	& 0.011 &	0.010	& 0.010	& 0.008 \\
0.008	& 0.008	& 0.010	& 0.013	& 0.007 \\
0.011	& 0.012	& 0.010	& 0.007 &	0.008 \\
0.010	& 0.013	& 0.009	& 0.011	& 0.010
\end{pmatrix},
\end{equation*}
\normalsize
and
\scriptsize
\begin{equation*}
\Delta \textrm{Im}(U_5^{r}) = \begin{pmatrix}
0 &	0.009 & 0.009	& 0.013	& 0.015 \\
0.008	& 0.005 &	0.010	& 0.011	& 0.011 \\
0.014	& 0.012	& 0.012	& 0.006	& 0.009 \\
0.009	& 0.007	& 0.013	& 0.010 &	0.011 \\
0.009	& 0.008	& 0.010	& 0.009 & 0.011
\end{pmatrix}.
\end{equation*}
\normalsize
The similarity between the sampled $U_{5}^t$ and the reconstructed $U_{5}^r$ can then be quantified by the gate fidelity $F=|Tr(U_5^t U_5^{r\dagger})|/5 =0.950\pm0.002$.

Such a value of the gate fidelity $F$ between the sampled unitary $U_{5}^t$ and the reconstructed one $U_{5}^r$ has to be related to the fabrication tolerances estimated from the data shown in \fig{LOelements}. To analyze the effect of fabrication errors in the full device, we sampled $N=10000$ random unitaries from a gaussian distribution centered around $U_{5}^t$. More specifically, the parameters of the internal phase shifters and beam-splitters have been randomly picked from a gaussian distribution centered at the parameter values of $U_{5}^t$, with standard deviation equal to the fabrication error. We obtained a value for the average gate fidelity between $U_{5}^t$ and the sampled unitaries equal to: $F=0.968 \pm 0.020$. Such a value is compatible within one standard deviation with the experimental value of $F=0.950\pm0.002$ obtained between $U_{5}^t$ and $U_{5}^r$. 

\section{7-mode chip}

The matrix $U_{7}^{t}$, used in \sec{validation}, was sampled from the random phases ensemble described in \sec{bosonnew_b_ensembles}:

\scriptsize
\begin{equation*}
\textrm{Re}(U_7^{t}) = \begin{pmatrix}
 0.4425 & -0.1165 & -0.1488 & 0.4638 & 0.1579 & 0.0794 & 0. \\
 -0.1399 & -0.4259 & -0.1446 & 0.0255 & -0.0794 & 0.1579 & 0. \\
 -0.0407 & 0.0883 & 0.5283 & 0.2971 & -0.1533 & -0.0246 & 0.1383 \\
 0.6001 & 0.3919 & -0.205 & -0.4029 & -0.2782 & -0.1281 & 0.2082 \\
 -0.1749 & 0.0259 & -0.2427 & 0.1622 & -0.1493 & -0.2798 & -0.0683 \\
 -0.0259 & -0.1749 & 0.073 & -0.1255 & 0.2164 & -0.3516 & 0.1798 \\
 0. & 0. & -0.0576 & -0.2433 & -0.4469 & 0.1336 & -0.5942 \\
\end{pmatrix},
\end{equation*}
\normalsize
and
\scriptsize
\begin{equation*}
\textrm{Im}(U_7^{t})= \begin{pmatrix}
 0 & -0.271 & -0.6244 & 0.1661 & -0.0794 & 0.1579 & 0 \\
 0.5437 & -0.5791 & 0.2894 & 0.0161 & -0.1579 & -0.0794 & 0 \\
 -0.253 & -0.3246 & -0.0445 & -0.3622 & -0.1533 & -0.4588 & 0.2082 \\
 0.0265 & -0.2588 & 0.1053 & -0.0722 & -0.1775 & -0.1004 & -0.1383 \\
 0.0259 & 0.1749 & -0.1019 & 0.4779 & -0.2428 & -0.6334 & -0.2194 \\
 -0.1749 & 0.0259 & 0.1096 & 0.1911 & -0.6708 & 0.2806 & 0.3694 \\
 0 & 0 & -0.2433 & 0.0576 & -0.0615 & -0.0134 & 0.548 \\
\end{pmatrix},
\end{equation*}
\normalsize
where the global phase was fixed so as to make the upper-left element of $U_{7}^{t}$ real. The parameter specifications are given in \tabl{Parameters7}, where L and M label the position of each phase shifter within the interferometer as follows: in the ordering of \fig{walkinterf}, the modes are labeled sequentially as $\textrm{M}_1$, $\textrm{M}_2$, and so on, from top to bottom, while the layers of phase shifters are labeled as $\textrm{L}_1$, $\textrm{L}_2$ etc., from left to right. Note that the chip has an odd numbers of modes, so one mode is idle in every layer of 50:50 beam splitters. The convention is that the first layer has beam splitters between modes 1 and 2, 3 and 4, 5 and 6, and so on, and the bottom mode is idle. 

\begin{table}[ht]\centering \scriptsize
\begin{tabular}{|c||c|c|c|c|}
\hline 
 & L1 & L2 & L3 & L4\\
\hline
M1 & 1.5253 & 0.6993 & 2.7776 & 1.8087 \\
\hline
M2 & 2.6182 & 1.5267 & 2.956 & 1.6449 \\
\hline
M3 & 1.9217 & 2.8131 & 0.6705 & 1.7497 \\
\hline
M4 & 0.7217 & 0.4718 & 0.9392 & 1.5706 \\
\hline
M5 & 2.9256 & 1.3138 & 2.1079 & 2.8032 \\
\hline
M6 & 0.4974 & 1.4759 & 0.3152 & 0.7684 \\
\hline
M7 & 2.0089 & 2.9217 & 1.347 & 2.025 \\
\hline
\end{tabular}
\caption{Specification of phase shifts (in rad) for the 7-mode chip from the random phases ensemble. The labeling notation is defined in the text.}
\label{tab:Parameters7}
\end{table}

We have also applied the refined Laing-O'Brien algorithm described in \sec{bosonnew_b_tomography} to obtain a reconstructed matrix $U_7^{r}$. Since the reconstruction method only obtains the unitary up to a round of arbitrary phases at the input and output modes, we have multiplied each row and column of the optimized unitary so as to obtain the highest gate fidelity with $U_{7}^{t}$, making the matrices easier to compare. The reconstructed unitary $U_{7}^{r}$ is given by:

\scriptsize
\begin{equation*}
\textrm{Re}(U_7^{r}) = \begin{pmatrix}
 0.4452 & -0.1619 & -0.0803 & 0.3911 & 0.1092 & 0.0209 & -0.0081 \\
 -0.1373 & -0.4556 & -0.1317 & 0.0791 & -0.0755 & 0.0837 & 0.0001 \\
 -0.0416 & 0.0536 & 0.4685 & 0.3431 & -0.1754 & -0.055 & 0.0912 \\
 0.6524 & 0.3148 & -0.1802 & -0.3467 & -0.2841 & -0.2602 & 0.2024 \\
 -0.1626 & -0.0474 & -0.2752 & 0.1311 & -0.1485 & -0.2213 & -0.0441 \\
 -0.0704 & -0.1444 & 0.0106 & -0.106 & 0.3 & -0.3859 & 0.1749 \\
 0.0001 & -0.0067 & -0.0768 & -0.2394 & -0.4011 & 0.0415 & -0.6603 \\
\end{pmatrix},
\end{equation*}
\normalsize
and
\scriptsize
\begin{equation*}
\textrm{Im}(U_7^{r}) = \begin{pmatrix}
 0 & -0.2345 & -0.7213 & 0.1138 & 0.0039 & 0.1245 & 0.0016 \\
 0.3705 & -0.7084 & 0.2681 & -0.0592 & -0.1077 & -0.1061 & 0.0001 \\
 -0.4088 & -0.1521 & -0.0212 & -0.3293 & -0.2072 & -0.4918 & 0.179 \\
 0.0715 & -0.1923 & 0.1248 & -0.1292 & -0.15 & -0.1598 & -0.1038 \\
 0.0324 & 0.154 & -0.0879 & 0.5484 & -0.0376 & -0.6209 & -0.292 \\
 -0.1169 & 0.0211 & 0.0282 & 0.2579 & -0.7122 & 0.2034 & 0.2517 \\
 0 & -0.007 & -0.1605 & 0.0906 & -0.1107 & -0.0338 & 0.5391 \\
\end{pmatrix}.
\end{equation*}
\normalsize

We have implemented a Monte Carlo simulation to check the consistency of our reconstruction method, in the same way as for the 5-mode chip described earlier in this Appendix. The error bars for the real and for the imaginary parts or $U_{7}^r$ are respectively:
\scriptsize
\begin{equation}
\Delta \textrm{Re}(U_7^{r}) = \begin{pmatrix}
0.011	& 0.021	& 0.025	& 0.011	& 0.008	& 0.008	& 0.008 \\
0.02	& 0.02	& 0.012	& 0.012	& 0.009	& 0.008	& 0.005 \\
0.019	& 0.012	& 0.009	& 0.01	& 0.012	& 0.016	& 0.014 \\
0.009	& 0.01	& 0.018	& 0.01	& 0.01	& 0.012	& 0.01 \\
0.008	& 0.008	& 0.011	& 0.016	& 0.014	& 0.021	& 0.019 \\
0.011	& 0.007	& 0.008	& 0.016	& 0.022	& 0.013	& 0.018 \\
0.006	& 0.007	& 0.012	& 0.016	& 0.014	& 0.016	& 0.018
\end{pmatrix},
\end{equation}
\normalsize
and
\scriptsize
\begin{equation}
\Delta \textrm{Im}(U_7^{r}) = \begin{pmatrix}
0.	& 0.019	& 0.008	& 0.018	& 0.01	& 0.01	& 0.009 \\
0.014	& 0.017	& 0.017	& 0.014	& 0.009	& 0.008	& 0.006 \\
0.01	& 0.013	& 0.018	& 0.013	& 0.013	& 0.01	& 0.014 \\
0.023	& 0.012	& 0.011	& 0.016	& 0.013	& 0.011	& 0.012 \\
0.01	& 0.009	& 0.012	& 0.012	& 0.013	& 0.01	& 0.017 \\
0.007	& 0.009	& 0.009	& 0.009	& 0.012	& 0.015	& 0.014 \\
0.007	& 0.007	& 0.01	& 0.018	& 0.018	& 0.016	& 0.022
\end{pmatrix}.
\end{equation}
\normalsize
The similarity between the sampled $U_7^t$ and the reconstructed $U_7^r$ can then be quantified by the gate fidelity $F=|Tr(U_7^t U_7^{r\dagger})|/7 =0.9750 \pm 0.0022$.

\section{9-mode chip}

The matrix $U_{9}^{t}$, used in \sec{validation}, was sampled from the random phases ensemble described in \sec{bosonnew_b_ensembles}:
\scriptsize
\begin{equation*}
\textrm{Re}(U_9^{t}) = \begin{pmatrix}
 0.1737 & 0.3764 & -0.2099 & 0.0618 & 0.156 & 0.0832 & 0. & 0. & 0. \\
 0.2093 & -0.1192 & -0.193 & 0.188 & -0.1875 & 0.0353 & -0.1083 & -0.0624 & 0. \\
 -0.3979 & -0.2573 & 0.1942 & 0.2589 & 0.1764 & -0.0378 & 0.0624 & -0.1083 & 0. \\
 -0.0198 & 0.1469 & 0.3669 & -0.1077 & -0.2168 & 0.4123 & 0.1051 & -0.2844 & -0.146 \\
 -0.0195 & 0.3428 & -0.5808 & -0.2214 & -0.0762 & 0.2985 & 0.319 & -0.0395 & 0.0997 \\
 -0.0433 & -0.1173 & -0.0321 & 0.0354 & 0.2657 & 0.2604 & -0.2291 & -0.535 & 0.4547 \\
 0.1173 & -0.0433 & -0.0772 & -0.2149 & 0.2085 & -0.3808 & 0.0738 & 0.2366 & 0.0849 \\
 0. & 0. & -0.1216 & 0.0289 & 0.362 & -0.3342 & 0.5234 & -0.4918 & -0.0882 \\
 0. & 0. & -0.0289 & -0.1216 & -0.0368 & 0.1977 & -0.0343 & -0.1498 & -0.6324 \\
\end{pmatrix},
\end{equation*}
\normalsize
and
\scriptsize
\begin{equation*}
\textrm{Im}(U_9^{t}) = \begin{pmatrix}
 0 & -0.6425 & -0.3237 & -0.4474 & -0.0832 & 0.156 & 0 & 0 & 0 \\
 0.7566 & -0.2241 & -0.0309 & 0.4142 & -0.0529 & 0.0619 & 0.0624 & -0.1083 & 0 \\
 0.3374 & -0.1194 & 0.3813 & -0.4619 & 0.2872 & 0.194 & 0.1083 & 0.0624 & 0 \\
 -0.1551 & -0.1639 & 0.0947 & 0.2141 & -0.0766 & 0.2954 & 0.5338 & 0.0884 & 0.0997 \\
 0.1591 & 0.3216 & 0.2771 & -0.1147 & 0.186 & 0.0874 & 0.0228 & 0.0695 & 0.146 \\
 -0.1173 & 0.0433 & -0.1696 & 0.1752 & 0.1475 & 0.2771 & -0.3511 & 0.0202 & 0.0303 \\
 -0.0433 & -0.1173 & -0.0718 & 0.2822 & 0.4033 & 0.3424 & 0.1966 & 0.495 & -0.1054 \\
 0 & 0 & 0.0289 & 0.1216 & -0.3066 & -0.1135 & 0.1344 & -0.0901 & -0.2633 \\
 0 & 0 & -0.1216 & 0.0289 & 0.4462 & 0.0092 & -0.2397 & -0.1164 & -0.4843 \\
\end{pmatrix},
\end{equation*}
\normalsize
where the global phase was fixed so as to make the upper-left element of $U_{9}^{t}$ real. The parameter specifications are given in \tabl{Parameters9}, where L and M label the position of each phase shifter within the interferometer in the same manner as for the 7-chip of the previous section of this Appendix. We did not apply the reconstruction method for the 9-mode chip, so every analysis was done directly with $U_9^t$.

\begin{table}[ht]\centering \scriptsize
\begin{tabular}{|c||c|c|c|c|c|}
\hline 
 & L1 & L2 & L3 & L4 & L5 \\
\hline
M1 & 2.6429 & 2.306 & 1.9381 & 0.857 & 1.5198 \\ \hline
M2 & 2.6199 & 0.5196 & 0.8464 & 2.9378 & 2.4199 \\ \hline
M3 & 0.4416 & 2.8766 & 2.4005 & 2.7948 & 0.2681 \\ \hline
M4 & 2.6066 & 0.6619 & 3.0257 & 2.1596 & 1.8241 \\ \hline
M5 & 2.9507 & 0.5332 & 0.274 & 2.8834 & 1.8991 \\ \hline
M6 & 0.9698 & 2.6471 & 1.2586 & 1.9846 & 2.3428 \\ \hline
M7 & 0.3772 & 2.3326 & 0.1867 & 0.8455 & 2.7346 \\ \hline
M8 & 0.5191 & 0.021 & 0.3205 & 1.3059 & 0.771 \\ \hline
M9 & 1.5579 & 1.9952 & 0.2173 & 2.8286 & 1.127 \\
\hline
\end{tabular}
\caption{Specification of phase shifts (in rad) for the 9-mode chip from the random phases ensemble. The labeling notation is equivalent to that of \tabl{Parameters7}.}
\label{tab:Parameters9}
\end{table}

\end{appendices}

\bibliographystyle{plainshort}

\begin{thebibliography}{100}

\bibitem{CompZoo}
Complexity Zoo
  (https://complexityzoo.uwaterloo.ca/Complexity\ensuremath{\_}Zoo).

\bibitem{Aaronson2005}
S.~Aaronson.
\newblock Quantum computing, postselection, and probabilistic polynomial-time.
\newblock {\em {Proceedings of the Royal Society of London. Series A:
  Mathematical, Physical and Engineering Sciences}} {\bf 461}, 3473--3482,
  (2005).

\bibitem{Aaronson2013b}
S.~Aaronson and A.~Arkhipov.
\newblock BosonSampling is far from uniform.
\newblock {\em {arXiv:1309.7460 \ensuremath{[}quant-ph\ensuremath{]}}} (2013).

\bibitem{Aaronson2013a}
S.~Aaronson and A.~Arkhipov.
\newblock The computational complexity of linear optics.
\newblock {\em {Theory of Computing}} {\bf 4}, 143--252, (2013).

\bibitem{Aaronsonprivcomm}
S.~Aaronson and A.~Arkhipov.
\newblock Private communication, (2013).

\bibitem{Aharonov2008a}
D.~Aharonov, M.~Ben-Or, and E.~Eban.
\newblock Interactive proofs for quantum computations.
\newblock {\em {arXiv:0810.5375 \ensuremath{[}quant-ph\ensuremath{]}}} (2008).

\bibitem{Aharonov2008b}
D.~Aharonov and M.~Ben-Or.
\newblock Fault-Tolerant Quantum Computation with Constant Error Rate.
\newblock {\em {SIAM Journal on Computing}} {\bf 38}, 1207--1282, (2008).

\bibitem{Aharonov2007}
D.~Aharonov, W.~van Dam, J.~Kempe, et~al.
\newblock Adiabatic Quantum Computation is Equivalent to Standard Quantum
  Computation.
\newblock {\em {SIAM Journal on Computing}} {\bf 37}, 166--194, (2007).

\bibitem{Anders2010}
J.~Anders, D.~K.~L. Oi, E.~Kashefi, D.~E. Browne, and E.~Andersson.
\newblock Ancilla-driven universal quantum computation.
\newblock {\em {Physical Review A}} {\bf 82}, 020301(R), (2010).

\bibitem{Arkhipov2012}
A.~Arkhipov and G.~Kuperberg.
\newblock The bosonic birthday paradox.
\newblock {\em {Geometry \& Topology Monographs}} {\bf 18}, 1--7, (2012).

\bibitem{livroBallentine}
L.~E. Ballentine.
\newblock {\em Quantum Mechanics: A Modern Development}.
\newblock {World Scientific Publishing Co.}, (1998).

\bibitem{Bardek1994}
V.~Bardek, M.~Dore\u{s}i\'{c}, and S.~Meljanac.
\newblock Anyons as quon particles.
\newblock {\em {Physical Review D}} {\bf 49}, 3059--3062, (1994).

\bibitem{Barreiro2005}
J.~T. Barreiro, N.~K. Langford, N.~A. Peters, and P.~G. Kwiat.
\newblock Generation of hyperentangled photon pairs.
\newblock {\em {Physical Review Letters}} {\bf 95}, 260501, (2005).

\bibitem{Beals2013}
R.~Beals, S.~Brierley, O.~Gray, et~al.
\newblock Efficient distributed quantum computing.
\newblock {\em {Proceedings of the Royal Society A: Mathematical, Physical and
  Engineering Science}} {\bf 469}, 20120686, (2013).

\bibitem{Beenakker2004}
C.~Beenakker, D.~DiVincenzo, C.~Emary, and M.~Kindermann.
\newblock Charge detection enables free-electron quantum computation.
\newblock {\em {Physical Review Letters}} {\bf 93}, 020501, (2004).

\bibitem{Bernstein1993}
E.~Bernstein and U.~Vazirani.
\newblock Quantum complexity theory.
\newblock In {\em {Proceedings of the Twenty-fifth Annual ACM Symposium on
  Theory of Computing}}, STOC '93, pp. 11--20, New York, NY, USA, (1993).
  {ACM}.

\bibitem{LivroBollobas}
B.~Bollob\'{a}s.
\newblock {\em Modern Graph Theory}.
\newblock {Springer}, (1998).

\bibitem{Bouland2013}
A.~Bouland and S.~Aaronson.
\newblock Generation of Universal Linear Optics by Any Beamsplitter.
\newblock {\em {arXiv:1310.6718 \ensuremath{[}quant-ph\ensuremath{]}}} (2013).

\bibitem{Boyajian2013}
W.~L. Boyajian, V.~Murg, and B.~Kraus.
\newblock Compressed simulation of evolutions of the XY model.
\newblock {\em {Physical Review A}} {\bf 88}, 052329, (2013).

\bibitem{Bravyi2002}
S.~B. Bravyi and A.~Y. Kitaev.
\newblock Fermionic quantum computation.
\newblock {\em {Annals of Physics}} {\bf 298}, 210--226, (2002).

\bibitem{Bremner2002}
M.~J. Bremner, C.~Dawson, J.~Dodd, et~al.
\newblock Practical scheme for quantum computation with any two-qubit
  entangling gate.
\newblock {\em {Physical Review Letters}} {\bf 89}, 247902, (2002).

\bibitem{Bremner2011}
M.~J. Bremner, R.~Jozsa, and D.~J. Shepherd.
\newblock Classical simulation of commuting quantum computations implies
  collapse of the polynomial hierarchy.
\newblock {\em {Proceedings of the Royal Society of London. Series A:
  Mathematical, Physical and Engineering Sciences}} {\bf 467}, 459--472,
  (2011).

\bibitem{Broadbent2009}
A.~Broadbent, J.~Fitzsimons, and E.~Kashefi.
\newblock Universal blind quantum computation.
\newblock In {\em {50th Annual IEEE Symposium on Foundations of Computer
  Science (FOCS)}}, pp. 517--526, (2009).

\bibitem{Brod2013}
D.~J. Brod and A.~M. Childs.
\newblock The computational power of matchgates and the XY interaction on
  arbitrary graphs.
\newblock {\em {arXiv:1308.1463 \ensuremath{[}quant-ph\ensuremath{]}, to appear
  in Quant. Inf. Comp.}} (2013).

\bibitem{Brod2011}
D.~J. Brod and E.~F. Galv\~{a}o.
\newblock Extending matchgates into universal quantum computation.
\newblock {\em {Physical Review A}} {\bf 84}, 022310, (2011).

\bibitem{Brod2012}
D.~J. Brod and E.~F. Galv\~{a}o.
\newblock Geometries for universal quantum computation with matchgates.
\newblock {\em {Physical Review A}} {\bf 86}, 052307, (2012).

\bibitem{Broome2013}
M.~A. Broome, A.~Fedrizzi, S.~Rahimi-Keshari, et~al.
\newblock Photonic Boson Sampling in a Tunable Circuit.
\newblock {\em {Science}} {\bf 339}, 794--798, (2013).
\newblock PMID: 23258411.

\bibitem{Browne2005}
D.~Browne and T.~Rudolph.
\newblock Resource-efficient linear optical quantum computation.
\newblock {\em {Physical Review Letters}} {\bf 95}, 010501, (2005).

\bibitem{Burgarth2010}
D.~Burgarth, K.~Maruyama, M.~Murphy, et~al.
\newblock Scalable quantum computation via local control of only two qubits.
\newblock {\em {Physical Review A}} {\bf 81}, 040303(R), (2010).

\bibitem{Carolan2013}
J.~Carolan, J.~D.~A. Meinecke, P.~Shadbolt, et~al.
\newblock On the experimental verification of quantum complexity in linear
  optics.
\newblock {\em {arXiv:1311.2913 \ensuremath{[}quant-ph\ensuremath{]}}} (2013).

\bibitem{Childs2005}
A.~Childs, D.~Leung, and M.~Nielsen.
\newblock Unified derivations of measurement-based schemes for quantum
  computation.
\newblock {\em {Physical Review A}} {\bf 71}, 032318, (2005).

\bibitem{LivroCohen}
C.~Cohen-Tannoudji, B.~Diu, and F.~Lalo\"{e}.
\newblock {\em Quantum Mechanics}, volume~2.
\newblock {Wiley-VCH}, (1992).

\bibitem{LivroCover}
T.~M. Cover and J.~A. Thomas.
\newblock {\em Elements of Information Theory (Wiley Series in
  Telecommunications and Signal Processing)}.
\newblock {Wiley-Interscience}, (1991).

\bibitem{Crespi2013a}
A.~Crespi, R.~Osellame, R.~Ramponi, et~al.
\newblock Anderson localization of entangled photons in an integrated quantum
  walk.
\newblock {\em {Nature Photonics}} {\bf 7}, 322--328, (2013).

\bibitem{Crespi2013b}
A.~Crespi, R.~Osellame, R.~Ramponi, et~al.
\newblock Integrated multimode interferometers with arbitrary designs for
  photonic boson sampling.
\newblock {\em {Nature Photonics}} {\bf 7}, 545--549, (2013).

\bibitem{daSilva2011}
R.~D. da~Silva, E.~F. Galv\~{a}o, and E.~Kashefi.
\newblock Closed timelike curves in measurement-based quantum computation.
\newblock {\em {Physical Review A}} {\bf 83}, 012316, (2011).

\bibitem{Deutsch1985}
D.~Deutsch.
\newblock Quantum Theory, the Church-Turing Principle and the Universal Quantum
  Computer.
\newblock {\em {Proceedings of the Royal Society of London. Series A:
  Mathematical, Physical and Engineering Sciences}} {\bf 400}, 97--117, (1985).

\bibitem{Deutsch1991}
D.~Deutsch.
\newblock Quantum mechanics near closed timelike lines.
\newblock {\em {Physical Review D}} {\bf 44}, 3197--3217, (1991).

\bibitem{DiVincenzo2000a}
D.~P. DiVincenzo, D.~Bacon, J.~Kempe, G.~Burkard, and K.~B. Whaley.
\newblock Universal quantum computation with the exchange interaction.
\newblock {\em {Nature}} {\bf 408}, 339--342, (2000).

\bibitem{Dzyaloshinsky1958}
I.~Dzyaloshinsky.
\newblock A thermodynamic theory of ``weak'' ferromagnetism of
  antiferromagnetics.
\newblock {\em {Journal of Physics and Chemistry of Solids}} {\bf 4}, 241--255,
  (1958).

\bibitem{Feynman1982}
R.~P. Feynman.
\newblock Simulating physics with computers.
\newblock {\em {International Journal of Theoretical Physics}} {\bf 21},
  467--488, (1982).

\bibitem{Fredkin1982}
E.~Fredkin and T.~Toffoli.
\newblock Conservative logic.
\newblock {\em {International Journal of Theoretical Physics}} {\bf 21},
  219--253, (1982).

\bibitem{LivroGarey}
M.~R. Garey and D.~S. Johnson.
\newblock {\em Computers and Intractability: A Guide to the Theory of
  NP-Completeness}.
\newblock {Freeman}, (1979).

\bibitem{Gilchrist2007}
A.~Gilchrist, A.~J.~F. Hayes, and T.~C. Ralph.
\newblock Efficient parity-encoded optical quantum computing.
\newblock {\em {Physical Review A}} {\bf 75}, 052328, (2007).

\bibitem{Gogolin2013}
C.~Gogolin, M.~Kliesch, L.~Aolita, and J.~Eisert.
\newblock BosonSampling in the light of sample complexity.
\newblock {\em {arXiv:1306.3995 \ensuremath{[}quant-ph\ensuremath{]}}} (2013).

\bibitem{LivroGoldreich}
O.~Goldreich.
\newblock {\em Computational Complexity: A Conceptual Perspective}.
\newblock {Cambridge University Press, Cambridge}, (2008).

\bibitem{Gottesman1999a}
D.~Gottesman.
\newblock The Heisenberg representation of quantum computers.
\newblock In {\em
  {Proceedings of the XXII International Colloquium on Group Theoretical
  Methods in Physics}} (editors S.~P. Corney, R.~Delbourgo, and P.~D. Jarvis), 32--43. {Cambridge, MA, International Press},
  (1999).

\bibitem{Gottesman1999b}
D.~Gottesman and I.~L. Chuang.
\newblock Demonstrating the viability of universal quantum computation using
  teleportation and single-qubit operations.
\newblock {\em {Nature}} {\bf 402}, 390--393, (1999).

\bibitem{Gottesman2009}
D.~Gottesman.
\newblock An introduction to quantum error correction and fault-tolerant
  quantum computation.
\newblock In {\em {Quantum Information Science and Its Contributions to
  Mathematics, Proceedings of Symposia in Applied Mathematics}}, volume~68,
  pp.~13, (2009).

\bibitem{Greenberg1991}
O.~W. Greenberg.
\newblock Particles with small violations of Fermi or Bose statistics.
\newblock {\em {Physical Review D}} {\bf 43}, 4111--4120, (1991).

\bibitem{Grover1996}
L.~K. Grover.
\newblock A Fast Quantum Mechanical Algorithm for Database Search.
\newblock In {\em {Proceedings of the Twenty-eighth Annual ACM Symposium on
  Theory of Computing}}, STOC '96, pp. 212--219, New York, NY, USA, (1996).
  {ACM}.

\bibitem{Han1997}
Y.~Han, L.~Hemaspaandra, and T.~Thierauf.
\newblock Threshold computation and cryptographic security.
\newblock {\em {SIAM Journal on Computing}} {\bf 26}, 59--78, (1997).

\bibitem{Hayes2004}
A.~J.~F. Hayes, A.~Gilchrist, C.~R. Myers, and T.~C. Ralph.
\newblock Utilizing encoding in scalable linear optics quantum computing.
\newblock {\em {Journal of Optics B: Quantum and Semiclassical Optics}} {\bf
  6}, 533, (2004).

\bibitem{Herrera2014}
F.~Herrera, Y.~Cao, S.~Kais, and K.~B. Whaley.
\newblock Infrared-dressed entanglement of cold open-shell polar molecules for
  universal matchgate quantum computing.
\newblock {\em {arXiv:1402.0381 \ensuremath{[}physics,
  physics:quant-ph\ensuremath{]}}} (2014).

\bibitem{Hirata2011}
Y.~Hirata, M.~Nakanishi, S.~Yamashita, and Y.~Nakashima.
\newblock An effcient conversion of quantum circuits to a linear nearest
  neighbor architecture.
\newblock {\em {Quantum Information and Computation}} {\bf 11}, 142--166,
  (2011).

\bibitem{Hoban2013}
M.~J. Hoban, J.~J. Wallman, H.~Anwar, et~al.
\newblock On the hardness of sampling and measurement-based classical
  computation.
\newblock {\em {arXiv:1304.2667 \ensuremath{[}quant-ph\ensuremath{]}}} (2013).

\bibitem{Hoffman1972}
D.~K. Hoffman, R.~C. Raffenetti, and K.~Ruedenberg.
\newblock Generalization of Euler angles to N-dimensional orthogonal matrices.
\newblock {\em {Journal of Mathematical Physics}} {\bf 13}, 528--533, (1972).

\bibitem{Hong1987}
C.~K. Hong, Z.~Y. Ou, and L.~Mandel.
\newblock Measurement of subpicosecond time intervals between two photons by
  interference.
\newblock {\em {Physical Review Letters}} {\bf 59}, 2044--2046, (1987).

\bibitem{Imamoglu1999}
A.~Imamo\={g}lu, D.~D. Awschalom, G.~Burkard, et~al.
\newblock Quantum information processing using quantum dot spins and
  cavity-QED.
\newblock {\em {Physical Review Letters}} {\bf 83}, 4204--4207, (1999).

\bibitem{Jerrum2004}
M.~Jerrum, A.~Sinclair, and E.~Vigoda.
\newblock A Polynomial-time Approximation Algorithm for the Permanent of a
  Matrix with Nonnegative Entries.
\newblock {\em {Journal of the ACM}} {\bf 51}, 671--697, (2004).

\bibitem{Jordan1928}
P.~Jordan and E.~Wigner.
\newblock \"{U}ber das Paulische \"{A}quivalenzverbot.
\newblock {\em {Z. Phys.}} {\bf 47}, 631--651, (1928).

\bibitem{Jouguet2013}
P.~Jouguet, S.~Kunz-Jacques, A.~Leverrier, P.~Grangier, and E.~Diamanti.
\newblock Experimental demonstration of long-distance continuous-variable
  quantum key distribution.
\newblock {\em {Nature Photonics}} {\bf 7}, 378--381, (2013).

\bibitem{Jozsa2010}
R.~Jozsa, B.~Kraus, A.~Miyake, and J.~Watrous.
\newblock Matchgate and space-bounded quantum computations are equivalent.
\newblock {\em {Proceedings of the Royal Society of London. Series A:
  Mathematical, Physical and Engineering Sciences}} {\bf 466}, 809--830,
  (2010).

\bibitem{Jozsa2008b}
R.~Jozsa and A.~Miyake.
\newblock Matchgates and classical simulation of quantum circuits.
\newblock {\em {Proceedings of the Royal Society of London. Series A:
  Mathematical, Physical and Engineering Sciences}} {\bf 464}, 3089--3106,
  (2008).

\bibitem{Jozsa2003}
R.~Jozsa and N.~Linden.
\newblock On the role of entanglement in quantum-computational speed-up.
\newblock {\em {Proceedings of the Royal Society of London. Series A:
  Mathematical, Physical and Engineering Sciences}} {\bf 459}, 2011--2032,
  (2003).

\bibitem{Jozsa2014}
R.~Jozsa and M.~Van~den Nest.
\newblock Classical simulation complexity of extended Clifford circuits.
\newblock {\em {Quantum Information and Computation}} {\bf 14}, 633--648,
  (2014).

\bibitem{Kalai2009}
G.~Kalai.
\newblock Quantum Computers: Noise Propagation and Adversarial Noise Models.
\newblock {\em {arXiv:0904.3265 \ensuremath{[}quant-ph\ensuremath{]}}} (2009).

\bibitem{Kalai2011}
G.~Kalai.
\newblock How Quantum Computers Fail: Quantum Codes, Correlations in Physical
  Systems, and Noise Accumulation.
\newblock {\em {arXiv:1106.0485 \ensuremath{[}quant-ph\ensuremath{]}}} (2011).

\bibitem{Kasteleyn1961}
P.~Kasteleyn.
\newblock The statistics of dimers on a lattice.
\newblock {\em {Physica}} {\bf 27}, 1209--1225, (1961).

\bibitem{Kempe2001b}
J.~Kempe, D.~Bacon, D.~DiVincenzo, and K.~Whaley.
\newblock Encoded universality from a single physical interaction.
\newblock {\em {Quantum Information and Computation}} {\bf 1}, 33--55, (2001).

\bibitem{Kempe2001a}
J.~Kempe, D.~Bacon, D.~A. Lidar, and K.~Whaley.
\newblock Theory of decoherence-free fault-tolerant universal quantum
  computation.
\newblock {\em {Physical Review A}} {\bf 63}, 042307, (2001).

\bibitem{Kempe2002}
J.~Kempe and K.~B. Whaley.
\newblock Exact gate sequences for universal quantum computation using the XY
  interaction alone.
\newblock {\em {Physical Review A}} {\bf 65(5)}, 052330, (2002).

\bibitem{Khaneja2001}
N.~Khaneja, R.~Brockett, and S.~J. Glaser.
\newblock Time optimal control in spin systems.
\newblock {\em {Physical Review A}} {\bf 63}, 032308, (2001).

\bibitem{Kitaev1997}
A.~Y. Kitaev.
\newblock Quantum computations: algorithms and error correction.
\newblock {\em {Russian Mathematical Surveys}} {\bf 52}, 1191--1249, (1997).

\bibitem{Knill2001a}
E.~Knill.
\newblock Fermionic linear optics and matchgates.
\newblock {\em {arXiv:quant-ph/0108033}} (2001).

\bibitem{Knill2002}
E.~Knill.
\newblock Quantum gates using linear optics and postselection.
\newblock {\em {Physical Review A}} {\bf 66}, 052306, (2002).

\bibitem{Knill1998}
E.~Knill and R.~Laflamme.
\newblock Power of one bit of quantum information.
\newblock {\em {Physical Review Letters}} {\bf 81}, 5672--5675, (1998).

\bibitem{Knill2001b}
E.~Knill, R.~Laflamme, and G.~J. Milburn.
\newblock A scheme for efficient quantum computation with linear optics.
\newblock {\em {Nature}} {\bf 409}, 46--52, (2001).

\bibitem{LivroKok}
P.~Kok and B.~W. Lovett.
\newblock {\em Introduction to Optical Quantum Information Processing}.
\newblock {Cambridge University Press, Cambridge}, (2010).

\bibitem{Kok2007}
P.~Kok, W.~J. Munro, K.~Nemoto, et~al.
\newblock Linear optical quantum computing with photonic qubits.
\newblock {\em {Reviews of Modern Physics}} {\bf 79}, 135--174, (2007).

\bibitem{Kraus2011}
B.~Kraus.
\newblock Compressed quantum simulation of the Ising model.
\newblock {\em {Physical Review Letters}} {\bf 107}, 250503, (2011).

\bibitem{Kraus2001}
B.~Kraus and J.~I. Cirac.
\newblock Optimal creation of entanglement using a two-qubit gate.
\newblock {\em {Physical Review A}} {\bf 63}, 062309, (2001).

\bibitem{Kwiat1995}
P.~G. Kwiat, K.~Mattle, H.~Weinfurter, et~al.
\newblock New High-Intensity Source of Polarization-Entangled Photon Pairs.
\newblock {\em {Physical Review Letters}} {\bf 75}, 4337--4341, (1995).

\bibitem{Laing2012}
A.~Laing and J.~L. O'Brien.
\newblock Super-stable tomography of any linear optical device.
\newblock {\em {arXiv:1208.2868 \ensuremath{[}quant-ph\ensuremath{]}}} (2012).

\bibitem{Leverrier2013}
A.~Leverrier and R.~Garc\'{\i}a-Patr\'{o}n.
\newblock Does Boson Sampling need fault-tolerance?
\newblock {\em {arXiv:1309.4687 \ensuremath{[}quant-ph\ensuremath{]}}} (2013).

\bibitem{Lidar2001}
D.~A. Lidar and L.-A. Wu.
\newblock Reducing constraints on quantum computer design by encoded selective
  recoupling.
\newblock {\em {Physical Review Letters}} {\bf 88}, 017905, (2001).

\bibitem{Lloyd1996}
S.~Lloyd.
\newblock Universal Quantum Simulators.
\newblock {\em {Science}} {\bf 273}, 1073--1078, (1996).

\bibitem{Lloyd2011}
S.~Lloyd, L.~Maccone, R.~Garcia-Patron, et~al.
\newblock Closed Timelike Curves via Postselection: Theory and Experimental
  Test of Consistency.
\newblock {\em {Physical Review Letters}} {\bf 106}, 040403, (2011).

\bibitem{Lund2013}
A.~P. Lund, A.~Laing, S.~Rahimi-Keshari, et~al.
\newblock Boson Sampling from Gaussian states.
\newblock {\em {arXiv:1305.4346 \ensuremath{[}quant-ph\ensuremath{]}}} (2013).

\bibitem{Ma2012}
X.-S. Ma, T.~Herbst, T.~Scheidl, et~al.
\newblock Quantum teleportation over 143 kilometres using active feed-forward.
\newblock {\em {Nature}} {\bf 489}, 269--273, (2012).

\bibitem{Makhlin2002}
Y.~Makhlin.
\newblock Nonlocal properties of two-qubit gates and mixed states and
  optimization of quantum computations.
\newblock {\em {Quantum Information Processing}} {\bf 1}, 243--252, (2002).

\bibitem{Martinlopez2012}
E.~Mart\'{\i}n-L\'{o}pez, A.~Laing, T.~Lawson, et~al.
\newblock Experimental realization of Shor's quantum factoring algorithm using
  qubit recycling.
\newblock {\em {Nature Photonics}} {\bf 6}, 773--776, (2012).

\bibitem{Meany2012}
T.~Meany, M.~Delanty, S.~Gross, et~al.
\newblock Non-classical interference in integrated 3D multiports.
\newblock {\em {Optics Express}} {\bf 20}, 26895, (2012).

\bibitem{Mezzadri2007}
F.~Mezzadri.
\newblock How to generate random matrices from the classical compact groups.
\newblock {\em {Notices of the AMS}} {\bf 54}, 592--604, (2007).

\bibitem{Mohseni2008}
M.~Mohseni, A.~T. Rezakhani, and D.~A. Lidar.
\newblock Quantum-process tomography: Resource analysis of different
  strategies.
\newblock {\em {Physical Review A}} {\bf 77}, 032322, (2008).

\bibitem{Moriya1960}
T.~Moriya.
\newblock Anisotropic superexchange interaction and weak ferromagnetism.
\newblock {\em {Physical Review}} {\bf 120}, 91--98, (1960).

\bibitem{Motes2013}
K.~R. Motes, J.~P. Dowling, and P.~P. Rohde.
\newblock Spontaneous parametric down-conversion photon sources are scalable in
  the asymptotic limit for boson sampling.
\newblock {\em {Physical Review A}} {\bf 88}, 063822, (2013).

\bibitem{Mozyrsky2001}
D.~Mozyrsky, V.~Privman, and M.~L. Glasser.
\newblock Indirect interaction of solid-state qubits via two-dimensional
  electron gas.
\newblock {\em {Physical Review Letters}} {\bf 86}, 5112--5115, (2001).

\bibitem{Nayak2008}
C.~Nayak, S.~H. Simon, A.~Stern, M.~Freedman, and S.~Das~Sarma.
\newblock Non-Abelian anyons and topological quantum computation.
\newblock {\em {Reviews of Modern Physics}} {\bf 80}, 1083--1159, (2008).

\bibitem{LivroNielsen}
M.~A. Nielsen and I.~L. Chuang.
\newblock {\em Quantum computation and quantum information}.
\newblock {Cambridge University Press, Cambridge}, (2000).

\bibitem{Nielsen2004}
M.~A. Nielsen.
\newblock Optical quantum computation using cluster states.
\newblock {\em {Physical Review Letters}} {\bf 93}, 040503, (2004).

\bibitem{LivroPapadimitriou}
C.~H. Papadimitriou.
\newblock {\em Computational Complexity}.
\newblock {Addison-Wesley}, (1994).

\bibitem{Peruzzo2011}
A.~Peruzzo, A.~Laing, A.~Politi, T.~Rudolph, and J.~L. O'Brien.
\newblock Multimode quantum interference of photons in multiport integrated
  devices.
\newblock {\em {Nature Communications}} {\bf 2}, 224, (2011).

\bibitem{Proctor2013}
T.~J. Proctor, E.~Andersson, and V.~Kendon.
\newblock Universal quantum computation by the unitary control of ancilla
  qubits and using a fixed ancilla-register interaction.
\newblock {\em {Physical Review A}} {\bf 88}, 042330, (2013).

\bibitem{Quiroga1999}
L.~Quiroga and N.~F. Johnson.
\newblock Entangled Bell and Greenberger-Horne-Zeilinger states of excitons in
  coupled quantum dots.
\newblock {\em {Physical Review Letters}} {\bf 83}, 2270--2273, (1999).

\bibitem{Rahimi-Keshari2013}
S.~Rahimi-Keshari, M.~A. Broome, R.~Fickler, et~al.
\newblock Direct characterization of linear-optical networks.
\newblock {\em {Optics Express}} {\bf 21}, 13450--13458, (2013).

\bibitem{Ramelow2010}
S.~Ramelow, A.~Fedrizzi, A.~M. Steinberg, and A.~G. White.
\newblock Matchgate quantum computing and non-local process analysis.
\newblock {\em {New Journal of Physics}} {\bf 12}, 083027, (2010).

\bibitem{Raussendorf2012}
R.~Raussendorf and T.-C. Wei.
\newblock Quantum computation by local measurement.
\newblock {\em {Annual Review of Condensed Matter Physics}} {\bf 3}, 239--261,
  (2012).

\bibitem{Raussendorf2001}
R.~Raussendorf and H.~J. Briegel.
\newblock A one-way quantum computer.
\newblock {\em {Physical Review Letters}} {\bf 86}, 5188--5191, (2001).

\bibitem{Reck1994}
M.~Reck, A.~Zeilinger, H.~J. Bernstein, and P.~Bertani.
\newblock Experimental realization of any discrete unitary operator.
\newblock {\em {Physical Review Letters}} {\bf 73}, 58--61, (1994).

\bibitem{Rohde2012}
P.~P. Rohde and T.~C. Ralph.
\newblock Error tolerance of the boson-sampling model for linear optics quantum
  computing.
\newblock {\em {Physical Review A}} {\bf 85}, 022332, (2012).

\bibitem{Rudolph2002}
T.~Rudolph and L.~Grover.
\newblock A 2 rebit gate universal for quantum computing.
\newblock {\em {arXiv:quant-ph/0210187}} (2002).

\bibitem{LivroSaleh}
B.~E.~A. Saleh and M.~C. Teich.
\newblock {\em Fundamentals of photonics}.
\newblock {Wiley-Interscience}, Hoboken, N.J., 2nd edition, (2007).

\bibitem{Sansoni2012}
L.~Sansoni, F.~Sciarrino, G.~Vallone, et~al.
\newblock Two-particle bosonic-fermionic quantum walk via integrated photonics.
\newblock {\em {Physical Review Letters}} {\bf 108}, 010502, (2012).

\bibitem{Sansoni2010}
L.~Sansoni, F.~Sciarrino, G.~Vallone, et~al.
\newblock Polarization Entangled State Measurement on a Chip.
\newblock {\em {Physical Review Letters}} {\bf 105}, 200503, (2010).

\bibitem{Scheel2004}
S.~Scheel.
\newblock Permanents in linear optical networks.
\newblock {\em {arXiv:quant-ph/0406127}} (2004).

\bibitem{Scheel2003}
S.~Scheel, K.~Nemoto, W.~Munro, and P.~Knight.
\newblock Measurement-induced nonlinearity in linear optics.
\newblock {\em {Physical Review A}} {\bf 68}, 032310, (2003).

\bibitem{Shchesnovich2013}
V.~Shchesnovich.
\newblock Sufficient bound on the mode mismatch of single photons for
  scalability of the boson sampling computer.
\newblock {\em {arXiv:1311.6796 \ensuremath{[}quant-ph\ensuremath{]}}} (2013).

\bibitem{Shen2014}
C.~Shen, Z.~Zhang, and L.-M. Duan.
\newblock Scalable Implementation of Boson Sampling with Trapped Ions.
\newblock {\em {Physical Review Letters}} {\bf 112}, 050504, (2014).

\bibitem{Shi2003}
Y.~Shi.
\newblock Both Toffoli and controlled-NOT need little help to do universal
  quantum computation.
\newblock {\em {Quantum Information and Computation}} {\bf 3}, 84--92, (2003).

\bibitem{Shor1996}
P.~Shor.
\newblock Fault-tolerant quantum computation.
\newblock In {\em {Proceedings of the IEEE Symposium on Foundations of Computer
  Science (FOCS)}}, pp.~56, (1996).

\bibitem{Shor1997}
P.~W. Shor.
\newblock Polynomial-time algorithms for prime factorization and discrete
  logarithms on a quantum computer.
\newblock {\em {SIAM Journal on Computing}} {\bf 26}, 1484--1509, (1997).

\bibitem{Shor1995}
P.~W. Shor.
\newblock Scheme for reducing decoherence in quantum computer memory.
\newblock {\em {Physical Review A}} {\bf 52}, R2493--R2496, (1995).

\bibitem{Simon1997}
D.~R. Simon.
\newblock On the Power of Quantum Computation.
\newblock {\em {SIAM Journal on Computing}} {\bf 26}, 1474--1483, (1997).

\bibitem{Simon1990}
R.~Simon and N.~Mukunda.
\newblock Minimal three-component SU(2) gadget for polarization optics.
\newblock {\em {Physics Letters A}} {\bf 143}, 165--169, (1990).

\bibitem{Spagnolo2013a}
N.~Spagnolo, C.~Vitelli, L.~Aparo, et~al.
\newblock Three-photon bosonic coalescence in an integrated tritter.
\newblock {\em {Nature Communications}} {\bf 4}, 1606, (2013).

\bibitem{Spagnolo2013c}
N.~Spagnolo, C.~Vitelli, M.~Bentivegna, et~al.
\newblock Efficient experimental validation of photonic BosonSampling against
  the uniform distribution.
\newblock {\em {arXiv:1311.1622 \ensuremath{[}physics,
  physics:quant-ph\ensuremath{]}}} (2013).

\bibitem{Spagnolo2013b}
N.~Spagnolo, C.~Vitelli, L.~Sansoni, et~al.
\newblock General rules for bosonic bunching in multimode interferometers.
\newblock {\em {Physical Review Letters}} {\bf 111}, 130503, (2013).

\bibitem{Spring2013}
J.~B. Spring, B.~J. Metcalf, P.~C. Humphreys, et~al.
\newblock Boson Sampling on a Photonic Chip.
\newblock {\em {Science}} {\bf 339}, 798--801, (2013).
\newblock PMID: 23258407.

\bibitem{Steane1996}
A.~Steane.
\newblock Multiple-Particle Interference and Quantum Error Correction.
\newblock {\em {Proceedings of the Royal Society of London. Series A:
  Mathematical, Physical and Engineering Sciences}} {\bf 452}, 2551--2577,
  (1996).

\bibitem{Svetlichny2011}
G.~Svetlichny.
\newblock Time Travel: Deutsch vs. Teleportation.
\newblock {\em {International Journal of Theoretical Physics}} {\bf 50},
  3903--3914, (2011).

\bibitem{Szameit2007}
A.~Szameit, F.~Dreisow, T.~Pertsch, S.~Nolte, and A.~T\"{u}nnermann.
\newblock Control of directional evanescent coupling in fs laser written
  waveguides.
\newblock {\em {Optics Express}} {\bf 15}, 1579--1587, (2007).

\bibitem{Temperley1961}
H.~Temperley and M.~E. Fisher.
\newblock Dimer problems in statistical mechanics\textemdash{}an exact result.
\newblock {\em {Philosophical Magazine}} {\bf 6}, 1061--1063, (1961).

\bibitem{Terhal2002}
B.~M. Terhal and D.~P. DiVincenzo.
\newblock Classical simulation of noninteracting-fermion quantum circuits.
\newblock {\em {Physical Review A}} {\bf 65}, 032325, (2002).

\bibitem{Terhal2004}
B.~M. Terhal and D.~P. DiVincenzo.
\newblock Adaptive quantum computation, constant depth quantum circuits and
  Arthur-Merlin games.
\newblock {\em {Quantum Information and Computation}} {\bf 4}, 134--145,
  (2004).

\bibitem{Tichy2011}
M.~C. Tichy.
\newblock {\em Entanglement and interference of identical particles}.
\newblock Phd thesis, {University of Freiburg}, (2011).

\bibitem{Tichy2010}
M.~C. Tichy, M.~Tiersch, F.~de~Melo, F.~Mintert, and A.~Buchleitner.
\newblock Zero-Transmission Law for Multiport Beam Splitters.
\newblock {\em {Physical Review Letters}} {\bf 104}, 220405, (2010).

\bibitem{Tillmann2013}
M.~Tillmann, B.~Dakic, R.~Heilmann, et~al.
\newblock Experimental boson sampling.
\newblock {\em {Nature Photonics}} {\bf 7}, 540--544, (2013).

\bibitem{Toda1989}
S.~Toda.
\newblock On the computational power of PP and \ensuremath{\oplus}P.
\newblock {\em {Proceedings of the IEEE Symposium on Foundations of Computer
  Science (FOCS)}} pp. 514--519, (1989).

\bibitem{Troyansky1996}
L.~Troyansky and N.~Tishby.
\newblock Permanent uncertainty: On the quantum evaluation of the determinant
  and the permanent of a matrix.
\newblock In {\em {Proceedings of PhysComp}}, (1996).

\bibitem{Vala2002}
J.~Vala and K.~B. Whaley.
\newblock Encoded universality for generalized anisotropic exchange
  Hamiltonians.
\newblock {\em {Physical Review A}} {\bf 66}, 022304, (2002).

\bibitem{Valiant1979}
L.~G. Valiant.
\newblock The complexity of computing the permanent.
\newblock {\em {Theoretical Computer Science}} {\bf 8}, 189--201, (1979).

\bibitem{Valiant2002}
L.~G. Valiant.
\newblock Quantum circuits that can be simulated classically in polynomial
  time.
\newblock {\em {SIAM Journal on Computing}} {\bf 31}, 1229--1254, (2002).

\bibitem{Valle2009}
G.~D. Valle, R.~Osellame, and P.~Laporta.
\newblock Micromachining of photonic devices by femtosecond laser pulses.
\newblock {\em {Journal of Optics A: Pure and Applied Optics}} {\bf 11},
  013001, (2009).

\bibitem{Dam2009}
W.~van Dam and M.~Howard.
\newblock Tight noise thresholds for quantum computation with perfect
  stabilizer operations.
\newblock {\em {Physical Review Letters}} {\bf 103}, 170504, (2009).

\bibitem{vanDam2001}
W.~van Dam, M.~Mosca, and U.~Vazirani.
\newblock How powerful is adiabatic quantum computation?
\newblock In {\em {42nd Annual IEEE Symposium on Foundations of Computer
  Science (FOCS)}}, pp. 279--287, (2001).

\bibitem{Nest2011a}
M.~van~den Nest.
\newblock Quantum matchgate computations and linear threshold gates.
\newblock {\em {Proceedings of the Royal Society of London. Series A:
  Mathematical, Physical and Engineering Sciences}} {\bf 467}, 821--840,
  (2011).

\bibitem{Nest2011b}
M.~van~den Nest.
\newblock Simulating quantum computers with probabilistic methods.
\newblock {\em {Quantum Information and Computation}} {\bf 11}, 784--812,
  (2011).

\bibitem{Nest2013}
M.~Van~den Nest.
\newblock Universal Quantum Computation with Little Entanglement.
\newblock {\em {Physical Review Letters}} {\bf 110}, 060504, (2013).

\bibitem{Wootters1989}
W.~K. Wootters and B.~D. Fields.
\newblock Optimal state-determination by mutually unbiased measurements.
\newblock {\em {Annals of Physics}} {\bf 191}, 363--381, (1989).

\bibitem{Yariv1973}
A.~Yariv.
\newblock Coupled-mode theory for guided-wave optics.
\newblock {\em {IEEE Journal of Quantum Electronics}} {\bf 9}, 919--933,
  (1973).

\bibitem{Yin2012}
J.~Yin, J.-G. Ren, H.~Lu, et~al.
\newblock Quantum teleportation and entanglement distribution over
  100-kilometre free-space channels.
\newblock {\em {Nature}} {\bf 488}, 185--188, (2012).

\bibitem{Yoran2003}
N.~Yoran and B.~Reznik.
\newblock Deterministic linear optics quantum computation with single photon
  qubits.
\newblock {\em {Physical Review Letters}} {\bf 91}, 037903, (2003).

\bibitem{Zanardi2000}
P.~Zanardi, C.~Zalka, and L.~Faoro.
\newblock Entangling power of quantum evolutions.
\newblock {\em {Physical Review A}} {\bf 62}, 030301, (2000).

\bibitem{Zhang2003}
J.~Zhang, J.~Vala, S.~Sastry, and K.~Whaley.
\newblock Geometric theory of nonlocal two-qubit operations.
\newblock {\em {Physical Review A}} {\bf 67}, 042313, (2003).

\bibitem{Zheng2000}
S.-B. Zheng and G.-C. Guo.
\newblock Efficient scheme for two-atom entanglement and quantum information
  processing in cavity QED.
\newblock {\em {Physical Review Letters}} {\bf 85}, 2392--2395, (2000).

\end{thebibliography}

\end{document}